\pdfoutput=1
\documentclass[11pt]{article}

\usepackage[utf8]{inputenc}
\usepackage[fontsize=11pt]{scrextend}
\usepackage[margin=1in]{geometry}
\usepackage[round]{natbib}
\usepackage{amsmath,amsthm,amssymb}
\IfFileExists{dsfont.sty}{\usepackage{dsfont}\newcommand{\ind}{\mathds{1}}}{\newcommand{\ind}{\mathbf{1}}}
\usepackage[pagewise]{lineno}
\usepackage{graphicx}
\usepackage{float}
\usepackage{microtype}
\usepackage{setspace}
\usepackage{tabularx}
\usepackage{longtable}
\usepackage{booktabs}
\usepackage{subcaption}
\usepackage{nicefrac}
\usepackage{etoolbox}
\usepackage{xcolor}
\usepackage{colortbl}
\usepackage{arydshln}

\makeatletter
\renewenvironment{proof}[1][\proofname]{\par
  \pushQED{\qed}%
  \normalfont
  \topsep6\p@\@plus6\p@\relax
  \trivlist
  \item[\hskip\labelsep
        \bfseries #1\@addpunct{.}]\ignorespaces
}{%
  \popQED\endtrivlist\@endpefalse
}
\makeatother

\definecolor{depAssumption}{HTML}{E5E7EB}
\definecolor{depLemma}{HTML}{BBF7D0}
\definecolor{depProposition}{HTML}{BFDBFE}
\definecolor{depTheorem}{HTML}{FCA5A5}
\definecolor{depCorollary}{HTML}{FED7AA}

\usepackage{tikz}
\usepackage{hyperref}
\hypersetup{
  hidelinks,
  pdfborder={0 0 0}
}
\usepackage{url}
\usepackage[english]{isodate}
\usepackage{enumitem}
\usepackage{titlesec}

\titleclass{\subsubsubsection}{straight}[\subsubsection]
\newcounter{subsubsubsection}[subsubsection]
\renewcommand{\thesubsubsubsection}{\thesubsubsection.\arabic{subsubsubsection}}

\titleformat{\subsubsubsection}
  {\normalfont\normalsize\bfseries}{\thesubsubsubsection}{1em}{}
\titlespacing*{\subsubsubsection}
  {0pt}{2ex plus .5ex minus .2ex}{1ex plus .2ex}

\usepackage{newpxtext,newpxmath}

\graphicspath{{./Images/}}

\newcommand{\cmmnt}[1]{}

\newcommand{\wM}{w_{\scriptscriptstyle\mathrm{M}}}

\theoremstyle{plain}
\newtheorem{theorem}{Theorem}
\newtheorem{proposition}{Proposition}
\newtheorem{lemma}{Lemma}
\newtheorem{corollary}{Corollary}[theorem]

\newtheorem{assumption}{Assumption}

\theoremstyle{plain}
\newtheorem{definition}{Definition}

\theoremstyle{definition}
\newtheorem{example}{Example}
\theoremstyle{plain}

\title{\textbf{Network Origins of the Money-Output Relation}}

\author{Vipin P. Veetil\thanks{Economics Area, Indian Institute of Management Kozhikode, Kerala 673 570, India.}}

\date{\small\printdayoff\today}
\begin{document}
\maketitle
\begin{abstract}
\setstretch{1.3}

This paper develops a dynamic production-network model in which monetary non-neutrality
emerges under fully flexible prices. Firms finance intermediate-input purchases
from current nominal balances, and intermediate goods arrive with heterogeneous lead times
across suppliers. A monetary shock then has two distinct effects. The first is
miscoordination: because production uses dated input bundles assembled under past nominal
conditions, the shock disturbs the composition of usable inputs (within and across firms) and thereby generates a
deadweight loss. The second is reallocation: because the incidence of the shock is tilted
toward smaller firms, which are disproportionately downstream, the shock redistributes
resources along the supply chain. More specifically, positive shocks push resources downstream, negative shocks
pull them upstream. This means that under a negative shock, both the miscoordination and the reallocation effect reduce downstream output and
therefore GDP. Under a positive shock, reallocation raises GDP while miscoordination reduces
it. The money--output relation is therefore
sign-asymmetric: negative shocks generate sharp contractions, while positive shocks generate at most mild expansions, and sometimes none. Computational experiments on a reconstructed US production network generate output declines after a monetary contraction that are about three times as large as the increases after a monetary expansion. 

\vspace{1cm}
\end{abstract}

\newpage
\setstretch{1.4}
\section{Introduction}
\label{sec:introduction}

The fundamental problem of monetary theory is the relation between money and output. In the
long run, the two appear to have little to do with each other. The number of bridges an
economy can build, the amount of bread it can produce, and the physicians it can train do
not depend on how many notes of a peculiar print happen to be in circulation. They
depend on technology and economic organization. Put differently, a larger circulation of certain
portraits may flatter imperial vanity but cannot increase national wealth. This much was
well understood by the classical economists and has remained, in one form or another, a
foundational proposition of monetary thought. The formidable problem, then as now, is the
short run. Negative monetary shocks tend to generate contractions in GDP, while positive
shocks generate expansions. How is it that instruments wholly powerless to alter national
wealth in the long run can nevertheless change output in the short run? This paper proposes
an answer.

Consider a supply-chain economy in which firms purchase intermediate inputs from one another, with upstream firms selling largely to other firms and downstream firms selling largely for final consumption.  Suppose that not all inputs can be combined contemporaneously:
some come like larks, others like slow freight. An automobile manufacturer, for instance,
may assemble a body frame ordered three months ago with tyres ordered last week and paint
ordered just yesterday, while a baker makes today's loaf with fresh yeast delivered this
morning and flour stocked weeks ago. Production, therefore, has a temporal structure.
Inputs flow through the network with heterogeneous lead times, so output at any date depends
on goods ordered under different past conditions. One of those conditions is the
distribution of nominal balances across firms. If a monetary shock changes all firms'
balances equiproportionally, relative positions are preserved, and there is no reason for
firms' real input orders to change. Once this equiproportionality is disturbed, however,
some firms can order more of certain intermediates, while others must order less. Some firms
expand, others contract. But these expansions and contractions do not offset one another,
for a reason at once subtle and simple. The input-bundles that become usable today were chosen
under different past nominal conditions and, after the shock, are no longer the ones firms
would have selected had relative balances remained aligned over time. With new money, the
baker can buy more yeast to make a lighter loaf but cannot bake many more loaves, for he
remains constrained by flour ordered weeks ago. Meanwhile, far away, a brewery now receives
less yeast and must cut production sharply. The baker's expansion is limited by commitments
inherited from an earlier monetary setting, while the brewer's contraction is severe because
plans laid in the past are now frustrated by the absence of a complementary input. What is gained
in one place therefore does not make up for what is lost in another.

These losses, however, are only half the story. Monetary shocks also shift
economic activity along the supply chain. Positive shocks push it downstream, while negative
shocks pull it upstream. This mechanism arises from the interaction of firm size with
the production network. Small firms tend to sit closer to the household, and small firms are
also more exposed to monetary shocks, which means that monetary shocks strike the supply chain
unevenly. A positive shock strengthens the purchasing power of firms close to final demand,
while a negative shock weakens it. Real resources at any date are fixed, so when downstream
firms bid more, fewer inputs remain for upstream use.
Production is thereby redistributed along the supply chain itself, with one part expanding as another contracts. The relevant rivalry is therefore not simply between the baker and the brewer,
but between the baker and the steel manufacturer, and among all the firms in between. Put differently, firms
affect one another not only through direct competition for the same inputs, but also through
competition for inputs used in the production of other inputs. The baker and the steel
manufacturer both draw directly on coal or electricity, but they are also linked more
indirectly. The baker depends on flour sacks, delivery trucks, and oven repairs. The steel
manufacturer depends on refractory bricks, machine parts, and industrial haulage. These are
different goods but their production relies on common inputs such as fuel, freight,
chemicals, and machinery. Nor is this redistribution costless. Each
firm keeps its workers in the short run, so the firms that win inputs
end up with more than their workers can use to best effect, while the
firms that lose inputs have too few for the workers they keep.

Our approach yields a rich and intricate set of dynamics, not only for GDP but for
production throughout the supply chain. To see those dynamics clearly, we must turn
to the formal structure of the model.  We
consider an economy with a finite number of firms and a representative household. The
household supplies a fixed quantity of labor and demands goods from firms. We impose an
additional structure on the production network that ties a firm's size to its location along
the supply chain: as one moves toward firms closer to final demand, the size distribution
shifts downward in the sense of first-order stochastic dominance. Firms
produce with a technology that aggregates intermediate inputs through CES and then combines
that compound with labor through a Cobb--Douglas nest. Intermediate inputs arrive with
delay. Each good is associated with a unique lead time drawn from an integer-valued
distribution with finite support. Prices are fully flexible, and firms set prices to clear
markets at each time step. Under standard regularity conditions on the production network, this setting admits a unique general equilibrium,
with nominal variables determined only up to a normalization of the price level, so that
equilibrium real allocations are independent of the aggregate nominal stock of money. We
also establish local stability of the stationary benchmark. 

Having thus constructed a network economy that remains stable despite dated input use, we
introduce monetary shocks as direct changes in firms' nominal balances.\footnote{Monetary disturbances that
take the form of direct, non-uniform changes in nominal balances have a long lineage. In
limited-participation and transaction-based models, this non-uniformity arises because new
money initially reaches only some agents or markets and only then diffuses through the
economy
\citep{grossman1983transaction,rotemberg1984monetary,lucas1990liquidity,fuerst1992liquidity,christiano1992liquidity}.
Spatial and network models make the same point of entry more explicit by letting injected
money enter particular locations or firms and then propagate through trading linkages
\citep{anthonisen2010monetary,mandel2018cantillon,mandel2021monetary}.}
 A positive shock
is an injection of cash into the firm sector, while a negative shock is an extraction. We
do not assume these changes are distributed uniformly across firms.\footnote{A large empirical literature shows that smaller firms are more exposed to monetary disturbances \citep{gertler1994,kashyap1993monetary,kashyap1994credit,bernanke1995credit}. There are many reasons for this. For one, large firms are better collateralized, depend less on banks and can borrow in the commercial-paper market \citep{gertler1994,bernanke1995credit}. And they have deeper internal funds. This matters because a firm with more capital of its own raises outside finance more easily, so that a tightening of credit falls hardest on poorly capitalized firms \citep{holmstrom1997financial}. Section~4.1 of \citet{mandel2021monetary} presents a brief summary.} We take this empirical regularity as
exogenous and parameterize the uneven incidence of monetary shocks through a concave
function of equilibrium money holdings. Because firms face cash-in-advance constraints, the shock's heterogeneous impact translates into heterogeneous
changes in input demand, with smaller firms responding more strongly. The concentration of
smaller firms downstream then gives those demand
shifts a directional character. Positive shocks enable downstream firms to bid resources away from upstream uses, while
negative shocks do the reverse. Resource reallocation along the supply chain is therefore
inseparable from the uneven incidence of the initial monetary disturbance. The same uneven
incidence is also the origin of miscoordination: miscoordination arises from
relative-price distortions, and such distortions require a non-uniform monetary impulse.
Were a cash injection to take the Humean form of an equiproportional change in all nominal
balances, firms would alter demand equiproportionally, and relative prices would remain
unchanged; neither reallocation nor miscoordination would arise. 

The aforenoted mathematical structure allows us to prove that a monetary shock, whether positive or negative, generates a deadweight
loss to the economy, with losses decreasing in the spectral gap of the production network
and in the substitutability of inputs. We show that, though the miscoordination effect pervades the whole economy,
the reallocation effect can dominate locally in some part of the supply-chain. More specifically, for an output aggregator
sufficiently biased toward one end of the supply-chain, the reallocative component can
outweigh miscoordination. It is precisely this possibility that allows GDP to rise in
response to a positive monetary shock, even though downstream firms too suffer from
miscoordination.

The temporal structure of production also shapes how a monetary disturbance is
felt over time. An economy built on longer and more roundabout input chains is
insulated for a while, since the disturbance takes time to reach usable
production; but that same depth, by pulling apart the fast and slow adjustments
the shock sets in motion, enlarges the cumulative loss it eventually inflicts.
Such an economy absorbs a monetary disturbance more slowly without thereby
absorbing it more cheaply.

Perhaps the most insightful of our formal results is that positive monetary shocks allow an
economy to borrow from the future. This may seem
somewhat odd prima facie, for such an inter-temporal reallocation of goods and services is
no part of the intention of monetary authorities. The unintended borrowing is a
consequence of the fact that today's upstream goods are tomorrow's downstream goods: going
up the supply chain is, in effect, going `back to the future'. A positive monetary shock
that raises GDP today does so by bidding resources away from upstream production, and hence
equivalently from the intermediate goods that are tomorrow's GDP. Society as a whole is
borrowing from itself. 

We measure the size of these effects with agent-based computational
experiments on the model calibrated to the United States production
network. The network carries the universe of firms in the United States,
aggregates to the input--output table, and embeds within it a
supply-chain structure. 
For a shock of the same size, the short-run decline in output after
a negative shock is about three times the largest rise after a positive one. It is worth noting that the sizeable asymmetry between output responses to positive and negative monetary shocks arises in a model with no asymmetry in the adjustment of prices or wages.

\subsection{Related literature}

Our paper is closely related to a growing literature that embeds production networks in
monetary models including \citet{bouakez2009}, \citet{nakamura2010monetary},
\citet{ozdagli2017production}, and \citet{pasten2019propagation}. Collectively, they argue
that input--output linkages, together with heterogeneity in sectoral size and price rigidity,
shape the distribution of monetary-policy effects across sectors, and that these effects do
not wash out in the aggregate.\footnote{One of the earliest, and admirably pioneering, contributions along these lines is \citet{anthonisen2018sticky}, who shows that with sticky prices in a
network economy, a monetary expansion can even leave output below its flexible-price
level at every location and date.}
 Put simply, production networks shape the macroeconomic consequences of
monetary shocks.\footnote{We should also note the nascent literature
that introduces networks into models of optimal monetary policy, more specifically 
\citet{lao2022optimal} and \citet{schaab2023hankio}. Here the question is not how networks
amplify the real effects of monetary policy, but how policy should be designed once one
recognizes the role of production networks in shaping real outcomes.}

Our contribution begins from a somewhat different perspective on monetary transmission. In
much of the literature just noted, networks aggravate nominal frictions and strengthen the
consequences of price stickiness. But the introduction of networks does not typically open a
new window through which to view monetary disturbances. The network remains a mere
amplification mechanism. In our setting, by contrast, the network is the wellspring of
monetary non-neutrality. This epistemic difference originates from the fact that, in our model, the
network is the structure through which monetary shocks disrupt the steady-state flow of
intermediate inputs.

What gives the network this role is a feature absent from standard monetary
models: the inputs entering a single production plan are delayed not uniformly
but by different amounts. Production at a date is therefore not a contemporaneous
object but a layered one, assembled from orders placed under different past
nominal conditions. A monetary shock consequently does more than move today's
spending; it disturbs the relation between today's spending and the input
commitments inherited from last week and last month. It is this layered time
profile---not the mere existence of delivery lags---that lets money be neutral in
the long-run equilibrium yet disruptive along the transition, as it scrambles the
vintages out of which current production is assembled.

The model used here builds on the foundational model of transient dynamics on production
networks proposed by \citet{gualdi2016emergence}. They study a  network economy
with CES technologies in which firms set prices and quantities through decentralized
interactions in segmented markets. We overlay dated inputs on that
decentralized network setting. The network thus acquires an explicit time dimension, and neither the existence nor the
local stability of equilibrium can be characterized without analyzing the interaction
between network propagation and input delays. More generally, the entire time profile of
how shocks---whether monetary or real---affect aggregate variables depends on the manner in
which input delays are embedded in the production network. We are not the first to introduce
delivery lags and time-to-build in intermediates into production networks.
\citet{taschereau2025echoes} and \citet{leng2025bullwhip} use such features to study the
macroeconomic effects of productivity and aggregate-demand shocks; both find that lags
amplify those effects. What distinguishes our paper from these  contributions is
the role played by delays in monetary transmission and how these delays interact with the
temporal structure of the delivery of intermediate inputs.

It is worth noting that this paper extends our previous work on the real consequences of
the flow of money through the production network. In \citet{mandel2018cantillon}, we showed
that the network percolation of a monetary shock can generate directionally anomalous price
movements, in the sense that some firms temporarily raise prices in response to monetary
contractions and lower them in response to expansions. In \citet{mandel2021monetary}, we
showed that the aggregate price level itself can move in the ``wrong'' direction, thereby
presenting a network-based resolution of the price puzzle. In \citet{veetil2026costinflation},
we showed that inflation can generate sizeable relative distortions even when the `size of
price changes' is nearly uncorrelated with the rate of inflation. Taken together, those papers
established that the transient passage of money through a production network can produce
rich, counterintuitive, and empirically consistent price dynamics. The present paper
turns to the corresponding output dynamics.

\subsection{Organization of the paper}

Section~\ref{sec:model} presents the model and its stationary
equilibrium. Section~\ref{sec:inefficiency} studies the production
inefficiency that a monetary shock causes within firms, and
Section~\ref{sec:reallocation} the reallocation of activity along the
supply chain. Section~\ref{sec:asymmetry} lets the two operate jointly,
measures the deadweight loss, made up of the loss from production inefficiency
and the loss from misallocation, and derives the asymmetric GDP response. Section~\ref{sec:quant} takes the model to a
reconstructed United States production network and to a family of
synthetic supply chains. Section~\ref{sec:conclusion}
concludes. Appendix~\ref{app:proofs} collects the proofs.
Appendix~\ref{app:analytical} states and proves the supplementary lemmas, and Appendix~\ref{app:construction} presents the calibration of the simulated supply chains and a glossary of notation.


\section{The Model}
\label{sec:model}

We build the model up from the individual firm. A firm produces with
labor and with intermediate inputs bought from other firms, and these
inputs arrive with lead times, so that the bundle a firm uses at a
date was ordered at several past dates. Production is thus tied both
to the network of who buys from whom and to the calendar of when each
order was placed. A representative household supplies labor and buys
from some of the firms. The economy rests at a stationary equilibrium
until a one-time monetary shock lands unevenly on the balances
of firms and sets off a decentralized adjustment. Along that
adjustment prices clear every market at every date. But firms keep
spending their balances along the supplier shares they had before the
shock.

\subsection{Economic environment}
\label{subsec:environment}

There are \(n\) firms, indexed by \(N:=\{1,\dots,n\}\), and one
representative household, indexed by \(h\).\footnote{Vectors are bold
lowercase and matrices bold uppercase letters, \(\mathbf 1\) is the
all-ones vector and \(\mathbf I\) the identity matrix. For a
firm-level variable \(z\), \(z_{i,t}\) is its value at firm
\(i\in N\) and date \(t\in\mathbb Z\), and a star marks its value at
the stationary equilibrium, \(z_i^\ast\).}

\subsubsection{Production technology}
\label{subsubsec:technology}
\label{subsubsec:lags}

Firm \(i\) combines labor and a composite of intermediate inputs in a
Cobb--Douglas top nest,
\begin{equation}
q_i
=
l_i^{\,\beta_i}X_i^{\,1-\beta_i},
\qquad i\in N
\label{eq:top_cd}
\end{equation}
where \(q_i\) is output, \(l_i\) is labor input, \(X_i\) the intermediate-input composite
and \(\beta_i\in(0,1)\) the labor share of firm \(i\). The labor share is a time-invariant
parameter of the firm's technology. Cost minimization makes \(\beta_i\) also
the share of the firm's revenue paid out as wages at the stationary
equilibrium. The labor shares are bounded away from \(0\) and \(1\),
\begin{equation}
\beta_i\in[\underline\beta,\overline\beta]\subset(0,1)
\qquad\text{for every }i\in N
\label{eq:beta_bounds}
\end{equation}
for fixed bounds \(\underline\beta\) and \(\overline\beta\).

Labor is immediately usable in production, while intermediate goods
arrive with lead times. Let \(\mathcal S_i\subseteq N\) be the set of firm \(i\)'s
suppliers and \(x_{ji}\) the quantity of good \(j\) used by firm
\(i\). Each purchase that firm \(i\) makes from a supplier \(j\) is
associated with a single lead time
\(k_{ji}\in\{\underline k,\underline k+1,\dots,\widehat{k}\}\). The lead
time \(k_{ji}\) is the number of periods that pass between the placing of an order for good
\(j\) by firm \(i\) and its use in firm \(i\)'s production. Here
\(\underline k\ge0\) is the minimum and \(\widehat k\) the maximum lead
time in the economy. The lead time is a property of the purchase, so the same good may reach one of its buyers sooner than it reaches another. The
inputs of a firm that share a lead time form a \emph{vintage layer} of
its bundle,
\begin{equation}
\mathcal S_i^{(k)}:=\{\,j\in\mathcal S_i:\;k_{ji}=k\,\},
\qquad k=\underline k,\dots,\widehat k
\label{eq:vintage_layer}
\end{equation}
and the lead times at which the firm has at least one supplier are its
\emph{active} vintages,
\(\mathcal K_i:=\{k:\mathcal S_i^{(k)}\neq\varnothing\}\). Every firm
buys at least one input at the minimum lead time of the economy,
\begin{equation}
\underline k\in\mathcal K_i
\qquad
\text{for every } i\in N
\label{eq:zero_lag_share}
\end{equation}
so that production first reacts to a shock at date \(\underline k\).

The technology nests the inputs of a firm by vintage. Inputs of a
common vintage are aggregated by an inner CES with curvature \(\rho\)
and supplier weights \(\nu_{ji}>0\),
\begin{equation}
X_i^{(k)}
:=
\Bigl(\sum_{j\in\mathcal S_i^{(k)}}\nu_{ji}\,x_{ji}^{\rho}\Bigr)^{1/\rho},
\qquad k\in\mathcal K_i
\label{eq:within_vintage_bundle}
\end{equation}
and the active vintage composites are then aggregated by an outer CES
with curvature \(\rho_c\) and vintage weights \(\chi_i^{(k)}\),
\begin{equation}
X_i
=
\Bigl(\sum_{k\in\mathcal K_i}\chi_i^{(k)}\,\bigl(X_i^{(k)}\bigr)^{\rho_c}\Bigr)^{1/\rho_c},
\qquad
\chi_i^{(k)}>0\ \text{ for }k\in\mathcal K_i,\quad\sum_{k\in\mathcal K_i}\chi_i^{(k)}=1
\label{eq:materials_bundle}
\end{equation}
The supplier weights \(\nu_{ji}\) and the vintage weights
\(\chi_i^{(k)}\) are parameters of firm \(i\)'s
technology.\footnote{We set \(\chi_i^{(k)}:=0\) for an inactive
vintage, so that a sum over \(k=\underline k,\dots,\widehat k\) is a
sum over \(\mathcal K_i\).}

The inner curvature \(\rho\) governs substitution among
contemporaneously ordered, same-vintage inputs. The outer curvature
\(\rho_c\) governs substitution across vintages, between vintage
composites assembled under different past conditions. In the short
run the intermediate inputs of a firm are, on most of the evidence, complements
(\citealp{barrot2016input,atalay2017sectoral,boehm2019input,cevallosfujiy2024production};
Table~\ref{tab:elasticities}). Both curvatures are accordingly
negative. And inputs of different
vintages are poorer substitutes than inputs of the same vintage. That
is, a firm can more readily replace one contemporaneously available supplier with another than substitute a freshly delivered input for one
ordered under earlier conditions. We therefore restrict the two
curvatures to
\begin{equation}
-\overline\rho\;\le\;\rho_c\;\le\;\rho\;\le\;-\underline\rho,
\qquad
0<\underline\rho<\overline\rho<\infty
\label{eq:rho_domain}
\end{equation}
and call \(\rho-\rho_c\ge0\) the \emph{curvature gap}.\footnote{When
the two curvatures coincide, \(\rho_c=\rho\), the composite collapses
to the flat CES
\(X_i=(\sum_{j\in\mathcal S_i}\chi_i^{(k_{ji})}\nu_{ji}\,x_{ji}^{\rho})^{1/\rho}\),
in which the vintage partition survives only through the vintage
weights \(\chi_i^{(k_{ji})}\). We refer to \(\rho_c=\rho\) as the flat
case. The flat case removes the difference between the within-vintage
and the cross-vintage penalty on a mismatched bundle, though not the
dating of orders, which the lead-time profile carries.}

\subsubsection{Production network and supplier shares}
\label{subsubsec:network}

The supplier sets \(\{\mathcal S_i\}_{i\in N}\) define a directed
production network, in which a link runs from firm \(j\) to firm \(i\)
when firm \(i\) uses the good of firm \(j\) as an intermediate input,
\(j\in\mathcal S_i\). Given a price vector
\(\mathbf p=(p_j)_{j\in N}\), cost minimization decides how firm \(i\)
divides its spending on intermediate inputs among its suppliers, and we
write
\begin{equation}
a_{ji}
:=
\frac{p_jx_{ji}}{\sum_{u\in \mathcal S_i}p_ux_{ui}},
\qquad
a_{ji}\ge0,
\qquad
\sum_{j\in \mathcal S_i}a_{ji}=1
\label{eq:supplier_share_def}
\end{equation}
for the share paid to supplier \(j\). The technology nests the inputs
by vintage, so this division runs in two steps. Firm \(i\) divides its
input spending first across its vintage layers and then, within each
vintage layer, across the suppliers that deliver at that lead time. At the
stationary equilibrium
each supplier share is accordingly the product of the cross-vintage
share \(s_i^{(k)}\) of the vintage layer and the within-vintage share
\(\widetilde a_{ji}^{(k)}\) of the supplier,
\begin{equation}
a_{ji}
=
s_i^{(k_{ji})}\,\widetilde a_{ji}^{(k_{ji})},
\qquad
\widetilde a_{ji}^{(k)}
=
\frac{\nu_{ji}\,x_{ji}^{\rho}}{\bigl(X_i^{(k)}\bigr)^{\rho}},
\qquad
s_i^{(k)}
=
\frac{\chi_i^{(k)}\bigl(X_i^{(k)}\bigr)^{\rho_c}}{X_i^{\rho_c}}
\label{eq:share_quantity_form}
\end{equation}
The within-vintage shares sum to one inside each vintage layer, so the cross-vintage share \(s_i^{(k)}\) of the outer nest is also the share of firm \(i\)'s input spending that arrives with lead time \(k\).

The supplier shares form the \(n\times n\) supplier-share matrix
\[
\mathbf A:=(a_{ji})_{j,i\in N}
\]
with \(a_{ji}=0\) whenever \(j\notin\mathcal S_i\). The support of
\(\mathbf A\) is fixed by the supplier sets and its nonzero entries by
cost minimization at the stationary equilibrium. Each column of
\(\mathbf A\) records how one firm divides its spending on inputs, so
it sums to one.

\subsubsection{Household preferences}
\label{subsubsec:household}

The representative household supplies one unit of labor inelastically
at the nominal wage \(w\) and consumes the differentiated goods
produced by firms. Its preferences are Cobb--Douglas,
\begin{equation}
U(c_1,\dots,c_n)
=
\prod_{i\in N} c_i^{\,\gamma_i},
\qquad
\gamma_i \geq 0,
\qquad
\sum_{i\in N}\gamma_i=1
\label{eq:hh_pref}
\end{equation}
where \(c_i\) is its consumption of good \(i\). The wage is the
household's only income, so at any prices it spends the fixed share
\(\gamma_i\) of the wage on good \(i\). The household need not buy from
every firm, and
\(\operatorname{supp}\boldsymbol\gamma:=\{i\in N:\gamma_i>0\}\) is the
set of firms it buys from.

Every firm's output nonetheless reaches the household through some
finite chain of sales. That is, for every firm \(j\) there are firms
\(i_1,\dots,i_r\) with \(j\in\mathcal S_{i_1}\),
\(i_1\in\mathcal S_{i_2}\), \dots, and
\(i_r\in\operatorname{supp}\boldsymbol\gamma\). We maintain this
reachability of final demand throughout. It is the household's
spending, passed up these chains of buyers, that pays for every firm's
sales. Reachability therefore makes every firm's stationary sales
strictly positive.

\subsection{The stationary competitive equilibrium}
\label{subsec:equilibrium}

With the technology, the network and the household in place, we now
turn to the stationary equilibrium from which the monetary shock
departs and against which we measure its effects. A stationary
competitive equilibrium consists of goods prices
\(\mathbf p=(p_i)_{i\in N}\), outputs \(\mathbf q=(q_i)_{i\in N}\),
intermediate-input orders \(\mathbf x=(x_{ji})_{j,i\in N}\), labor
inputs \(\mathbf l=(l_i)_{i\in N}\) and a nominal wage \(w\), all
constant across dates. At these values firms minimize cost and price at
unit cost, the household maximizes utility, and the markets for goods
and labor clear.

\begin{proposition}[Existence and uniqueness of the stationary equilibrium]
\label{prop:existence_stationary}
Suppose that every firm has a non-empty supplier set, that the labor
shares satisfy \eqref{eq:beta_bounds} and the curvatures
\eqref{eq:rho_domain}, and that final demand is reachable from every
firm. Then a stationary competitive equilibrium exists and is unique up
to the choice of the nominal wage.
\end{proposition}

\noindent
Proof in Appendix~\ref{app:proof_existence}.

The lead times enter the stationary equilibrium only through the
way the technology nests inputs by vintage. The dating of orders itself does not affect the stationary equilibrium. When every order is the same at every date, an input
ordered several periods ago is identical to one ordered today, and the
production identities \eqref{eq:top_cd}--\eqref{eq:materials_bundle}
hold at every date as if all inputs arrived at the same time. The stationary
equilibrium is also efficient among stationary allocations. This is
the First Welfare Theorem applied date by date. Markets clear
competitively, production has constant returns, and there are no
externalities, taxes or market power.

At the stationary equilibrium a firm's sales in one period are the balance it spends in the next period. The vector
\(\mathbf m^\ast=(m_i^\ast)_{i\in N}\) of stationary nominal sales,
\(m_i^\ast=p_i^\ast q_i^\ast\), is therefore also the vector of balances that firms
hold at the start of each period, and
\(\overline M:=\mathbf 1^\top\mathbf m^\ast\) is the money stock, the aggregate stock of
firm balances. A firm spends the fraction \(1-\beta_i\) of its revenue
on intermediate inputs and the fraction \(\beta_i\) on labor, so it
pays the fraction \((1-\beta_i)a_{ji}\) of its revenue to supplier
\(j\). We collect these firm-to-firm expenditure shares in
\begin{equation}
\mathbf A_\beta:=\mathbf A\,\mathrm{diag}(\mathbf 1-\boldsymbol\beta)
\label{eq:Ahat_def}
\end{equation}
where \(\boldsymbol\beta:=(\beta_i)_{i\in N}\). A firm's sales are what
the household spends on its good plus what its buyers spend on it,
\begin{equation}
m_i
=
\gamma_i w
+
\sum_{j\in N}(1-\beta_j)\,a_{ij}\,m_j,
\qquad
\mathbf m
=
w\boldsymbol\gamma
+
\mathbf A_\beta\mathbf m
\label{eq:firm_sales_accounting}
\end{equation}
Every unit of final spending is thus passed up the supply chain, each
firm passing on the fraction \(1-\beta_i\) of what it receives and
paying the rest to its workers. The Leontief inverse
\((\mathbf I-\mathbf A_\beta)^{-1}\) sums these successive rounds of spending.

\subsection{Monetary shock}
\label{subsec:money}

The only exogenous disturbance in the model is a one-time monetary
shock to firms' nominal balances. The economy is at the stationary
equilibrium and is hit at date \(t=0\) by a shock of total size
\(\pi\overline M\), the fraction \(\pi\) of the money stock.
The shock is a monetary expansion for \(\pi>0\) and a monetary contraction for
\(\pi<0\). We study
small shocks, \(|\pi|\le\overline\pi\), where the bound
\(\overline\pi>0\) is set by the constants of the
assumptions.\footnote{The constants of the assumptions are the fixed
constants in which the assumptions and the bounds
\eqref{eq:beta_bounds}, \eqref{eq:zero_lag_share},
\eqref{eq:rho_domain} and \eqref{eq:log_size_bound} are stated. A constant set by them takes the
same value in every economy that satisfies the assumptions, whatever
its number of firms. The bound \(\overline\pi\) keeps every balance
positive after the shock and keeps the economy in the neighborhood of
the stationary equilibrium on which the local approximations of the
paper apply. Every ``for \(|\pi|\) sufficiently small'' in the paper
means \(|\pi|\le\overline\pi\) unless a statement gives its own
radius. Appendix~\ref{app:proofs} lists the constants of the
assumptions and the conditions that define \(\overline\pi\).}

We measure the effects of the shock by how far it moves each variable
from its stationary value. For a variable \(z_{i,t}\) with
\(z_i^\ast>0\), the absolute and the proportional deviation are
\[
\Delta z_{i,t}\;:=\;z_{i,t}-z_i^\ast,
\qquad
\widehat z_{i,t}\;:=\;\frac{\Delta z_{i,t}}{z_i^\ast}
\]
A hat on a variable marks its proportional form, and a variable without a hat is in its absolute form. A
variable also depends on the size \(\pi\) of the shock. Its first-order
response per unit of shock carries a dot, and its curvature response
carries two dots,
\[
\dot z_{i,t}\;:=\;\left.\frac{\partial z_{i,t}(\pi)}{\partial\pi}\right|_{\pi=0},
\qquad
\ddot z_{i,t}\;:=\;\left.\frac{\partial^2 z_{i,t}(\pi)}{\partial\pi^2}\right|_{\pi=0}
\]
and both inherit the units of \(z\).\footnote{The stationary value does
not depend on \(\pi\), so dots and hats commute. We write the
composed objects as \(\dot{\widehat z}_{i,t}\) and
\(\ddot{\widehat z}_{i,t}\). Quantities introduced as proportional
deviations from the outset carry the hat in their names. All
conventions extend componentwise to vector- and matrix-valued objects.
An order symbol refers to the limit of a vanishing shock,
\(\pi\to0\), with the economy held fixed, and carries no uniformity
over dates, firms or the number of firms unless that is stated.}

The shock does not land on firms in proportion to their balances.
Smaller firms are more exposed to monetary shocks. We capture
this exposure by letting firm \(i\) receive the share \(\zeta_i\) of
the shock,
\(\Delta m_{i,0}=\pi\overline M\zeta_i\), where
\begin{equation}
\zeta_i
:=
\frac{(m_i^\ast)^{\theta}}{\sum_{j\in N}(m_j^\ast)^{\theta}},
\qquad
\theta\in(0,1)
\label{eq:zeta_def}
\end{equation}
The map \(x\mapsto x^\theta\) is concave for \(\theta\in(0,1)\), so the
incidence vector \(\boldsymbol\zeta=(\zeta_i)_{i\in N}\) is less
concentrated than the incidence vector of the proportional benchmark, obtained at \(\theta=1\).
Smaller firms therefore receive a larger transfer relative to their
stationary balances than under proportional incidence, which is to say
a relatively larger injection when \(\pi>0\) and a relatively larger
withdrawal when \(\pi<0\). Since \(\boldsymbol\zeta\) does not depend
on \(\pi\), the pattern of incidence across firms is the same for
shocks of every size. We characterize how unevenly the shock lands
with Assumption~\ref{assump:diffuse_incidence}.

\begin{assumption}[Diffuse incidence]
\label{assump:diffuse_incidence}
The shock lands unevenly, more heavily on smaller firms relative
to their balances, and by a bounded degree. That is, the concavity parameter
\(\theta\in(0,1)\) of the incidence rule \eqref{eq:zeta_def} satisfies
\begin{equation}
1-\theta
\;\le\;
\overline\epsilon_\theta
\;\le\;
\tfrac12
\label{eq:diffuse_incidence_bound}
\end{equation}
for the fixed constant \(\overline\epsilon_\theta\).
\end{assumption}

At the proportional benchmark \(\theta=1\) the shock rescales every
balance in the same proportion,
\(\overline M\boldsymbol\zeta(1)=\mathbf m^\ast\), and money is
neutral. Money is non-neutral in the model because it lands unevenly,
not because its quantity changes.\footnote{The cap
\(\overline\epsilon_\theta\) in \eqref{eq:diffuse_incidence_bound}
holds \(1-\theta\) away from \(1\). Together with small \(|\pi|\), it
keeps the joint parameter \(|\pi|(1-\theta)\) small enough for the
marginal analysis.}

An uneven shock misaligns balances. It leaves some firms with more than
their stationary share of the money stock and others with less. How unevenly the shock misaligns balances depends on the slack \(1-\theta\) and on how
dispersed the sizes of firms are. Write
\(\ell_i:=\log m_i^\ast-\sum_{j\in N}(m_j^\ast/\overline M)\log m_j^\ast\)
for firm \(i\)'s centered log size, and take the range of firm sizes to
be bounded by a fixed constant,
\begin{equation}
\max_{i\in N}|\ell_i|\;\le\;L_\ell
\label{eq:log_size_bound}
\end{equation}
The incidence ratio of firm \(i\), the share of the shock it receives
relative to its share of the money stock, is
\(\overline M\zeta_i/m_i^\ast=e^{-(1-\theta)\ell_i}/Z\), with \(Z:=\sum_{j\in N}(m_j^\ast/\overline M)\,e^{-(1-\theta)\ell_j}\). The incidence ratio falls with the size of the firm. The misalignment that the shock
creates at firm \(i\), per unit of the shock and relative to the
firm's balance, is the incidence ratio less one,
\begin{equation}
\iota_i
:=
\frac{\overline M\zeta_i}{m_i^\ast}-1
\label{eq:incidence_ratio}
\end{equation}
The misalignment \(\iota_i\) is positive at the firms whose share of the shock exceeds their share of the money stock, and negative at the others. To leading order it is the slack times the
firm's centered log size, with the sign reversed,
\begin{equation}
\iota_i
\;=\;
-(1-\theta)\,\ell_i+O\bigl((1-\theta)^2\bigr)
\label{eq:incidence_slack_expansion}
\end{equation}
uniformly over firms.\footnote{Every ``to leading order in
\(1-\theta\)'' in the paper refers to the expansion
\eqref{eq:incidence_slack_expansion}, whose order symbol refers to the
limit \(\theta\to1\) around the proportional benchmark, with the
economy held fixed.} We call
\begin{equation}
C_\zeta:=\frac{\max_{i\in N}|\iota_i|}{1-\theta},
\qquad\text{so that}\qquad
|\iota_i|\;\le\;C_\zeta\,(1-\theta)
\quad\text{for every }i\in N
\label{eq:impact_misalignment_constant}
\end{equation}
the \emph{incidence range} of the economy, per unit of slack. The incidence range measures how unevenly the shock lands, not how large the shock is.\footnote{The log-size bound \eqref{eq:log_size_bound} imposes no particular shape on the size distribution. Every finite economy
satisfies it at the constant \(L_\ell\) read off its own sizes,
however heavy-tailed. A wider range of firm sizes raises the incidence
range \(C_\zeta\), and with it the constants of the results.}

\subsection{Decentralized adjustment after a monetary shock}
\label{subsec:stability}
\label{subsec:temporal_depth_formal}

Once the shock has misaligned balances, no one coordinates their
return to the stationary distribution. Every firm spends the balance
it holds, and every market clears at the price that equates the
spending on a good to the output available. What a firm receives in one period is the balance it spends in the next period. Over the horizon we study, firms
do not revise their supplier relationships, so the supplier-share
matrix \(\mathbf A\) stays fixed. We characterize the labor market
along the adjustment with
Assumption~\ref{assump:short_run_nominal_rigidity}.

\begin{assumption}[Labor-market protocol]
\label{assump:short_run_nominal_rigidity}
For every firm \(i\in N\) and every date \(t\ge 0\),
\textnormal{(a)} the labor input is held at its stationary value,
\(l_{i,t}=l_i^\ast\), and \textnormal{(b)} the nominal wage is indexed to the money stock,
\begin{equation}
\wM \;=\; \frac{w^\ast}{\overline M}\,\mathbf 1^\top\mathbf m_t
\label{eq:nominal_wage_indexation}
\end{equation}
\end{assumption}

The shock at \(t=0\) is the only change to the money stock. The money stock rises once, from \(\overline M\) to \((1+\pi)\overline M\), and
is then conserved as balances circulate among firms. The indexed wage
is therefore \(\wM=(1+\pi)w^\ast\), constant along the post-shock
path. The two parts of
Assumption~\ref{assump:short_run_nominal_rigidity} together fix each
firm's wage bill after the shock at
\(\wM l_i^\ast=(1+\pi)\beta_im_i^\ast\). The whole of the adjustment of firms to the shock therefore falls on their intermediate inputs, the inputs that carry the vintage structure.

Each period unfolds in the same sequence, which
Appendix~\ref{app:within_period} records in full. Firms and the
household first express nominal demand. At date \(t\) firm \(i\) holds
the balance \(m_{i,t}\), hires its stationary labor \(l_i^\ast\) at the
indexed wage \(\wM\), and pays for its intermediate purchases out of
the balance that remains, a cash-in-advance constraint. Since the wage
bill is fixed, every change in a firm's balance falls on its
intermediate-input expenditure \(e_{i,t}:=m_{i,t}-\wM l_i^\ast\). The firm
divides that expenditure among its suppliers at the stationary shares,
so its real order on supplier \(j\) is
\(x_{ji,t}=a_{ji}\bigl(m_{i,t}-\wM l_i^\ast\bigr)/p_{j,t}\), equation
\eqref{eq:dynamic_real_alloc}. The household spends its wage on final
goods according to \(\boldsymbol\gamma\), and the nominal demand
\(d_{j,t}\) for good \(j\) is what its buyers and the household spend
on it, equation \eqref{eq:dyn_nominal_demand}. Firms then set prices to
clear goods markets given the output available at date \(t\),
\(p_{j,t}=d_{j,t}/q_{j,t}\), equation \eqref{eq:dyn_price_update}.
Production takes place last, through the technology
\eqref{eq:top_cd}--\eqref{eq:materials_bundle} with the stationary
labor and with each input evaluated at its order date \(t-k_{ji}\),
equation \eqref{eq:dynamic_output}.

These within-period actions imply the law of motion for nominal
balances. Since firms receive the nominal demand directed toward their
output as revenue, next period's balances equal current nominal sales,
\begin{equation}
\mathbf m_{t+1}
\;=\;
\mathbf A\bigl(\mathbf m_t-\wM\mathbf l^\ast\bigr)+\boldsymbol\gamma \wM
\label{eq:firm_balance_law_pre}
\end{equation}
With no further shocks after \(t=0\), the money stock is
conserved at the post-shock level,
\(\mathbf 1^\top\mathbf m_{t+1}=\mathbf 1^\top\mathbf m_t=(1+\pi)\overline M\)
for \(t\ge0\).

What moves in \eqref{eq:firm_balance_law_pre} is the distribution of
money across firms. The wage bill and the household's spending have
both risen in proportion to the money stock. A firm that holds more
than its stationary share of the money stock therefore spends the whole
of its excess on intermediate inputs, and passes it on to its suppliers
in proportion to its supplier shares. A firm that holds less passes on
its shortfall in the same way. Which is why a misalignment of balances is passed on through \(\mathbf A\) without loss, while a unit of final spending
loses the labor share at every stage of \(\mathbf A_\beta\). Per unit
of the shock, the part of the disturbance that moves balances from some
firms to others is the \emph{cross-sectional balance misalignment}
\begin{equation}
\dot{\mathbf m}_t^\perp
\;:=\;
\dot{\mathbf m}_t-\mathbf m^\ast
\label{eq:m_perp_def}
\end{equation}
The misalignment follows the exact law
\(\dot{\mathbf m}^\perp_{t+1}=\mathbf A\dot{\mathbf m}^\perp_t\). At
impact it equals the impact misalignment
\(\dot{\mathbf m}_0^\perp=\overline M\boldsymbol\zeta-\mathbf m^\ast\),
the difference between the distribution of the shock across firms and the stationary distribution of balances.

We now derive the speed at which a misalignment of balances dies out. What some
firms hold in excess others lack, so a misalignment \(\mathbf v\) always sums to zero. The columns of
\(\mathbf A\) sum to one, and a round of purchasing therefore
preserves the sum,
\(\mathbf 1^\top\mathbf A\mathbf v=\mathbf 1^\top\mathbf v=0\). The
matrix \(\mathbf A\) thus maps the zero-sum vectors of the
conservation subspace
\(\mathcal Z_{\mathbf A}:=\{\mathbf v:\mathbf 1^\top\mathbf v=0\}\)
into itself. The eigenvalue one of \(\mathbf A\) records only that the money stock, the total of balances, is conserved. The speed at which a misalignment dies out is set by the
eigenvalues of \(\mathbf A\) on \(\mathcal Z_{\mathbf A}\). We write
\(\lambda_2(\mathbf A)\) for the largest modulus among them, the
spectral radius of \(\mathbf A\) on \(\mathcal Z_{\mathbf A}\), and
call it the \emph{nominal redistribution rate}. After \(t\) rounds of
purchasing a misalignment has shrunk by a factor of the order of
\(\lambda_2(\mathbf A)^t\).

A given sum of money is a large misalignment for a small firm and a
small one for a large firm. We therefore measure a misalignment firm by
firm, relative to the firm's stationary balance. With
\(\mathbf D_m:=\mathrm{diag}(\mathbf m^\ast)\) the diagonal matrix of
stationary balances,
\[
\|\mathbf v\|_{m}\;:=\;\max_{i\in N}\frac{|v_i|}{m_i^\ast}\;=\;\|\mathbf D_m^{-1}\mathbf v\|_\infty
\]
is the largest proportional misalignment at any firm. We characterize
the speed at which money is redistributed with
Assumption~\ref{assump:monotone_decay_spectrum}.

\begin{assumption}[Nominal redistribution rate]
\label{assump:monotone_decay_spectrum}
For the fixed constants \(\overline\lambda\in(0,1)\) and
\(C_\lambda\ge1\), \textnormal{(a)} the nominal redistribution rate is
bounded away from one,
\[
\lambda_2(\mathbf A)\;\le\;\overline\lambda
\]
and \textnormal{(b)} after \(t\) rounds of purchasing a misalignment is
at most \(C_\lambda\lambda_2(\mathbf A)^t\) times its initial size,
\begin{equation}
\|\mathbf A^t\mathbf v\|_{m}\;\le\;C_\lambda\,\lambda_2(\mathbf A)^t\,\|\mathbf v\|_{m}
\qquad\text{for every }\mathbf v\in\mathcal Z_{\mathbf A}\text{ and every }t\ge0
\label{eq:power_bound}
\end{equation}
\end{assumption}

Part~(a) is natural in a supply chain whose firms are arranged in tiers, each buying from the tiers above it. There the slowest
pattern of misalignment relaxes steadily, without rotating or
alternating in sign. The content of part~(b) is that the
multiple \(C_\lambda\) does not depend on the number of firms. In such a supply chain a misalignment placed at the retail end travels upstream one tier per
period without decaying until it reaches the top of the supply chain. The
constant \(C_\lambda\) must therefore allow for as many rounds without decay as the supply chain has tiers.\footnote{The assumption is imposed on the supplier-share
matrix \(\mathbf A\) of \eqref{eq:share_quantity_form}, which is
column-stochastic with the support of the supplier sets whatever the
two curvatures. Part~(b) holds for a single economy with some constant
whenever the eigenvalues of largest modulus on
\(\mathcal Z_{\mathbf A}\) are semisimple. Part~(a) also makes the
eigenvalue one of \(\mathbf A\) simple and every other eigenvalue
smaller than one in modulus. The matrix \(\mathbf A\) then has
exactly one closed sourcing class, a set of firms that buy only from
one another, and that class is aperiodic. In such a supply chain it is the top of the supply chain, the producers of raw inputs who buy from one another, and a misalignment of balances passed up the supply chain ends up circulating among them.}

Balances are not the only link between one date and the next. Orders placed at different dates are in transit at the same time. Let \(x_{ji,s}\) be the
quantity that buyer \(i\) orders from supplier \(j\) at date \(s\). The
order becomes usable at date \(s+k_{ji}\), and at every date \(t\) in
between, \(s<t\le s+k_{ji}\), it is outstanding, with \(s+k_{ji}-t\)
periods of lead time remaining. The lead-time profile
\(\mathbf k=(k_{ji})_{j\in\mathcal S_i,\,i\in N}\) thereby turns the
sequence of order placements into a pipeline of outstanding orders,
which we record at the start of each period, before that period's
orders are placed, sorted by the lead time remaining.

\begin{definition}[Dated-input pipeline]
\label{def:pipeline}
For each date \(t\) and each \(k=0,1,\dots,\widehat k-1\), the pipeline
layer at remaining lead time \(k\) is the set of orders placed before date
\(t\) that mature \(k\) periods after it,
\begin{equation}
\boldsymbol{\mathcal X}_t^{(k)}
:=
\bigl\{x_{ji,t-k_{ji}+k}:\;j,i\in N,\;k_{ji}\ge k+1\bigr\}
\label{eq:pipeline_layer}
\end{equation}
The pipeline state at date \(t\) is the tuple of pipeline layers
\(\boldsymbol{\mathcal X}_t:=(\boldsymbol{\mathcal X}_t^{(0)},\dots,\boldsymbol{\mathcal X}_t^{(\widehat k-1)})\).
\end{definition}

The bottom pipeline layer, \(k=0\), holds the orders that mature at
date \(t\). With the zero-lead-time orders placed at date \(t\) itself, which
are used within the period and are never outstanding, it forms the
dated input bundle \(\{x_{ji,t-k_{ji}}:j,i\in N\}\) from which firms
produce at date \(t\). The higher pipeline layers hold the orders that have not yet matured. Before the shock every order is the stationary one, so each
pipeline layer holds stationary orders,
\begin{equation}
\boldsymbol{\mathcal X}^{(k)\ast}
=
\bigl\{x_{ji}^\ast:j,i\in N,\,k_{ji}\ge k+1\bigr\},
\qquad k=0,1,\dots,\widehat k-1
\label{eq:stationary_pipeline}
\end{equation}
Dated production thus adds a second state variable to the balances,
the time profile of the outstanding orders. The post-shock state at
date \(t\) is the pair
\(\mathbf S_t:=(\mathbf m_t,\boldsymbol{\mathcal X}_t)\), current
balances and the pipeline, taken at the start of period \(t\), before
the within-period sequence places that period's orders.\footnote{The
orders placed at date \(t\), including those with zero lead time, which are
used within the period, follow from the state through the
within-period sequence (Lemma~\ref{lem:propagation_linearization}).}
At impact, after the date-\(0\) shock and before any post-shock
production, the balances are
\(\mathbf m_0=\mathbf m^\ast+\pi\overline M\boldsymbol\zeta\) and the
pipeline is at its stationary value \(\boldsymbol{\mathcal X}^\ast\),
filled by orders placed at dates \(-k_{ji},\dots,-1\).

The question is whether the economy as a whole, and not only the
distribution of money, returns to a stationary equilibrium. Balances
return to their stationary distribution at the nominal redistribution
rate. But the adjustment of quantities has a feedback of its own. An input that arrives out of its stationary proportion to the firm's other inputs changes the firm's output and its clearing price. These in turn change the quantities that its customers order at that date, orders that mature only after their own lead times. Each round of this real feedback is damped, however. A firm's
labor is fixed, so its output moves by only the fraction \(1-\beta_i\)
of a proportional change in all its intermediate inputs. And each round takes at most \(\widehat k\) periods, the longest lead time. The real disturbance therefore dies out as the misalignment of balances does, though more slowly the longer the longest lead time. Dated inputs can thus prolong the real adjustment
beyond the horizon over which money is redistributed. Write
\(\widehat{\mathbf S}_t(\pi)\) for the state in proportional deviations
from the post-shock stationary values, with balance coordinates
\((m_{i,t}-(1+\pi)m_i^\ast)/m_i^\ast\) and pipeline coordinates
\((x_{ji,s}-x_{ji}^\ast)/x_{ji}^\ast\) over the outstanding orders.

\begin{proposition}[Local convergence after a monetary shock]
\label{prop:stability_stationary}
Suppose that Assumptions~\ref{assump:diffuse_incidence}
and~\ref{assump:monotone_decay_spectrum} hold. Then for \(|\pi|\le\overline\pi\) every intermediate-input expenditure stays
strictly positive, and the post-shock trajectory
\(\{\mathbf S_t(\pi)\}_{t\ge 0}\) exists at every date. And there are a
contraction rate \(\kappa\in(0,1)\) and a constant
\(C_S\), both set by the constants of the assumptions, such that for
\(|\pi|\le\overline\pi\)
\begin{equation}
\bigl\|\widehat{\mathbf S}_t(\pi)\bigr\|_\infty
\;\le\;
C_S\,\kappa^{\,t}\,|\pi|\,(1-\theta),
\qquad t\ge 0
\label{eq:post_shock_geometric_decay}
\end{equation}
with the same bound for outputs, prices and the orders placed at date
\(t\).
\end{proposition}

\noindent
Proof in Appendix~\ref{app:proof_stability}.

The simple reading of Proposition~\ref{prop:stability_stationary} is
that money is neutral in the long run. After the shock the economy
converges to a new stationary equilibrium, in which every price and
every balance has changed in proportion to the money stock and every real quantity returns to its stationary value. Money matters only along the transition from the old stationary equilibrium to the new one.


\section{Miscoordination within Firms}
\label{sec:inefficiency}

In the long run a monetary shock leaves every real quantity unchanged. But monetary shocks are
capable of distorting the input combinations, and consequently
production, in the process of transition from the old stationary equilibrium to the new one. We now derive the form of the loss of output due to this
distortion. At every date a firm produces from a bundle of inputs ordered at
different dates, before the shock or at different points of the
transition. And the shock has reached the firm's suppliers by different
routes. The inputs of the bundle therefore move by different proportions from their stationary levels. The technology absorbs this mismatch by substituting among
the inputs. But the inputs are poor substitutes for one another. The
firm therefore produces less than it would have produced had every
input moved by the average proportional change of the bundle. The
output lost in this way is the loss from production inefficiency, one of the two
parts of the deadweight loss of a monetary shock. Miscoordination within firms
is the output that the firms together lose to the mismatch of the
intermediate inputs within their bundles.

A firm's output responds to its bundle in two ways, to the scale of the
bundle and to its composition. The scale is how much the inputs have risen or fallen on average from their stationary levels. The composition is how far they have moved out of proportion with one another when compared to the proportions in which they are used at the stationary equilibrium. The two responses separate in a
second-order expansion of log output around the stationary equilibrium,
written with the supplier-averaging operator \(\mathcal F_i\), which we
now define.

\begin{definition}[Supplier averaging]
\label{def:supplier_operator}
\label{def:buyer_specific_lag_exposure}
For each firm \(i\in N\), the supplier-averaging operator
\(\mathcal F_i\) sends a supplier-indexed quantity
\(y=(y_j)_{j\in\mathcal S_i}\) to its stationary-share-weighted average over
firm \(i\)'s suppliers,
\begin{equation}
\mathcal F_i(y):=\sum_{j\in\mathcal S_i}a_{ji}\,y_{j}
\label{eq:bundle_scale_def}
\end{equation}
\end{definition}

The dated input bundle of firm \(i\) at date \(t\),
\(\mathbf x_{i,t}=(x_{ji,t})_{j\in\mathcal S_i}\), holds the inherited
orders of the suppliers with \(k_{ji}\ge1\) and the current orders of
those with \(k_{ji}=0\). The proportional change of input \(j\) is
\(\widehat x_{ji,t}=(x_{ji,t}-x_{ji}^\ast)/x_{ji}^\ast\). The deviation
of the bundle from the stationary bundle,
\(\Delta\mathbf x_{i,t}:=\mathbf x_{i,t}-\mathbf x_i^\ast\), is the
stationary bundle scaled by the average proportional change
\(\mathcal F_i(\widehat x_{i,t})\), plus a composition residual
\(\mathbf z_{i,t}\). The residual is orthogonal to the gradient
\(\nabla_i^\ast:=D_{\mathbf x}\log q_i(\mathbf x_i^\ast,l_i^\ast)\) of
log output at the stationary bundle,
\begin{equation}
\Delta\mathbf x_{i,t}=\mathcal F_i(\widehat x_{i,t})\,\mathbf x_i^\ast+\mathbf z_{i,t},
\qquad \nabla_i^\ast\!\cdot\mathbf z_{i,t}=0
\label{eq:scale_composition_raw}
\end{equation}
The second-order expansion of \(\log q_{i,t}\) around
\(\mathbf x_i^\ast\) is then
\begin{equation}
\begin{aligned}
\Delta\log q_{i,t}
={}&
(1-\beta_i)\,\mathcal F_i(\widehat x_{i,t})
-\tfrac12(1-\beta_i)\{\mathcal F_i(\widehat x_{i,t})\}^2\\
&-\tfrac12(1-\beta_i)
\Bigl[
(1-\rho)\,V_i^{\mathrm w}(\widehat x_{i,t})
+(1-\rho_c)\,V_i^{\mathrm b}(\widehat x_{i,t})
\Bigr]
+
o(\|\Delta\mathbf x_{i,t}\|^2)
\end{aligned}
\label{eq:log_q_scale_composition}
\end{equation}
(Lemma~\ref{lem:production_linearization}). Here
\(\mathcal F_i^{(k)}(y):=\sum_{j\in\mathcal S_i^{(k)}}\widetilde a_{ji}^{(k)}y_j\)
is the average of a supplier-indexed quantity \(y\) over vintage layer
\(k\) under the within-vintage shares. The within-vintage variance and
the between-vintage variance of \(y\) are
\begin{equation}
\begin{aligned}
V_i^{\mathrm w}(y)&:=\sum_{k}s_i^{(k)}\sum_{j\in\mathcal S_i^{(k)}}\widetilde a_{ji}^{(k)}\bigl(y_j-\mathcal F_i^{(k)}(y)\bigr)^2,\\
V_i^{\mathrm b}(y)&:=\sum_{k}s_i^{(k)}\bigl(\mathcal F_i^{(k)}(y)-\mathcal F_i(y)\bigr)^2
\end{aligned}
\label{eq:vintage_variances}
\end{equation}

The expansion \eqref{eq:log_q_scale_composition} separates the response
of output to the scale of the dated input bundle from its response to the
composition of the bundle. Its first two terms form the second-order
expansion of \((1-\beta_i)\log\bigl(1+\mathcal F_i(\widehat x_{i,t})\bigr)\).
They depend on the bundle only through its scale, and they record how
much the output of firm \(i\) moves because its dated input bundle rises or falls on average from its stationary level. The third term depends on the bundle only through the
dispersion of the proportional input changes within the bundle, split
by vintage as in \eqref{eq:vintage_variances}. It records how much
output is lost because the bundle is internally misaligned, with more of some dated inputs and less of others than the stationary proportions call for. Holding the average proportional change of the bundle fixed, a wider dispersion of the proportional changes of its inputs raises the loss in the third term and leaves the first two terms unchanged. The output lost through the third
term is the loss from production inefficiency.

\begin{definition}[Loss from production inefficiency]
\label{def:production_inefficiency}
The loss from production inefficiency of firm \(i\) at date \(t\) is the output
lost through the composition term of \eqref{eq:log_q_scale_composition},
\begin{equation}
\mathcal D_{i,t}(\pi)
:=
\tfrac12(1-\beta_i)
\Bigl[
(1-\rho)\,V_i^{\mathrm w}(\widehat x_{i,t})
+(1-\rho_c)\,V_i^{\mathrm b}(\widehat x_{i,t})
\Bigr]
=\tfrac12(1-\beta_i)\bigl\langle\widehat x_{i,t},\widehat x_{i,t}\bigr\rangle_i
\label{eq:firm_inefficiency_def}
\end{equation}
where
\(\langle y,y\rangle_i:=(1-\rho)\,V_i^{\mathrm w}(y)+(1-\rho_c)\,V_i^{\mathrm b}(y)\)
is the curvature-weighted dispersion of a supplier-indexed quantity
\(y\). A weighting is a vector \(\boldsymbol\omega=(\omega_i)_{i\in N}\) with
\(\omega_i\ge0\) and \(\sum_{i\in N}\omega_i=1\), and the firms with \(\omega_i>0\)
form its support. The loss from production inefficiency at
date \(t\) under the weighting \(\boldsymbol\omega\) is
\begin{equation}
\mathcal D_t^{(\boldsymbol\omega)}(\pi)
:=
\sum_{i\in N} \omega_i\,\mathcal D_{i,t}(\pi)
\label{eq:weighted_inefficiency_def}
\end{equation}
\end{definition}

The two variance components of the loss from production inefficiency
\eqref{eq:firm_inefficiency_def}, \(V_i^{\mathrm w}\) and
\(V_i^{\mathrm b}\), carry different weights, \(1-\rho\) and
\(1-\rho_c\), because the technology combines them in different nests.
When the inputs of one vintage arrive in proportions that differ from their stationary proportions, the technology absorbs the mismatch by substituting
among them, and the poorer the substitutes, the more output is lost.
When the vintage composites themselves have moved by different proportions from their stationary levels, the substitution has to run across vintages, and since \(\rho_c\le\rho\) this substitution is harder than the one within a vintage. 

\begin{lemma}[Sign of the loss from production inefficiency]
\label{lem:sign_production_inefficiency}
At every firm \(i\), date \(t\) and shock \(\pi\),
\(\mathcal D_{i,t}(\pi)\ge0\), with equality exactly when the dated
inputs of firm \(i\) all move in the same proportion. Hence
\(\mathcal D_t^{(\boldsymbol\omega)}(\pi)\ge0\) under every weighting
\(\boldsymbol\omega\) of the firms, with equality exactly when, at
every firm in the support of \(\boldsymbol\omega\), the dated inputs all move in the same proportion.
\end{lemma}

\noindent
Proof in Appendix~\ref{app:proof_sign_production_inefficiency}.

\subsection{Cumulative loss from production inefficiency}
\label{subsec:cumulative_loss}

The loss from production inefficiency is not confined to the date of the impact of the monetary shock. The pipeline of dated orders carries the disturbance to the dates after impact, and the inputs of a bundle keep arriving in proportions that differ from their stationary proportions as long as the transition lasts.
The output that a monetary shock destroys within firms is therefore the
loss from production inefficiency summed over every date of the transition.
Define the \emph{cumulative loss from production inefficiency} as
\begin{equation}
\mathcal D^{(\boldsymbol\omega)}(\pi)
:=
\sum_{t=0}^{\infty}\mathcal D_t^{(\boldsymbol\omega)}(\pi)
\label{eq:cumulative_inefficiency_def}
\end{equation}

We now derive the size of the cumulative loss from production inefficiency due to
a monetary shock. At each date the loss from production inefficiency is a
dispersion of the proportional changes of a firm's inputs. It is
therefore of the order of the square of those changes. The changes
themselves are of the order of the shock times the slack \(1-\theta\),
and they die out geometrically along the transition
(Proposition~\ref{prop:stability_stationary}). Which means that the losses from production inefficiency at successive
dates sum to a finite value. 

\begin{theorem}[Order of the cumulative loss from production inefficiency]
\label{theorem:finite_cumulative_inefficiency}
Suppose that Assumptions~\ref{assump:diffuse_incidence} and~\ref{assump:monotone_decay_spectrum}
hold and fix a weighting \(\boldsymbol\omega\). Then for every
\(|\pi|\le\overline\pi\) the cumulative loss from production inefficiency
\(\mathcal D^{(\boldsymbol\omega)}(\pi)\) is finite and at most a fixed
multiple of \((1-\theta)^2\pi^2\), the multiple set by the constants of
the assumptions. As \(\pi\to0\),
\begin{equation}
\mathcal D^{(\boldsymbol\omega)}(\pi)=\pi^2\,C^{(\boldsymbol\omega)}+o(\pi^2),
\qquad
C^{(\boldsymbol\omega)}:=\tfrac12\sum_{t\ge0}\sum_{i\in N}\omega_i(1-\beta_i)\,
\bigl\langle\dot{\widehat x}_{i,t},\dot{\widehat x}_{i,t}\bigr\rangle_i
\label{eq:cumulative_inefficiency_expansion}
\end{equation}
where \(\dot{\widehat x}_{i,t}:=(\dot{\widehat x}_{ji,t})_{j\in\mathcal S_i}\)
is the vector of first-order responses of the proportional changes
of the dated inputs of firm \(i\). The coefficient \(C^{(\boldsymbol\omega)}\) is
strictly positive exactly when, at some date, the first-order responses
of some firm in the support of \(\boldsymbol\omega\) are not equal across
its suppliers.
\end{theorem}

\noindent
Proof in Appendix~\ref{app:proof_finite_cumulative_inefficiency}.

The simple reading of
Theorem~\ref{theorem:finite_cumulative_inefficiency} is that the cumulative loss from production inefficiency of a small shock grows with the square of the shock. A small
monetary expansion and a small monetary contraction of the same size destroy the same
output within firms. This is because they misalign the dated inputs in mirror-image patterns. Naturally, the leading coefficient
\(C^{(\boldsymbol\omega)}\) rises linearly with the two
curvature weights \(1-\rho\) and \(1-\rho_c\), and so, at a given
\(\rho\), with the curvature gap \(\rho-\rho_c\)
(Lemma~\ref{lem:inefficiency_gap} of Appendix~\ref{app:analytical}).\footnote{The comparison across curvatures holds the stationary allocation fixed through the normalized nests of Lemma~\ref{lem:normalized_nests} of Appendix~\ref{app:analytical}, which rescale the CES weights so that every pair of curvatures reproduces the same stationary equilibrium.}

\subsection{Lead times and the loss from production inefficiency}
\label{subsec:lag_profile}

Naturally, the loss from production inefficiency due to monetary
shocks depends on the distribution of intermediate-input lead times. We now
state formally how the loss from production
inefficiency depends on the variance and the mean of the distribution of these lead times. Variance is
the simpler case, and by far. An increase in the variance of the lead times causes input bundle combinations to become more mismatched across the dates on
which the inputs were ordered. And production depends on the
composition of the input bundle in such a way, through equation
\eqref{eq:firm_inefficiency_def}, that a mismatch between inputs of
different vintages is penalized at the outer curvature \(1-\rho_c\).
This can be seen in the third term of the expansion
\eqref{eq:log_q_scale_composition} and in the between-vintage term
\((1-\rho_c)V_i^{\mathrm b}\) of \eqref{eq:firm_inefficiency_def}.

The mismatch within a bundle has two sources, the dates on which its
inputs were ordered and the routes by which the shock reached its
suppliers. To separate them, we split the response of each dated input
into a part that moves with its lead time and a part that does not. For a lead-time profile \(\mathbf k\), a firm \(i\) and a
date \(t\), let
\(\dot{\widehat x}^{\mathbf k}_{i,t}=(\dot{\widehat x}^{\mathbf k}_{ji,t-k_{ji}})_{j\in\mathcal S_i}\)
be the first-order response of the dated inputs of firm \(i\) under
\(\mathbf k\). Write \(\operatorname{Var}_{\mathcal F_i}(\mathbf k)\) for
the variance of the lead time under the supplier shares of firm \(i\).
The regression slope of the input response on the lead time and its
residual are
\begin{equation}
\varsigma^{\mathbf k}_{i,t}:=\frac{\operatorname{Cov}_{\mathcal F_i}(\dot{\widehat x}^{\mathbf k}_{i,t},\mathbf k)}{\operatorname{Var}_{\mathcal F_i}(\mathbf k)}
\ \ (\varsigma^{\mathbf k}_{i,t}:=0\text{ if }\operatorname{Var}_{\mathcal F_i}(\mathbf k)=0),
\qquad
u^{\mathbf k}_{ji,t}:=\dot{\widehat x}^{\mathbf k}_{ji,t-k_{ji}}-\mathcal F_i(\dot{\widehat x}^{\mathbf k}_{i,t})-\varsigma^{\mathbf k}_{i,t}\{k_{ji}-\mathcal F_i(\mathbf k)\}
\label{eq:lag_regression_slope}
\end{equation}
We call \(\varsigma^{\mathbf k}_{i,t}\{k_{ji}-\mathcal F_i(\mathbf k)\}\)
the \emph{lead-time component} of the response and
\(u^{\mathbf k}_{ji,t}\) its \emph{routing component}. The routing
component carries the part of the mismatch that comes from the routes by
which the shock reached the suppliers. We say that a lead-time profile
\(\mathbf k'\) is a \emph{buyerwise mean-preserving lead-time-variance
increase} of \(\mathbf k\) if every firm faces the same mean lead time and
a weakly larger lead-time variance under \(\mathbf k'\) than under
\(\mathbf k\). That is, \(\mathcal F_i(\mathbf k')=\mathcal F_i(\mathbf k)\)
and \(\operatorname{Var}_{\mathcal F_i}(\mathbf k')\ge\operatorname{Var}_{\mathcal F_i}(\mathbf k)\)
for every \(i\in N\). Write
\(\mathcal D_t^{(\boldsymbol\omega)}(\pi,\mathbf k)\) for the loss from
production inefficiency at date \(t\) under the profile \(\mathbf k\), and
\(\langle y,y\rangle_{i,\mathbf k}\) for the curvature-weighted dispersion
of \eqref{eq:firm_inefficiency_def} taken under the vintage partition of
\(\mathbf k\).

\begin{corollary}[Loss from production inefficiency under a wider lead-time profile]
\label{cor:firm_level_vintage_spread_aggregate_inefficiency}
Let \(\mathbf k'\) be a buyerwise mean-preserving lead-time-variance increase of
\(\mathbf k\) on the same network, with the same stationary shares, labor shares
and incidence. Fix a date \(t\) at which the slopes agree buyer by
buyer, \(\varsigma^{\mathbf k'}_{i,t}=\varsigma^{\mathbf k}_{i,t}=:\varsigma_{i,t}\). Then
\begin{equation}
\begin{aligned}
\lim_{\pi\to0}\frac{2\bigl[\mathcal D_t^{(\boldsymbol\omega)}(\pi,\mathbf k')-\mathcal D_t^{(\boldsymbol\omega)}(\pi,\mathbf k)\bigr]}{\pi^2}
={}&
(1-\rho_c)\sum_{i\in N}\omega_i(1-\beta_i)\,\varsigma_{i,t}^2\bigl(\operatorname{Var}_{\mathcal F_i}(\mathbf k')-\operatorname{Var}_{\mathcal F_i}(\mathbf k)\bigr)\\
&+\sum_{i\in N}\omega_i(1-\beta_i)\Bigl[\bigl\langle u^{\mathbf k'}_{i,t},u^{\mathbf k'}_{i,t}\bigr\rangle_{i,\mathbf k'}-\bigl\langle u^{\mathbf k}_{i,t},u^{\mathbf k}_{i,t}\bigr\rangle_{i,\mathbf k}\Bigr]
\end{aligned}
\label{eq:lag_variance_comparison}
\end{equation}
and the first sum is non-negative. Each curvature-weighted form is
taken under the vintage partition of its own lead-time profile.
\end{corollary}

\noindent
Proof in Appendix~\ref{app:proof_firm_spread_aggregate}.

The simple reading of
Corollary~\ref{cor:firm_level_vintage_spread_aggregate_inefficiency}
is that a
wider spread of lead times around an unchanged mean raises the loss from
production inefficiency at a date by the outer curvature \(1-\rho_c\) times the increase in
the variance of the lead times. Each firm's increase in lead-time variance is
weighted by how strongly its input orders respond to the lead time.
The exception is when the part of the mismatch that comes from the
routes by which the shock reached the suppliers falls by enough to
offset this rise in the loss from production inefficiency.

The question of what happens to the loss from production inefficiency with an increase in
the mean of the distribution of intermediate-input lead times is a
more difficult one. To ascertain the problem at hand, we must begin by
recognizing that a monetary shock sets in motion two intertwined processes, each of which propagates itself and generates the other. A monetary expansion influences money balances, therefore
prices and production, and therefore future money balances. In this
way a present monetary expansion generates a decaying wave of monetary
disturbances that runs through the system. But by influencing input
purchases and production, it also generates real disturbances, since
one firm's output is another firm's input. These waves of real
disturbances also course through the firm network. An increase in the mean of the distribution of intermediate-input lead times separates the monetary and the real waves in time. The monetary waves disturb prices and input orders, but
they do not immediately affect production, since production depends on
orders placed long before the shock occurred. And by the time the
`new' disturbed orders come to affect production (and prices), the
original monetary wave may have all but dampened out. Whether an
increase in the mean lead time of intermediate inputs increases
the loss from production inefficiency depends on whether the monetary and the real waves
dampen or amplify each other. If they dampen each other, then an increase in the mean lead time, by separating the monetary and the real waves in time, increases
the loss from production inefficiency. If they amplify one another, then separating them
reduces the loss from production inefficiency generated by a monetary shock.

Both waves can be read off the response of every dated input. The
first-order response of the real order that firm \(i\) places on
supplier \(j\) at a date \(s\ge0\) is the response of the expenditure of
firm \(i\) net of the response of the demand for good \(j\), plus the
response of the output of \(j\),
\(\dot{\widehat x}_{ji,s}=(\dot{\widehat e}_{i,s}-\dot{\widehat d}_{j,s})+\dot{\widehat q}_{j,s}\)
(Lemma~\ref{lem:order_response_decomposition} of
Appendix~\ref{app:analytical}). Evaluated at the order dates of the dated
input bundle of firm \(i\) at date \(t\), the first term is the monetary
wave and the second the real wave,
\begin{equation}
M_{ji,t}:=\dot{\widehat e}_{i,t-k_{ji}}-\dot{\widehat d}_{j,t-k_{ji}},
\qquad
Q_{ji,t}:=\dot{\widehat q}_{j,t-k_{ji}},
\qquad j\in\mathcal S_i
\label{eq:waves_main}
\end{equation}
The real wave arrives in rounds. Each round has travelled through one
more lead time than the round before. Write \(Q^{(r)}\) for the
\(r\)-th round, so that \(Q=\sum_{r\ge1}Q^{(r)}\)
(equation \eqref{eq:real_wave_rounds} of Appendix~\ref{app:proof_lag_shift}).
For collections \(y=(y_{i,t})\) and \(z=(z_{i,t})\) of supplier vectors
indexed by firm and date, and a weighting \(\boldsymbol\omega\), the
cumulative curvature-weighted form is
\begin{equation}
\langle y,z\rangle_H:=\sum_{t\ge0}\sum_{i\in N}\omega_i(1-\beta_i)\,\langle y_{i,t},z_{i,t}\rangle_i,
\qquad
\|y\|_H^2:=\langle y,y\rangle_H
\label{eq:H_form}
\end{equation}
whenever the sums converge absolutely. Here \(\langle y,z\rangle_i\) is
the curvature-weighted covariance \eqref{eq:curvature_weighted_cov}, whose
value at \(z=y\) is the curvature-weighted dispersion of
\eqref{eq:firm_inefficiency_def}. Let \(\mathbf k'\) be the profile in which every lead time is longer by
\(\Delta k\ge1\) periods, and \(C^{(\boldsymbol\omega)}(\mathbf k)\) the
leading coefficient of Theorem~\ref{theorem:finite_cumulative_inefficiency}
under the profile \(\mathbf k\). The longer lead times delay the monetary
wave by \(\Delta k\) periods and the \(r\)-th round of the real wave by
\((r+1)\Delta k\) periods. We therefore count the monetary wave as the
wave of round zero, \(W^{(0)}:=M\), and write \(W^{(r)}:=Q^{(r)}\) for
\(r\ge1\). The translation \(\mathsf S_\Delta\) delays a collection by
\(\Delta\) dates, \((\mathsf S_\Delta z)_t=z_{t-\Delta}\) for
\(t\ge\Delta\) and zero before.

\begin{corollary}[Loss from production inefficiency under a uniform lead-time shift]
\label{cor:lag_shift_inefficiency}
Suppose that Assumptions~\ref{assump:diffuse_incidence}
and~\ref{assump:monotone_decay_spectrum} hold under \(\mathbf k'\), with
the same stationary shares, labor shares and incidence as under
\(\mathbf k\). Then for every weighting \(\boldsymbol\omega\)
\begin{equation}
C^{(\boldsymbol\omega)}(\mathbf k')-C^{(\boldsymbol\omega)}(\mathbf k)
=\sum_{0\le r<r'}\Bigl[\bigl\langle W^{(r)},\mathsf S_{(r'-r)\Delta k}W^{(r')}\bigr\rangle_H
-\bigl\langle W^{(r)},W^{(r')}\bigr\rangle_H\Bigr]
\label{eq:lag_shift_exact}
\end{equation}
If, in addition, every purchase has the same lead time and the impact
misalignment \(\dot{\mathbf m}_0^\perp\) is an eigenvector of \(\mathbf A\)
with a real eigenvalue \(\lambda\in(0,1)\), then
\begin{equation}
C^{(\boldsymbol\omega)}(\mathbf k')-C^{(\boldsymbol\omega)}(\mathbf k)
=-\sum_{0\le r<r'}\bigl(1-\lambda^{(r'-r)\Delta k}\bigr)\bigl\langle W^{(r)},W^{(r')}\bigr\rangle_H
\label{eq:lag_shift_geometric}
\end{equation}
\end{corollary}

\noindent
Proof in Appendix~\ref{app:proof_lag_shift}.

The simple reading of Corollary~\ref{cor:lag_shift_inefficiency} is that
lengthening every lead time by \(\Delta k\) periods changes the cumulative
loss from production inefficiency by exactly the change in the
comovements between waves that the longer lead times shift apart in
time. These are the comovement of the monetary wave with each round of
the real wave, and the comovements of the rounds of the real wave with
one another. The dispersion that each wave causes on its own does not
change, for a wave that arrives later is still the same wave. Whether
the monetary and the real waves dampen or amplify each other is a
question of alignment. It depends on whether the direction in which
intermediate-input demand travels through the network coincides with the
direction in which money travels from the firms on which the monetary
expansion lands.

The second part of Corollary~\ref{cor:lag_shift_inefficiency} says how
large this change is in the simplest case, in which every purchase has
the same lead time and the misalignment of balances dies out at a single
rate \(\lambda\). There every wave dies out at the rate \(\lambda\) as
well. Two waves that are \(r'-r\) rounds apart are delayed relative to
each other by \((r'-r)\Delta k\) periods. By the time the later wave
arrives, the earlier one has lost the fraction
\(1-\lambda^{(r'-r)\Delta k}\) of its strength, and the longer lead times
remove this fraction of their comovement. The fraction is larger the
faster balances re-equilibrate, the longer the added lead time
\(\Delta k\) and the more rounds separate the two waves. The loss from
production inefficiency therefore rises with the lead times when the
waves that they pull apart dampen each other, since the offset between
them shrinks. And it falls when they amplify each other.

In the simplest case the alignment of the two waves is easy to see. A
supplier on which the monetary expansion lands heavily keeps receiving
more than its share of spending while the misalignment of balances dies
out. Its buyers bid up the price of its good, and each of them orders
less of that good relative to its other inputs. This is the monetary
wave. But the same supplier also has more money to spend on its own
inputs, and a lead time later it produces more, which holds its price
down. This is the real wave. The two waves therefore move the price of
the supplier's good, and the orders placed with it, in opposite
directions. They dampen each other whenever averaging over suppliers does
not make the shock more uneven (Lemma~\ref{lem:wave_dampening} of
Appendix~\ref{app:proof_lag_shift}). That is, the shock is spread at
least as unevenly across the suppliers of a firm as its averages over the
firms further up the supply chain from which each of them buys, directly
or indirectly. The shock lands unevenly because firms differ in size. The
dampening condition therefore compares two dispersions among the
suppliers of a firm, that of their own sizes and that of the average
sizes of the firms further up the supply chain from which each of them
buys.

\section{Monetary Reallocation}
\label{sec:reallocation}
Monetary shocks distort the input combinations within firms, and consequently production. We now
formalize how the incidence and propagation of monetary shocks also
shift economic activity from one part of the supply chain to another.
A monetary shock lands more heavily on small firms, relative to their
balances. When the firms close to final demand are smaller than the other buyers of the inputs they use, a positive shock lets them bid these inputs away from the firms further up the supply chain. Downstream output rises at impact, upstream firms receive fewer inputs than before the shock, and the pipeline of dated orders carries both movements to later dates. When the resulting fall in
upstream output reaches the firms close to final demand, it can turn
their early expansion into a contraction. A negative shock produces the mirror image.

We begin by introducing a measure of
a firm's position within the supply chain. A firm's downstreamness is the
extent to which its sales are connected, directly or indirectly, to
household expenditure through chains of firm-to-firm transactions.

\begin{definition}[Downstreamness]
\label{def:downstreamness}
Let
\begin{equation}
\boldsymbol\Omega
:=
\mathbf D_m^{-1}\mathbf A_\beta\mathbf D_m,
\qquad
\Omega_{ji}
=
\frac{(1-\beta_i)\,a_{ji}\,m_i^\ast}{m_j^\ast}
=
\frac{p_j^\ast x_{ji}^\ast}{m_j^\ast}
\label{eq:output_share_matrix}
\end{equation}
be the output-share matrix, whose entry \(\Omega_{ji}\) is the share of the
stationary sales of supplier \(j\) bought by firm \(i\). Let
\begin{equation}
r_j
:=
\frac{\gamma_j w^\ast}{m_j^\ast}
=
\frac{c_j^\ast}{q_j^\ast}
\label{eq:retail_absorption_def}
\end{equation}
be the consumption share of firm \(j\), the share of its output sold to the household, and fix
a discount \(\delta\in(0,1)\), the same for every economy. The
downstreamness \(\widetilde\psi_j\) of firm \(j\) is the share of its
output sold to the household plus \(\delta\) times the sales-weighted
average downstreamness of its customers,
\begin{equation}
\widetilde\psi_j
=
r_j+\delta\sum_{i\in N}\Omega_{ji}\,\widetilde\psi_i,
\qquad
\widetilde{\boldsymbol\psi}
=
\delta\,\boldsymbol\Omega\,\widetilde{\boldsymbol\psi}
+
\mathbf r
\label{eq:downstream_raw_def}
\end{equation}
We call the downstreamness of a firm relative to that of the most
downstream firm,
\(\psi_i:=\widetilde\psi_i/\max_{j\in N}\widetilde\psi_j\), its
downstreamness index.
\end{definition}

Put simply, \(\widetilde\psi_j\) is the \(\delta\)-discounted share of the output of firm \(j\) that reaches the household, when that output is followed from each seller to its buyers in proportion to the output shares. The seed
\(\mathbf r\) records direct sales to the household, and each further
round adds the exposure through customers, then customers' customers,
and so on, discounted once per round.\footnote{The vector
\(\widetilde{\boldsymbol\psi}\) is the fixed point of the contraction
\(\mathbf v\mapsto\mathbf r+\delta\boldsymbol\Omega\mathbf v\) on
\([0,1]^n\), and it is strictly positive because final demand is
reachable from every firm. Without the discount the series gives every
firm the value one, since all of a firm's output eventually reaches the household. Without the division by own sales as well, it returns the
stationary sales vector
\((\mathbf I-\mathbf A_\beta)^{-1}\boldsymbol\gamma=\mathbf m^\ast/w^\ast\).
The division by own sales removes the effect of firm size, and the discount lets the
measure distinguish one position from another.}

\subsection{Supply-chain structure}
\label{subsec:downstreamness}

In a supply chain money flows from downstream to upstream, and
intermediate inputs flow in the opposite direction. We now impose the
four assumptions that give the production network this shape and so
make it a supply chain in earnest.

The first assumption gives sourcing a direction. Fix a level
\(\psi_T\) and call the firms at downstreamness \(\psi_T\) or below the
top tier,
\begin{equation}
\mathcal T:=\{i\in N:\psi_i\le\psi_T\}
\label{eq:top_tier}
\end{equation}

\begin{assumption}[Upstream sourcing]
\label{assump:transmission_delay}
For the fixed constant \(\Delta_\psi>0\), every firm outside the top tier
is at least \(\Delta_\psi\) more downstream than its suppliers are on
average, each supplier weighted by its share of the firm's input expenditure,
\begin{equation}
\psi_i-\mathcal F_i(\psi)
\;\ge\;
\Delta_\psi
\qquad\text{for every }i\in N\setminus\mathcal T
\label{eq:transmission_delay}
\end{equation}
\end{assumption}

One round of sourcing climbs at least \(\Delta_\psi\) up the
supply chain.\footnote{The climb
cannot hold at every firm of the top tier, since the most upstream
firms have no suppliers upstream of them. The top tier
is the part of the supply chain over which the climb per round of sourcing
has fallen below \(\Delta_\psi\), as it must as the top is approached.}
The constant \(\Delta_\psi\) is the height of a tier of the supply chain. This climb
of at least one tier per round of sourcing is what makes the supply
chain a `chain'.

We also assume that retailers are
predominantly final-goods producers, in that only a small fraction
\(1-\epsilon_h\) of their output is sold to other firms.

\begin{assumption}[Consumption-share floor at the retail tier]
\label{assump:retail_outflow_cap}
For the fixed constant \(\epsilon_h\in(0,1)\), the household buys at
least the share
\(\epsilon_h\) of the output of every firm it buys from,
\begin{equation}
r_i\;\ge\;\epsilon_h
\qquad
\text{for every }i\in\operatorname{supp}\boldsymbol\gamma
\label{eq:retail_consumption_share}
\end{equation}
The constant \(\epsilon_h\)
lies above the downstreamness discount of
Definition~\ref{def:downstreamness}, \(\epsilon_h>\delta\).
\end{assumption}

Note that Assumption~\ref{assump:retail_outflow_cap} places the firms
that the household buys from strictly downstream of every other firm,
\begin{equation}
\min_{i\in\operatorname{supp}\boldsymbol\gamma}\widetilde\psi_i
\;>\;
\max_{j\notin\operatorname{supp}\boldsymbol\gamma}\widetilde\psi_j
\label{eq:retail_tier_primacy}
\end{equation}
No other firm sells enough to retailers to be more downstream than a retailer.\footnote{By \eqref{eq:downstream_raw_def},
every retailer has \(\widetilde\psi_i\ge r_i\ge\epsilon_h\), while
every firm that the household does not buy from has
\(\widetilde\psi_j\le\delta\max_{l\in N}\widetilde\psi_l\le\delta<\epsilon_h\).}
We call the firms that the household buys from,
\(\mathcal R:=\operatorname{supp}\boldsymbol\gamma\), the \emph{retail
tier}, the network's consumer-facing tier. Write
\(\psi_R:=\min_{i\in\mathcal R}\psi_i\) for the position of the least
downstream retailer, so that \(\psi_R\ge\epsilon_h\), while every firm off the retail tier has a position \(\psi_j\le\delta\). We call
\(1-\epsilon_h\) the \emph{retail slack}, the largest share of its
output that a retailer may sell to other firms.

\begin{assumption}[Time to build]
\label{assump:lag_position_sorting}
For the fixed constant \(\epsilon_k\in(0,1)\) and a lead time
\(k_\ast\in\{\underline k+1,\dots,\widehat k\}\), every retailer places
all but the share \(\epsilon_k\) of its input expenditure on suppliers
whose inputs arrive within fewer than \(k_\ast\) periods. And every
firm of the top tier places all but the share \(\epsilon_k\) on
suppliers whose inputs take \(k_\ast\) periods or more to arrive,
\begin{equation}
\sum_{k<k_\ast}s_i^{(k)}\ge1-\epsilon_k\ \text{ for every }i\in\mathcal R,
\qquad
\sum_{k\ge k_\ast}s_i^{(k)}\ge1-\epsilon_k\ \text{ for every }i\in\mathcal T
\label{eq:time_to_build}
\end{equation}
\end{assumption}

The assumption means lead times lengthen along the supply chain away from final demand. The retail tier produces mostly with inputs that arrive within fewer than \(k_\ast\) periods, and the top tier mostly with inputs that take \(k_\ast\) periods or more to arrive.\footnote{Inventory productivity is materially lower upstream in
manufacturing supply networks \citep{agrawal_osadchiy_2024}, which is what
longer effective inventory horizons upstream would produce.}
Under Assumption~\ref{assump:lag_position_sorting} the lead times of the firms between the retail tier and the top tier may be long or short and mixed within a firm. And in the retail tier and in the top tier a firm may hold the share \(\epsilon_k\) of its inputs at any lead time.

For a fixed constant \(s_R\in(\epsilon_k,1)\), write
\begin{equation}
\mathcal W
:=
\Bigl\{t\in\{\underline k,\dots,\widehat k\}:
\sum_{i\in\mathcal R}m_i^\ast(1-\beta_i)\,s_i^{(t)}
\;\ge\;
s_R\sum_{i\in\mathcal R}m_i^\ast(1-\beta_i)\Bigr\}
\label{eq:retail_lag_dates}
\end{equation}
for the \emph{retail lead-time dates}, the lead times on which the retail tier
places at least the share \(s_R\) of its input expenditure. Here
\(s_i^{(k)}\) of \eqref{eq:share_quantity_form} is the share that
firm \(i\) places on suppliers of lead time \(k\). Under
Assumption~\ref{assump:lag_position_sorting} the retail lead-time dates all
fall before \(k_\ast\). This is because the retail tier places at most
the share \(\epsilon_k<s_R\) of its expenditure on lead times of \(k_\ast\)
or more.

At impact firm \(i\) holds the balance \(m_i^\ast+\pi\overline M\zeta_i\)
and pays the wage bill \((1+\pi)\beta_im_i^\ast\), so its
intermediate-input expenditure is
\((1+\pi)(1-\beta_i)m_i^\ast+\pi\,\iota_im_i^\ast\). Per unit of shock
and relative to its stationary level, the expenditure moves by
\(1+\sigma_i\), where
\begin{equation}
\sigma_i:=\frac{\iota_i}{1-\beta_i}
\label{eq:expenditure_impulse}
\end{equation}
is the expenditure impulse of firm \(i\), the misalignment \eqref{eq:incidence_ratio} at firm \(i\) divided by its intermediate-input share. The one in \(1+\sigma_i\) is the rise of the money stock, which every firm shares. And \(\sigma_i\) is the extra rise that the uneven incidence of the shock gives the firm's expenditure, since its excess balance goes wholly into its purchases.
The nominal demand for a good moves with the expenditures of its
buyers. Per unit of shock and relative to its stationary level, the
demand for good \(j\) therefore moves by \(1+\overline\sigma_j\), where
\begin{equation}
\overline\sigma_j
:=
\sum_{l\in N}\Omega_{jl}\,\sigma_l
\label{eq:buyer_average_impulse}
\end{equation}
averages \(\sigma\) over the buyers of good \(j\), each counted by its
purchases. The household enters this average with its share \(r_j\) and an expenditure impulse of zero. This is because its expenditure moves with
the money stock exactly. The quantity of a good that a buyer obtains is set by its expenditure relative to the total demand for the good. Which means that
what matters for the real order of firm \(i\) from supplier \(j\) is
the difference \(\sigma_i-\overline\sigma_j\). For a firm \(i\) and a lead time
\(k\), write
\begin{equation}
\xi_i^{(k)}
:=
\sum_{j\in\mathcal S_i^{(k)}}
\widetilde a^{(k)}_{ji}\,\bigl(\sigma_i-\overline\sigma_j\bigr)
\label{eq:vintage_impulse_def}
\end{equation}
for the \emph{relative expenditure impulse} of firm \(i\) on its
inputs of lead time \(k\). It is the excess of the expenditure impulse of firm \(i\) over the
average expenditure impulse of the other buyers of those inputs. This
excess is averaged over its suppliers of lead time \(k\) with the
within-vintage shares of \eqref{eq:share_quantity_form}. For a set of firms
\(\mathcal C\subseteq N\), write
\begin{equation}
\xi_{\mathcal C}^{(k)}
:=
\frac{\sum_{i\in\mathcal C}m_i^\ast(1-\beta_i)\,s_i^{(k)}\,\xi_i^{(k)}}
{e_{\mathcal C}^{(k)}},
\qquad
e_{\mathcal C}^{(k)}
:=
\sum_{i\in\mathcal C}m_i^\ast(1-\beta_i)\,s_i^{(k)}
\label{eq:set_impulse}
\end{equation}
for the relative expenditure impulse of \(\mathcal C\) on its inputs
of lead time \(k\), the relative expenditure impulses of its firms averaged over their expenditure on those inputs. The relative expenditure impulse of a firm is positive when the shock lands more heavily on it, relative to its balance, than on the other buyers of its inputs, and negative when the shock lands more lightly on it than on them.

\begin{assumption}[Incidence--position sorting]
\label{assump:impact_forcing_ordering}
For the fixed constant \(c_\xi>0\), at every retail lead-time date the shock
lands more heavily on the retail tier than on the other buyers of the
inputs it orders at that lead time. And the shock does so by the margin \(c_\xi\) per unit of
slack,
\begin{equation}
\xi_{\mathcal R}^{(k)}\;\ge\;c_\xi\,(1-\theta)
\qquad\text{for every }k\in\mathcal W
\label{eq:size_position_sorting}
\end{equation}
\end{assumption}

When firm size declines along the supply chain toward final demand, which we call the size--position sorting, the incidence rule turns the ordering of sizes into an ordering of relative expenditure impulses. The
incidence ratio \(e^{-(1-\theta)\ell_i}/Z\) falls with size, so a firm smaller than the other buyers of its inputs receives a larger expenditure impulse than they do. To leading order in the slack, the relative expenditure impulse \(\xi_i^{(k)}\) of a firm on its inputs of lead time \(k\) is the slack times the shortfall of its log size
below the average log size of the other buyers of those inputs. Each
log size enters divided by the intermediate-input share of its firm
(Lemma~\ref{lem:impulse_exposure} of
Appendix~\ref{app:lemmas_reallocation}). The assumption says that
the retail tier is smaller than the firms with which it shares its input markets, stated in terms of the relative expenditure impulses that the shock produces.\footnote{The average firm in the consumer-facing sectors, retail
trade, food service and personal services, is far smaller than in the
upstream primary and intermediate-goods sectors.} The assumption compares the
retail tier with the other buyers of the inputs it buys, and with
nothing else.\footnote{The assumption says nothing about the sizes of firms elsewhere in the supply chain, about how size varies from one stage to the next, or about any single retailer. A retailer larger than every firm in the economy is
allowed, since it enters the average only in proportion to what it
buys in the markets it shares with non-retail firms or in which the
household buys. A market in which retailers bid only against one
another at the same lead time, and in which the household does not buy,
contributes nothing to the relative expenditure impulse of the retail tier. In such a market what one retailer gains another loses. The assumption
fails only where retailers buy mainly in markets of their
own, so that the comparison is empty, or where the retailers that buy
alongside upstream firms are the larger ones.}

\subsection{Impact reallocation}
\label{subsec:impact_reallocation}

To measure where activity moves across the supply chain, we introduce a
one-parameter family of aggregate log-output objects. The parameter
controls how much weight is placed on different parts of the supply chain, and lets us track whether a monetary shock shifts production
toward or away from final demand. Position is measured by the
downstreamness index \(\boldsymbol\psi\) of
Definition~\ref{def:downstreamness}.

\begin{definition}[Position-tilted log-output index]
\label{def:position_tilted_output_index}
For \(b\in[0,\infty]\), the tilt weights and the
position-tilted log-output index are
\begin{equation}
\omega_i(b)
:=
\frac{m_i^\ast\,\widehat\psi_i^{\,b}}
{\sum_{j\in N}m_j^\ast\,\widehat\psi_j^{\,b}},
\qquad
\widehat\psi_i:=\min\Bigl\{\frac{\psi_i}{\psi_R},\,1\Bigr\},
\qquad
\mathcal Q_t(b)
:=
\sum_{i\in N}
\omega_i(b)\log q_{i,t}
\label{eq:w_b_def}
\end{equation}
with \(\widehat\psi_i^{\,\infty}:=\lim_{b\to\infty}\widehat\psi_i^{\,b}\),
which is one on the retail tier and zero off it.
\end{definition}

The capped position \(\widehat\psi_i\) is the downstreamness index of firm
\(i\) relative to that of the least downstream retailer, capped at
one, so it is one on the whole retail tier and at most
\(\delta/\epsilon_h<1\) off it. At \(b=0\) the weights are the sales
shares \(m_i^\ast/\overline M\), which are proportional to the Domar
weights \(m_i^\ast/w^\ast\), and the index is the sales-weighted
aggregate of log outputs. As \(b\) grows the weights shift toward firms closer to final demand, and at \(b=\infty\) they are the sales shares
within the retail tier, so the index is the sales-weighted log output
of the retail tier. We study the first-order response of the index to
the shock,
\[
D_\pi\mathcal Q_t(b)
:=
\frac{\partial \mathcal Q_t(b,\pi)}{\partial\pi}\Bigg|_{\pi=0}
\]
At first order the output of a firm responds only to the scale of its dated input bundle, and not to its composition. This is because the composition residual
\(\mathbf z_{i,t}\) is orthogonal to the gradient of log output,
\(\nabla_i^\ast\!\cdot\mathbf z_{i,t}=0\) in
\eqref{eq:scale_composition_raw}. And the loss from production inefficiency
\(\mathcal D_{i,t}\) of \eqref{eq:firm_inefficiency_def}, a quadratic
form in the bundle deviation, is of second order in the shock. The
response of the index is therefore the tilted average of the responses
of the bundle scales,
\(D_\pi\mathcal Q_t(b)=\sum_{i\in N}\omega_i(b)(1-\beta_i)\mathcal F_i(\dot{\widehat x}_{i,t})\).
The results of this section concern this first-order response alone. They hold in the general economy, in which the loss from production inefficiency is also present.

The first bundle that the shock affects is the one that matures at date
\(\underline k\), the shortest lead time in the economy, which we take to be
positive throughout this section. The shortest-lead-time part of that
bundle was ordered at the impact date. The rest of it, and every input
from which the suppliers produced those orders, was ordered before the
shock. The impact-date orders are therefore exchanged at prices that the shock moves, against outputs that were fixed before the shock. The real order of
firm \(i\) from a shortest-lead-time supplier \(j\) moves by the
difference \(\sigma_i-\overline\sigma_j\). The part of its dated input bundle that the shock affects therefore moves by
\begin{equation}
\xi_{i,\underline k}
:=
\sum_{j\in\mathcal S_i^{(\underline k)}}
a_{ji}\,\bigl(\sigma_i-\overline\sigma_j\bigr)
\;=\;
s_i^{(\underline k)}\,\xi_i^{(\underline k)}
\label{eq:impact_forcing_def}
\end{equation}
the relative expenditure impulse \eqref{eq:vintage_impulse_def} of
firm \(i\) on its shortest-lead-time inputs, weighted by the share of
its expenditure that those inputs carry. This relative expenditure impulse is a property of the economy, fixed by the network \(\mathbf A\), the lead-time profile, the
stationary allocation, the labor shares and the incidence
parameter \(\theta\).\footnote{The relative expenditure impulse of a
firm vanishes when none of its suppliers sells to the household and
all of each supplier's firm customers carry the same expenditure
impulse. When both hold on every active shortest-lead-time link, money is
neutral at impact.}

The log output of firm \(i\) at date \(\underline k\) moves by its
intermediate-input share times its relative expenditure impulse, and
on its inputs of lead time \(k\) alone by
\begin{equation}
\phi_i:=(1-\beta_i)\,\xi_{i,\underline k},
\qquad
\phi_i^{(k)}:=(1-\beta_i)\,\xi_i^{(k)},
\qquad
\phi_i=s_i^{(\underline k)}\phi_i^{(\underline k)}
\label{eq:scaled_impulse_def}
\end{equation}
the \emph{scaled} relative expenditure impulses. The scaled relative expenditure impulse \(\phi_i^{(k)}\) on a vintage of lead time \(k>\underline k\) enters production at date \(k\), when the orders placed at impact on that vintage mature. At date \(\underline k\) the position-tilted log-output index moves by the
tilted average of the scaled relative expenditure impulses.

Every impulse is bounded by the incidence range. A relative
expenditure impulse is an average of differences of two expenditure
impulses, each at most \(C_\zeta(1-\theta)/(1-\overline\beta)\) in
absolute value, so at every firm and every active vintage
\begin{equation}
\bigl|\xi_i^{(k)}\bigr|\;\le\;C_\xi\,(1-\theta),
\qquad
\bigl|\phi_i^{(k)}\bigr|\;\le\;C_\phi\,(1-\theta),
\qquad
C_\xi:=\frac{2C_\zeta}{1-\overline\beta},
\qquad
C_\phi:=(1-\underline\beta)\,C_\xi
\label{eq:forcing_upper_bound}
\end{equation}
At a retail lead-time date the retail tier places at least the share
\(s_R\) of its input expenditure on the inputs of that lead time. Under
Assumption~\ref{assump:impact_forcing_ordering} the sales-weighted
output response of the retail tier on those inputs is therefore at
least
\begin{equation}
\sum_{i\in\mathcal R}m_i^\ast\,s_i^{(k)}\phi_i^{(k)}
\;\ge\;
c_R\,(1-\theta)\sum_{i\in\mathcal R}m_i^\ast,
\qquad
c_R:=s_R\,(1-\overline\beta)\,c_\xi
\qquad\text{for every }k\in\mathcal W
\label{eq:retail_impulse_bound}
\end{equation}
the \emph{retail margin} \(c_R\) per unit of the sales of the retail tier
(Appendix~\ref{app:proof_impact_reallocation}).

Call a tilt \(b\in[0,\infty]\) \emph{downstream-concentrated} if its
weights place mass at least \(1-\overline\epsilon_w\) on the retail
tier, \(\sum_{i\in\mathcal R}\omega_i(b)\ge1-\overline\epsilon_w\), for
a threshold \(\overline\epsilon_w\in(0,1)\) set by the constants of the
assumptions and the incidence range (Appendix~\ref{app:reallocation}).
The retail-tier mass \(\sum_{i\in\mathcal R}\omega_i(b)\) is
nondecreasing in \(b\) and rises to one. This is because every
capped position off the retail tier is below one
(Appendix~\ref{app:proof_impact_reallocation}). The
downstream-concentrated tilts therefore form an interval
\([b_w,\infty]\), nonempty for every threshold.

\begin{proposition}[Impact reallocation]
\label{prop:impact_reallocation}
Suppose that Assumptions~\ref{assump:diffuse_incidence},
\ref{assump:retail_outflow_cap}
and~\ref{assump:impact_forcing_ordering} hold, that
\(\underline k\ge1\), and that the shortest lead time is a retail lead-time date,
\(\underline k\in\mathcal W\). Then for every downstream-concentrated
tilt \(b\)
\[
D_\pi\mathcal Q_{\underline k}(b)\;\ge\;\tfrac12\,c_R\,(1-\theta)\;>\;0
\]
The impact response of every firm is its scaled relative expenditure
impulse, \(D_\pi\log q_{i,\underline k}=\phi_i\). A firm on which the
shock lands more lightly than on the other buyers of its shortest-lead-time
inputs by the margin, \(\xi_i^{(\underline k)}\le-c_\xi(1-\theta)\),
contracts,
\[
D_\pi\log q_{i,\underline k}\;\le\;-s_i^{(\underline k)}(1-\beta_i)\,c_\xi\,(1-\theta)\;<\;0
\]
And a set of firms \(\mathcal C\) on which the shock lands more heavily by
the margin, \(\xi_{\mathcal C}^{(\underline k)}\ge c_\xi(1-\theta)\),
expands in the sales-weighted aggregate,
\[
\sum_{i\in\mathcal C}m_i^\ast\,D_\pi\log q_{i,\underline k}
\;\ge\;
c_\xi\,(1-\theta)\,e_{\mathcal C}^{(\underline k)}
\;>\;0
\]
For a negative shock all signs reverse.
\end{proposition}

\noindent
Proof in Appendix~\ref{app:proof_impact_reallocation}.

The simple reading of Proposition~\ref{prop:impact_reallocation} is
that a positive monetary shock moves activity toward final demand at
impact. A firm on which the shock lands more heavily than on the other
buyers of its shortest-lead-time inputs outbids them for outputs that
were fixed before the shock. The bound on the contraction of a firm is proportional to its shortest-lead-time share \(s_i^{(\underline k)}\). This is
because only that part of its bundle was ordered after the shock,
\(\phi_i=s_i^{(\underline k)}\phi_i^{(\underline k)}\). The last statement of the proposition, on a set of firms, applies the incidence--position sorting at any point
of the supply chain. Consider the firms whose downstreamness exceeds some
level. What these firms gain at impact is what the firms further up
the supply chain, and the household, give up in the markets they share. Their gain is positive whenever the shock lands more heavily on the buyers nearer final demand in those markets than on the other buyers there, whatever it does elsewhere. Taking the level just below the downstreamness of the least downstream retailer gives the retail tier.

\subsection{Persistence of the impact reallocation}
\label{subsec:finite_horizon_persistence}

A positive monetary shock expands activity downstream at impact, and
it contracts the firms on which it lands more lightly than on the other buyers of their shortest-lead-time inputs. A negative shock generates the mirror image. The question is how long
these signs survive as the shock propagates through the production
network, and why they might reverse. As it turns out, what can reverse these signs is the upstream contraction itself. The mechanism is as follows. The firms far
from final demand order less at impact when the shock lands more lightly on them than on the other buyers of their inputs. Their reduced output travels down the supply chain, one
lead time per link, until it enters the output of the retail tier.
If this reduced output arrives while the retailers' impact-date orders are still maturing, it offsets part of the downstream expansion. What matters is how long the output of a firm takes to reach final demand. Let \(F_i(T)\) be the share of firm \(i\)'s output that
reaches final demand within \(T\) periods. A unit that firm \(i\) sells to the household reaches final demand immediately. A unit that it sells to
a firm \(j\) arrives there after the lead time \(k_{ij}\) of that
purchase, enters firm \(j\)'s output, and moves on in the same way.
Which means that for \(T\ge0\)
\begin{equation}
F_i(T)
=
r_i+\sum_{j\in N}\Omega_{ij}\,F_j(T-k_{ij}),
\qquad
F_i(T)=0\ \text{ for }T<0
\label{eq:delivery_time}
\end{equation}
with \(r_i\) the consumption share of firm \(i\) and
\(\Omega_{ij}\) the share of its output sold to firm \(j\). The share \(F_i(T)\)
rises with \(T\) to one. The function \(T\mapsto F_i(T)\) is the
distribution of the delivery time of firm \(i\)'s output to final
demand, the position of the firm along the supply chain measured in periods
rather than in the discount of Definition~\ref{def:downstreamness}.
Followed back through the chains of purchases, each unit of the household's expenditure pays firm \(i\), in total, its Domar weight \(m_i^\ast/w^\ast\). Of this total, the part spent on output that
\(i\) delivers to final demand within \(T\) periods is
\((m_i^\ast/w^\ast)F_i(T)\) (Lemma~\ref{lem:gdp_pass_through} of
Appendix~\ref{app:lemmas_reallocation}).

The first-order response of the position-tilted log-output index at the dates after impact is governed by the same relative expenditure impulses, evaluated
through the pipeline at the dates on which the maturing orders were
placed. After impact the excess balances that the monetary expansion
placed on small firms move, round by round, to their suppliers. This
is because the cross-sectional balance misalignment propagates by the exact law
\(\dot{\mathbf m}^\perp_{t+1}=\mathbf A\dot{\mathbf m}^\perp_t\). At an order date \(s\) after impact the expenditure of a buyer \(i\) therefore
moves by \(\dot{\widehat e}_{i,s}\) per unit of shock, which reflects the part of this excess that the buyer holds at date \(s\). The demand for each of its inputs
\(j\) moves by \(\dot{\widehat d}_{j,s}\), which reflects the parts of the excess held by the other buyers of that input. Let
\begin{equation}
\Xi_{ij,s}:=\dot{\widehat e}_{i,s}-\dot{\widehat d}_{j,s}
\label{eq:dated_contrast}
\end{equation}
be the \emph{dated contrast} of buyer \(i\) on supplier \(j\) at order
date \(s\), the date-\(s\) counterpart of the difference
\(\sigma_i-\overline\sigma_j\) that sets the impact-date orders in
\eqref{eq:impact_forcing_def}. The dated contrasts at date \(s\) are set by the misalignment of balances after \(s\) rounds of purchasing. Evaluated at the order dates of a bundle, the dated contrasts are the monetary wave, \(M_{ji,t}=\Xi_{ij,t-k_{ji}}\). The first-order response of the position-tilted log-output index at
every date is a non-negative combination of dated contrasts,
\begin{equation}
D_\pi\mathcal Q_t(b)
=
\sum_{s=0}^{t}\;\sum_{(i,j):\,j\in\mathcal S_i}
\varpi^{(t,s)}_{ij}(b)\,\Xi_{ij,s}
\label{eq:forced_representation}
\end{equation}
with weights \(\varpi^{(t,s)}_{ij}(b)\ge0\) whose total at any date is at
most \((1-\underline\beta)/\underline\beta\)
(Lemma~\ref{lem:supplier_recursion}). The weight \(\varpi^{(t,s)}_{ij}(b)\) is the tilt mass that reaches the link from buyer \(i\) to supplier \(j\) at order date \(s\) along paths of suppliers whose lead times sum to \(t-s\). The paths of one link carry the dated contrasts of the buyer itself on its suppliers, at the dates on which it ordered from them. The longer paths carry the dated contrasts of its suppliers on their own suppliers, and so on up the supply chain, each at the date on which the corresponding order was placed. Every path is discounted by the pass-through factors
\(1-\beta\) of the firms along it, and this discount bounds the total
tilt mass at a date.

Note that the impact-date relative expenditure impulse \(\xi_{i,\underline k}\) does not
disappear at the impact date. An order
placed at date \(0\) on a supplier of lead time \(k\) matures at date \(k\). The dated contrast of the impact date therefore enters production at date \(k\) at every firm with a supplier of lead time \(k\). The pipeline carries the impact-date contrasts to later dates one vintage at a time, until the orders of the longest lead time have matured. Write
\(I_t(b)\) for the part of \eqref{eq:forced_representation} that comes from the dated contrasts of the impact date, the terms with \(s=0\), and \(J_t(b):=D_\pi\mathcal Q_t(b)-I_t(b)\) for the part that comes from the dated contrasts of the dates after impact. Write \(t_{\mathcal W}:=\max\mathcal W\) for the last
retail lead-time date.

An input ordered at date \(0\) from a supplier of lead time \(k\) arrives at
date \(k\) and enters the buyer's output, which reaches final demand
in the delivery time of \eqref{eq:delivery_time}. Over an interval of
dates \([T_1,T_2]\), the expenditure on the purchases whose product
reaches the household within the interval is, relative to GDP,
\begin{equation}
\Gamma[T_1,T_2]
:=
\sum_{l\in N}\sum_{k\in\mathcal K_l}\Gamma_{lk}[T_1,T_2],
\qquad
\Gamma_{lk}[T_1,T_2]
:=
\frac{m_l^\ast}{w^\ast}(1-\beta_l)\,s_l^{(k)}\bigl[F_l(T_2-k)-F_l(T_1-1-k)\bigr]
\label{eq:delivered_purchases}
\end{equation}
Each purchase is counted by the share of the buyer's output delivered to
the household in the interval. For a set of firms \(\mathcal C\), \(\Gamma_{\mathcal C}[T_1,T_2]\) is the sum in
\eqref{eq:delivered_purchases} over the purchases of the firms of
\(\mathcal C\).

Call the firms within a fixed distance \(\Delta_+>0\) upstream of the
retail tier its \emph{neighborhood},
\begin{equation}
\mathcal G:=\{i\in N\setminus\mathcal R:\psi_i\ge\psi_R-\Delta_+\}
\label{eq:retail_neighborhood}
\end{equation}
with \(\psi_R\) the position of the least downstream retailer, and take the top tier to lie beyond the neighborhood, \(\psi_T<\psi_R-\Delta_+\). Call the retail
tier with its neighborhood, \(\mathcal R^{+}:=\mathcal R\cup\mathcal G\),
the \emph{neighborhood of final demand}. The shock may land more
lightly on a firm of the neighborhood than on the other buyers of its
inputs, but by no more than a tolerance
\(\overline\epsilon_\xi(1-\theta)\),\footnote{Firms that sell mostly to retailers lie just below the retail tier in downstreamness, since most
of what they make is one step from the household. Many of them are
small firms that buy fuel, materials and machinery in markets shared
with upstream firms, where they are the smaller buyers.
Condition~\eqref{eq:hypothesis_N} limits how much larger than the other
buyers of its inputs any firm of the neighborhood may be.}
\begin{equation}
\xi_i^{(k)}\;\ge\;-\overline\epsilon_\xi\,(1-\theta)
\quad\text{for every }i\in\mathcal G\text{ and every active vintage }k
\tag{N}
\label{eq:hypothesis_N}
\end{equation}
We also ask that the firms beyond the neighborhood of final demand
have delivered little to the household by the last retail lead-time date,
and that the retail slack \(1-\epsilon_h\) be small, by the thresholds
\(\overline\epsilon_F\) and \(\overline\epsilon_h\),
\begin{equation}
\Gamma_{N\setminus\mathcal R^{+}}[0,t_{\mathcal W}]\le\overline\epsilon_F,
\qquad
1-\epsilon_h\le\overline\epsilon_h
\tag{L}
\label{eq:hypothesis_L}
\end{equation}
The tolerance \(\overline\epsilon_\xi\) and the two thresholds are set
by the constants of the assumptions and the incidence range
(Appendix~\ref{app:reallocation}).

\begin{proposition}[Sign persistence over a finite horizon]
\label{prop:finite_horizon_sign_persistence}
Suppose that Assumptions~\ref{assump:diffuse_incidence},
\ref{assump:monotone_decay_spectrum}, \ref{assump:retail_outflow_cap}
and~\ref{assump:impact_forcing_ordering} hold, that
\(\underline k\ge1\), and that conditions~\eqref{eq:hypothesis_N}
and~\eqref{eq:hypothesis_L} hold. Then for every
downstream-concentrated tilt \(b\) and every date \(t\in\mathcal W\)
\begin{equation}
I_t(b)\;\ge\;\tfrac12\,c_R\,(1-\theta)
\label{eq:persistence_impact_component}
\end{equation}
and hence \(D_\pi\mathcal Q_t(b)\ge\tfrac14c_R(1-\theta)>0\)
at every date \(t\in\mathcal W\) at which the post-impact part of the
response is small, in the sense that
\begin{equation}
|J_t(b)|\;\le\;\tfrac14\,c_R\,(1-\theta)
\tag{P}
\label{eq:hypothesis_P}
\end{equation}
\end{proposition}

\noindent
Proof in Appendix~\ref{app:proof_finite_horizon_sign_persistence}.

The simple reading of Proposition~\ref{prop:finite_horizon_sign_persistence} is that the impact advantage of the retail tier holds on every retail lead-time date. The
persistence of the downstream expansion comes from the pipeline of
dated orders.\footnote{Condition~\eqref{eq:hypothesis_P}
bounds the post-impact part \(J_t(b)\) of the response. After impact
the excess balances move on from the retailers to their suppliers.
\(J_t(b)\) collects the dated contrasts of the retailers on their suppliers at the dates after impact, and those of the firms further up their supply chains. In words, condition~\eqref{eq:hypothesis_P} says that these dated contrasts do not outweigh the maturing impact advantage of the retail tier. The condition holds at the
shortest lead time, where \(J_{\underline k}(b)=0\), since no order placed
after impact has matured by then. At a later retail lead-time date it holds
when the orders placed after impact carry little of the tilt mass at
that date. It also holds when the excess balances they were paid from
are spread evenly across the buyers of each input
(Appendix~\ref{app:proof_finite_horizon_sign_persistence}).} At each retail lead-time date a vintage of the orders that the retail tier placed at impact matures, orders placed when the shock let it outbid the other buyers of those inputs. The bound \eqref{eq:persistence_impact_component} is
what this maturing impact advantage is worth once two leakages are
netted out. One is what the firms beyond the neighborhood of final demand have
delivered to the household by the retail lead-time dates, and the other is
what retailers sell to other firms.

\subsection{Reversal of the impact reallocation}
\label{subsec:sign_reversing_irf}

A positive monetary shock keeps downstream output above its stationary level on the retail lead-time dates. But firms at the upstream end of the supply chain order less at impact when the shock lands more lightly on them than on the other buyers of their inputs, and these orders lower their output as they mature.
We now turn to the dates beyond the retail lead-time dates. Upstream output is the input of downstream output through
long chains of buyer--seller relations, so the
reduction in upstream production arrives at downstream firms, after the lead times of those relations, as a
shortage of dated inputs. The initial expansion of downstream output can then give way to a contraction. Whether it does depends on when the shortage arrives, on how much of it survives the passage down the supply chain, and on whether what survives outweighs what still keeps downstream output above its stationary level. We date the
arrival of the shortage, bound what survives of it, and state the comparison with what still keeps downstream output above its stationary level as an explicit condition.

Note that the reversal of the downstream expansion is driven by two
different forces, which equation \eqref{eq:forced_representation}
separates. One of them originates in
the monetary wave. At impact the excess balances are held by the small, predominantly downstream firms that received them. One round of purchasing later they are held by the suppliers of those firms, and so on up the supply chain. The dated contrasts of a retailer on its suppliers, positive at impact, therefore fall as the excess balances move from the retailer to the other customers of its suppliers, who then outbid it for the same goods. How far these dated contrasts fall, and whether they change sign, is a property of the particular supply chain, given by equation \eqref{eq:forced_representation} and the path of the misalignment of balances. At the dates beyond its longest lead time the bundle of a retailer contains only orders placed after the impact date. The dated contrasts that enter that bundle are therefore those of the dates after impact.

The other force driving the reversal is the propagation down the supply chain of the contraction of the firms at its upstream end. The top tier
\(\mathcal T\) orders less at impact when the shock lands more lightly
on its firms than on the other buyers of their inputs by the margin
\(c_\xi\) at every vintage,
\begin{equation}
\xi_i^{(k)}\;\le\;-c_\xi\,(1-\theta)
\qquad\text{for every }i\in\mathcal T\text{ and every active vintage }k
\tag{U}
\label{eq:hypothesis_U}
\end{equation}
These orders then lower its output as they mature, as in
Proposition~\ref{prop:impact_reallocation}, and we call the top tier
the contracting region when condition~\eqref{eq:hypothesis_U} holds.
For a set of firms \(\mathcal C\), write
\begin{equation}
I_t^{\mathcal C}(b)
:=
\sum_{(i,j):\,i\in\mathcal C}\varpi^{(t,0)}_{ij}(b)\,\Xi_{ij,0},
\qquad
\mu_t(\mathcal C)
:=
\sum_{(i,j):\,i\in\mathcal C}\frac{\varpi^{(t,0)}_{ij}(b)}{1-\beta_i}
\label{eq:arriving_contraction}
\end{equation}
for the part of the impact-date component \(I_t(b)\) that the dated contrasts of the impact date at buyers in \(\mathcal C\) carry to date \(t\),
and for the tilt mass that carries this part. The tilt mass is scaled
by the buyers' intermediate-input shares. The date-\(t\) response \eqref{eq:forced_representation} then
splits into three parts, the part carried by the contracting region,
the part carried by the other firms and the post-impact part,
\(D_\pi\mathcal Q_t(b)=I_t^{\mathcal T}(b)+I_t^{N\setminus\mathcal T}(b)+J_t(b)\).
We call \(-I_t^{\mathcal T}(b)\) the \emph{arriving contraction}. The
constant \(C_\phi\) of \eqref{eq:forcing_upper_bound} bounds the scaled
relative expenditure impulses, \(|\phi_i^{(k)}|\le C_\phi(1-\theta)\) at every firm and vintage.

\begin{theorem}[Two-phase reallocative impulse response]
\label{thm:sign_reversing_irf}
Suppose that Assumptions~\ref{assump:diffuse_incidence},
\ref{assump:monotone_decay_spectrum}, \ref{assump:transmission_delay},
\ref{assump:retail_outflow_cap}, \ref{assump:lag_position_sorting}
and~\ref{assump:impact_forcing_ordering} hold, that
\(\underline k\ge1\), and that conditions~\eqref{eq:hypothesis_U},
\eqref{eq:hypothesis_N} and~\eqref{eq:hypothesis_L} hold. Then for
every downstream-concentrated tilt \(b\) the first-order response of
the position-tilted log-output index is positive,
\(D_\pi\mathcal Q_t(b)\ge\tfrac14c_R(1-\theta)>0\)
at every date \(t\in\mathcal W\) at which condition~\eqref{eq:hypothesis_P}
of Proposition~\ref{prop:finite_horizon_sign_persistence} holds. There
are a horizon \(\overline\tau\) and a constant \(\underline\mu>0\), both set by the
constants of the assumptions, such that at some date
\(\tau^-\in[\underline k+k_\ast,\overline\tau]\), beyond every retail
lead-time date, the arriving contraction is bounded below,
\[
-I_{\tau^-}^{\mathcal T}(b)\;\ge\;\underline\mu\,(1-\overline\beta)\,c_\xi\,(1-\theta)\;>\;0
\]
And the horizon \(\overline\tau\) is a bounded multiple of the longest lead time
\(\widehat k\). Suppose that at the date \(\tau^-\) the tilt mass on the contracting region's impact-date
orders outweighs the tilt mass on the other impact-date orders together
with the post-impact part of the response,
\begin{equation}
(1-\overline\beta)\,c_\xi\,(1-\theta)\,\mu_{\tau^-}(\mathcal T)
\;>\;
C_\phi\,(1-\theta)\,\mu_{\tau^-}(N\setminus\mathcal T)+J_{\tau^-}(b)
\tag{R}
\label{eq:hypothesis_R}
\end{equation}
Then \(D_\pi\mathcal Q_{\tau^-}(b)<0\), so that the early expansion is
followed by a reversal. For a shock of the opposite sign the first-order
response is the mirror image.
\end{theorem}

\noindent
Proof in Appendix~\ref{app:proof_sign_reversing_irf}.

The simple reading of Theorem~\ref{thm:sign_reversing_irf} is that the
impact reallocation carries the seed of its own reversal, and the
dated inputs are what carry it. The early expansion is produced with
inputs bid away from upstream firms, and those inputs reach the
retailers over the lead times on which they were ordered. The shortage that
ends the expansion is the reduced output that the upstream firms, deprived of inputs, go on to deliver, and that output reaches the
retail tier only after the lead times of the supply chain. Without dated inputs
money would still move activity at impact, but no pipeline of
outstanding orders would carry the expansion to later dates or bring the shortage down the supply chain. The two phases, the early expansion and the
reversal, are the borrowing from the future that the introduction
describes, written in the model's own terms.

Upstream sourcing (Assumption~\ref{assump:transmission_delay}) sets the horizon \(\overline\tau\). The shortage of upstream output travels one link per
lead time, and no link takes longer than the longest lead time. Upstream sourcing bounds the number of rounds of sourcing that separate the retail tier from the contracting region, so the contraction reaches the
retail tier within a bounded, datable horizon.
The time to build (Assumption~\ref{assump:lag_position_sorting}) sets the earliest date at which the contraction can arrive. Of the orders that the contracting region placed at impact, all but a
share \(\epsilon_k\) mature \(k_\ast\) periods or more later, and the reduced output they produce reaches a customer one lead time after
that. The contraction they carry arrives no earlier than
\(\underline k+k_\ast\), when every retail lead-time date has passed, and the
date \(\tau^-\) of the theorem is one at which at least the mass
\(\underline\mu\) of it has arrived
(Lemma~\ref{lem:layered_routing}). The expansion runs on the retail lead-time
dates and the reversal falls after them. The order of the two phases is set by the lead times of the retail tier and of the top tier, and the lead times of the firms in between do not enter it.

Whether the reversal occurs turns on
condition~\eqref{eq:hypothesis_R}. Its left side is the tilt mass \(\mu_{\tau^-}(\mathcal T)\) that the orders placed at impact by the firms of the contracting region carry to the date \(\tau^-\), at least \(\underline\mu\), times the margin \((1-\overline\beta)c_\xi\) by which those firms contract. The mass
\(\underline\mu\) is what survives the passage from the retail tier to
the contracting region, discounted by the pass-through at every round
of sourcing and by the share \(\epsilon_k\) of the orders placed at impact by the firms of the contracting region that mature before \(k_\ast\).
Condition~\eqref{eq:hypothesis_R} is therefore easiest to meet in short
supply chains with broad sourcing, where few rounds of sourcing lead from the
retail tier to the contracting region. The right side of
condition~\eqref{eq:hypothesis_R} bounds what else moves the position-tilted log-output index at the date \(\tau^-\), namely the impact-date orders of the firms outside the
contracting region, which carry the mass
\(\mu_{\tau^-}(N\setminus\mathcal T)\), and the post-impact part
\(J_{\tau^-}(b)\). The post-impact part \(J_{\tau^-}(b)\) reinforces the reversal when the dated contrasts that it collects at the dates after impact are negative. For a retailer this happens once the excess balances have moved from it to the other customers of its suppliers, who then outbid it for the same goods.

\section{Monetary Non-neutrality}
\label{sec:asymmetry}

A monetary shock misaligns a firm's dated input bundle and, at the same
time, moves purchasing power along the supply chain. The retailer that gains balances relative to the other buyers of its inputs, and so bids away from them a bundle larger than its stationary one, assembles that bundle from orders placed before and after the shock. The scale of its bundle and the composition of that bundle move together, and both affect what the household consumes. We now derive how
the deadweight loss and the reallocation combine in GDP along the transition from the old stationary equilibrium to the new one, and why their combined
response differs with the sign of the shock.

Throughout, GDP is the household-consumption aggregate
\begin{equation}
\Delta\log Y_t
:=
\sum_{i\in\mathcal R}\gamma_i\,\Delta\log c_{i,t}
\label{eq:gdp_def_main}
\end{equation}
the log change of the Cobb--Douglas consumption index at the household
shares, summed over the firms the household buys from. With the
household's Cobb--Douglas preferences, this aggregate is also the log
change in its flow utility. GDP is measured over intervals of dates, a
quarter at a time. And GDP over an interval is the average of the log
changes \(\Delta\log Y_t\) over its dates.\footnote{A quarterly figure
sums the output of the quarter, so its proportional change is the
average of the proportional changes \(Y_t/Y^\ast-1\) of its dates rather
than of their logarithms. The two agree to first order in the shock.}

Labor is fixed along the transition, so the
goods produced at a date are either consumed or ordered into the
pipeline again. Which means that the GDP response
\(f_t(\pi):=\Delta\log Y_t(\pi)\) at a date is what the economy draws
out of the goods in transit less what it wastes,
\begin{equation}
f_t(\pi)=\Pi_t(\pi)-\Lambda_t(\pi),
\qquad
\Pi_t:=\frac{P_t-P_{t+1}}{w^\ast}
\label{eq:resource_identity}
\end{equation}
Here \(P_t\) is the excess value of the goods in transit, the value at
stationary prices of the orders outstanding at date \(t\) in excess of the stationary pipeline. And \(\Pi_t\) is the drawdown of the goods in transit, the fall in their excess value from date \(t\) to date \(t+1\). The
\emph{deadweight loss} \(\Lambda_t\) is the output and the welfare that
the economy wastes at that date, in units of stationary GDP. It is nonnegative at every date and has two parts. The
first is the loss from production inefficiency, the output lost because
the inputs within a firm's bundle are out of their stationary
proportions with one another. The second is the \emph{loss from
misallocation}, the output lost because a firm's bundle as a whole is
larger or smaller than at the stationary equilibrium, relative to the
workers it employs. The loss from misallocation also includes the welfare lost because
consumption across goods and dates departs from its stationary
proportions.\footnote{Appendix~\ref{app:lemmas_asymmetry} defines the
excess value of the goods in transit and the deadweight loss, splits the
deadweight loss into its two parts and proves the identity
\eqref{eq:resource_identity} (Lemma~\ref{lem:resource_accounting_date}).
To second order in the shock, the first part is the loss from production
inefficiency of every firm, weighted by its Domar weight. The second is
a weighted sum of the squared changes in the scale of the bundles and in
consumption.}

Summed over an interval of dates that begins with the shock, the
drawdowns of the goods in transit add up to the fall in their excess
value by the end of the interval. That excess value is zero at the
shock, for the orders then in transit were placed before it. GDP over
such an interval therefore measures what the economy has drawn out of
transit by the end of the interval, less what it has wasted over the
interval.

\subsection{Asymmetric monetary non-neutrality}
\label{subsec:asymmetric_nonneutrality}

At first order the deadweight loss vanishes, and GDP moves with the
reallocation by the same amount after a monetary expansion and after a
monetary contraction of the same size, with opposite signs. The
difference between the responses to the monetary expansion and to the
monetary contraction is therefore of second order in the shock. It comes
from two forces acting together, the first of which is the deadweight
loss. The deadweight loss lowers GDP whatever the sign of the shock. The
second force is the reallocation, and it is larger after a monetary
contraction than after a monetary expansion of the same size. This is
because, measured against the money stock after the shock, the monetary
contraction shifts more balances across firms than the monetary
expansion, and so moves more output along the supply chain. Which means that on the early
dates at which the reallocation raises GDP, a monetary contraction
lowers GDP by more than a monetary expansion of the same size raises it.

A response that rises at first order in the shock and is concave at
second order is smaller in absolute value after a monetary expansion
than after a monetary contraction of the same size. Whether the
money--output relation is asymmetric therefore turns on the signs of the
first and second derivatives of the GDP response at zero,
\begin{equation}
f_t(\pi)=\dot f_t\,\pi+\tfrac12\,\ddot f_t\,\pi^2+o(\pi^2),
\qquad
\dot f_t=A_t+\dot\Theta_t
\label{eq:gdp_expansion_coefficients}
\end{equation}
The first derivative \(\dot f_t\) is the reallocation effect as it
appears in GDP. It is the sum of the \emph{reallocation impulse}
\(A_t:=\sum_{i\in\mathcal R}\gamma_i(1-\beta_i)\,\mathcal F_i(\dot{\widehat x}_{i,t})\),
the first-order response of the retailers' input-bundle scales at the
household shares, and of the first-order part \(\dot\Theta_t\) of the
\emph{consumption--output wedge}
\(\Theta_t:=\sum_{i\in\mathcal R}\gamma_i\,\Delta\log(c_{i,t}/q_{i,t})\).
The consumption--output wedge is of the order of the retail slack, for a
retailer sells at most that share of its output to other firms
(Lemma~\ref{lem:gdp_labor_income_anchor}).

The reallocation is larger after a monetary contraction because real
quantities respond to the shift of balances across firms measured
against the money stock after the shock. A monetary expansion of size \(\pi\) raises every nominal
magnitude by the factor \(1+\pi\) and, in addition, moves balances
across firms by \(\pi h_\theta\). Here
\(h_\theta:=\overline M\boldsymbol\zeta-\mathbf m^\ast\) is the impact
misalignment, the value \(\dot{\mathbf m}_0^\perp\) of the
cross-sectional balance misalignment at impact. Dividing every nominal
magnitude by \(1+\pi\) leaves every real order, output and consumption
unchanged. What remains is an economy at the stationary wage whose
balances start at \(\mathbf m^\ast+\tilde\pi\,h_\theta\), with
\(\tilde\pi:=\pi/(1+\pi)\) the \emph{normalized shock}, the size of the
shock measured against the money stock after it. The real response is
therefore a fixed function of \(\tilde\pi\,h_\theta\). For every date
\(t\) there is a map \(\mathcal Y_t\) with \(\mathcal Y_t(\mathbf 0)=0\),
independent of \(\pi\), such that
\begin{equation}
f_t(\pi)=\mathcal Y_t\bigl(\tilde\pi\,h_\theta\bigr),
\qquad
\dot f_t=D\mathcal Y_t(\mathbf 0)\,h_\theta,
\qquad
\ddot f_t=-2\,\dot f_t+D^2\mathcal Y_t(\mathbf 0)[h_\theta,h_\theta]
\label{eq:concave_scaling}
\end{equation}
(Lemma~\ref{lem:concave_scaling}).

A monetary expansion shifts balances within a money stock that has grown
by the factor \(1+\pi\), and a monetary contraction of the same size
shifts them within one that has shrunk by the factor \(1-\pi\). Either way, the reallocation changes the real output of downstream
firms in proportion to the shift of balances, measured against the money
stock after the shock. The gain in real output after a
monetary expansion therefore grows less than in proportion to the
monetary expansion, and the fall after a monetary contraction grows more
than in proportion to the monetary contraction. In \eqref{eq:concave_scaling} the larger reallocation after a monetary
contraction appears as the curvature \(-2\dot f_t\), which the
normalized shock contributes on its own. Every other second-order effect of the network, the deadweight loss
among them, enters the residual curvature \(D^2\mathcal Y_t(\mathbf 0)[h_\theta,h_\theta]\) of the GDP response. This residual curvature is of the order of the square of the slack, for it is
quadratic in the impact misalignment, and the impact misalignment is of
the order of the slack. At a date at which a monetary expansion raises
GDP by a margin of the order of the slack, a small slack therefore makes
the second derivative \(\ddot f_t\) of the GDP response
negative.\footnote{The loss from production inefficiency of the
retailers and the curvature of their log output in the scale of their
bundles both lower \(\ddot f_t\). But the curvatures in the shock of
their input-bundle scales and of the consumption--output wedge can have either sign
(Lemma~\ref{lem:gdp_expansion}). The impact misalignment obeys
\(\|h_\theta\|_m\le C_\zeta(1-\theta)\) by
\eqref{eq:impact_misalignment_constant}. At a date with
\(\dot f_t\ge c\,(1-\theta)\) and \(c>0\), a slack below
\(c/(KC_\zeta^2)\), with \(K\) a bound on the Hessian of
\(\mathcal Y_t\) at zero, gives \(\ddot f_t\le-c\,(1-\theta)\)
(Lemma~\ref{lem:concave_scaling}). To leading order in the slack, the
GDP response is then the first-order response applied to the normalized
shock, \(f_t(\pi)=\dot f_t\,\pi/(1+\pi)\), up to a remainder of the
order of \((1-\theta)^2\pi^2\). At the proportional benchmark
\(\theta=1\) the impact misalignment vanishes, and so do \(\dot f_t\) and
\(\ddot f_t\).}

Such a margin appears over the first dates at which the shock reaches
production. At impact the inputs of every firm are already in transit,
for they were ordered before the shock. No firm therefore changes its
output before the orders it places at impact mature. Before the
shortest lead time \(\underline k\), GDP moves only through the
consumption--output wedge, the change in the share of retail output that
the household buys. At the shortest lead time the first vintage of the
orders that the retail tier placed at impact matures. The retail tier placed them when the shock let it outbid the
other buyers of those inputs, and GDP now takes up this impact advantage.
Nothing else has reached the household by then, neither an order placed
after impact nor the changed output of a supplier.

We call the dates from the shock to the shortest lead time, \(0\) to
\(\underline k\), the \emph{first reporting interval}. Write \(f_+(\pi)\) for GDP over
the first reporting interval, the average of the GDP responses
\(f_t(\pi)\) at its dates, and \(\Lambda_+(\pi)\) for the deadweight loss
averaged in the same way. The impact advantage of the retail tier enters GDP
at the household's weights rather than at the retailers' sales, and the
two weightings differ by the retail slack. Net of this difference in
weights and of the consumption--output wedge, the impact advantage of the
retail tier bounds the first-order response of GDP over the first
reporting interval from below. We call the coefficient of the slack
\(1-\theta\) in this lower bound the \emph{GDP margin} \(c_Y\). The GDP margin is positive when the retail slack is small
relative to the retail margin \(c_R\)
(Lemma~\ref{lem:first_interval_gdp}). The excess value \(P_t\) of the
goods in transit is a fixed function of the normalized shock, as GDP is,
for it values real quantities at fixed prices
(Lemma~\ref{lem:resource_accounting_date}). Let \(K_+\) bound the Hessian
at zero of the average of the maps \(\mathcal Y_t\) over the first
reporting interval, in the norm of Lemma~\ref{lem:concave_scaling}. And
let \(K_P\) bound, in the same norm, the Hessian at zero of the excess
value of the goods in transit at the end of the interval, per date of
the interval and in units of stationary GDP.

\begin{theorem}[Asymmetric monetary non-neutrality]
\label{thm:asymmetric_nonneutrality}
Suppose that Assumptions~\ref{assump:diffuse_incidence},
\ref{assump:monotone_decay_spectrum}, \ref{assump:retail_outflow_cap}
and~\ref{assump:impact_forcing_ordering} hold, that
\(\underline k\ge1\), that the shortest lead time is a retail lead-time
date, \(\underline k\in\mathcal W\), and that the GDP margin is positive,
\(c_Y>0\). Let the slack satisfy \(1-\theta\le c_Y/(K_+C_\zeta^2)\). Then
there is \(\pi_\ast>0\) such that, for every \(0<\pi\le\pi_\ast\), the
monetary expansion raises GDP over the first reporting interval, the
monetary contraction lowers it, and the response to the monetary
contraction is the larger in absolute value,
\begin{equation}
f_+(\pi)>0>f_+(-\pi),
\qquad
|f_+(-\pi)|>|f_+(\pi)|
\label{eq:asymmetry_inequality_main}
\end{equation}
Suppose that the slack also satisfies \(1-\theta\le c_Y/(K_PC_\zeta^2)\).
Then \(\pi_\ast\) can be chosen so that, in addition, the response to the
monetary contraction exceeds that to the monetary expansion in absolute
value by more than the sum of their deadweight losses,
\[
|f_+(-\pi)|-|f_+(\pi)|\;>\;\Lambda_+(\pi)+\Lambda_+(-\pi)
\]
\end{theorem}

\noindent
Proof in Appendix~\ref{app:proof_asymmetric_nonneutrality}.

The simple reading of Theorem~\ref{thm:asymmetric_nonneutrality} is
that money is non-neutral, and asymmetrically so, although every goods
price clears its market at every date. Nothing in the adjustment of
prices or wages differs between a monetary expansion and a monetary
contraction. The larger reallocation after a monetary contraction
creates the difference between the responses to a monetary expansion
and to a monetary contraction, and the deadweight loss widens that
difference.\footnote{Of the deadweight loss of a date, only the part
that falls on the output the household consumes, the loss on consumed output, widens this difference at that date. The other part, the loss in
transit, falls on the output that goes into transit. It lowers the goods
put into transit by its full amount and raises the drawdown \(\Pi_t\) of the goods in transit by as much,
so it leaves GDP at the date unchanged. Summed over the transition, it
lowers GDP by its full amount (Lemma~\ref{lem:consumed_transit}).} The
second part of the theorem says more. When the slack is small enough,
the monetary contraction puts more goods into transit over the first
reporting interval than the monetary expansion draws out of transit. Which means that
the deadweight loss of the monetary expansion and the monetary
contraction is only a part of the difference between their responses.

The asymmetry between the responses to a monetary expansion and to a
monetary contraction of the same size is carried by the even part of the
GDP response, \(\overline f_t:=\tfrac12\bigl(f_t(\pi)+f_t(-\pi)\bigr)\).
This is because the first-order response cancels when the two responses
are added. The identity
\eqref{eq:resource_identity} splits the even part of the GDP response exactly into the even parts of
the drawdown \(\Pi_t\) of the goods in transit and of the deadweight loss
\(\Lambda_t\). The first derivative of the drawdown of the goods in transit
at zero is \(\dot f_t\), for the deadweight loss is
second order in the shock. Then
\begin{equation}
-\overline f_t
\;=\;
\underbrace{\tfrac12\bigl[\Lambda_t(\pi)+\Lambda_t(-\pi)\bigr]}_{\text{deadweight loss}}
\;+\;
\underbrace{\dot f_t\,\pi^2}_{\text{normalization term}}
\;-\;
\underbrace{\tfrac12\pi^2\,D^2\mathcal Y^{\Pi}_t(\mathbf 0)[h_\theta,h_\theta]}_{\substack{\text{curvature term of the drawdown}\\\text{of the goods in transit}}}
\;+\;o(\pi^2)
\label{eq:even_part}
\end{equation}
(Lemma~\ref{lem:resource_accounting_date}). The negative of the even part of the GDP response is the deadweight loss,
averaged over the two signs of the shock, plus the normalization term,
less the curvature term of the drawdown of the goods in transit.
And this is precisely the sense in which the asymmetry comes from the
deadweight loss and the reallocation acting together.
The deadweight loss arises from the curvatures of the technology and of
the household's preferences, and it lowers GDP after a monetary
expansion and after a monetary contraction alike. The normalization term
is what the larger reallocation after a monetary contraction contributes
to the even part. It belongs to the reallocation
effect, and it would be present even with perfectly substitutable
inputs. The curvature term of the drawdown of the goods in transit moves second-order gains and
losses from one date to another, among them the losses of upstream
firms, which reach GDP only when their reduced output is delivered
downstream. It can have either sign, but it is second order in the
slack.\footnote{The same decomposition holds exactly for every shock
with \(|\pi|\le\overline\pi\) (Lemma~\ref{lem:gap_exact}). The
normalization term is then the first-order response times half the
difference in size between the normalized shocks of the monetary
contraction and of the monetary expansion. On its own, at a date with
\(\dot f_t>0\), the normalization term makes the ratio \(|f_t(-\pi)|/f_t(\pi)\) of the
absolute responses to the monetary contraction and to the monetary
expansion equal to \((1+\pi)/(1-\pi)\). The deadweight loss and the
curvature term of the drawdown of the goods in transit move this ratio by
a term of the order of \((1-\theta)\pi\) (Lemma~\ref{lem:concave_scaling}). Where the first-order response is
negative, the normalization term makes the rise in GDP after the
monetary contraction exceed the fall after the monetary expansion. What remains of the curvature term of the drawdown of the goods in transit, once the loss in transit cancels against it, can have either sign even at the shortest lead
time (Lemma~\ref{lem:shortest_lag_curvature} and
Example~\ref{ex:drawdown_curvature_sign}).}

The same two forces decide what larger shocks do to GDP. Beyond the
threshold \(\pi_\ast\), the identity \eqref{eq:resource_identity} says
that a monetary expansion raises GDP at a date exactly when the drawdown
of the goods in transit exceeds the deadweight loss at that size of the
shock. To leading order this comparison is between the reallocative
gain, the part of the reallocation effect that survives in GDP, and the
deadweight loss. The reallocative gain grows in proportion to the shock.
This is because the shift of nominal balances across firms is itself
proportional to the monetary expansion. The deadweight loss, by
contrast, grows with the square of the shock, for it is the curvature
of the technology and of the household's preferences acting on input and
consumption changes that are themselves proportional to the monetary
expansion. A monetary expansion large enough for the deadweight loss to
exceed the reallocative gain would therefore lower GDP, as a monetary
contraction does. But after a monetary contraction the reallocation
effect and the deadweight loss both lower GDP over the first reporting
interval.

The decomposition of the GDP response into the drawdown of the goods in
transit and the deadweight loss bears on how a GDP figure is to be
interpreted. When a monetary expansion raises GDP, it does so by
drawing resources toward downstream firms, where they weigh most
heavily in final expenditure. The same shock still inflicts a
deadweight loss, a loss from production inefficiency at the firms whose inputs no
longer fit together and a loss from misallocation across firms. And the
downstream gain is purchased in part through an upstream contraction.
A rise in measured GDP after a monetary expansion is therefore
consistent with a fall in upstream production and with inputs spread
less efficiently, within firms and across them, than at the stationary
equilibrium.

\subsection{The intertemporal effects of monetary shocks}
\label{subsec:intertemporal_reallocation}

Later in the transition the upstream contraction that a monetary
expansion sets off (Theorem~\ref{thm:sign_reversing_irf}) reaches the
firms whose output enters GDP. The deadweight loss then deepens the fall
in GDP that this contraction brings, for it is nonnegative at every
date.

What the early gain in GDP costs over the whole transition follows from the
resource identity \eqref{eq:resource_identity}, without any of the
reallocation hypotheses. GDP above its stationary level is financed by
running down the goods in transit, and the goods in transit return to
their stationary value as the transition dies out.

\begin{proposition}[Intertemporal resource accounting]
\label{prop:resource_accounting}
Suppose that Assumptions~\ref{assump:diffuse_incidence} and~\ref{assump:monotone_decay_spectrum}
hold. Then the excess value of the goods in transit starts and ends at
zero, \(P_0=0\) and \(P_t\to0\), and the first-order responses sum to zero,
\(\sum_{t\ge0}\dot f_t=0\). And the GDP responses sum to minus the deadweight
loss of the transition,
\begin{equation}
\sum_{t\ge0}f_t(\pi)
\;=\;
-\sum_{t\ge0}\Lambda_t(\pi)
\;\le\;0
\label{eq:cumulative_gdp}
\end{equation}
with strict inequality when household consumption departs from its
stationary value at some date.
\end{proposition}

\noindent
Proof in Appendix~\ref{app:proof_resource_accounting}.

The simple reading of Proposition~\ref{prop:resource_accounting} is
that a monetary shock cannot raise GDP cumulated over the whole
transition. At
first order the deadweight loss vanishes. And the drawdowns \(\Pi_t\) of the goods in transit sum to zero, for their excess value starts and ends at zero. The gain in GDP at one date is therefore matched, to first order, by a loss of the same size at other dates. What remains once the gains and losses of GDP at the different dates are netted is the deadweight loss of the transition, second
order in the shock. Summed over the transition,
the even part of the GDP response is minus the deadweight loss averaged
over the monetary expansion and the monetary contraction,
\(-\sum_{t\ge0}\overline f_t=\tfrac12\sum_{t\ge0}\bigl[\Lambda_t(\pi)+\Lambda_t(-\pi)\bigr]\),
by \eqref{eq:cumulative_gdp} at \(\pi\) and at \(-\pi\). The even part of the drawdown of the goods in transit, which carries the
normalization term and the curvature term of \eqref{eq:even_part}, sums
to zero as well, since these drawdowns sum to zero at every size of the
shock. And without
goods in transit there is nothing to draw down. With every lead time zero the pipeline is empty,
\(P_t\equiv0\), and \eqref{eq:resource_identity} says that GDP cannot
rise at any date (Lemma~\ref{lem:no_dated_inputs} of Appendix~\ref{app:lemmas_reallocation}).

Under the conditions of Theorem~\ref{thm:asymmetric_nonneutrality}, a
monetary expansion raises GDP over the first reporting interval. By
\eqref{eq:cumulative_gdp}, GDP cumulated over the rest of the transition
then falls by at least as much as GDP cumulated over the first reporting
interval rose, for the
deadweight loss of the transition is nonnegative.

Write
\[
G_Y(\pi):=\sum_{t\ge0}\max\{f_t(\pi),0\},
\qquad
L_Y(\pi):=\sum_{t\ge0}\max\{-f_t(\pi),0\},
\qquad
\mathcal I(\pi):=\frac{L_Y(\pi)}{G_Y(\pi)}-1
\]
for the cumulative gain of the monetary expansion, its cumulative loss, and the
implicit interest rate.

\begin{proposition}[Implicit interest cost of a monetary expansion]
\label{prop:implicit_interest}
Suppose that Assumptions~\ref{assump:diffuse_incidence} and~\ref{assump:monotone_decay_spectrum}
hold and that GDP moves at some date, \(f_t(\pi)\ne0\) for some
\(t\). Then the cumulative loss exceeds the cumulative gain by the deadweight
loss of the transition,
\[
L_Y(\pi)-G_Y(\pi)\;=\;\sum_{t\ge0}\Lambda_t(\pi)\;>\;0
\]
If the GDP responses accumulated up to some date \(T_0\) are positive,
\(\sum_{t\le T_0}f_t(\pi)>0\), then GDP lies below its stationary level
at some date after \(T_0\), within a horizon set by that accumulated
response and the constants of Proposition~\ref{prop:stability_stationary}.
If the first-order response is positive at some date, \(\dot f_{t_1}>0\),
then for every sufficiently small \(\pi>0\) the implicit interest rate is positive
and of the order of the shock, \(0<\mathcal I(\pi)=O(\pi)\).
\end{proposition}

\noindent
Proof in Appendix~\ref{app:proof_implicit_interest}.

The simple reading of Proposition~\ref{prop:implicit_interest} is
that the economy repays what a monetary expansion lets it borrow,
with interest. The implicit interest rate \(\mathcal I(\pi)\) is measured in summed log
deviations, neither annualized nor discounted. Its name comes from the
analogy with a loan. The early gain is output moved into the present by running down the goods in transit, and the later loss repays that output with interest. At first order the economy repays exactly what it borrowed.
The interest paid is the deadweight loss of the transition, the loss from production
inefficiency and the loss from misallocation summed over its dates. And the interest grows with the square of the shock, while the early gain grows in proportion to the shock. Under the conditions of Theorem~\ref{thm:asymmetric_nonneutrality}, the
first-order response of GDP is positive at some date of the first
reporting interval, so the implicit interest rate is positive and of the
order of the shock.

\section{Computational Experiments}
\label{sec:quant}
\label{sec:quant-economies}
\label{sec:quant-validation}
\label{sec:quant-setup}
\label{sec:quant-incidence}

Our analytical results hold for small shocks, at fixed supplier shares,
under a condition on the spectrum of the supplier-share matrix, and on
networks whose range of firm sizes is bounded. The computational
experiments of this section relax these restrictive assumptions. We also measure the size of monetary non-neutrality in a
calibrated economy of the size and structure of the United
States.\footnote{The measurement is a computational experiment in the
sense of \citet{kydland1996computational}. The network, the firm sizes and the within-vintage curvature are taken
from data, and the mapping from the monetary shock to the response of
GDP is that of the model.}

We work with two kinds of supply chains, the first of which consists of synthetic supply chains built in tiers, with lead times that lengthen away
from final demand and with firm size falling toward final demand. Put simply, the firms in the synthetic supply chains are placed with a granular structure that meets the theoretical setting.\footnote{Firms are placed in five tiers, from the top of the
supply chain to the retail tier. Each firm draws its suppliers from the tiers upstream of its own, with a probability that decays in the
distance up the supply chain, and its expenditure shares across them
are independent exponential draws normalized to sum to one. Firm sizes
within a tier are Pareto, and the household buys from the retail tier
with weights that are exponential draws normalized to sum to one. The
rank correlation between firm size and downstreamness is \(-0.35\).} The second is the US supply chain, reconstructed using the algorithm of \citet{bhattathiripad2026reconstructing}. The supply chain generated by their algorithm has no lead times for the delivery of intermediate inputs. We therefore impose a distribution of lead times on the reconstructed supply chain, with shorter lead times for buyers closer to final demand. The degrees,
supplier shares and firm sizes of the reconstructed US supply chain are heavy-tailed.\footnote{Table~\ref{tab:baseline} of Appendix~\ref{app:construction} lists every
parameter of both supply chains.}

We generate the synthetic supply chains at five sizes, from \(10^3\) to
\(10^5\) firms, with twenty independent realizations at each size,
obtained from different random seeds. And we reconstruct the US supply
chain at four sizes, from \(10^4\) firms to the full census of
\(6.46\times10^{6}\) firms. At the baseline the within-vintage curvature
is set to \(\rho=-0.8\), the value implied by the most direct firm-level
estimate of input substitutability (Appendix~\ref{app:calibration}).\footnote{We set the supplier
weights of the within-vintage nest \eqref{eq:within_vintage_bundle} on
both supply chains to \(\nu_{ji}=a_{ji}^{\,1-\rho}\), so that the
cost-minimizing expenditure shares equal \(a_{ji}\) at equal input
prices for every curvature \(\rho\). A sweep over \(\rho\) then varies
complementarity while holding technology fixed.}

Before the shock, every economy is simulated until it reaches its
stationary equilibrium. Labor is then held at its stationary value at
every firm and the wage is indexed to the money stock, as in
Assumption~\ref{assump:short_run_nominal_rigidity}, so that firms adjust
to the shock only through their intermediate inputs.\footnote{The
initial transient, \(300\) months and \(450\) at \(\rho\le-1.1\), is
discarded, and the response is recorded over a window of \(150\)
months.}

We then shock each economy once. The monetary expansion is \(5\%\) of the money
stock, \(\pi=0.05\), distributed by the size-biased incidence rule
\eqref{eq:zeta_def} at \(\theta=0.5\), and the monetary contraction is its mirror
image. Unlike in the model, the household holds balances and receives a share of the shock.\footnote{The household, with balance \(M_{\mathrm h}\),
receives the share
\(M_{\mathrm h}^{\theta_{\mathrm h}}/(M_{\mathrm h}^{\theta_{\mathrm h}}+\overline M^{\theta_{\mathrm h}})\)
of the shock, the two-agent form of the incidence rule
\eqref{eq:zeta_def}, with \(\theta_{\mathrm h}=0.9\). At the
stationary equilibrium \(M_{\mathrm h}/\overline M=0.4\) on every
supply chain, so the household's share of the shock is \(0.3\).} The
household buys from the
retail tier, so its share adds to the downstream tilt of the initial
demand shift. The incidence rule and the household's share are applied
in the same way to monetary expansions and to monetary contractions, so neither contributes
an asymmetry of its own.\footnote{Two further details of the protocol differ from the model. The transfer to each firm is
capped so that no balance moves by more than half of itself, under
either sign of the shock, and the cap binds only in the extreme tail of
the size distribution. And along the transition a firm's wage bill is kept between a fifth and four fifths of its balance. This bound never binds on the synthetic supply chain at the baseline, and on the US supply chain it binds at firms holding at most \(0.2\%\) of the money stock.} We report the percentage deviation of GDP, the household's consumption aggregate \eqref{eq:gdp_def_main}, from its pre-shock level, the median of the five months before the shock. Each quarterly point is the GDP response cumulated over the three months of the quarter and divided by three, the counterpart of GDP over an interval in the theory.

\subsection{Experiments on synthetic supply chains and the robustness of analytical results}
\label{sec:quant-robustness}
\label{sec:quant-mor}

We begin with the response of GDP on the synthetic supply chain, which
Figure~\ref{fig:quant_irfsyn} reports at its five sizes. It shows the two
phases of Theorem~\ref{thm:sign_reversing_irf}, an early rise and a
delayed fall below the pre-shock level. And the response to a monetary
contraction is larger in absolute value than the response to a monetary
expansion of the same size. After the monetary expansion GDP rises in the first quarter and falls
slightly below its pre-shock level in the second. From the third quarter
to the sixth it is above its pre-shock level, with a peak of \(0.2\%\)
in the fourth. It then falls below its
pre-shock level from the seventh quarter, to a trough of \(0.2\%\) in the tenth, and returns to its pre-shock level by about the seventeenth quarter.
The monetary contraction traces the mirror pattern with a deeper
trough, of \(0.3\%\) in the fourth quarter. Across the twenty realizations of each size the paths vary little, so
the asymmetry of the GDP response does not depend on a particular
realization of the network.

\begin{figure}[t]
\centering
\includegraphics[width=0.9\textwidth]{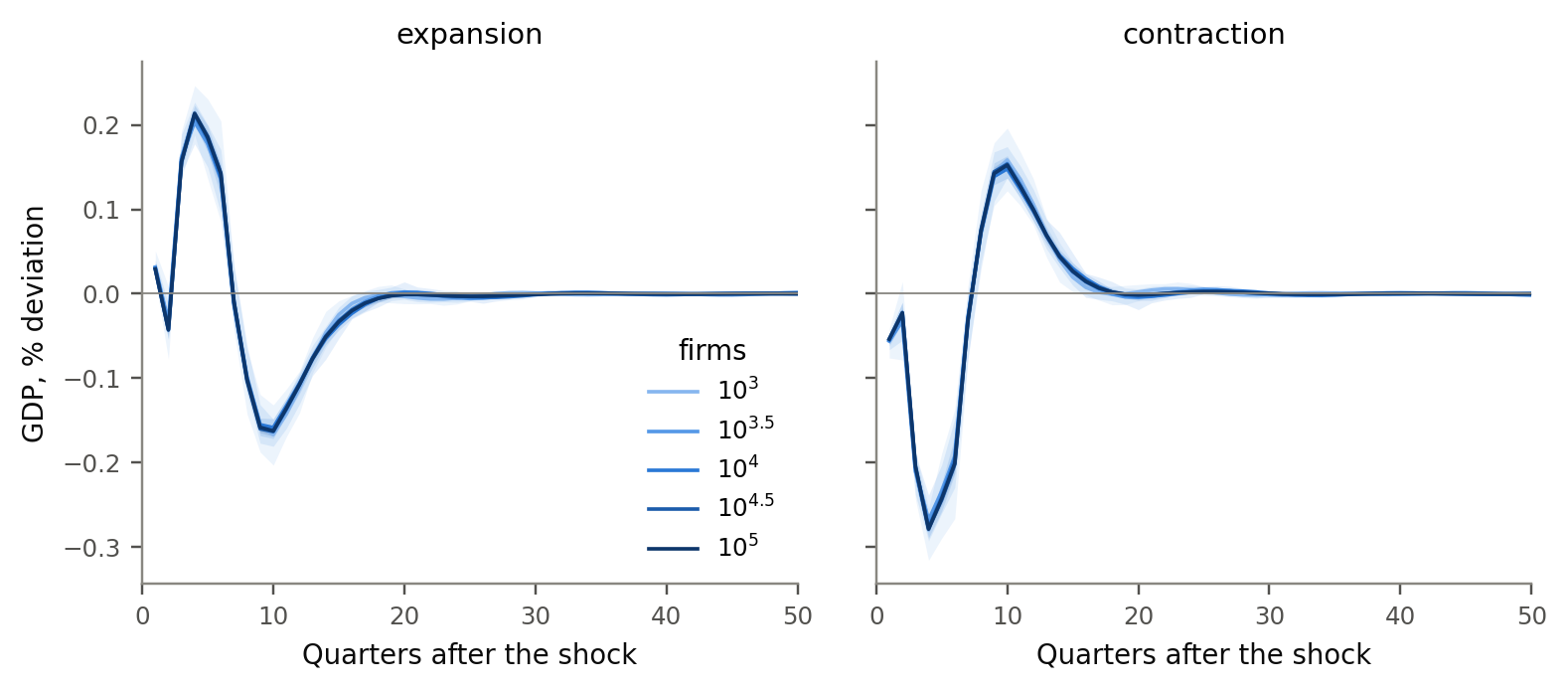}
\caption{Quarterly response of GDP to a monetary expansion and a
monetary contraction of \(\pi=0.05\), synthetic supply chain. Lines are means over twenty realizations, bands the tenth to ninetieth percentiles.}
\label{fig:quant_irfsyn}
\end{figure}

We summarize the asymmetry of the GDP response by the \emph{peak
ratio}, the deepest fall in GDP after a monetary contraction divided by
the highest rise after a monetary expansion of the same
size. On the synthetic supply chain the peak ratio is \(1.3\) at every size
from \(10^3\) to \(10^5\) firms.

Proposition~\ref{prop:implicit_interest} says that a monetary expansion
borrows GDP from the future and repays it at a strictly positive
implicit interest rate \(\mathcal I(\pi)=L_Y/G_Y-1\), set by the
deadweight loss of the transition. The size of this rate can be measured on the synthetic supply chain, for its transition is complete within the window
(Figure~\ref{fig:quant_irfsyn}). Summed over the window at \(10^5\)
firms, the GDP gained is \(0.7\) percentage-point-quarters and the GDP lost \(0.9\), so that \(\mathcal I(\pi)\) is about \(0.3\). For every
unit of GDP that the monetary expansion brings forward, the economy
gives up about \(1.3\) units later. The monetary expansion is therefore
contractionary once its later fall below the pre-shock level is counted, by \(0.2\) percentage-point-quarters, and the monetary contraction
lowers GDP cumulated over the window by \(0.3\). Both figures are the
same to the first decimal at every size from \(10^3\) to \(10^5\)
firms.

We now vary the size of the shock and the within-vintage curvature
one at a time, holding everything else at the baseline. The sweeps are carried out on the US supply chain at \(10^5\) firms and
on three realizations of the synthetic supply chain at \(10^4\) firms,
for which we report the median.\footnote{At \(\rho=-1.4\) one of the three realizations does not return to a
stationary state within the window after the monetary contraction. The
median is then the larger of the peak ratios of the other two.} Figure~\ref{fig:quant_sweeps} reports the peak ratio against the size of the shock and against the within-vintage curvature. The decomposition \eqref{eq:even_part} says how the peak
ratio should move with the size of the shock. The even part
\(\overline f_t\) of the GDP response grows with the square of the
shock and the first-order response in proportion to it, so the excess
of the peak ratio over unity should grow in proportion to the shock.
It does. The peak ratio tends to unity as the shock vanishes, at \(1.1\) on the US supply chain and \(1.0\) on the synthetic supply chain at
\(\pi=0.002\), and reaches \(7.3\) and \(1.7\) at \(\pi=0.10\).
Regressing the logarithm of its excess over unity on the logarithm of
the shock size over this fiftyfold range gives an elasticity of \(1.0\) on the US supply chain and \(1.1\) on
the synthetic supply chain. The within-vintage curvature sets how much output a mismatched bundle
loses, for the loss from production inefficiency
\eqref{eq:firm_inefficiency_def} scales with the curvature weights
\(1-\rho\) and \(1-\rho_c\). The asymmetry should therefore rise as inputs
become more complementary. The peak ratio rises
from \(2.4\) at \(\rho=0.3\) to \(6.0\) at \(\rho=-1.4\) on the US
supply chain and from \(1.2\) to \(1.5\) on the synthetic supply chain. And the
asymmetry survives when inputs are substitutes, outside the domain on
which the theory is proved.

\begin{figure}[t]
\centering
\includegraphics[width=0.9\textwidth]{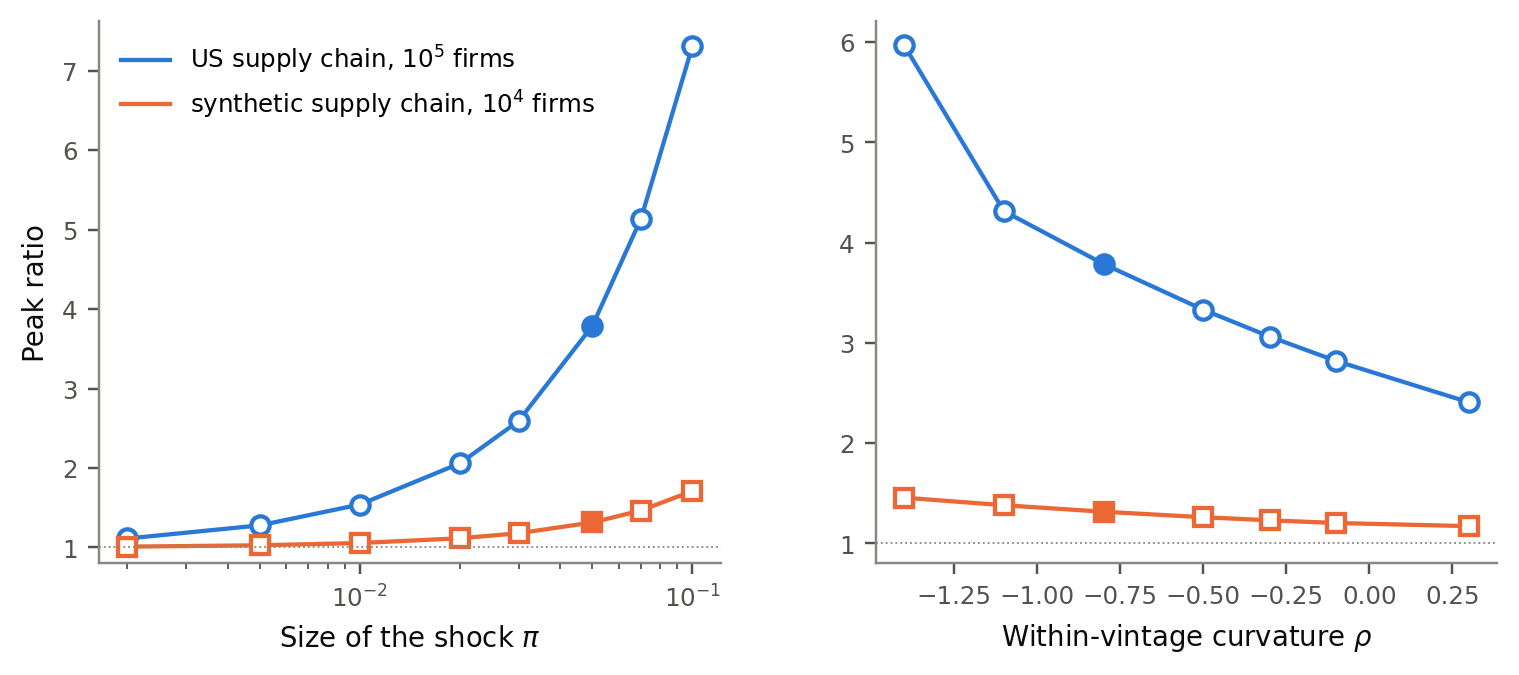}
\caption{Peak ratio against the size of the monetary shock (left) and
the within-vintage curvature (right). The synthetic supply chain is the median over three realizations. Filled markers are the baseline.}
\label{fig:quant_sweeps}
\end{figure}

\subsection{The size of monetary non-neutrality in the US economy}
\label{sec:quant-size}

We now measure the size of monetary non-neutrality in the US economy.
Figure~\ref{fig:quant_us_response} shows that on the US supply chain,
too, the response of GDP exhibits the two phases of
Theorem~\ref{thm:sign_reversing_irf}. The figure reports the US supply
chain at its four sizes and superimposes on the response to the
monetary expansion the sign-reversed response to the monetary
contraction. In the first quarter GDP falls after the monetary
expansion. Few of the orders placed at impact have matured by the end of the first quarter, and the small firms that buy alongside the household outbid it
for retail output produced from orders placed before the shock. In the
second quarter most of the orders that the retail tier placed at impact mature.
GDP then rises after the monetary expansion, by \(0.5\%\) at the full census, and falls after the monetary contraction, by \(1.5\%\), about three times as far. Were the response symmetric in the sign of the shock, the response to the monetary expansion and the sign-reversed response to the monetary contraction would coincide. The gap between them is
\(-2\overline f_t\), twice the size of the even part \(\overline f_t\)
of the GDP response. By \eqref{eq:even_part}, \(-\overline f_t\) splits
into the deadweight loss, the normalization term and the curvature term of the drawdown of the goods in transit. The even part is concentrated in the first year. From the fourth quarter after the monetary expansion, the reduced output of the upstream firms reaches the retail tier and holds GDP below its pre-shock level over the following two years. At \(10^6\) firms and at the full census this shortfall of GDP is about a tenth of a percent. At these two sizes GDP then follows a slow oscillation that has not decayed by the end of the window, down to between a quarter and a third of a percent below its pre-shock level in the tenth year.\footnote{The experiments follow a single monetary shock. In the data monetary shocks of different sizes arrive one after another, and the slow oscillations they set off overlap. An oscillation of this length therefore tends not to show up in the data.}

\begin{figure}[t]
\centering
\includegraphics[width=0.9\textwidth]{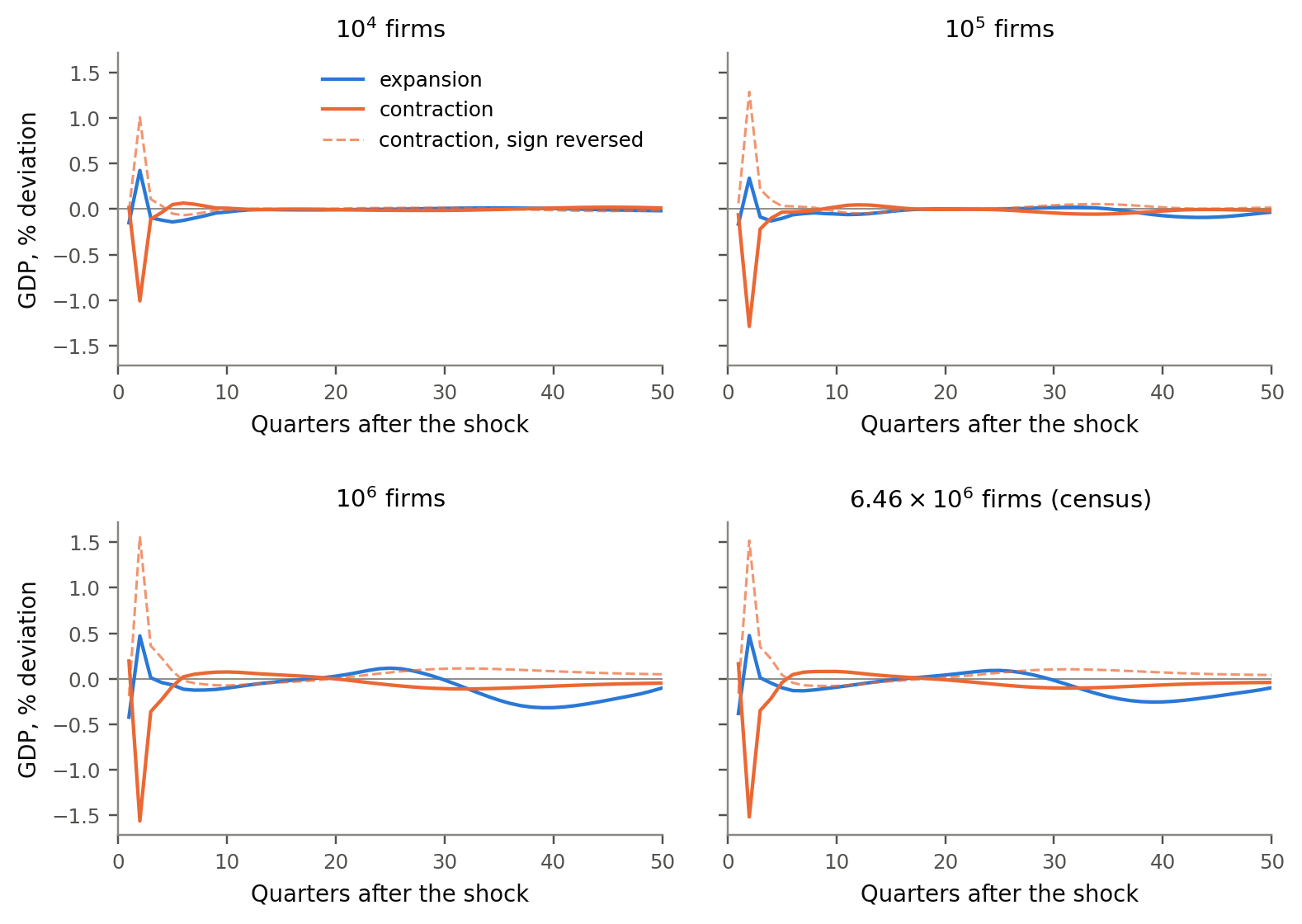}
\caption{Quarterly response of GDP to a monetary expansion and a
monetary contraction of \(\pi=0.05\), US supply chain.}
\label{fig:quant_us_response}
\end{figure}

The transition on the US supply chain is therefore not complete within
the window, so GDP cumulated over the window depends on where the window
ends. On the US supply chain we report the peak responses instead.

Both responses peak in the second quarter at every size of the US supply
chain, so the peak ratio compares them within one quarter. It is
\(2.4\) at \(10^4\) firms, and from \(10^5\) firms to the full census,
nearly two orders of magnitude, it lies between \(3.2\) and \(3.8\), with
\(3.2\) at the full census. Which means that at the full census the peak
ratio on the US supply chain is about two and a half times that on the
synthetic supply chain.

Figure~\ref{fig:quant_peaks} reports the peak responses of GDP against
the number of firms. On the US supply chain the peak fall in GDP after a monetary contraction
grows from \(1.0\%\) at \(10^4\) firms to \(1.5\) to \(1.6\%\) from
\(10^6\) firms to the full census. The peak rise after a monetary expansion stays between
\(0.3\) and \(0.5\%\). Note that monetary non-neutrality does not shrink as the economy
grows, which is consistent with the theory. The peak responses are the same on the synthetic supply chain at \(10^3\) firms and at \(10^5\), and on the US supply chain they are as large at the full census as
at \(10^6\) firms. The constants in every bound of the theory are set by the assumptions,
and the number of firms enters the hypotheses only through proportions. A
monetary shock is not an idiosyncratic disturbance, and no law of
large numbers dilutes it. The shock moves real quantities only through
the misalignment of balances that its uneven incidence creates, a
cross-sectional object, positive at some firms and negative at others,
with a size-weighted sum of zero. Were the firm-level values of this misalignment independent draws, their weighted sums would shrink at the rate
\(1/\sqrt n\), and money would be neutral in a large enough economy.
But the incidence rule apportions one aggregate shock by size, so the
firm-level values are not independent draws. The relative expenditure
impulse of a firm is a fixed function of its size relative to the
sizes of the other buyers of its inputs. The sums that matter for output, the average expenditure impulse of the buyers of a good and the household-consumption aggregate, are weighted means of firm-level impulses. Doubling the number
of firms doubles the terms in each sum and leaves the weighted mean
unchanged. Nor can the deadweight loss average out, for it is a sum of
nonnegative losses at the individual firms and in the household's
consumption of the individual goods.

\begin{figure}[t]
\centering
\includegraphics[width=0.62\textwidth]{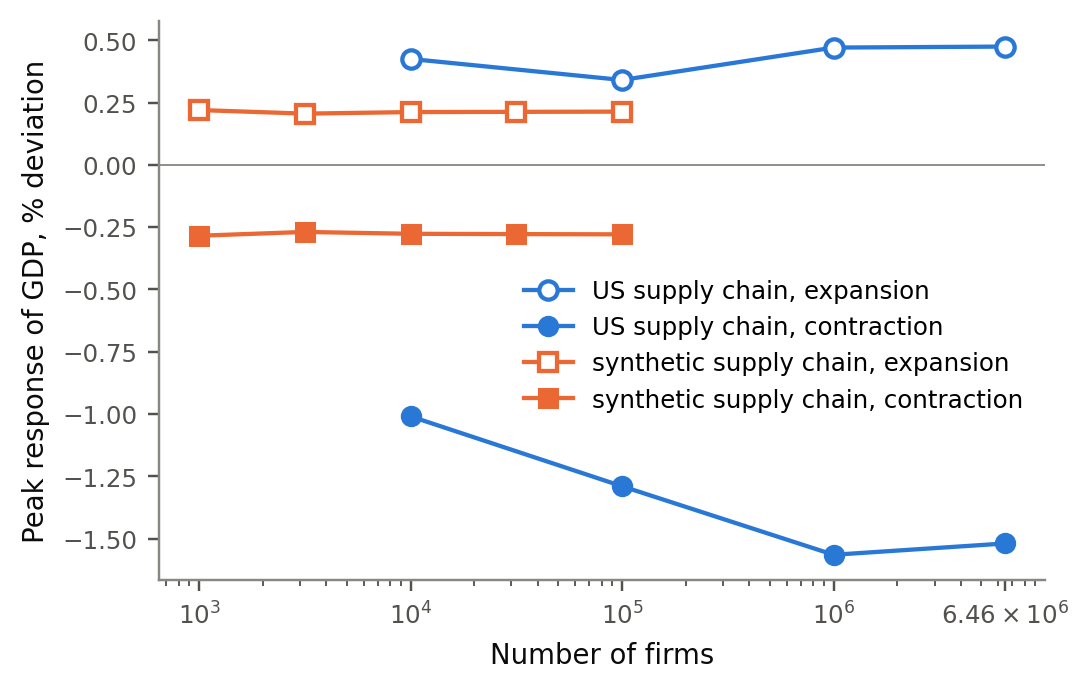}
\caption{Peak quarterly response of GDP to a monetary expansion and a
monetary contraction of \(\pi=0.05\), against the number of firms. The
synthetic supply chain is the mean over twenty realizations.}
\label{fig:quant_peaks}
\end{figure}

\label{sec:quant-anatomy}
To see where along the supply chain output moves, we measure, firm by
firm, the two parts of the response that the theory separates, the part
that changes sign with the shock and the part that does not. We simulate the same economy under \(+\pi\) and under \(-\pi\). Half the difference
between the two output responses of a firm is its antisymmetric
component, and half their sum is its symmetric component. The
antisymmetric component is, to first order, the reallocation effect,
the first-order term of the expansion
\eqref{eq:gdp_expansion_coefficients}, and in the aggregate the
symmetric component is the even part \(\overline f_t\) of the GDP response.
Figure~\ref{fig:quant_o1decomp} reports the first-quarter decomposition
at \(10^5\) firms by decile of upstreamness, the average number of
production stages that a firm's output passes through before it reaches
final demand, computed from the Leontief inverse as in
\citet{antras2018upstreamness}. Upstreamness falls as the downstreamness index of Definition~\ref{def:downstreamness} rises, and it is
the measure an empirical test would construct from input--output
accounts.

The antisymmetric component rises toward final demand, from about
zero over the upstream half of the supply chain to \(3.0\%\) in the most
downstream decile, the pattern of the impact reallocation of
Proposition~\ref{prop:impact_reallocation}. The output of the upstream deciles has changed little in the first quarter, since their inputs arrive with the
longest lead times and the orders they placed at impact mature later.
In the most downstream decile the antisymmetric component is about nine
times the peak rise of GDP after the monetary expansion. The ordering
by position owes nothing to price setting, since prices are equally
flexible everywhere in the model. The symmetric component is negative
in every decile. Under a monetary expansion it offsets part of what
the antisymmetric component adds at a downstream firm, while under a
monetary contraction it deepens the contraction. The contraction at
that firm therefore exceeds its expansion by twice the absolute value of the symmetric component. Over the deciles, the mean absolute symmetric component is \(23\%\) of the mean absolute antisymmetric component on the synthetic supply chain and \(36\%\) on the US supply chain, in
line with the larger peak ratio of the US supply chain.\footnote{The synthetic supply chain is tiered,
so its firms take only a few distinct values of upstreamness and its
deciles collapse into four groups. The ratio for the synthetic supply chain is taken over these four groups at \(10^4\) firms.}

\begin{figure}[t]
\centering
\includegraphics[width=0.62\textwidth]{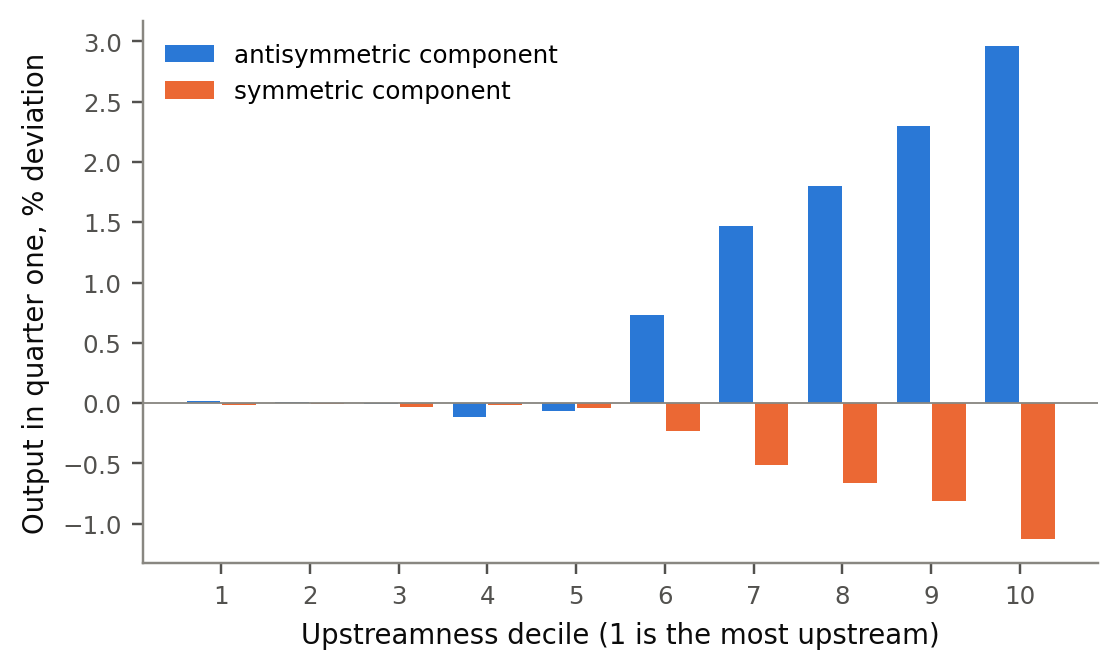}
\caption{First-quarter response of output by upstreamness decile, split
into the components antisymmetric and symmetric in the sign of the
shock, US supply chain at \(10^5\) firms.}
\label{fig:quant_o1decomp}
\end{figure}

\section{Conclusion}
\label{sec:conclusion}

About a quarter millennium ago, \citeauthor{hume1742} asked why changes in money
affect output, even though such changes can seem little more than a change in bookkeeping,
like moving from the Arabian notation, which requires few characters, to the Roman, which
requires many. The question has proved remarkably difficult.\footnote{As \citet{lucas1996nobel} put it, ``So much thought has been devoted to this
question and so much evidence is available that one might reasonably assume that it had been
solved long ago.''}  While the long and organized effort devoted to this problem has produced a wide range of
empirical methods and an extensive body of facts, it has
produced only a handful of theories. These theories differ in important respects in their treatment of information, market structure, and price formation, but they
 share a common generative mechanism: monetary shocks move the quantity of labor used
in production in the direction of the shock. This approach is natural once one begins from
the premise that monetary shocks change output \emph{broadly construed}.   If monetary shocks do not alter productivity, an increase in production \emph{everywhere} must come from greater use of inputs. In principle, that could mean more intermediate inputs or more labor. In practice, however, the former is much harder to generate on impact, so the adjustment is usually made to fall on labor.

In this paper, we begin from a different premise, a Wicksellian position, if you will.
\citet[p.~143]{wicksell1936} argued that a positive monetary shock can raise output in one
sector ``only at the expense of the other branches of production from which labor and liquid
capital have to be drawn.'' He went on to say that a general expansion of goods and services
across all sectors is impossible.
Starting from this position, the problem becomes one of explaining why positive monetary
shocks can increase output measured as GDP, and what this means for production in other
parts of the economy. This is precisely what we theorized in this paper.

We developed a supply-chain network model in which monetary shocks disproportionately affect
small firms, which happen to be overrepresented downstream. Firms produce using
dated inputs, so the supply chain is overlaid with a time structure.  In this environment, positive monetary shocks tend to reallocate resources
from upstream producers toward downstream producers, while negative shocks tend to shift
resources in the opposite direction. Such reallocation is not costless however. Because firms bid
for resources using newly acquired nominal balances in decentralized markets, the resulting
changes in resource availability are generally not proportional to the equilibrium
combinations in which inputs are efficiently used. This generates a deadweight loss. The reallocation of intermediate inputs, and therefore production, when combined with the deadweight loss generates a contraction in GDP in response to negative monetary shocks, and a milder expansion in response to positive ones.

The welfare implications of our model differ sharply from those emphasized in the two main
traditions of monetary non-neutrality. In the New Keynesian tradition,
monetary policy matters because nominal rigidities prevent the economy from attaining the
relevant efficient or constrained-efficient allocation, so appropriate monetary stabilization
can improve welfare by offsetting those distortions \citep{gali2015monetary}. In the New
Classical tradition, monetary disturbances are inefficient because they lead agents to act
on misleading or imperfectly understood price signals, thereby distorting labor supply and
production \citep{lucas1972expectations}. Our mechanism is different. We
assume no price stickiness: markets clear at every date, including along the transition from
one equilibrium to another. Monetary non-neutrality arises here not because prices fail to
adjust, nor because agents misread price signals, but because the reallocation of purchasing
power across the supply chain alters the temporal composition of production and moves
resources in temporally inappropriate proportions. A positive
monetary shock can raise output today by drawing on tomorrow's productive base, while a
negative shock can preserve more of that base for later use. But because such intertemporal
reallocations are costly, society does not borrow and lend across time at par. Output is
lost in transit. Monetary shocks, therefore, are generally welfare-reducing. The magnitude of that welfare loss depends on the implied intertemporal cost of
reallocation, which is itself shaped by the structure of the economy: the architecture of
the supply chain and the size distribution of firms. In economies with long chains of flow of intermediate inputs between firms, the welfare consequences of monetary
changes cannot be deduced from aggregate objects such as the elasticity of
labor supply or an index of price stickiness. The dynamics of monetary non-neutrality, and therefore its welfare consequences, emerge from the structural features of an economy intricately tied to questions of industrial organization.

\newpage
\bibliography{ref}
\bibliographystyle{aea}

\newpage
\appendix
\counterwithin{figure}{section}
\counterwithin{table}{section}
\renewcommand{\theHfigure}{\thesection.\arabic{figure}}
\renewcommand{\theHtable}{\thesection.\arabic{table}}
\section{Mathematical Appendix}
\label{app:proofs}

Throughout this appendix, constants depend only on the constants of
the assumptions, and not on the shock size \(\pi\), the date \(t\) or the
number of firms, unless stated otherwise. The constants of the
assumptions are
\(\underline\beta,\overline\beta,\underline\rho,\overline\rho,
\overline\lambda,C_\lambda,L_\ell,\overline\epsilon_\theta,\underline k\)
and \(\widehat k\), and in the proofs of
Sections~\ref{sec:reallocation} and~\ref{sec:asymmetry} also
\(\delta,\psi_T,\Delta_+,c_\xi,\Delta_\psi,\epsilon_h,s_R,\epsilon_k\)
and \(k_\ast\). An economy is admissible if it satisfies
Assumptions~\ref{assump:diffuse_incidence}--\ref{assump:monotone_decay_spectrum}
and the bounds \eqref{eq:beta_bounds}, \eqref{eq:zero_lag_share},
\eqref{eq:rho_domain} and \eqref{eq:log_size_bound} at these values. On
this class the incidence ratio \(\overline M\zeta_i/m_i^\ast\) lies in
\([e^{-2(1-\theta)L_\ell},e^{2(1-\theta)L_\ell}]\), so the incidence range
obeys \(C_\zeta\le2L_\ell e^{2\overline\epsilon_\theta L_\ell}\). The bound \(\overline\pi\) is the smallest of three radii set by these constants. They are the
radii of the geometric decay of Appendix~\ref{app:proof_stability}, of
the no-shutdown margin \eqref{eq:no_shutdown_margin} and of the bound on
the consumption--output wedge in Lemma~\ref{lem:gdp_labor_income_anchor}. The threshold \(\pi_\ast\) of
Theorem~\ref{thm:asymmetric_nonneutrality}, the order constant of
Proposition~\ref{prop:implicit_interest} and the Hessian bounds \(K\),
\(K_+\) and \(K_P\) depend on the economy under study.
Assumption~\ref{assump:short_run_nominal_rigidity} and the bounds above are
maintained in every proof, and each statement lists the further
hypotheses its proof uses. The lemmas are stated and proved in
Appendix~\ref{app:analytical}, an equation number that carries a letter
refers to Appendices~\ref{app:analytical} and~\ref{app:construction}, and the dependency map of the results
is in Appendix~\ref{app:depgraph}.

\subsection{Proofs of Section \ref{sec:model}: The Model}
\label{app:model}

\subsubsection{The within-period sequence}
\label{app:within_period}

The within-period sequence is as follows. Firms and the household first express nominal demand. At date \(t\) firm \(i\) holds the balance \(m_{i,t}\),
hires its stationary labor \(l_i^\ast\) at the indexed wage \(\wM\), and
pays for its intermediate purchases out of the balance that remains, a
cash-in-advance constraint. Its intermediate-input expenditure is
therefore
\begin{equation}
\sum_{j\in N}p_{j,t}x_{ji,t}
\;=\;
m_{i,t}-\wM l_i^\ast
\label{eq:dynamic_intermediate_budget}
\end{equation}
Supplier shares are not revised in the short run, so the expenditure is
allocated across suppliers at the stationary shares,
\begin{equation}
p_{j,t}x_{ji,t}
\;=\;
a_{ji}\bigl(m_{i,t}-\wM l_i^\ast\bigr),
\qquad i,j\in N
\label{eq:dynamic_nom_alloc}
\end{equation}
and the real order on each supplier is the nominal order deflated by
the supplier's price,
\begin{equation}
x_{ji,t}
\;=\;
a_{ji}\frac{m_{i,t}-\wM l_i^\ast}{p_{j,t}},
\qquad i,j\in N
\label{eq:dynamic_real_alloc}
\end{equation}
Labor-market clearing \(\sum_{i\in N}l_i^\ast=1\) holds by construction.

The household spends its wage income on final goods according to
\(\boldsymbol\gamma\), so nominal demand for good \(j\) is the
intermediate demand of its buyers plus household demand,
\begin{equation}
d_{j,t}
\;=\;
\sum_{i\in N}a_{ji}\bigl(m_{i,t}-\wM l_i^\ast\bigr)+\gamma_j \wM,
\qquad j\in N
\label{eq:dyn_nominal_demand}
\end{equation}

Prices then clear goods markets given the output available at date \(t\),
\begin{equation}
p_{j,t}
\;=\;
\frac{d_{j,t}}{q_{j,t}},
\qquad j\in N
\label{eq:dyn_price_update}
\end{equation}

Production then uses the stationary labor and the dated inputs
through the nests \eqref{eq:within_vintage_bundle}--\eqref{eq:materials_bundle},
\begin{equation}
q_{i,t}
\;=\;
(l_i^\ast)^{\beta_i}
\left[
\sum_{k\in\mathcal K_i}\chi_i^{(k)}
\left(
\sum_{j\in \mathcal S_i^{(k)}}
\nu_{ji}\,x_{ji,t-k}^{\rho}
\right)^{\rho_c/\rho}
\right]^{\frac{1-\beta_i}{\rho_c}},
\qquad i\in N
\label{eq:dynamic_output}
\end{equation}
where the purchases in vintage layer \(k\) of firm \(i\) share the lead
time \(k_{ji}=k\), so that \(x_{ji,t-k_{ji}}=x_{ji,t-k}\).

\subsubsection{Proof of existence and uniqueness of the stationary equilibrium}
\label{app:proof_existence}

\begin{proof}[Proof of Proposition~\ref{prop:existence_stationary}]
At a stationary equilibrium an order placed \(k\) periods ago equals an
order placed today. Hence the dated production identity
\eqref{eq:dynamic_output} reduces at every date to the static composite
\eqref{eq:materials_bundle}, and existence reduces to a fixed point of the
static unit-cost system. Fix \(w=1\). The two elasticities \(\eta=1/(1-\rho)\) and
\(\eta_c=1/(1-\rho_c)\) lie in \((0,1)\) on the domain
\eqref{eq:rho_domain}, and the argument below uses only \(\rho,\rho_c<1\),
\(\rho,\rho_c\neq0\). At supplier prices \(\mathbf p\), the unit cost
of firm \(i\)'s vintage composite of lead time \(k\) is the CES price index of
that vintage layer. And the unit cost of the intermediate composite is
the outer CES price index of the active vintage layers,
\begin{equation}
\overline p_i^{(k)}(\mathbf p)
:=
\Bigl(\sum_{j\in\mathcal S_i^{(k)}}\nu_{ji}^{\,\eta}\,p_j^{\,1-\eta}\Bigr)^{1/(1-\eta)},
\qquad
\overline p_i(\mathbf p)
:=
\Bigl(\sum_{k\in\mathcal K_i}\bigl(\chi_i^{(k)}\bigr)^{\eta_c}
\bigl[\overline p_i^{(k)}(\mathbf p)\bigr]^{1-\eta_c}\Bigr)^{1/(1-\eta_c)}
\label{eq:ces_price_indices}
\end{equation}
Shephard's lemma applied to the two indices gives the within-vintage
and cross-vintage shares at those prices,
\begin{equation}
\begin{gathered}
\widetilde a_{ji}^{(k)}(\mathbf p)
=
\frac{\nu_{ji}^{\,\eta}\,p_j^{\,1-\eta}}{\sum_{u\in\mathcal S_i^{(k)}}\nu_{ui}^{\,\eta}\,p_u^{\,1-\eta}},
\qquad
s_i^{(k)}(\mathbf p)
=
\frac{(\chi_i^{(k)})^{\eta_c}[\overline p_i^{(k)}(\mathbf p)]^{1-\eta_c}}
{\sum_{r\in\mathcal K_i}(\chi_i^{(r)})^{\eta_c}[\overline p_i^{(r)}(\mathbf p)]^{1-\eta_c}},\\
a_{ji}(\mathbf p)=s_i^{(k_{ji})}(\mathbf p)\,\widetilde a_{ji}^{(k_{ji})}(\mathbf p)
\end{gathered}
\label{eq:shares_from_prices}
\end{equation}
which at the stationary prices is the two-stage form
\eqref{eq:share_quantity_form} evaluated at the stationary bundle. The
unit cost of gross output at the numeraire wage is the Cobb--Douglas
unit cost of the top nest,
\begin{equation}
\Psi_i(\mathbf p)
:=
\beta_i^{-\beta_i}(1-\beta_i)^{-(1-\beta_i)}\,
\bigl[\overline p_i(\mathbf p)\bigr]^{1-\beta_i},
\qquad i\in N
\label{eq:stationary_price_map}
\end{equation}
and under competitive pricing a stationary price vector is a positive
fixed point of \(\mathbf p=\Psi(\mathbf p)\).

We first show that the log-unit-cost map is a global contraction. Set
\(\check\Psi_i(\check{\mathbf p}):=\log \Psi_i(e^{\check{\mathbf p}})\) for
\(\check{\mathbf p}=\log\mathbf p\in\mathbb R^n\). By Shephard's lemma the
derivative of \(\log \overline p_i\) with respect to \(\log p_j\) is the
supplier share \(a_{ji}(e^{\check{\mathbf p}})\) of
\eqref{eq:shares_from_prices}, non-negative and summing to one over the
suppliers of \(i\). Hence
\[
\frac{\partial \check\Psi_i}{\partial \check p_j}(\check{\mathbf p})
=
(1-\beta_i)\,a_{ji}(e^{\check{\mathbf p}})\;\ge\;0,
\qquad
\sum_{j\in N}\Bigl|\frac{\partial \check\Psi_i}{\partial \check p_j}(\check{\mathbf p})\Bigr|
=
1-\beta_i\;\le\;1-\underline\beta
\]
for every \(\check{\mathbf p}\in\mathbb R^n\). For any pair \(\check{\mathbf p},\check{\mathbf p}'\),
the mean-value theorem along the segment joining them and the row-sum bound
give
\[
\bigl\|\check\Psi(\check{\mathbf p})-\check\Psi(\check{\mathbf p}')\bigr\|_\infty
\le
(1-\underline\beta)\,\bigl\|\check{\mathbf p}-\check{\mathbf p}'\bigr\|_\infty
\]
so \(\check\Psi\) is a contraction on \((\mathbb R^n,\|\cdot\|_\infty)\). By
Banach's fixed-point theorem it has a unique fixed point
\(\check{\mathbf p}^\ast\), and \(\mathbf p^\ast:=e^{\check{\mathbf p}^\ast}\) is the
unique strictly positive solution of \(\mathbf p=\Psi(\mathbf p)\).

We next recover the shares, the sales and the allocation from the fixed
point. Evaluate the cost-minimizing shares at the fixed point,
\eqref{eq:shares_from_prices}, and assemble \(\mathbf A\). Its columns sum
to one, \(\mathbf A_\beta=\mathbf A\,\mathrm{diag}(\mathbf 1-\boldsymbol\beta)\)
has column sums \(1-\beta_i\le1-\underline\beta\), so
the spectral radius of \(\mathbf A_\beta\) is at most \(1-\underline\beta<1\),
\(\mathbf I-\mathbf A_\beta\) is invertible, and
\(\mathbf m^\ast=(\mathbf I-\mathbf A_\beta)^{-1}\boldsymbol\gamma
=\sum_{s\ge0}\mathbf A_\beta^{\,s}\boldsymbol\gamma\).
The \(j\)-th entry of \(\mathbf A_\beta^{\,s}\boldsymbol\gamma\) is
positive exactly when a chain of \(s\) sales leads from \(j\) to a firm the
household buys from, so reachability of final demand gives
\(\mathbf m^\ast\gg\mathbf 0\). Recover the allocation by
\begin{equation}
q_i^\ast=\frac{m_i^\ast}{p_i^\ast},
\qquad
l_i^\ast=\beta_i m_i^\ast,
\qquad
x_{ji}^\ast=(1-\beta_i)a_{ji}\frac{m_i^\ast}{p_j^\ast},
\qquad
c_i^\ast=\frac{\gamma_i}{p_i^\ast}
\label{eq:stationary_quantities}
\end{equation}
Cost minimization holds by construction,
since the orders \(x_{ji}^\ast=(1-\beta_i)a_{ji}m_i^\ast/p_j^\ast\) are the
cost-minimizing bundle at \(\mathbf p^\ast\) scaled to the intermediate-input
expenditure, and the labor input is the Cobb--Douglas share of revenue at the
numeraire wage. Zero profit holds because \(p_i^\ast\) is the unit cost.
The production identity holds because, at the cost-minimizing bundle, the
composite equals expenditure divided by its unit cost,
\(X_i(\mathbf x_i^\ast)=(1-\beta_i)m_i^\ast/\overline p_i(\mathbf p^\ast)\), and
\(l_i^\ast=\beta_im_i^\ast\), so
\((l_i^\ast)^{\beta_i}X_i(\mathbf x_i^\ast)^{1-\beta_i}
=m_i^\ast/\Psi_i(\mathbf p^\ast)=m_i^\ast/p_i^\ast=q_i^\ast\).
Goods-market clearing is the sales identity, in which firm \(j\)'s nominal
sales \(m_j^\ast\) equal intermediate demand \(\sum_i(1-\beta_i)a_{ji}m_i^\ast\)
plus household demand \(\gamma_j\), which is \eqref{eq:firm_sales_accounting}.
Labor-market clearing follows from
\(\mathbf 1^\top(\mathbf I-\mathbf A_\beta)=\boldsymbol\beta^\top\), so that
\(\sum_il_i^\ast=\boldsymbol\beta^\top\mathbf m^\ast=\mathbf 1^\top\boldsymbol\gamma=1\).

Uniqueness follows in the same order. The unit costs determine
\(\mathbf p^\ast\) uniquely, \(\mathbf p^\ast\) determines the shares, the
shares determine \(\mathbf m^\ast\) through the invertible sales system, and
the recovery formulas determine the allocation and the pipeline. The unit
cost \(\Psi_i\) is homogeneous of degree one in \((w,\mathbf p)\), so
replacing the numeraire \(w=1\) by any \(w>0\) rescales
\(\mathbf p^\ast\) and \(\mathbf m^\ast\) by \(w\) and leaves every real
quantity in \eqref{eq:stationary_quantities} unchanged.

Finally, time-invariant prices and balances make every order
time-invariant through the order rule, \(x_{ji,s}=x_{ji}^\ast\) for every
\(s\in\mathbb Z\). Hence the pipeline consists of the stationary pipeline
layers \eqref{eq:stationary_pipeline} and the dated production identity
\eqref{eq:dynamic_output} reduces at every date to the static identity
verified above. The constructed allocation and prices form a stationary
competitive equilibrium, unique up to the choice of the nominal wage.
\end{proof}

\subsubsection{Proof of local convergence after a monetary shock}
\label{app:proof_stability}

\begin{proof}[Proof of Proposition~\ref{prop:stability_stationary}]
Throughout, \(\varepsilon_\pi:=|\pi|(1-\theta)\), and \(\lambda_2\) stands for \(\lambda_2(\mathbf A)\).
Fix a contraction rate \(\kappa\) with
\begin{equation}
\max\bigl\{\overline\lambda,\;(1-\underline\beta)^{1/\widehat k}\bigr\}
\;<\;\kappa\;<\;1
\label{eq:contraction_rate_choice}
\end{equation}
Both terms on the left are below one, so such a \(\kappa\) exists, and
any choice serves, with the constant \(C_S\) depending on it. When
\(\widehat k=0\) any \(\kappa\in(\overline\lambda,1)\) serves.

We first show that the nominal block is exact. Let
\(\mathbf m_t^\perp:=\mathbf m_t-(1+\pi)\mathbf m^\ast\) be the misalignment of
balances from the post-shock stationary distribution. Subtracting the
scaled stationary identity
\((1+\pi)\mathbf m^\ast=\mathbf A\bigl((1+\pi)\mathbf m^\ast-(1+\pi)\,\mathrm{diag}(\boldsymbol\beta)\,\mathbf m^\ast\bigr)+\boldsymbol\gamma(1+\pi)w^\ast\)
from the balance law \eqref{eq:firm_balance_law_pre}, in which the wage bill
is \(\wM l_i^\ast=(1+\pi)\beta_im_i^\ast\), gives
\begin{equation}
\mathbf m_{t+1}^\perp=\mathbf A\,\mathbf m_t^\perp,
\qquad
\mathbf m_0^\perp=\pi\bigl(\overline M\boldsymbol\zeta-\mathbf m^\ast\bigr),
\qquad
\mathbf 1^\top\mathbf m_t^\perp=0
\label{eq:stab_nominal_block}
\end{equation}
Since \(\mathbf m_t^\perp\in\mathcal Z_{\mathbf A}\), the power bound
\eqref{eq:power_bound} gives
\(\|\mathbf m_t^\perp\|_m\le C_\lambda\lambda_2^t\|\mathbf m_0^\perp\|_m\), that is,
in proportional coordinates \(\widehat{\mathbf m}_t^\perp:=\mathbf D_m^{-1}\mathbf m_t^\perp\),
\(\|\widehat{\mathbf m}_t^\perp\|_\infty\le C_\lambda\lambda_2^t\|\widehat{\mathbf m}_0^\perp\|_\infty\). At
impact \(1+\widehat m_{i,0}^\perp/\pi=\overline M\zeta_i/m_i^\ast\), which lies in
\([e^{-2(1-\theta)L_\ell},e^{2(1-\theta)L_\ell}]\) by the log-size bound,
so \(|\widehat m_{i,0}^\perp|\le|\pi|(e^{2(1-\theta)L_\ell}-1)\le C_\zeta|\pi|(1-\theta)\)
with \(C_\zeta\) the incidence range of
\eqref{eq:impact_misalignment_constant}. Hence
\begin{equation}
\|\widehat{\mathbf m}_t^\perp\|_\infty
\;\le\;
C_\lambda C_\zeta\,\varepsilon_\pi\,\lambda_2^t,
\qquad t\ge0
\label{eq:stab_nominal_envelope}
\end{equation}

We next show that expenditures and demands stay positive. The
intermediate-input expenditure
and the nominal demand are, in units of their post-shock stationary values,
\[
\frac{e_{i,t}}{(1+\pi)(1-\beta_i)m_i^\ast}
=
1+\frac{\widehat m_{i,t}^\perp}{(1+\pi)(1-\beta_i)},
\qquad
\frac{d_{j,t}}{(1+\pi)m_j^\ast}
=
\frac{m_{j,t+1}}{(1+\pi)m_j^\ast}
=
1+\frac{\widehat m_{j,t+1}^\perp}{1+\pi}
\]
where \(d_{j,t}=m_{j,t+1}\) because nominal sales equal nominal demand under
market clearing. The bound \(\overline\pi\) is at most \(\tfrac12\) and small enough that
\(2C_\lambda C_\zeta\overline\pi\,\overline\epsilon_\theta/(1-\overline\beta)\le\tfrac12\),
this being one of the conditions that define it. Then by \eqref{eq:stab_nominal_envelope} both ratios lie in
\([\tfrac12,\tfrac32]\) at every date, and every expenditure and every demand is
strictly positive. Every balance also exceeds the wage bill by a fixed
fraction of itself. The bound \(e_{i,t}\ge\tfrac12(1+\pi)(1-\beta_i)m_i^\ast\)
and \(m_{i,t}=(1+\pi)\beta_im_i^\ast+e_{i,t}\) give
\(e_{i,t}/m_{i,t}\ge(1-\beta_i)/(1+\beta_i)\ge(1-\overline\beta)/(1+\overline\beta)\),
which is the no-shutdown margin
\begin{equation}
\frac{\wM\,l_i^\ast}{m_{i,t}}
\;\le\;
1-c_\beta,
\qquad
c_\beta:=\frac{1-\overline\beta}{1+\overline\beta}
\label{eq:no_shutdown_margin}
\end{equation}
at every firm and date.

We now derive an exact identity for the orders. By the order rule
\eqref{eq:dynamic_real_alloc} and the clearing rule
\eqref{eq:dyn_price_update}, the quantity buyer \(i\) orders from supplier
\(j\) at date \(t\) is \(x_{ji,t}=a_{ji}e_{i,t}q_{j,t}/d_{j,t}\), and at the
stationary equilibrium \(x_{ji}^\ast=a_{ji}e_i^\ast q_j^\ast/d_j^\ast\). Dividing,
the common factor \(1+\pi\) cancels, and
\begin{equation}
\begin{aligned}
\check x_{ji,t}:=\log\frac{x_{ji,t}}{x_{ji}^\ast}
&\;=\;
\check q_{j,t}+\check\Xi_{ij,t},
\qquad
\check q_{j,t}:=\log\frac{q_{j,t}}{q_j^\ast},\\
\check\Xi_{ij,t}
&:=\log\Bigl(1+\frac{\widehat m_{i,t}^\perp}{(1+\pi)(1-\beta_i)}\Bigr)-\log\Bigl(1+\frac{\widehat m_{j,t+1}^\perp}{1+\pi}\Bigr)
\end{aligned}
\label{eq:stab_order_identity}
\end{equation}
Since \(|\log(1+v)|\le2|v|\) for \(|v|\le\tfrac12\) and
\(1+\pi\ge\tfrac12\), the bounds on the two ratios and
\eqref{eq:stab_nominal_envelope} give
\begin{equation}
\max_{j,i}|\check\Xi_{ij,t}|
\;\le\;
C_{\check\Xi}\,\varepsilon_\pi\,\lambda_2^t,
\qquad
C_{\check\Xi}:=\frac{8\,C_\lambda C_\zeta}{1-\overline\beta}
\label{eq:stab_forcing_bound}
\end{equation}
Orders placed before the shock carry no deviation, so \(\check x_{ji,s}=0\),
\(\check q_{j,s}=0\) and \(\check\Xi_{ij,s}=0\) for \(s<0\).

We now show that log output is a delayed contraction of the log inputs.
With labor fixed at \(l_i^\ast\), the dated production identity
\eqref{eq:dynamic_output} and the two-stage shares
\eqref{eq:share_quantity_form} evaluated at the stationary bundle give
\[
\check q_{i,t}
=
\mathcal L_i\bigl((\check x_{ji,t-k_{ji}})_{j\in\mathcal S_i}\bigr),
\qquad
\mathcal L_i(\check{\mathbf x})
:=
\frac{1-\beta_i}{\rho_c}
\log\Bigl\{\sum_{k\in\mathcal K_i}s_i^{(k)}
\Bigl[\sum_{j\in\mathcal S_i^{(k)}}\widetilde a_{ji}^{(k)}e^{\rho \check x_{ji}}\Bigr]^{\rho_c/\rho}\Bigr\}
\]
because \(X_i^{(k)}(\mathbf x)/X_i^{(k)}(\mathbf x^\ast)
=(\sum_j\widetilde a_{ji}^{(k)}e^{\rho \check x_{ji}})^{1/\rho}\) and
\(X_i(\mathbf x)/X_i(\mathbf x^\ast)=(\sum_ks_i^{(k)}[X_i^{(k)}(\mathbf x)/X_i^{(k)}(\mathbf x^\ast)]^{\rho_c})^{1/\rho_c}\).
The map \(\mathcal L_i\) satisfies \(\mathcal L_i(\mathbf 0)=0\), and its
partial derivatives are \((1-\beta_i)\) times the supplier shares of the
nested composite at the perturbed bundle, hence non-negative with sum
\(1-\beta_i\le1-\underline\beta\). By the mean-value theorem
\(|\mathcal L_i(\check{\mathbf x})-\mathcal L_i(\check{\mathbf x}')|\le(1-\underline\beta)\|\check{\mathbf x}-\check{\mathbf x}'\|_\infty\).
Zero-lead-time entries \(\check x_{ji,t}=\check q_{j,t}+\check\Xi_{ij,t}\) involve the current outputs,
so at each date \(\check{\mathbf q}_t\) solves a fixed-point problem. By
Lemma~\ref{lem:propagation_linearization} it is a contraction in those
coordinates, and the solution exists and is unique date by date.

We now establish the geometric envelope. Set
\(\kappa_{\beta}:=(1-\underline\beta)\kappa^{-\widehat k}\), which is below one by
\eqref{eq:contraction_rate_choice}, and \(C_q:=\kappa_{\beta}C_{\check\Xi}/(1-\kappa_{\beta})\). We show by
induction on \(t\) that \(\|\check{\mathbf q}_t\|_\infty\le C_q\varepsilon_\pi\kappa^t\). For negative dates
the bound holds with zero deviations. At date \(t\), the entries of the dated input bundle are
\(\check x_{ji,t-k_{ji}}=\check q_{j,t-k_{ji}}+\check\Xi_{ij,t-k_{ji}}\) with \(k_{ji}\in\{0,\dots,\widehat k\}\), so by
the Lipschitz bound and \eqref{eq:stab_forcing_bound}, using
\(\lambda_2\le\kappa\),
\[
\begin{aligned}
\|\check{\mathbf q}_t\|_\infty
&\;\le\;
(1-\underline\beta)\Bigl(\max\bigl\{\|\check{\mathbf q}_t\|_\infty,\;\max_{1\le k\le \widehat k}\|\check{\mathbf q}_{t-k}\|_\infty\bigr\}+C_{\check\Xi}\varepsilon_\pi\kappa^{t-\widehat k}\Bigr)\\
&\;\le\;
(1-\underline\beta)\Bigl(\max\{\|\check{\mathbf q}_t\|_\infty,\;C_q\varepsilon_\pi\kappa^{t-\widehat k}\}+C_{\check\Xi}\varepsilon_\pi\kappa^{t-\widehat k}\Bigr)
\end{aligned}
\]
Put \(\Upsilon:=\kappa^{-t}\|\check{\mathbf q}_t\|_\infty\) and divide by \(\kappa^t\). Since
\(\kappa^{-\widehat k}\ge1\),
\[
\Upsilon\;\le\;\kappa_{\beta}\bigl(\max\{\Upsilon,\,C_q\varepsilon_\pi\}+C_{\check\Xi}\varepsilon_\pi\bigr)
\]
If \(\Upsilon>C_q\varepsilon_\pi\) this reads \(\Upsilon\le \kappa_{\beta}(\Upsilon+C_{\check\Xi}\varepsilon_\pi)\), hence
\(\Upsilon\le \kappa_{\beta}C_{\check\Xi}\varepsilon_\pi/(1-\kappa_{\beta})=C_q\varepsilon_\pi\), a contradiction. Therefore
\(\Upsilon\le C_q\varepsilon_\pi\), which completes the induction and gives
\begin{equation}
\max_i|\check q_{i,t}|\le C_q\,\varepsilon_\pi\,\kappa^t,
\qquad
\max_{j,i}|\check x_{ji,t}|\le(C_q+C_{\check\Xi})\,\varepsilon_\pi\,\kappa^t,
\qquad t\ge0
\label{eq:stab_real_envelope}
\end{equation}
When \(\widehat k=0\) every input is ordered at date \(t\), the fixed-point inequality
is \(\|\check{\mathbf q}_t\|_\infty\le(1-\underline\beta)(\|\check{\mathbf q}_t\|_\infty+C_{\check\Xi}\varepsilon_\pi\lambda_2^t)\), and
\(\|\check{\mathbf q}_t\|_\infty\le(1-\underline\beta)C_{\check\Xi}\varepsilon_\pi\lambda_2^t/\underline\beta\), so any \(\kappa\in(\lambda_2,1)\)
serves.

It remains to convert the log bounds into proportional deviations. An
order outstanding in the
start-of-period state at date \(t\) was placed at some date
\(s\in\{t-\widehat k,\dots,t-1\}\), so its log deviation is at most
\((C_q+C_{\check\Xi})\varepsilon_\pi\kappa^{s}\le(C_q+C_{\check\Xi})\kappa^{-\widehat k}\varepsilon_\pi\kappa^t\).
Reducing \(\overline\pi\) further so that every log deviation on the
trajectory is at most \(\tfrac12\), the inequality
\(|e^{z}-1|\le2|z|\) for \(|z|\le\tfrac12\) converts the log bounds into
proportional-deviation bounds. Together with the balance bound
\eqref{eq:stab_nominal_envelope}, in which \(\lambda_2\le\kappa\), this gives
\eqref{eq:post_shock_geometric_decay} with
\[
C_S:=\max\Bigl\{C_\lambda C_\zeta,\;2(C_q+C_{\check\Xi})\kappa^{-\widehat k}\Bigr\}
\]
and the same envelope, with the constants \(2C_q\) and \(2(C_q+C_{\check\Xi})\), for
outputs and for the orders placed at date \(t\). Prices follow from
\(p_{j,t}/((1+\pi)p_j^\ast)=(1+\widehat m_{j,t+1}^\perp/(1+\pi))e^{-\check q_{j,t}}\). Every
constant depends only on \(C_\lambda\), \(L_\ell\),
\(\overline\epsilon_\theta\), \(\underline\beta\), \(\overline\beta\), \(\widehat k\)
and \(\kappa\), and not on the number of firms.
\end{proof}

\subsection{Proofs of Section \ref{sec:inefficiency}: Miscoordination within Firms}
\label{app:proofs_inefficiency}

\subsubsection{Proof of the sign of the loss from production inefficiency}
\label{app:proof_sign_production_inefficiency}

\begin{proof}[Proof of Lemma~\ref{lem:sign_production_inefficiency}]
Both curvatures lie below one and \(\beta_i<1\), so the weights
\((1-\beta_i)(1-\rho)\) and \((1-\beta_i)(1-\rho_c)\) in
\eqref{eq:firm_inefficiency_def} are strictly positive, and each of
\(V_i^{\mathrm w}\) and \(V_i^{\mathrm b}\) in
\eqref{eq:vintage_variances} is a combination of squares with
non-negative weights. Hence \(\mathcal D_{i,t}(\pi)\ge0\), with equality
exactly when both variances vanish. Every active supplier carries a
strictly positive within-vintage share and every nonempty vintage layer
a strictly positive cross-vintage share. Hence
\(V_i^{\mathrm w}(\widehat x_{i,t})=0\) exactly when
\(\widehat x_{ji,t}=\mathcal F_i^{(k)}(\widehat x_{i,t})\) for every
supplier \(j\) of every nonempty vintage layer \(k\), and
\(V_i^{\mathrm b}(\widehat x_{i,t})=0\) exactly when
\(\mathcal F_i^{(k)}(\widehat x_{i,t})=\mathcal F_i(\widehat x_{i,t})\)
for every nonempty vintage layer. Both hold exactly when \(\widehat x_{ji,t}\)
takes the same value at every \(j\in\mathcal S_i\), that is, when the
dated inputs of firm \(i\) all move in the same proportion. The
statement for \(\mathcal D_t^{(\boldsymbol\omega)}(\pi)\) follows, since it
is the sum of the \(\mathcal D_{i,t}(\pi)\) with the non-negative
weights \(\omega_i\).
\end{proof}

\subsubsection{Proof of the order of the cumulative loss from production inefficiency}
\label{app:proof_finite_cumulative_inefficiency}

\begin{proof}[Proof of Theorem~\ref{theorem:finite_cumulative_inefficiency}]
We first record the envelope on the dated inputs.
Proposition~\ref{prop:stability_stationary} bounds the
proportional-deviation state, whose coordinates include the inherited
pipeline entries, by \(C_S\kappa^t|\pi|(1-\theta)\), and it bounds the
orders placed at date \(t\), which are the zero-lead-time dated inputs, by a
bound of the same form. Together these two bounds give, for every proportional
change of a dated input, inherited or current,
\begin{equation}
\max_{i\in N,\,j\in\mathcal S_i}\bigl|\widehat x_{ji,t}(\pi)\bigr|
\;\le\;
C_S\,\kappa^t\,|\pi|\,(1-\theta),
\qquad t\ge0,\ |\pi|\le\overline\pi
\label{eq:usable_input_envelope}
\end{equation}
The
supplier count \(|\mathcal S_i|\) enters no bound below, since
\(\mathcal F_i\) averages against weights that sum to one.

We next prove the bound on the cumulative loss from production inefficiency. By the
firm-level definition
\eqref{eq:firm_inefficiency_def}, the loss from production inefficiency is a non-negative combination of
the within- and between-vintage variance components
\eqref{eq:vintage_variances}, whose sum is the total within-bundle variance
\(V_i^{\mathrm w}(\widehat x_{i,t})+V_i^{\mathrm b}(\widehat x_{i,t})=\sum_{j\in\mathcal S_i}a_{ji}\bigl(\widehat x_{ji,t}-\mathcal F_i(\widehat x_{i,t})\bigr)^2\le\mathcal F_i(\widehat x_{i,t}^{\,2})\).
And both curvature weights are at most \(1+\overline\rho\) on the curvature
domain, so
\[
\mathcal D_{i,t}(\pi)
\;\le\;
\tfrac12(1-\beta_i)(1+\overline\rho)\,\mathcal F_i(\widehat x_{i,t}^{\,2})
\;\le\;
\tfrac12(1-\underline\beta)(1+\overline\rho)\,C_S^2\kappa^{2t}\pi^2(1-\theta)^2
\]
by \eqref{eq:usable_input_envelope}. The \(\boldsymbol\omega\)-weighted loss from production
inefficiency at a date is a weighted average of the firms' losses from production inefficiency and obeys the same
bound. And summing the geometric series over \(t\ge0\) shows that the
series is finite for every \(|\pi|\le\overline\pi\) and gives the bound
of the theorem, \(\mathcal D^{(\boldsymbol\omega)}(\pi)\le\overline C^{\mathcal D}(1-\theta)^2\pi^2\) with
\(\overline C^{\mathcal D}=(1-\underline\beta)(1+\overline\rho)C_S^2/\{2(1-\kappa^2)\}\),
a constant set by the constants of the assumptions.

We now prove the expansion \eqref{eq:cumulative_inefficiency_expansion}. Fix a date
\(t\) and a firm \(i\). Differentiability
at zero gives \(\widehat x_{i,t}(\pi)/\pi\to\dot{\widehat x}_{i,t}\), and the
curvature-weighted form is a continuous quadratic form on the finite
supplier vector, so \(\mathcal D_{i,t}(\pi)/\pi^2\to\tfrac12(1-\beta_i)\langle\dot{\widehat x}_{i,t},\dot{\widehat x}_{i,t}\rangle_i\).
The per-date ratios \(\mathcal D_t^{(\boldsymbol\omega)}(\pi)/\pi^2\) are
dominated, by the firm-level bound above, by the summable sequence
\(\tfrac12(1-\underline\beta)(1+\overline\rho)C_S^2(1-\theta)^2\kappa^{2t}\),
so dominated convergence for the counting measure on dates gives
\(\mathcal D^{(\boldsymbol\omega)}(\pi)/\pi^2\to C^{(\boldsymbol\omega)}\), which is
\eqref{eq:cumulative_inefficiency_expansion}. The coefficient is a sum of non-negative
terms, each \(\tfrac12\omega_i(1-\beta_i)\) times a curvature-weighted variance of
the first-order responses. It therefore vanishes if and only if every such
variance vanishes at every supported firm and date, that is, if and only if
the first-order responses are constant across suppliers there, which is
the characterization of the coefficient in the statement. The bound has
the same constant for every normalized
weight vector, and it depends on the economy only through \(C_S\), \(\kappa\),
\(\underline\beta\) and \(\overline\rho\), so it is uniform over any class on
which the envelope is uniform.
\end{proof}

\subsection{Proofs of Section \ref{sec:reallocation}: Monetary Reallocation}
\label{app:reallocation}

The threshold of the downstream-concentrated tilts is
\begin{equation}
\overline\epsilon_w
:=
\frac{c_R}{8\,(c_R+C_\phi/\underline\beta)}
\;\in\;(0,1)
\label{eq:downstream_threshold}
\end{equation}
and the tolerance of condition~\eqref{eq:hypothesis_N} and the two
thresholds of condition~\eqref{eq:hypothesis_L} are
\begin{equation}
\overline\epsilon_\xi:=\frac{c_R\,\underline\beta}{8\,(1-\underline\beta)^2},
\qquad
\overline\epsilon_F:=\frac{c_R\,\epsilon_h}{8\,C_\xi},
\qquad
\overline\epsilon_h:=\Biggl[\frac{c_R\bigl(1-(1-\underline\beta)^{1/2}\bigr)}{8\,C_\phi\,(1-\underline\beta)^{1/2}}\Biggr]^2
\label{eq:persistence_thresholds}
\end{equation}

\subsubsection{Proof of impact reallocation}
\label{app:proof_impact_reallocation}

By the ordering \eqref{eq:retail_tier_primacy}, \(\psi_R\ge\epsilon_h\) and
\(\psi_j\le\delta\) off the retail tier, so the capped position of
Definition~\ref{def:position_tilted_output_index} is
\(\widehat\psi_i=1\) on \(\mathcal R\) and
\(\widehat\psi_j\le\delta/\epsilon_h<1\) off it. Write
\(M_{\mathcal R}:=\sum_{i\in\mathcal R}m_i^\ast\) for the sales of the
retail tier. The retail-tier mass of the tilt weights
\eqref{eq:w_b_def} is
\[
\omega_{\mathcal R}(b):=\sum_{i\in\mathcal R}\omega_i(b)
=
\frac{M_{\mathcal R}}{M_{\mathcal R}+\sum_{j\notin\mathcal R}m_j^\ast\widehat\psi_j^{\,b}}
\]
which is nondecreasing in \(b\), since every \(\widehat\psi_j^{\,b}\)
off the retail tier is nonincreasing, and rises to \(\omega_{\mathcal R}(\infty)=1\).
With \(a_m:=(\overline M-M_{\mathcal R})/M_{\mathcal R}\), the sales
off the retail tier relative to those on it,
\(\omega_{\mathcal R}(b)\ge1/(1+a_m(\delta/\epsilon_h)^b)\), so the
downstream-concentrated tilts form an interval \([b_w,\infty]\) with
\(b_w\le\log\bigl(a_m/\overline\epsilon_w\bigr)/\log(\epsilon_h/\delta)\).
On the retail tier the weights are flat,
\(\omega_i(b)=m_i^\ast/(M_{\mathcal R}+\sum_{j\notin\mathcal R}m_j^\ast\widehat\psi_j^{\,b})=\omega_{\mathcal R}(b)\,m_i^\ast/M_{\mathcal R}\)
for every \(i\in\mathcal R\) and every \(b\).

\begin{proof}[Proof of Proposition~\ref{prop:impact_reallocation}]
We first compute the dated responses of expenditure and demand. Fix a date
\(s\ge0\). By the nominal recursion \eqref{eq:stab_nominal_block} in the
proof of Proposition~\ref{prop:stability_stationary}, the balances are \(\mathbf m_s=(1+\pi)\mathbf m^\ast+\mathbf m_s^\perp\) with
\(\mathbf m_s^\perp=\mathbf A^s\mathbf m_0^\perp\) and
\(\mathbf m_0^\perp=\pi(\overline M\boldsymbol\zeta-\mathbf m^\ast)\), and the wage
bill is \(\wM l_i^\ast=(1+\pi)\beta_im_i^\ast\). Firm \(i\)'s
intermediate-input expenditure is therefore
\(e_{i,s}=m_{i,s}-\wM l_i^\ast=(1+\pi)(1-\beta_i)m_i^\ast+m^\perp_{i,s}\), and its
first-order proportional response is
\begin{equation}
\dot{\widehat e}_{i,s}
=
1+\frac{\dot m^\perp_{i,s}}{(1-\beta_i)\,m_i^\ast},
\qquad
\dot{\mathbf m}_s^\perp=\mathbf A^s\bigl(\overline M\boldsymbol\zeta-\mathbf m^\ast\bigr)
\label{eq:centered_budget_impulse}
\end{equation}
which at \(s=0\) is
\(\dot{\widehat e}_{i,0}=1+\iota_i/(1-\beta_i)=1+\sigma_i\), the
neutral common mode plus the expenditure impulse
\eqref{eq:expenditure_impulse}. Nominal demand
\eqref{eq:dyn_nominal_demand} for good \(j\) is
\(d_{j,s}=\sum_la_{jl}e_{l,s}+\gamma_j\wM\). Differentiating at \(\pi=0\)
and dividing by the stationary sales identity
\(m_j^\ast=\sum_la_{jl}(1-\beta_l)m_l^\ast+\gamma_jw^\ast\),
\begin{equation}
\dot{\widehat d}_{j,s}
=
\frac{\sum_la_{jl}(1-\beta_l)m_l^\ast\,\dot{\widehat e}_{l,s}+\gamma_jw^\ast}{m_j^\ast}
=
\sum_{l\in N}\Omega_{jl}\,\dot{\widehat e}_{l,s}+r_j
\label{eq:co_customer_average}
\end{equation}
with \(\Omega_{jl}=(1-\beta_l)a_{jl}m_l^\ast/m_j^\ast\) the output shares
\eqref{eq:output_share_matrix} and \(r_j=\gamma_jw^\ast/m_j^\ast\) the
consumption share \eqref{eq:retail_absorption_def}. The household enters
with weight \(r_j\) and a response of one, because its expenditure
\(\gamma_j\wM\) is indexed to the money stock. Since
\(\sum_l\Omega_{jl}=1-r_j\), the dated contrast \eqref{eq:dated_contrast}
is, at every date,
\[
\Xi_{ij,s}
=
\dot{\widehat e}_{i,s}-\dot{\widehat d}_{j,s}
=
\frac{\dot m^\perp_{i,s}}{(1-\beta_i)\,m_i^\ast}
-\sum_{l\in N}\Omega_{jl}\,\frac{\dot m^\perp_{l,s}}{(1-\beta_l)\,m_l^\ast}
\]
the misalignment of buyer \(i\) in proportional coordinates net of the
output-share-weighted misalignment of the other customers of \(j\). At
\(s=0\), \(\dot{\widehat d}_{j,0}=1+\overline\sigma_j\) with
\(\overline\sigma_j\) the average \eqref{eq:buyer_average_impulse}, so
the impact-date contrast is \(\Xi_{ij,0}=\sigma_i-\overline\sigma_j\). And
averaging over \(i\)'s suppliers of lead time \(k\) with the
within-vintage shares gives the relative expenditure impulse
\eqref{eq:vintage_impulse_def} and, scaled by \(1-\beta_i\), the
scaled relative expenditure impulse \eqref{eq:scaled_impulse_def}. Summed over a set of
firms with the sales weights, by \eqref{eq:set_impulse},
\begin{equation}
\sum_{i\in\mathcal C}m_i^\ast\,s_i^{(k)}\phi_i^{(k)}
\;=\;
\sum_{i\in\mathcal C}m_i^\ast(1-\beta_i)\,s_i^{(k)}\,\xi_i^{(k)}
\;=\;
e_{\mathcal C}^{(k)}\,\xi_{\mathcal C}^{(k)}
\label{eq:set_impulse_identity}
\end{equation}

We next bound the size of the contrasts. By
\eqref{eq:impact_misalignment_constant},
\(|\sigma_i|\le C_\zeta(1-\theta)/(1-\overline\beta)\) at every firm,
and \(|\overline\sigma_j|\) obeys the same bound as an average of
\(|\sigma_l|\) over the firm customers of \(j\) with total weight
\(1-r_j\). Every impact-date contrast is therefore at most
\(C_\xi(1-\theta)\) in absolute value, and so is every relative
expenditure impulse \(\xi_i^{(k)}\), which is an average of contrasts,
while \(|\phi_i^{(k)}|\le(1-\underline\beta)C_\xi(1-\theta)=C_\phi(1-\theta)\).
This is \eqref{eq:forcing_upper_bound}. The scaled relative expenditure impulse
\(\phi_i=s_i^{(\underline k)}\phi_i^{(\underline k)}\) obeys the same bound,
since \(0<s_i^{(\underline k)}\le1\). At a retail lead-time date
\(k\in\mathcal W\), the identity \eqref{eq:set_impulse_identity} on the
retail tier, Assumption~\ref{assump:impact_forcing_ordering} and the
definition \eqref{eq:retail_lag_dates} of the retail lead-time dates give
\[
\sum_{i\in\mathcal R}m_i^\ast\,s_i^{(k)}\phi_i^{(k)}
=
e_{\mathcal R}^{(k)}\,\xi_{\mathcal R}^{(k)}
\;\ge\;
s_R\sum_{i\in\mathcal R}m_i^\ast(1-\beta_i)\;c_\xi(1-\theta)
\;\ge\;
s_R(1-\overline\beta)\,c_\xi\,(1-\theta)\,M_{\mathcal R}
\]
which is the retail impulse bound \eqref{eq:retail_impulse_bound}.

We now compute the impact response. Firm \(i\)'s
dated input bundle at date \(\underline k\) consists of the orders placed at the
dates \(\underline k-k_{ji}\le0\), \(j\in\mathcal S_i\). The orders on suppliers
of lead time \(k_{ji}>\underline k\) were placed before the shock and are independent
of \(\pi\). The orders on the shortest-lead-time suppliers, \(k_{ji}=\underline k\),
were placed at date \(0\) at the price \(p_{j,0}=d_{j,0}/q_{j,0}\). And
\(q_{j,0}\) is produced from the labor \(l_j^\ast\) and from inputs ordered
at the dates \(-k_{lj}\le-\underline k<0\), so it too is independent of
\(\pi\), and \(\dot{\widehat p}_{j,0}=\dot{\widehat d}_{j,0}\). By the order
rule \eqref{eq:dynamic_real_alloc} and \eqref{eq:centered_budget_impulse},
\(\dot{\widehat x}_{ji,0}=\dot{\widehat e}_{i,0}-\dot{\widehat d}_{j,0}\)
for every \(j\in\mathcal S_i^{(\underline k)}\). The input-bundle scale
\eqref{eq:bundle_scale_def}, evaluated at the order dates, is therefore at date
\(\underline k\)
\[
\mathcal F_i(\dot{\widehat x}_{i,\underline k})
=
\sum_{j\in\mathcal S_i^{(\underline k)}}a_{ji}\bigl(\dot{\widehat e}_{i,0}-\dot{\widehat d}_{j,0}\bigr)
=
\xi_{i,\underline k}
\]
and the first-order production identity
\(D_\pi\log q_{i,t}=(1-\beta_i)\mathcal F_i(\dot{\widehat x}_{i,t})\) of
Lemma~\ref{lem:production_linearization} gives
\(D_\pi\log q_{i,\underline k}=(1-\beta_i)\xi_{i,\underline k}=\phi_i\) at every firm. No
cascade enters, because no supplier's output has moved by the time the
date-\(0\) orders are placed. Aggregating with the tilt weights, which
are flat on the retail tier, \(\omega_i(b)=\omega_{\mathcal R}(b)m_i^\ast/M_{\mathcal R}\)
for \(i\in\mathcal R\),
\begin{equation}
D_\pi\mathcal Q_{\underline k}(b)
=
\frac{\omega_{\mathcal R}(b)}{M_{\mathcal R}}\sum_{i\in\mathcal R}m_i^\ast s_i^{(\underline k)}\phi_i^{(\underline k)}
+\sum_{i\notin\mathcal R}\omega_i(b)\phi_i
\;\ge\;
(1-\overline\epsilon_w)\,c_R\,(1-\theta)-\overline\epsilon_wC_\phi(1-\theta)
\label{eq:impact_bound_positive_lag}
\end{equation}
by the retail impulse bound \eqref{eq:retail_impulse_bound} at the lead time
\(\underline k\), a retail lead-time date by hypothesis. This is because the retail
tier carries tilt mass \(\omega_{\mathcal R}(b)\ge1-\overline\epsilon_w\), and the
bound \eqref{eq:forcing_upper_bound} holds off it. At a firm with
\(\xi_i^{(\underline k)}\le-c_\xi(1-\theta)\),
\begin{equation}
D_\pi\log q_{i,\underline k}=\phi_i=s_i^{(\underline k)}(1-\beta_i)\,\xi_i^{(\underline k)}\le-s_i^{(\underline k)}(1-\beta_i)\,c_\xi(1-\theta)<0
\label{eq:impact_upstream_positive_lag}
\end{equation}
and for a set of firms \(\mathcal C\) with
\(\xi_{\mathcal C}^{(\underline k)}\ge c_\xi(1-\theta)\), the identity
\eqref{eq:set_impulse_identity} at the lead time \(\underline k\) gives
\begin{equation}
\sum_{i\in\mathcal C}m_i^\ast D_\pi\log q_{i,\underline k}
=
\sum_{i\in\mathcal C}m_i^\ast s_i^{(\underline k)}\phi_i^{(\underline k)}
=
e_{\mathcal C}^{(\underline k)}\,\xi_{\mathcal C}^{(\underline k)}
\;\ge\;
c_\xi(1-\theta)\,e_{\mathcal C}^{(\underline k)}
\label{eq:impact_cut_positive_lag}
\end{equation}

The bound \eqref{eq:impact_bound_positive_lag} gives
\(D_\pi\mathcal Q_{\underline k}(b)\ge\tfrac12c_R(1-\theta)\)
as soon as \(\overline\epsilon_w\le c_R/(2(c_R+C_\phi))\), which the
threshold \eqref{eq:downstream_threshold} satisfies, since
\(8(c_R+C_\phi/\underline\beta)\ge2(c_R+C_\phi)\). The bounds
\eqref{eq:impact_upstream_positive_lag} and
\eqref{eq:impact_cut_positive_lag} hold without a threshold.

Finally, every object in the argument is linear in the per-unit shock, so
a negative shock reverses both signs.
\end{proof}

\subsubsection{Proof of sign persistence over a finite horizon}
\label{app:proof_finite_horizon_sign_persistence}

\begin{proof}[Proof of Proposition~\ref{prop:finite_horizon_sign_persistence}]
We use the path expansion \eqref{eq:forced_representation} of
Lemma~\ref{lem:supplier_recursion}, whose proof in
Appendix~\ref{app:lemmas_reallocation} does not depend on this
one, grouping the supplier paths by their last buyer. For \(r\ge0\) let
\(\mathcal P'=(i_0\to\cdots\to i_r)\) be a path of \(r\) links, each \(i_m\) a
supplier of \(i_{m-1}\), ending at the buyer \(l=i_r\), with summed
lead time \(k(\mathcal P'):=\sum_{m=1}^rk_{i_mi_{m-1}}\), the sum of the lead
times of its links. The tilt mass of the path \(\mathcal P'\) is
\(\mu(\mathcal P'):=\omega_{i_0}(b)(1-\beta_{i_0})\prod_{m=1}^{r-1}(1-\beta_{i_m})\prod_{m=1}^{r}a_{i_mi_{m-1}}\)
for \(r\ge1\) and \(\mu(\mathcal P'):=\omega_{i_0}(b)\) for \(r=0\). The lemma says
that extending \(\mathcal P'\) by one link \(l\to j\) carries the mass
\(\mu(\mathcal P')(1-\beta_l)a_{jl}\) to the contrast \(\Xi_{lj,\,t-k(\mathcal P')-k_{jl}}\), and
that \(D_\pi\mathcal Q_t(b)\) is the sum of these terms over all paths
with \(k(\mathcal P')+k_{jl}\le t\). Write \(\mathcal P\) for \(\mathcal P'\) extended by the link \(l\to j\), and \(k(\mathcal P):=k(\mathcal P')+k_{jl}\) for its summed lead time. When the contrast is evaluated at the impact
date, \(k_{jl}=t-k(\mathcal P')=:k\), the links of that vintage aggregate, by the
two-stage shares \(a_{jl}=s_l^{(k)}\widetilde a^{(k)}_{jl}\) of
\eqref{eq:share_quantity_form}, to
\begin{equation}
\sum_{j\in\mathcal S_l^{(k)}}\mu(\mathcal P')\,(1-\beta_l)\,a_{jl}\,\Xi_{lj,0}
=
\mu(\mathcal P')\,(1-\beta_l)\,s_l^{(k)}\,\xi_l^{(k)}
=
\mu(\mathcal P')\,s_l^{(k)}\phi_l^{(k)}
\label{eq:vintage_grouping}
\end{equation}
The sum in \eqref{eq:vintage_grouping} is the scaled relative expenditure impulse of the buyer \(l\) on its
inputs of lead time \(k\), weighted by its cross-vintage share \(s_l^{(k)}\) and by the tilt mass of the path \(\mathcal P'\). Hence
\begin{equation}
I_t(b)
=
\sum_{r\ge0}\;\sum_{\mathcal P':\,k(\mathcal P')\le t-\underline k}\mu(\mathcal P')\,s_{i_r}^{(t-k(\mathcal P'))}\phi_{i_r}^{(t-k(\mathcal P'))}
\label{eq:impact_component_paths}
\end{equation}
where \(s_l^{(k)}=0\) when the lead time \(k\) is not active at \(l\). Two
bounds on the masses are used below. Summing \(\mu(\mathcal P')\) over all paths of
\(r\) links gives at most \((1-\underline\beta)^r\), because the
supplier shares of each buyer sum to one and each pass-through factor is
at most \(1-\underline\beta\). The paths that start off the retail tier, whose
tilt mass is at most \(\overline\epsilon_w\), have mass at most
\(\overline\epsilon_w(1-\underline\beta)^r\) at each length. And the mass of a path
is linear in the weight of its starting firm. Hence by
Lemma~\ref{lem:gdp_pass_through} the paths that start on the retail
tier, end at the buyer \(l\) and have summed lead time at most \(L\) have
mass at most \(\epsilon_h^{-1}(m_l^\ast/w^\ast)F_l(L)\)
over all lengths together. Here \(F_l\) is the distribution of the delivery time
\eqref{eq:delivery_time}.

We first bound the impact-date component \(I_t(b)\). Fix
\(t\in\mathcal W\). The \(r=0\) terms
of \eqref{eq:impact_component_paths} are the retailers' own maturing
orders, \(\sum_i\omega_i(b)s_i^{(t)}\phi_i^{(t)}\). The tilt weights are flat
on the retail tier, \(\omega_i(b)=\omega_{\mathcal R}(b)m_i^\ast/M_{\mathcal R}\), so the
retail impulse bound \eqref{eq:retail_impulse_bound} at the retail lead-time
date \(t\) gives at least \(\omega_{\mathcal R}(b)c_R(1-\theta)\ge(1-\overline\epsilon_w)c_R(1-\theta)\)
on the retail tier. Off
\(\mathcal R\) the tilt mass is at most \(\overline\epsilon_w\),
\(s_i^{(t)}\le1\), and \(\phi_i^{(t)}\ge-C_\phi(1-\theta)\) by
\eqref{eq:forcing_upper_bound}. Hence the \(r=0\) terms are at least
\([(1-\overline\epsilon_w)c_R-\overline\epsilon_wC_\phi](1-\theta)\).
For \(r\ge1\), the scaled relative expenditure impulse of a buyer in the
neighborhood \(\mathcal G\) is at least
\(-(1-\underline\beta)\overline\epsilon_\xi(1-\theta)\) by
\eqref{eq:scaled_impulse_def} under condition~\eqref{eq:hypothesis_N},
and elsewhere it is at least \(-C_\phi(1-\theta)\). The
total mass of the \(r\)-link paths is at most \((1-\underline\beta)^r\). Of the paths
ending outside \(\mathcal R^{+}=\mathcal G\cup\mathcal R\), those starting
off the retail tier have mass at most \(\overline\epsilon_w(1-\underline\beta)^r\).
Those starting on the retail tier and ending at a buyer
\(l\notin\mathcal R^{+}\) with the vintage \(k\) at \(l\) have summed lead time
\(t-k\), so over all lengths together they have mass at most
\(\epsilon_h^{-1}(m_l^\ast/w^\ast)F_l(t-k)\). And their impulse obeys
\(|\phi_l^{(k)}|=(1-\beta_l)|\xi_l^{(k)}|\le(1-\beta_l)C_\xi(1-\theta)\)
by \eqref{eq:forcing_upper_bound}. Their terms are therefore
at least
\begin{align*}
-C_\xi(1-\theta)\,\epsilon_h^{-1}
\sum_{l\notin\mathcal R^{+}}\sum_{k\in\mathcal K_l}\frac{m_l^\ast}{w^\ast}(1-\beta_l)\,s_l^{(k)}\,F_l(t-k)
&\;=\;
-C_\xi(1-\theta)\,\epsilon_h^{-1}\,\Gamma_{N\setminus\mathcal R^{+}}[0,t]\\
&\;\ge\;
-C_\xi(1-\theta)\,\epsilon_h^{-1}\,\Gamma_{N\setminus\mathcal R^{+}}[0,t_{\mathcal W}]
\end{align*}
by the definition \eqref{eq:delivered_purchases} of the delivered
purchases, which are nondecreasing in the horizon, and
\(t\le t_{\mathcal W}\). The
mass of those ending on the retail tier is at most
\(\min\{(1-\underline\beta)^r,1-\epsilon_h\}\). To prove this bound,
write \(\mu^{(r)}(l)\) for the mass of the \(r\)-link paths ending at the
buyer \(l\) and \(\overline M(b):=\sum_{j\in N}m_j^\ast\widehat\psi_j^{\,b}\).
Then \(\mu^{(0)}(l)=\omega_l(b)\le m_l^\ast/\overline M(b)\) and
\(\mu^{(r)}(l)=\sum_{l'}\mu^{(r-1)}(l')(1-\beta_{l'})a_{ll'}\), so by
induction, with \(m_{l'}^\ast(1-\beta_{l'})a_{ll'}=m_l^\ast\Omega_{ll'}\),
\(\mu^{(r)}(l)\le(m_l^\ast/\overline M(b))\sum_{l'}\Omega_{ll'}=(1-r_l)m_l^\ast/\overline M(b)\)
for every \(r\ge1\), where \(r_l\) is the consumption share \eqref{eq:retail_absorption_def} of firm \(l\). On the retail tier \(1-r_l\le1-\epsilon_h\) and
\(m_l^\ast/\overline M(b)=\omega_l(b)\), so summing over it gives
\(\mu^{(r)}(\mathcal R)\le(1-\epsilon_h)\omega_{\mathcal R}(b)\le1-\epsilon_h\).
Since \(\min\{(1-\underline\beta)^r,1-\epsilon_h\}\le((1-\underline\beta)^r(1-\epsilon_h))^{1/2}\),
\(\sum_{r\ge1}\mu^{(r)}(\mathcal R)\le(1-\epsilon_h)^{1/2}(1-\underline\beta)^{1/2}/(1-(1-\underline\beta)^{1/2})\).
Summing over \(r\ge1\) with \(\sum_r(1-\underline\beta)^r=(1-\underline\beta)/\underline\beta\),
\begin{equation}
I_t(b)
\;\ge\;
\Bigl[c_R
-\overline\epsilon_w\Bigl(c_R+\frac{C_\phi}{\underline\beta}\Bigr)
-\frac{(1-\underline\beta)^2\,\overline\epsilon_\xi}{\underline\beta}
-\frac{C_\xi}{\epsilon_h}\,\Gamma_{N\setminus\mathcal R^{+}}[0,t_{\mathcal W}]
-\frac{C_\phi\,(1-\epsilon_h)^{1/2}(1-\underline\beta)^{1/2}}{1-(1-\underline\beta)^{1/2}}\Bigr](1-\theta)
\label{eq:persistence_impact_bound}
\end{equation}
The bound \eqref{eq:persistence_impact_component} follows as soon as each
of the four negative terms is at most \(c_R/8\), that is, under
\begin{gather*}
\overline\epsilon_w\le\frac{c_R}{8\,(c_R+C_\phi/\underline\beta)},
\qquad
\overline\epsilon_\xi\le\frac{c_R\,\underline\beta}{8\,(1-\underline\beta)^2},
\qquad
\Gamma_{N\setminus\mathcal R^{+}}[0,t_{\mathcal W}]\le\frac{c_R\,\epsilon_h}{8\,C_\xi},\\
1-\epsilon_h\le\Bigl[\frac{c_R\,(1-(1-\underline\beta)^{1/2})}{8\,C_\phi\,(1-\underline\beta)^{1/2}}\Bigr]^2
\end{gather*}
The first two hold with equality by the definitions
\eqref{eq:downstream_threshold} of the threshold of the
downstream-concentrated tilts and \eqref{eq:persistence_thresholds} of the
tolerance of condition~\eqref{eq:hypothesis_N}. The other two are the hypotheses
\(\Gamma_{N\setminus\mathcal R^{+}}[0,t_{\mathcal W}]\le\overline\epsilon_F\) and
\(1-\epsilon_h\le\overline\epsilon_h\) of the statement, with the
thresholds \eqref{eq:persistence_thresholds}.
The bound \eqref{eq:persistence_impact_component} is uniform over the dates of \(\mathcal W\), since the
delivered purchases of the firms beyond the neighborhood of final demand are taken at the last retail lead-time
date and the same four thresholds serve at every date.

We finally turn to the post-impact component \(J_t(b)\), which
collects the paths with \(k(\mathcal P)<t\) and carries contrasts at dates
\(s=t-k(\mathcal P)\ge1\). Every path has at least one link and every link a
lead time of at least \(\underline k\), so \(k(\mathcal P)\ge\underline k\). No
contrast of a date \(s\ge1\) therefore enters before the date
\(\underline k+1\), and \(J_{\underline k}(b)=0\). At every date,
\(|J_t(b)|\) is at most the tilt mass that the paths with \(k(\mathcal P)<t\)
carry, times the largest contrast \(|\Xi_{ij,s}|\) at the dates
\(1\le s\le t-\underline k\). Condition \eqref{eq:hypothesis_P}
together with \eqref{eq:persistence_impact_component} gives
\(D_\pi\mathcal Q_t(b)=I_t(b)+J_t(b)\ge\tfrac14c_R(1-\theta)\).
\end{proof}

\subsubsection{Proof of the two-phase reallocative impulse response}
\label{app:proof_sign_reversing_irf}

\begin{proof}[Proof of Theorem~\ref{thm:sign_reversing_irf}]
The positivity of \(D_\pi\mathcal Q_t(b)\) on the retail lead-time dates is
Proposition~\ref{prop:finite_horizon_sign_persistence}.

For the arriving contraction, set
\begin{equation}
r_a:=\Bigl\lceil\frac{2}{\Delta_\psi}\Bigr\rceil,
\qquad
\overline\tau:=(r_a+1)\,\widehat k,
\qquad
\underline\mu:=\frac{7\,(1-\epsilon_k)(1-\overline\beta)^{r_a}}{16\,(\overline\tau-\underline k-k_\ast+1)}
\label{eq:reversal_constants}
\end{equation}
with \(\Delta_\psi\) the height of a tier in
Assumption~\ref{assump:transmission_delay} and \(\epsilon_k\), \(k_\ast\)
the constants of Assumption~\ref{assump:lag_position_sorting}. Here
\(r_a\) is the number of rounds of sourcing within which the supplier
walk from the retail tier has reached the contracting region with
probability at least one half. And \(1-\epsilon_k\) is the share of the
contracting region's impact-date orders that mature \(k_\ast\) periods
or more later. Lemma~\ref{lem:layered_routing} gives the date
\(\tau^-\in[\underline k+k_\ast,\overline\tau]\), beyond the retail lead-time
dates, with \(\mu_{\tau^-}(\mathcal T)\ge\underline\mu\).
Grouping the impact-date atoms at each buyer \(l\in\mathcal T\) by vintage,
as in \eqref{eq:vintage_grouping}, and using
\eqref{eq:scaled_impulse_def}, which under condition~\eqref{eq:hypothesis_U}
gives
\(\phi_l^{(k)}=(1-\beta_l)\xi_l^{(k)}\le-(1-\overline\beta)c_\xi(1-\theta)\)
at every firm of \(\mathcal T\) and every active vintage,
\begin{align*}
-I_{\tau^-}^{\mathcal T}(b)
&=
-\sum_{(l,j):\,l\in\mathcal T}\varpi^{(\tau^-,0)}_{lj}(b)\,\Xi_{lj,0}
=
-\sum_{\mathcal P'}\mu(\mathcal P')\,s_{l}^{(k)}\phi_{l}^{(k)}\\
&\ge\;
(1-\overline\beta)\,c_\xi(1-\theta)\sum_{\mathcal P'}\mu(\mathcal P')\,s_l^{(k)}
=
(1-\overline\beta)\,c_\xi(1-\theta)\,\mu_{\tau^-}(\mathcal T)
\;\ge\;
\underline\mu\,(1-\overline\beta)\,c_\xi\,(1-\theta)
\end{align*}
Here the sums run over the paths \(\mathcal P'\) that end at a buyer
\(l\in\mathcal T\) with \(k:=\tau^--k(\mathcal P')\ge\underline k\), and the last
identity uses \(\sum_{j\in\mathcal S_l^{(k)}}a_{jl}=s_l^{(k)}\).

For the reversal, grouping the impact-date terms at each buyer
\(l\notin\mathcal T\) by vintage as in \eqref{eq:vintage_grouping},
\(I_{\tau^-}^{N\setminus\mathcal T}(b)=\sum_{\mathcal P'}\mu(\mathcal P')s_l^{(k)}\phi_l^{(k)}\le C_\phi(1-\theta)\sum_{\mathcal P'}\mu(\mathcal P')s_l^{(k)}=C_\phi(1-\theta)\,\mu_{\tau^-}(N\setminus\mathcal T)\)
by the bound \eqref{eq:forcing_upper_bound}. With the lower bound on
the arriving contraction, the split of \eqref{eq:forced_representation}
at \(\tau^-\) gives
\(D_\pi\mathcal Q_{\tau^-}(b)=I_{\tau^-}^{\mathcal T}(b)+I_{\tau^-}^{N\setminus\mathcal T}(b)+J_{\tau^-}(b)\le C_\phi(1-\theta)\,\mu_{\tau^-}(N\setminus\mathcal T)+J_{\tau^-}(b)-(1-\overline\beta)\,c_\xi(1-\theta)\,\mu_{\tau^-}(\mathcal T)\),
which is negative under \eqref{eq:hypothesis_R}. Linearity in the
per-unit shock gives the mirror image for a negative shock.
\end{proof}

\subsection{Proofs of Section \ref{sec:asymmetry}: Monetary Non-neutrality}
\label{app:asymmetry}

\subsubsection{Proof of asymmetric monetary non-neutrality}
\label{app:proof_asymmetric_nonneutrality}

\begin{proof}[Proof of Theorem~\ref{thm:asymmetric_nonneutrality}]
By \eqref{eq:concave_scaling} at each date, GDP over the first reporting
interval is
\(f_+(\pi)=(\underline k+1)^{-1}\sum_{t\le\underline k}\mathcal Y_t(\tilde\pi\,h_\theta)\),
a finite average of maps real-analytic near zero that vanish there
(Lemma~\ref{lem:concave_scaling}). So \(f_+(0)=0\),
\(\dot f_+\ge c_Y(1-\theta)\) by Lemma~\ref{lem:first_interval_gdp}, and
\[
\ddot f_+
=-2\dot f_++\frac{1}{\underline k+1}\sum_{t\le\underline k}D^2\mathcal Y_t(\mathbf 0)[h_\theta,h_\theta]
\;\le\;-2c_Y(1-\theta)+K_+C_\zeta^2(1-\theta)^2
\;\le\;-c_Y(1-\theta)\;<\;0
\]
by \(\|h_\theta\|_m\le C_\zeta(1-\theta)\), which is
\eqref{eq:impact_misalignment_constant}, and the slack condition.
Lemma~\ref{lem:concave_response_asymmetry} uses only these two
derivative signs of a function of \(\pi\) that vanishes at zero. Applied
to \(f_+\), it gives the two signs and the asymmetry, with its radius
intersected with the bound \(\overline\pi\) of
Proposition~\ref{prop:stability_stationary}.

For the second part, let the slack also satisfy
\(1-\theta\le c_Y/(K_PC_\zeta^2)\), and suppose that the two signs of the
first part hold. The identity \eqref{eq:gap_interval_exact} at
\(T=\underline k\) and the expansions \eqref{eq:gap_expansions} of
Lemma~\ref{lem:gap_exact}, divided by \(\underline k+1\), give
\begin{align*}
&|f_+(-\pi)|-|f_+(\pi)|-\bigl[\Lambda_+(\pi)+\Lambda_+(-\pi)\bigr]\\
&\qquad=\frac{2\pi^2}{1-\pi^2}\,\dot f_++\frac{P^\sharp_{\underline k+1}(\pi)+P^\sharp_{\underline k+1}(-\pi)}{(\underline k+1)w^\ast}
=\Bigl(2\dot f_++\frac{D^2\mathcal Y^{P}_{\underline k+1}(\mathbf 0)[h_\theta,h_\theta]}{(\underline k+1)w^\ast}\Bigr)\pi^2+O(\pi^4)
\end{align*}
The bilinear bound and \eqref{eq:impact_misalignment_constant} give
\(|D^2\mathcal Y^{P}_{\underline k+1}(\mathbf 0)[h_\theta,h_\theta]|/((\underline k+1)w^\ast)\le K_PC_\zeta^2(1-\theta)^2\le c_Y(1-\theta)\),
and \(\dot f_+\ge c_Y(1-\theta)\). The right side is therefore at least
\(c_Y(1-\theta)\pi^2+O(\pi^4)\), and it is positive once \(\pi_\ast\) is
reduced enough.
\end{proof}

A quarterly figure reports the average of the proportional changes,
\((\underline k+1)^{-1}\sum_{t\le\underline k}\bigl(e^{f_t(\pi)}-1\bigr)\)
over the first reporting interval. It has the same first derivative at
zero as \(f_+\), and its Hessian adds the squares of the first-order
responses, which are of the order of \((1-\theta)^2\). The first part of
Theorem~\ref{thm:asymmetric_nonneutrality} therefore holds for it as
well, with \(K_+\) raised by
\((\underline k+1)^{-1}\sum_{t\le\underline k}\|D\mathcal Y_t(\mathbf 0)\|_m^2\).

\subsubsection{Proof of intertemporal resource accounting}
\label{app:proof_resource_accounting}

\begin{proof}[Proof of Proposition~\ref{prop:resource_accounting}]
By Lemma~\ref{lem:resource_accounting_date} of
Appendix~\ref{app:analytical}, the identity \eqref{eq:resource_identity}
holds at every date, with \(\Lambda_t\ge0\), and the consumption loss
\(H_t\) is positive when \(c_{i,t}\ne c_i^\ast\) for some
\(i\in\mathcal R\).

Every order outstanding
at date \(0\) was placed before the shock, so \(P_0=0\). An order
placed at a date \(s\ge0\) deviates proportionally by at most
\(C_S\kappa^s|\pi|(1-\theta)\), by \eqref{eq:post_shock_geometric_decay}
and the bound it carries for the orders placed at date \(s\), and an
order placed before the shock does not deviate. Since
\(\sum_i\sum_{j\in\mathcal S_i}p_j^\ast x_{ji}^\ast=\sum_i(1-\beta_i)m_i^\ast\le\overline M\),
\(|P_t|\le\widehat k\,\overline MC_S\kappa^{t-\widehat k}|\pi|(1-\theta)\to0\).

In the notation of
Appendix~\ref{app:proof_stability}, household consumption obeys
\(c_{j,t}/c_j^\ast=(1+\pi)p_j^\ast/p_{j,t}=e^{\check q_{j,t}}/(1+\widehat m_{j,t+1}^\perp/(1+\pi))\),
so the envelopes \eqref{eq:stab_real_envelope} and
\eqref{eq:stab_nominal_envelope} give
\begin{equation}
|f_t(\pi)|
\;\le\;
\sum_{i\in\mathcal R}\gamma_i\Bigl|\log\frac{c_{i,t}}{c_i^\ast}\Bigr|
\;\le\;
C_c\,\kappa^t|\pi|(1-\theta),
\qquad
C_c:=C_q+4C_\lambda C_\zeta
\label{eq:consumption_envelope}
\end{equation}
Summing \eqref{eq:resource_identity} from \(0\) to \(T\),
\(\sum_{t\le T}f_t(\pi)=-P_{T+1}/w^\ast-\sum_{t\le T}\Lambda_t(\pi)\). As
\(T\to\infty\) the left side converges by
\eqref{eq:consumption_envelope} and \(P_{T+1}\to0\), so the
nondecreasing partial sums of the nonnegative losses converge to the
same limit, which is \eqref{eq:cumulative_gdp}. The inequality is
strict when consumption departs from its stationary value at some
date, since the consumption loss is then positive.

The maps \(\pi\mapsto f_t(\pi)\), \(P_t\)
and \(\Lambda_t\) are differentiable at zero
(Lemma~\ref{lem:propagation_linearization}). \(\Lambda_t\) is
nonnegative and vanishes at \(\pi=0\), so its derivative at zero
vanishes, and \eqref{eq:resource_identity} gives
\(w^\ast\dot f_t=\dot P_t-\dot P_{t+1}\), whence
\(w^\ast\sum_{t\le T}\dot f_t=-\dot P_{T+1}\). The bound on the orders
placed after the shock holds for every \(|\pi|\le\overline\pi\), so the
first-order responses of the orders obey \(|\dot x_{ji,s}|/x_{ji}^\ast\le C_S\kappa^s(1-\theta)\)
and \(\dot P_{T+1}\to0\), while \eqref{eq:consumption_envelope} gives
\(|\dot f_t|\le C_c\kappa^t(1-\theta)\). The series therefore converges
absolutely, to zero.
\end{proof}

\subsubsection{Proof of implicit interest cost of a monetary expansion}
\label{app:proof_implicit_interest}

\begin{proof}[Proof of Proposition~\ref{prop:implicit_interest}]
Since \(f_t=\max\{f_t,0\}-\max\{-f_t,0\}\) and
both series converge absolutely by \eqref{eq:consumption_envelope},
\(L_Y(\pi)-G_Y(\pi)=-\sum_tf_t(\pi)\), which is \(\sum_t\Lambda_t(\pi)\)
by \eqref{eq:cumulative_gdp}. When GDP moves at some date some consumed
quantity has moved, the corresponding consumption loss \(H_t\) is
positive, and so is the excess.

For the second claim, write \(G_0:=\sum_{t\le T_0}f_t(\pi)>0\). By
\eqref{eq:cumulative_gdp}, \(\sum_{t>T_0}f_t(\pi)\le-G_0<0\), so some
date after \(T_0\) has \(f_t(\pi)<0\). For the horizon, let \(T>T_0\)
satisfy \(C_c\kappa^{T+1}|\pi|(1-\theta)\le\tfrac12(1-\kappa)G_0\). Then
by \eqref{eq:consumption_envelope}
\(\sum_{t=T_0+1}^{T}f_t(\pi)\le-G_0+\sum_{t>T}|f_t(\pi)|\le-\tfrac12G_0\),
and one of the \(T-T_0\) terms is at most \(-G_0/(2(T-T_0))\).

For the last claim, fix the economy and write
\(\widehat{\mathbf x}_{i,t}\) for the proportional changes of the dated inputs
of firm \(i\) at date \(t\), as in \eqref{eq:usable_input_envelope}.
By \(p_j^\ast x_{ji}^\ast=(1-\beta_i)a_{ji}m_i^\ast\), the value \(E_t\) of
\eqref{eq:production_loss} in Appendix~\ref{app:lemmas_asymmetry} is
\(E_t=\sum_im_i^\ast\mathcal E_i(\widehat{\mathbf x}_{i,t})\), with
\(\mathcal E_i(\widehat{\mathbf x}):=(1-\beta_i)\mathcal F_i(\widehat{\mathbf x})-q_i\bigl((x_{ji}^\ast(1+\widehat x_{ji}))_j,l_i^\ast\bigr)/q_i^\ast+1\),
where \(q_i(\mathbf x_i,l_i)\) is the production function \eqref{eq:dynamic_output}.
Each \(\mathcal E_i\) is smooth on \(\|\widehat{\mathbf x}\|_\infty\le\tfrac12\),
nonnegative by the argument of
Lemma~\ref{lem:resource_accounting_date}, and zero at zero, so its
minimum is at zero and its gradient vanishes there. Hence
\(\mathcal E_i(\widehat{\mathbf x})\le C_{\mathcal E}\|\widehat{\mathbf x}\|_\infty^2\)
on that ball for a constant \(C_{\mathcal E}\). And \(z-1-\log z\le2(z-1)^2\) on
\([\tfrac12,\tfrac32]\) with \(|z-1|\le2|\log z|\). The dated-input
envelope \eqref{eq:usable_input_envelope} and
\eqref{eq:consumption_envelope} therefore give
\(E_t+H_t\le(C_{\mathcal E}\overline MC_S^2+8w^\ast C_c^2)\kappa^{2t}\pi^2(1-\theta)^2\),
and \(\sum_t(E_t+H_t)\le C_\Lambda\pi^2\) with
\(C_\Lambda:=(C_{\mathcal E}\overline MC_S^2+8w^\ast C_c^2)/(1-\kappa^2)\). If
\(\dot f_{t_1}>0\), then \(f_{t_1}(\pi)\ge\tfrac12\dot f_{t_1}\pi\) for every
sufficiently small \(\pi>0\), so \(G_Y(\pi)\ge\tfrac12\dot f_{t_1}\pi\), and
by the first part
\(0<\mathcal I(\pi)=\bigl(L_Y(\pi)-G_Y(\pi)\bigr)/G_Y(\pi)\le2C_\Lambda\pi/(w^\ast\dot f_{t_1})\).
\end{proof}

\appendix
\setcounter{section}{1}
\counterwithin{lemma}{section}
\renewcommand{\thelemma}{\thesection.\arabic{lemma}}
\renewcommand{\theHlemma}{\thesection.\arabic{lemma}}
\counterwithin{equation}{section}
\counterwithin{figure}{section}
\counterwithin{table}{section}
\renewcommand{\theHequation}{\thesection.\arabic{equation}}
\renewcommand{\theHfigure}{\thesection.\arabic{figure}}
\renewcommand{\theHtable}{\thesection.\arabic{table}}

\onehalfspacing
\newpage
\section{Analytical Results and Proofs}
\label{app:analytical}
\label{app:lemmas}

The conventions of Appendix~\ref{app:proofs} of the paper, in
particular the standing
Assumption~\ref{assump:short_run_nominal_rigidity} of
Section~\ref{subsec:stability} and the bound \(\overline\pi\) of
Section~\ref{subsec:money}, hold throughout without restatement.
\subsection{Lemmas for Section~\ref{sec:model}: The Model}
\label{app:lemmas_model}

Let \(\boldsymbol{\mathcal X}_t\) be the dated-input pipeline of
Definition~\ref{def:pipeline} and \(\mathbf S_t=(\mathbf
m_t,\boldsymbol{\mathcal X}_t)\) the augmented state. Since the lead-time
support is finite, the state space is finite-dimensional, and the
decentralized economy induces a first-order system \(\mathbf
S_{t+1}=\Phi(\mathbf S_t)\) with stationary state \(\mathbf
S^\ast=(\mathbf m^\ast,\boldsymbol{\mathcal X}^\ast)\). At the stationary state \(\mathbf S^\ast\) all
balances and pipeline entries are strictly positive
(Proposition~\ref{prop:existence_stationary}) and every balance exceeds
the wage bill by the intermediate expenditure,
\(m_i^\ast-w^\ast l_i^\ast=(1-\beta_i)m_i^\ast>0\). Intermediate-input
expenditures are therefore strictly positive on a neighborhood of
\(\mathbf S^\ast\), and along the local trajectory by the no-shutdown
margin \eqref{eq:no_shutdown_margin} of the paper. So
the production function, the clearing rule, the wage rule and the order
rule are smooth in a neighborhood of \(\mathbf S^\ast\), the transition
map \(\Phi\) is \(C^1\) there, and \(\mathbf J_\Phi^\ast:=D\Phi(\mathbf
S^\ast)\) is its Jacobian at the stationary state.

\begin{lemma}[Linearization of the transition map]
\label{lem:propagation_linearization}
In a neighborhood of the stationary state \(\mathbf S^\ast\) on which
intermediate-input expenditures are strictly positive, the within-period
clearing problem, which determines the current outputs, prices and
orders given the start-of-period state, has a unique positive
solution. The transition map \(\Phi\)
together with the dated-input observation maps
\(\mathbf x_{i,t}=\mathcal V_i(\mathbf S_t)\), which return firm \(i\)'s dated input bundle
at date \(t\) (inherited entries for \(k_{ji}\ge1\), current orders for
\(k_{ji}=0\)), are real-analytic there. The impact state
\(\mathbf S_0(\pi)=(\mathbf m^\ast+\pi\overline M\boldsymbol\zeta,\,
\boldsymbol{\mathcal X}^\ast)\) is affine in \(\pi\). Hence, for every
fixed finite date \(t\ge 0\), with
\(\mathbf J_\Phi^\ast:=D\Phi(\mathbf S^\ast)\) and
\(\dot{\mathbf S}_0:=D_\pi\mathbf S_0(\pi)|_{\pi=0}=(\overline M\boldsymbol\zeta,\,\mathbf 0)\),
\begin{equation}
\mathbf S_t(\pi)=\mathbf S^\ast+\pi\,\mathbf J_\Phi^{\ast\,t}\dot{\mathbf S}_0+O_t(\pi^2),
\qquad
\Delta\mathbf x_{i,t}(\pi)\;=\;\pi\,\dot{\mathbf x}_{i,t}\;+\;O_{i,t}(\pi^2),
\qquad
\dot{\mathbf x}_{i,t}:=D\mathcal V_i(\mathbf S^\ast)\,\mathbf J_\Phi^{\ast\,t}\dot{\mathbf S}_0
\label{eq:bundle_first_order_taylor}
\end{equation}
\end{lemma}

\begin{proof}[Proof of Lemma~\ref{lem:propagation_linearization}]
Fix a start-of-period state
\(\mathbf S=(\mathbf m,\boldsymbol{\mathcal X})\) with strictly positive
balances, expenditures and inherited orders, with the wage given by
the rule of Assumption~\ref{assump:short_run_nominal_rigidity}(b) of
the paper, \(\wM=(w^\ast/\overline M)\mathbf 1^\top\mathbf m\). By
\eqref{eq:dynamic_intermediate_budget} and \eqref{eq:dyn_nominal_demand}
the expenditures \(e_i=m_i-\wM l_i^\ast\) and the demands
\(d_j=\sum_ia_{ji}e_i+\gamma_j\wM\) are affine in \(\mathbf m\) and do
not involve quantities. For a supplier with \(k_{ji}\ge1\) the dated
input \(x_{ji}\) is a coordinate of \(\boldsymbol{\mathcal X}\). For a
zero-lead-time supplier the order rule \eqref{eq:dynamic_real_alloc} and the
clearing rule \eqref{eq:dyn_price_update} give
\(x_{ji}=a_{ji}e_iq_j/d_j\), so that
\(\log x_{ji}=\log(a_{ji}e_i/d_j)+\log q_j\). Write
\(\check{\mathbf q}=(\log(q_i/q_i^\ast))_{i\in N}\) for the log outputs relative to their stationary values. Substituting the zero-lead-time orders into
the production identity \eqref{eq:dynamic_output} expresses each
\(\check q_i\) as a function of \(\check{\mathbf q}\) and \(\mathbf S\),
\begin{equation}
\check q_i=\mathcal N_i(\check{\mathbf q};\mathbf S),
\qquad
\frac{\partial\mathcal N_i}{\partial\check q_j}(\check{\mathbf q};\mathbf S)
=(1-\beta_i)\,a_{ji}(\check{\mathbf q};\mathbf S)\,\ind_{\{k_{ji}=0\}}
\label{eq:within_period_fixed_point}
\end{equation}
Here \(a_{ji}(\check{\mathbf q};\mathbf S)\ge0\) is the elasticity of the
intermediate composite \(X_i\) with respect to the input \(x_{ji}\),
evaluated at the bundle in question, and
\(\sum_{j\in\mathcal S_i}a_{ji}(\check{\mathbf q};\mathbf S)=1\), since
the composite is homogeneous of degree one.
Hence \(\sum_j|\partial\mathcal N_i/\partial\check q_j|\le1-\beta_i\le1-\underline\beta<1\)
at every \(\check{\mathbf q}\in\mathbb R^n\), and the mean-value theorem gives
\(\|\mathcal N(\check{\mathbf q};\mathbf S)-\mathcal N(\check{\mathbf q}';\mathbf S)\|_\infty\le(1-\underline\beta)\|\check{\mathbf q}-\check{\mathbf q}'\|_\infty\).
By Banach's fixed-point theorem \eqref{eq:within_period_fixed_point}
has a unique solution \(\check{\mathbf q}=\check{\mathbf q}(\mathbf S)\), and the outputs, the
prices \(p_j=d_j/q_j\), the orders \(x_{ji}=a_{ji}e_i/p_j\) and the
next state, \(\mathbf m'=\mathbf d\) by the balance law and
\(\boldsymbol{\mathcal X}'\) the pipeline aged by one period, follow.

The map
\((\check{\mathbf q},\mathbf S)\mapsto\check{\mathbf q}-\mathcal N(\check{\mathbf q};\mathbf S)\)
is real-analytic on the open set of states with strictly positive
balances, expenditures and inherited orders, being composed of
exponentials, logarithms, powers of positive quantities, sums and
quotients with positive denominators. Its derivative in \(\check{\mathbf q}\)
at the solution is \(\mathbf I-\mathbf J_{\mathcal N}(\mathbf S)\), with
\(\mathbf J_{\mathcal N}(\mathbf S):=(\partial\mathcal N_i/\partial\check q_j)_{i,j}\)
non-negative with row sums at most \(1-\underline\beta<1\), so
\(\|\mathbf J_{\mathcal N}(\mathbf S)\|_\infty<1\) and \(\mathbf I-\mathbf J_{\mathcal N}(\mathbf S)\)
is invertible. By the analytic implicit-function theorem \(\check{\mathbf q}(\mathbf S)\), and
with it the prices, the orders, the transition map \(\Phi\) and the
observation maps \(\mathcal V_i\), are real-analytic on a neighborhood
\(\mathcal U\) of \(\mathbf S^\ast\), on which
expenditures stay positive.

The impact state
\(\mathbf S_0(\pi)=(\mathbf m^\ast+\pi\overline M\boldsymbol\zeta,\boldsymbol{\mathcal X}^\ast)\)
is affine in \(\pi\) with \(\mathbf S_0(0)=\mathbf S^\ast\), and
\(\Phi(\mathbf S^\ast)=\mathbf S^\ast\). Fix \(t\). Since \(\Phi\) is
continuous on \(\mathcal U\), there is \(\pi_t>0\) such that
\(\Phi^{\,u}(\mathbf S_0(\pi))\in\mathcal U\) for all \(u\le t\) and
\(|\pi|\le\pi_t\), and on that interval
\(\mathbf S_t(\pi)=\Phi^{\,t}(\mathbf S_0(\pi))\) is a composition of
real-analytic maps, hence real-analytic in \(\pi\). By the chain rule,
with every derivative evaluated at the fixed point \(\mathbf S^\ast\),
\(D_\pi\mathbf S_t(\pi)|_{\pi=0}=\mathbf J_\Phi^{\ast\,t}\dot{\mathbf S}_0\),
and Taylor's theorem gives the first display of
\eqref{eq:bundle_first_order_taylor} with a remainder constant that
depends on \(t\). Composing with \(\mathcal V_i\) and applying the chain rule
once more gives the second display, with
\(\dot{\mathbf x}_{i,t}=D\mathcal V_i(\mathbf S^\ast)\mathbf J_\Phi^{\ast\,t}\dot{\mathbf S}_0\).
The leading coefficient depends only on \(\mathbf J_\Phi^\ast\),
\(D\mathcal V_i(\mathbf S^\ast)\) and \(\dot{\mathbf S}_0\), which are fixed by
\(\mathbf A\), \(\boldsymbol\beta\), \(\boldsymbol\gamma\),
\(\boldsymbol\zeta\), the lead-time profile and the stationary allocation,
and not on \(\pi\).
\end{proof}

\subsection{Lemmas for Section~\ref{sec:inefficiency}: Miscoordination within Firms}
\label{app:lemmas_inefficiency}

The supplier shares \(a_{ji}=s_i^{(k_{ji})}\widetilde a_{ji}^{(k_{ji})}\) of
\eqref{eq:share_quantity_form} give the supplier-averaging operator
\(\mathcal F_i\) of Definition~\ref{def:supplier_operator} a two-stage
structure. With
\(\mathcal F_i^{(k)}(y):=\sum_{j\in\mathcal S_i^{(k)}}\widetilde a_{ji}^{(k)}y_j\)
the within-vintage average over vintage layer \(k\) and \(s_i^{(k)}\) the
cross-vintage share of that vintage layer, the tower property gives
\begin{equation}
\mathcal F_i(y)=\sum_{k=0}^{\widehat k}s_i^{(k)}\,\mathcal F_i^{(k)}(y)
\label{eq:operator_tower}
\end{equation}
For supplier-indexed quantities \(y=(y_j)_{j\in\mathcal S_i}\) and
\(z=(z_j)_{j\in\mathcal S_i}\), let
\(\operatorname{Cov}_{\mathcal F_i}(y,z):=\mathcal F_i(yz)-\mathcal F_i(y)\mathcal F_i(z)\),
with \(yz\) the entrywise product, and
\(\operatorname{Var}_{\mathcal F_i}(y):=\operatorname{Cov}_{\mathcal F_i}(y,y)\)
be the covariance and variance against the supplier-share weights of
firm \(i\). And let \(\operatorname{Cov}_{\mathcal F_i^{(k)}}\) and
\(\operatorname{Var}_{\mathcal F_i^{(k)}}\) be the same objects taken
within vintage layer \(k\). The law of total variance splits the supplier-share
variance into the within-vintage and the between-vintage variance
\eqref{eq:vintage_variances} of the paper,
\begin{equation}
\operatorname{Var}_{\mathcal F_i}(y)
=
\underbrace{\sum_{k}s_i^{(k)}\,\operatorname{Var}_{\mathcal F_i^{(k)}}(y)}_{\textstyle V_i^{\mathrm w}(y)}
\;+\;
\underbrace{\sum_{k}s_i^{(k)}\bigl(\mathcal F_i^{(k)}(y)-\mathcal F_i(y)\bigr)^2}_{\textstyle V_i^{\mathrm b}(y)}
\label{eq:ltv_decomposition}
\end{equation}
and the covariance splits in the same way,
\(\operatorname{Cov}_{\mathcal F_i}=\operatorname{Cov}^{\mathrm w}_{i}+\operatorname{Cov}^{\mathrm b}_{i}\),
with
\begin{equation}
\begin{aligned}
\operatorname{Cov}^{\mathrm w}_{i}(y,z)&:=\sum_k s_i^{(k)}\operatorname{Cov}_{\mathcal F_i^{(k)}}(y,z),\\
\operatorname{Cov}^{\mathrm b}_{i}(y,z)&:=\sum_k s_i^{(k)}\bigl(\mathcal F_i^{(k)}(y)-\mathcal F_i(y)\bigr)\bigl(\mathcal F_i^{(k)}(z)-\mathcal F_i(z)\bigr)
\end{aligned}
\label{eq:cov_split}
\end{equation}
The \emph{curvature-weighted covariance form} of the paper is
\begin{equation}
\langle y,z\rangle_i
:=
(1-\rho)\,\operatorname{Cov}^{\mathrm w}_{i}(y,z)
+(1-\rho_c)\,\operatorname{Cov}^{\mathrm b}_{i}(y,z)
\label{eq:curvature_weighted_cov}
\end{equation}
It is symmetric and bilinear, and positive semidefinite because both
weights are positive. Its quadratic form is the curvature-weighted
dispersion of Definition~\ref{def:production_inefficiency} of the paper,
\(\langle y,y\rangle_i=(1-\rho)V_i^{\mathrm w}(y)+(1-\rho_c)V_i^{\mathrm b}(y)\),
and in the flat case \(\rho_c=\rho\) the form reduces to
\((1-\rho)\operatorname{Cov}_{\mathcal F_i}\).

\begin{lemma}[Linearization of production at the stationary equilibrium]
\label{lem:production_linearization}
Let the active stationary inputs of firm \(i\) be strictly positive and let
its labor be held at \(l_i^\ast\). Then the input-bundle deviation
decomposes as in \eqref{eq:scale_composition_raw}, with
\(\nabla_i^\ast\!\cdot\mathbf z_{i,t}=0\), and the second-order expansion of
\(\log q_{i,t}\) around \(\mathbf x_i^\ast\) is
\eqref{eq:log_q_scale_composition}.
\end{lemma}

\begin{proof}[Proof of Lemma~\ref{lem:production_linearization}]
The production identity is
\[
\log q_i
=
\beta_i\log l_i^\ast
+
(1-\beta_i)\log X_i
\]
Since \(l_i^\ast\) is fixed in the input-bundle expansion, the gradient of
\(\log q_i\) with respect to \(\mathbf x_i\) is
\(\nabla_i^\ast=(1-\beta_i)D_{\mathbf x}\log X_i(\mathbf x_i^\ast)\). For the
two-level composite \eqref{eq:materials_bundle}, cost minimization at the
stationary equilibrium makes the elasticity of \(\log X_i\) with respect
to \(\log x_{ji}\) equal to the supplier share
\(s_i^{(k)}\widetilde a_{ji}^{(k)}=a_{ji}\) of
\eqref{eq:share_quantity_form}, by the envelope property
of CES nests. The supplier share is the cross-vintage share of the supplier's
vintage layer times the within-vintage share of the supplier. Hence \(\nabla_{ji}^\ast=(1-\beta_i)a_{ji}/x_{ji}^\ast\).

We next verify the scale--composition decomposition. With \(\mathcal F_i\)
the supplier-averaging operator of Definition~\ref{def:supplier_operator},
set \(\mathbf z_{i,t}:=\Delta\mathbf x_{i,t}-\mathcal F_i(\widehat x_{i,t})\,\mathbf x_i^\ast\),
so that \(\Delta\mathbf x_{i,t}=\mathcal F_i(\widehat x_{i,t})\,\mathbf x_i^\ast+\mathbf z_{i,t}\).
Since
\(\nabla_i^\ast\cdot\Delta\mathbf x_{i,t}=(1-\beta_i)\sum_{j\in\mathcal S_i}a_{ji}\widehat x_{ji,t}=(1-\beta_i)\mathcal F_i(\widehat x_{i,t})\)
and \(\nabla_i^\ast\cdot\mathbf x_i^\ast=(1-\beta_i)\sum_{j\in\mathcal S_i}a_{ji}=1-\beta_i\),
we have \(\nabla_i^\ast\cdot\mathbf z_{i,t}=0\), so \(\mathbf z_{i,t}\) lies in the
composition subspace \(\{\mathbf z:\nabla_i^\ast\cdot\mathbf z=0\}\). The
decomposition of the first-order response follows by applying the same
linear decomposition to \(\dot{\mathbf x}_{i,t}\).

The second-order term of the Taylor expansion of \(\log q_i\) around
\(\mathbf x_i^\ast\) is obtained by expanding \(\log X_i\) through the two
nests. Write \(\widehat x_{ji,t}=\Delta x_{ji,t}/x_{ji}^\ast\) for the relative
input deviation. For a CES aggregator with share weights \(\nu_m\ge0\),
\(\sum_m\nu_m=1\), curvature \(\varrho\), and relative input deviations
\(z_m\), the expansions
\((1+z)^{\varrho}=1+\varrho z+\tfrac12\varrho(\varrho-1)z^2+o(z^2)\)
and \(\log(1+u)=u-\tfrac12u^2+o(u^2)\), together with
\(\sum_m\nu_mz_m^2=\operatorname{Var}_\nu(z)+\overline z^{\,2}\),
give
\begin{equation}
\frac{1}{\varrho}\log\sum_m\nu_m(1+z_m)^{\varrho}
=
\overline z-\tfrac12\overline z^{\,2}
-\tfrac12(1-\varrho)\operatorname{Var}_\nu(z)
+o(\|z\|^2),
\qquad
\overline z:=\sum_m\nu_mz_m
\label{eq:ces_level_expansion}
\end{equation}
Applying \eqref{eq:ces_level_expansion} to the inner nest \(k\) (weights
\(\widetilde a_{ji}^{(k)}\), curvature \(\rho\)) and passing from
\(\log X_i^{(k)}\) to the relative deviation
\(\widehat X_{i,t}^{(k)}:=\Delta X_i^{(k)}/X_i^{(k)\ast}\), which adds back
\(+\tfrac12\{\mathcal F_i^{(k)}(\widehat x_{i,t})\}^2\) and cancels the
\(-\tfrac12\overline z^{\,2}\) term, gives
\[
\widehat X_{i,t}^{(k)}
=
\mathcal F_i^{(k)}(\widehat x_{i,t})
-\tfrac12(1-\rho)\operatorname{Var}_{\mathcal F_i^{(k)}}(\widehat x_{i,t})
+o(\|\widehat x_{i,t}\|^2)
\]
Applying \eqref{eq:ces_level_expansion} to the outer nest (weights
\(s_i^{(k)}\), curvature \(\rho_c\), inputs \(\widehat X_{i,t}^{(k)}\)),
\[
\Delta\log X_i
=
\sum_k s_i^{(k)}\widehat X_{i,t}^{(k)}
-\tfrac12\Bigl(\sum_k s_i^{(k)}\widehat X_{i,t}^{(k)}\Bigr)^2
-\tfrac12(1-\rho_c)\operatorname{Var}_{s_i}\!\bigl(\widehat X_{i,t}^{(k)}\bigr)
+o
\]
By the tower property \eqref{eq:operator_tower},
\(\sum_k s_i^{(k)}\mathcal F_i^{(k)}(\widehat x)=\mathcal F_i(\widehat x)\) and
\(\sum_k s_i^{(k)}\operatorname{Var}_{\mathcal F_i^{(k)}}(\widehat x)=V_i^{\mathrm w}(\widehat x)\),
so \(\sum_k s_i^{(k)}\widehat X_{i,t}^{(k)}=\mathcal F_i(\widehat x_{i,t})-\tfrac12(1-\rho)V_i^{\mathrm w}(\widehat x_{i,t})+o\),
and to leading order
\(\operatorname{Var}_{s_i}(\widehat X_{i,t}^{(k)})=\operatorname{Var}_{s_i}(\mathcal F_i^{(k)}(\widehat x_{i,t}))=V_i^{\mathrm b}(\widehat x_{i,t})\),
the between-vintage component of \eqref{eq:ltv_decomposition} (the inner-curvature
corrections to \(\widehat X_{i,t}^{(k)}\) enter \(\operatorname{Var}_{s_i}\) only
beyond second order). Substituting and discarding terms beyond second order,
\[
\Delta\log X_i
=
\mathcal F_i(\widehat x_{i,t})
-\tfrac12\{\mathcal F_i(\widehat x_{i,t})\}^2
-\tfrac12\bigl[(1-\rho)V_i^{\mathrm w}(\widehat x_{i,t})+(1-\rho_c)V_i^{\mathrm b}(\widehat x_{i,t})\bigr]+o
\]
and multiplying by \((1-\beta_i)\) yields the scale--composition form
\eqref{eq:log_q_scale_composition}.
\end{proof}

\begin{lemma}[Normalized nests and the fixed-share curvature comparison]
\label{lem:normalized_nests}
Consider economies that share the stationary allocation and differ only
in the curvature pair \((\rho,\rho_c)\). Fix a stationary equilibrium
\((\mathbf p^\ast,\mathbf q^\ast,\boldsymbol{\mathcal X}^\ast,\mathbf l^\ast,w^\ast)\)
of one admissible economy, with supplier shares
\(a_{ji}=s_i^{(k_{ji})}\widetilde a_{ji}^{(k_{ji})}\). For an admissible
pair \((\rho,\rho_c)\) let \(q_i(\mathbf x_i,l_i)\) be the production
function with the normalized composite
\begin{equation}
X_i^{(k)}(\mathbf x)
:=
X_i^\ast\Bigl(\sum_{j\in\mathcal S_i^{(k)}}\widetilde a_{ji}^{(k)}
\Bigl(\frac{x_{ji}}{x_{ji}^\ast}\Bigr)^{\rho}\Bigr)^{1/\rho},
\qquad
X_i(\mathbf x)
:=
X_i^\ast\Bigl(\sum_{k\in\mathcal K_i}s_i^{(k)}\Bigl(\frac{X_i^{(k)}(\mathbf x)}{X_i^\ast}\Bigr)^{\rho_c}\Bigr)^{1/\rho_c}
\label{eq:normalized_nests}
\end{equation}
in the top nest \eqref{eq:top_cd}. The normalized composite is the technology
\eqref{eq:within_vintage_bundle}--\eqref{eq:materials_bundle} of the
paper with the supplier weights
\(\nu_{ji}=\widetilde a_{ji}^{(k_{ji})}\bigl(X_i^\ast/x_{ji}^\ast\bigr)^{\rho}\)
and the vintage weights \(\chi_i^{(k)}=s_i^{(k)}\). Then at every
admissible pair the following hold. \textnormal{(i)} The stationary
levels and the elasticities at the stationary bundle are those of the
given economy, so the supplier shares and the matrix \(\mathbf A\)
are the same. \textnormal{(ii)} The stationary allocation is the
unique stationary equilibrium of the normalized economy up to the
choice of the nominal wage, with the same stationary balances, incidence and impact
state. \textnormal{(iii)} The first-order responses of the dated
inputs are the same.
\end{lemma}

\begin{proof}
(i) At \(\mathbf x_i^\ast\) every
ratio \(X_i^{(k)}(\mathbf x_i^\ast)/X_i^\ast\) equals one, so \(X_i(\mathbf x_i^\ast)=X_i^\ast\) and
\(q_i(\mathbf x_i^\ast,l_i^\ast)=q_i^\ast\). Logarithmic
differentiation at that point gives
\(\partial\log X_i^{(k)}/\partial\log x_{ji}=\widetilde a_{ji}^{(k)}\)
and \(\partial\log X_i/\partial\log X_i^{(k)}=s_i^{(k)}\), hence
\(\partial\log q_i/\partial\log x_{ji}=(1-\beta_i)a_{ji}\) and
\(\partial\log q_i/\partial\log l_i=\beta_i\). The supplier shares at
the stationary bundle are therefore the same at every pair, and so is
the matrix \(\mathbf A\).

(ii) By
the stationary identities
\(p_j^\ast x_{ji}^\ast=(1-\beta_i)a_{ji}p_i^\ast q_i^\ast\) and
\(w^\ast l_i^\ast=\beta_ip_i^\ast q_i^\ast\), the elasticities give
\(p_i^\ast\,\partial q_i/\partial x_{ji}=p_j^\ast\) and
\(p_i^\ast\,\partial q_i/\partial l_i=w^\ast\) at the stationary
inputs. Each production function \(q_i(\mathbf x_i,l_i)\) is concave and homogeneous of degree one, since a
CES nest with positive weights and negative curvature is concave,
increasing and homogeneous of degree one, and the Cobb--Douglas top
nest preserves the three properties under composition. Concavity and
Euler's identity give, for every positive input bundle \((\mathbf x_i,l_i)\),
\(q_i(\mathbf x_i,l_i)\le q_i(\mathbf x_i^\ast,l_i^\ast)+Dq_i(\mathbf x_i^\ast,l_i^\ast)[(\mathbf x_i-\mathbf x_i^\ast,l_i-l_i^\ast)]=Dq_i(\mathbf x_i^\ast,l_i^\ast)[(\mathbf x_i,l_i)]\),
so that revenue at the stationary prices is at most input cost, with
equality at the stationary inputs. Hence the stationary bundle is
cost-minimizing at \(\mathbf p^\ast\) and \(w^\ast\), and the unit cost
of every firm at \(\mathbf p^\ast\) is \(p_i^\ast\). Household
optimality and market clearing involve only prices and quantities,
which are unchanged. The stationary allocation is therefore a
stationary equilibrium of the normalized economy at every pair. And by
Proposition~\ref{prop:existence_stationary} of the paper, which applies
to every choice of positive nest weights, it is the unique one up to
the choice of the nominal wage. The stationary balances, the incidence
\(\boldsymbol\zeta\) and the impact state are then the same at every
pair.

(iii) The balance law
\eqref{eq:firm_balance_law_pre} depends on \(\mathbf A\), the wage bill
and the incidence alone. The linearized order and production
identities,
\(\dot{\widehat x}_{ji,s}=\dot{\widehat e}_{i,s}-\dot{\widehat d}_{j,s}+\dot{\widehat q}_{j,s}\)
(Lemma~\ref{lem:order_response_decomposition}) and
\(\dot{\widehat q}_{i,t}=(1-\beta_i)\sum_{j\in\mathcal S_i}a_{ji}\dot{\widehat x}_{ji,t-k_{ji}}\)
(Lemma~\ref{lem:production_linearization}), have coefficients and
initial conditions that are the same at every pair, and a unique
solution by Lemma~\ref{lem:propagation_linearization}, which handles
the zero-lead-time block. The first-order dated-input responses, and with
them the variances \(V_i^{\mathrm w}\) and \(V_i^{\mathrm b}\) of
\eqref{eq:vintage_variances} of the paper, are therefore the same across the
comparison.
\end{proof}

In the comparison of Lemma~\ref{lem:normalized_nests} the first-order
responses \(\dot{\widehat x}_{i,t}\) are the same at every curvature
pair, by part (iii). For a weighting \(\boldsymbol\omega\) write
\begin{equation}
V^{\mathrm w}:=\sum_{t\ge0}\sum_{i\in N}\omega_i(1-\beta_i)\,V_i^{\mathrm w}(\dot{\widehat x}_{i,t}),
\qquad
V^{\mathrm b}:=\sum_{t\ge0}\sum_{i\in N}\omega_i(1-\beta_i)\,V_i^{\mathrm b}(\dot{\widehat x}_{i,t})
\label{eq:gap_dispersion}
\end{equation}
for the cumulative within-vintage and between-vintage dispersions of
the responses.

\begin{lemma}[The leading coefficient is affine in the curvature gap]
\label{lem:inefficiency_gap}
Consider the comparison of Lemma~\ref{lem:normalized_nests}, and
suppose that Assumptions~\ref{assump:diffuse_incidence}
and~\ref{assump:monotone_decay_spectrum} of the paper hold. Then the
leading coefficient of Theorem~\ref{theorem:finite_cumulative_inefficiency} of
the paper is
\begin{equation}
C^{(\boldsymbol\omega)}(\rho,\rho_c)
=
\tfrac12(1-\rho)(V^{\mathrm w}+V^{\mathrm b})+\tfrac12(\rho-\rho_c)\,V^{\mathrm b}
\label{eq:inefficiency_gap_decomposition}
\end{equation}
\end{lemma}

\begin{proof}
By Lemma~\ref{lem:normalized_nests} the first-order responses
\(\dot{\widehat x}_{i,t}\) are the same at every curvature pair, so the
two dispersions \eqref{eq:gap_dispersion} are curvature-free. They are
finite by the envelope \eqref{eq:usable_input_envelope} of the paper,
each summand being at most the squared envelope. Substituting
\(\langle\dot{\widehat x}_{i,t},\dot{\widehat x}_{i,t}\rangle_i=(1-\rho)V_i^{\mathrm w}(\dot{\widehat x}_{i,t})+(1-\rho_c)V_i^{\mathrm b}(\dot{\widehat x}_{i,t})\)
into the coefficient \eqref{eq:cumulative_inefficiency_expansion} of the paper gives
\(C^{(\boldsymbol\omega)}=\tfrac12\{(1-\rho)V^{\mathrm w}+(1-\rho_c)V^{\mathrm b}\}\),
and adding and subtracting \(\tfrac12(1-\rho)V^{\mathrm b}\) gives
\eqref{eq:inefficiency_gap_decomposition}.
\end{proof}

For collections \(y=(y_{i,t})\) and \(z=(z_{i,t})\) of supplier vectors
indexed by firm and date, and a weighting \(\boldsymbol\omega\), let
\(\langle y,z\rangle_H\) and \(\|y\|_H\) be the cumulative
curvature-weighted form \eqref{eq:H_form} of the paper, whenever the sums
converge absolutely. The form
\(\langle\cdot,\cdot\rangle_H\) is symmetric, bilinear and positive
semidefinite, so \(|\langle y,z\rangle_H|\le\|y\|_H\|z\|_H\). With
\(\dot{\widehat e}_{i,s}\) the first-order response of the
intermediate-input expenditure of firm \(i\) and
\(\dot{\widehat p}_{j,s}\) that of the price of supplier \(j\)
(Lemma~\ref{lem:order_response_decomposition}), define
\begin{align}
\mathcal D_{\mathrm M}^{(\boldsymbol\omega)}(\pi)
&:=
\tfrac{\pi^2}{2}\sum_{t\ge 0}\sum_{i\in N}
\omega_i(1-\beta_i)\,
\bigl\langle \dot{\widehat e}_{i,\,t-k_{ji}},\,\dot{\widehat e}_{i,\,t-k_{ji}}\bigr\rangle_i,
\label{eq:inefficiency_monetary_component}\\[0.2em]
\mathcal D_{\mathrm I}^{(\boldsymbol\omega)}(\pi)
&:=
\tfrac{\pi^2}{2}\sum_{t\ge 0}\sum_{i\in N}
\omega_i(1-\beta_i)\,
\bigl\langle \dot{\widehat p}_{j,\,t-k_{ji}},\,\dot{\widehat p}_{j,\,t-k_{ji}}\bigr\rangle_i,
\label{eq:inefficiency_input_component}\\[0.2em]
\mathcal D_{\mathrm C}^{(\boldsymbol\omega)}(\pi)
&:=
\tfrac{\pi^2}{2}\sum_{t\ge 0}\sum_{i\in N}
\omega_i(1-\beta_i)\,
\bigl\langle \dot{\widehat e}_{i,\,t-k_{ji}},\,
\dot{\widehat p}_{j,\,t-k_{ji}}\bigr\rangle_i
\label{eq:inefficiency_cross_component}
\end{align}

\begin{lemma}[Nominal-expenditure and input-price decomposition of the cumulative loss from production inefficiency]
\label{lem:monetary_input_decomposition}
Suppose that Assumptions~\ref{assump:diffuse_incidence} and~\ref{assump:monotone_decay_spectrum} of the paper hold, let
\(\kappa\) be the contraction rate of Proposition~\ref{prop:stability_stationary}, and fix a weighting
\(\boldsymbol\omega\). Then
\begin{equation}
\mathcal D^{(\boldsymbol\omega)}(\pi)
=
\mathcal D_{\mathrm M}^{(\boldsymbol\omega)}(\pi)
+
\mathcal D_{\mathrm I}^{(\boldsymbol\omega)}(\pi)
-
2\,\mathcal D_{\mathrm C}^{(\boldsymbol\omega)}(\pi)
+
o(\pi^2)
\label{eq:monetary_input_decomposition}
\end{equation}
where
\(\bigl|\mathcal D_{\mathrm C}^{(\boldsymbol\omega)}\bigr|\le\sqrt{\mathcal D_{\mathrm M}^{(\boldsymbol\omega)}\,\mathcal D_{\mathrm I}^{(\boldsymbol\omega)}}\).
\end{lemma}

\begin{proof}
By the expansion \eqref{eq:cumulative_inefficiency_expansion} of
Theorem~\ref{theorem:finite_cumulative_inefficiency}, for every sufficiently
small \(\pi\),
\[
\mathcal D^{(\boldsymbol\omega)}(\pi)
=
\tfrac12\,\pi^2\sum_{t\ge 0}\sum_{i\in N}
\omega_i(1-\beta_i)\,\bigl\langle\dot{\widehat x}_{i,t},\dot{\widehat x}_{i,t}\bigr\rangle_i
+o(\pi^2)
\]
the date sum converging absolutely because
\(\|\dot{\widehat x}_{i,t}\|_\infty\le C_S\kappa^t(1-\theta)\) by the
envelope \eqref{eq:usable_input_envelope}. In the dated form of
\(\mathcal F_i\) each supplier is evaluated at its order date, so
\(\langle\dot{\widehat x}_{i,t},\dot{\widehat x}_{i,t}\rangle_i\) is the
curvature-weighted dispersion over \(j\in\mathcal S_i\) of the dated
responses \(\dot{\widehat x}_{ji,\,t-k_{ji}}\). By the order-response
decomposition \eqref{eq:bundle_response_decomp} of
Lemma~\ref{lem:order_response_decomposition}, read at supplier
\(j\)'s order date,
\[
\dot{\widehat x}_{ji,\,t-k_{ji}}
=
\dot{\widehat e}_{i,\,t-k_{ji}}-\dot{\widehat p}_{j,\,t-k_{ji}}
\]
Both terms are supplier-indexed quantities, the first varying over
\(j\) only through the order date. The form \(\langle\cdot,\cdot\rangle_i\)
is symmetric, bilinear and positive semidefinite, so
\(\langle y-z,y-z\rangle_i=\langle y,y\rangle_i+\langle z,z\rangle_i-2\langle y,z\rangle_i\),
and substituting into the expansion and grouping the three double sums
against \eqref{eq:inefficiency_monetary_component}--\eqref{eq:inefficiency_cross_component}
gives \eqref{eq:monetary_input_decomposition}, provided the three sums
converge.

A constant across suppliers drops out of the form, so at every date
\(t\ge\widehat k\), at which every order date in the bundle is post-shock, the
responses may be replaced by the centered responses
\(\dot{\widehat e}_{i,s}-1\) and \(\dot{\widehat p}_{j,s}-1\). By
\eqref{eq:centered_budget_impulse} and the nominal envelope
\eqref{eq:stab_nominal_envelope} of the paper,
\(|\dot{\widehat e}_{i,s}-1|\le C_\lambda C_\zeta\lambda_2(\mathbf A)^s(1-\theta)/(1-\overline\beta)\).
Since \(\dot{\widehat d}_{j,s}-1=\sum_l\Omega_{jl}(\dot{\widehat e}_{l,s}-1)\)
with \(\sum_l\Omega_{jl}\le1\), the same bound holds for
\(|\dot{\widehat d}_{j,s}-1|\), and since
\(\dot{\widehat p}_{j,s}=\dot{\widehat d}_{j,s}-\dot{\widehat q}_{j,s}\)
by \eqref{eq:price_response_identity} with
\(|\dot{\widehat q}_{j,s}|\le C_S\kappa^s(1-\theta)\) by the output
envelope of Proposition~\ref{prop:stability_stationary},
\(|\dot{\widehat p}_{j,s}-1|\le[C_\lambda C_\zeta/(1-\overline\beta)+C_S]\kappa^s(1-\theta)\),
using \(\lambda_2(\mathbf A)\le\kappa\). As
\(\langle y,y\rangle_i\le(1+\overline\rho)\max_jy_j^2\) for a supplier
vector \(y\), the date-\(t\) terms of
\(\mathcal D_{\mathrm M}^{(\boldsymbol\omega)}\) and
\(\mathcal D_{\mathrm I}^{(\boldsymbol\omega)}\) are bounded by constants
times \(\kappa^{2(t-\widehat k)}\) for \(t\ge\widehat k\), and the dates
before \(\widehat k\) are finitely many, so both sums converge. The Cauchy--Schwarz inequality for the form \eqref{eq:H_form} gives
the convergence of \(\mathcal D_{\mathrm C}^{(\boldsymbol\omega)}\) and
the bound
\(|\mathcal D_{\mathrm C}^{(\boldsymbol\omega)}|\le\sqrt{\mathcal D_{\mathrm M}^{(\boldsymbol\omega)}\mathcal D_{\mathrm I}^{(\boldsymbol\omega)}}\).
The remainder is the \(o(\pi^2)\) of
Theorem~\ref{theorem:finite_cumulative_inefficiency}.

\end{proof}

\subsection{The lead-time profile and the loss from production inefficiency}
\label{app:supplementary}

\subsubsection{An increase in the variance of the lead-time profile}
\label{app:proof_firm_spread_aggregate}

The first-order dated-input response \(\dot{\widehat x}^{\mathbf k}_{i,t}\),
the regression slope \(\varsigma^{\mathbf k}_{i,t}\) and its residual
\(u^{\mathbf k}_{i,t}\) are those of \eqref{eq:lag_regression_slope} of the
paper. The vintage partition of the bundle of firm \(i\), and with it the
curvature-weighted form \eqref{eq:curvature_weighted_cov}, depends on the
lead-time profile. Write \(\langle\cdot,\cdot\rangle_{i,\mathbf k}\) for the
form under the partition of \(\mathbf k\).

\begin{lemma}[Regression of the input response on the lead time]
\label{lem:firm_spread_aggregate_impulse_dispersion}
For a lead-time profile \(\mathbf k\), a firm \(i\) and a date \(t\), let
\(\varsigma^{\mathbf k}_{i,t}\) be the regression slope \eqref{eq:lag_regression_slope}
of the first-order dated-input response on the lead time and \(u^{\mathbf k}_{i,t}\)
its residual. Then \(u^{\mathbf k}_{i,t}\) has zero mean and is uncorrelated
with the lead time under \(\mathcal F_i\), and
\begin{equation}
\bigl\langle\dot{\widehat x}^{\mathbf k}_{i,t},\dot{\widehat x}^{\mathbf k}_{i,t}\bigr\rangle_{i,\mathbf k}
=(1-\rho_c)(\varsigma^{\mathbf k}_{i,t})^2\operatorname{Var}_{\mathcal F_i}(\mathbf k)+\bigl\langle u^{\mathbf k}_{i,t},u^{\mathbf k}_{i,t}\bigr\rangle_{i,\mathbf k}
\label{eq:firm_spread_aggregate_impulse_decomposition}
\end{equation}
\end{lemma}

\begin{proof}[Proof of Lemma~\ref{lem:firm_spread_aggregate_impulse_dispersion}]
The identities are those of weighted least squares against the
supplier-share weights. Write \(y_j:=\dot{\widehat x}^{\mathbf k}_{ji,t-k_{ji}}\),
\(\overline y:=\mathcal F_i(y)\), \(\overline k:=\mathcal F_i(\mathbf k)\) and
\(\varsigma:=\varsigma^{\mathbf k}_{i,t}\). By construction
\(u_j=y_j-\overline y-\varsigma(k_{ji}-\overline k)\), so
\(\mathcal F_i(u)=0\), and
\(\operatorname{Cov}_{\mathcal F_i}(u,\mathbf k)=\operatorname{Cov}_{\mathcal F_i}(y,\mathbf k)-\varsigma\operatorname{Var}_{\mathcal F_i}(\mathbf k)=0\)
when \(\operatorname{Var}_{\mathcal F_i}(\mathbf k)>0\), while for \(\operatorname{Var}_{\mathcal F_i}(\mathbf k)=0\) the lead time is constant on
the support and every covariance with it vanishes. Hence
\(\operatorname{Var}_{\mathcal F_i}(y)=\varsigma^2\operatorname{Var}_{\mathcal F_i}(\mathbf k)+\operatorname{Var}_{\mathcal F_i}(u)+2\varsigma\operatorname{Cov}_{\mathcal F_i}(\mathbf k,u)=\varsigma^2\operatorname{Var}_{\mathcal F_i}(\mathbf k)+\operatorname{Var}_{\mathcal F_i}(u)\).
For \eqref{eq:firm_spread_aggregate_impulse_decomposition}, the centered lead time is constant inside
every vintage layer, so its within-vintage variance and its within-vintage
covariance with any supplier vector vanish, and its between-vintage
variance is the total \(\operatorname{Var}_{\mathcal F_i}(\mathbf k)\). Its between-vintage covariance with
\(u\) therefore equals the total covariance, which is zero. Expanding the
form \(\langle\cdot,\cdot\rangle_{i,\mathbf k}\) of \(y-\overline y=\varsigma(\mathbf k-\overline k)+u\),
with the within- and between-vintage covariances taken under the
partition of \(\mathbf k\),
\[
\langle y,y\rangle_{i,\mathbf k}
=(1-\rho)\operatorname{Cov}^{\mathrm w}_i(u,u)
+(1-\rho_c)\bigl[\varsigma^2\operatorname{Var}_{\mathcal F_i}(\mathbf k)+\operatorname{Cov}^{\mathrm b}_i(u,u)\bigr]
=(1-\rho_c)\,\varsigma^2\operatorname{Var}_{\mathcal F_i}(\mathbf k)+\langle u,u\rangle_{i,\mathbf k}
\]
which is \eqref{eq:firm_spread_aggregate_impulse_decomposition}.
\end{proof}

\begin{proof}[Proof of Corollary~\ref{cor:firm_level_vintage_spread_aggregate_inefficiency}]
At the fixed date \(t\) the dated inputs are differentiable at \(\pi=0\)
(Lemma~\ref{lem:propagation_linearization}) and, under each profile,
the curvature-weighted form is a continuous quadratic form on the
finite supplier vector. Hence
\(\mathcal D_t^{(\boldsymbol\omega)}(\pi,\mathbf k)/\pi^2\to\tfrac12\sum_i\omega_i(1-\beta_i)\langle\dot{\widehat x}^{\mathbf k}_{i,t},\dot{\widehat x}^{\mathbf k}_{i,t}\rangle_{i,\mathbf k}\)
under each profile, with \(\dot{\widehat x}^{\mathbf k}_{i,t}\) the first-order
dated-input response under \(\mathbf k\). Subtracting the split
\eqref{eq:firm_spread_aggregate_impulse_decomposition} under \(\mathbf k'\) and
under \(\mathbf k\), with equal slopes
\(\varsigma_{i,t}\),
\begin{align*}
\lim_{\pi\to0}\frac{2[\mathcal D_t^{(\boldsymbol\omega)}(\pi,\mathbf k')-\mathcal D_t^{(\boldsymbol\omega)}(\pi,\mathbf k)]}{\pi^2}
&=
(1-\rho_c)\sum_i\omega_i(1-\beta_i)\,\varsigma_{i,t}^2\bigl(\operatorname{Var}_{\mathcal F_i}(\mathbf k')-\operatorname{Var}_{\mathcal F_i}(\mathbf k)\bigr)\\
&\qquad
+\sum_i\omega_i(1-\beta_i)\Bigl[\langle u^{\mathbf k'}_{i,t},u^{\mathbf k'}_{i,t}\rangle_{i,\mathbf k'}-\langle u^{\mathbf k}_{i,t},u^{\mathbf k}_{i,t}\rangle_{i,\mathbf k}\Bigr]
\end{align*}
which is \eqref{eq:lag_variance_comparison}. The first sum is
non-negative, since \(\operatorname{Var}_{\mathcal F_i}(\mathbf k')\ge \operatorname{Var}_{\mathcal F_i}(\mathbf k)\) at every firm.
\end{proof}

\subsubsection{A uniform shift of the lead-time profile}
\label{app:proof_lag_shift}
\label{app:wave_comovement}

By \eqref{eq:bundle_closed} of Lemma~\ref{lem:order_response_decomposition}
and \eqref{eq:supply_response_identity}, the first-order response of
the real order placed by firm \(i\) on supplier \(j\) at a date
\(s\ge0\) is
\(\dot{\widehat x}_{ji,s}=(\dot{\widehat e}_{i,s}-\dot{\widehat d}_{j,s})+\dot{\widehat q}_{j,s}\).
Evaluated at the order dates of the dated input bundle of firm \(i\) at
date \(t\), its two terms are the monetary wave \(M\) and the real wave
\(Q\) of \eqref{eq:waves_main} of the paper. With
\(\langle\cdot,\cdot\rangle_H\) the cumulative form \eqref{eq:H_form} of the
paper, write
\begin{equation}
C_{EQ}:=\sum_{t\ge0}\sum_{i\in N}\omega_i(1-\beta_i)\bigl\langle\dot{\widehat e}_{i,t-k_{ji}},Q_{i,t}\bigr\rangle_i,
\qquad
C_{DQ}:=\sum_{t\ge0}\sum_{i\in N}\omega_i(1-\beta_i)\bigl\langle\dot{\widehat d}_{j,t-k_{ji}},Q_{i,t}\bigr\rangle_i
\label{eq:DMR_def}
\end{equation}

\begin{lemma}[Wave decomposition of the cumulative loss from production inefficiency]
\label{lem:wave_comovement}
Suppose that Assumptions~\ref{assump:diffuse_incidence} and~\ref{assump:monotone_decay_spectrum} of the paper hold, and fix a weighting
\(\boldsymbol\omega\). Then \(\|M\|_H\) and \(\|Q\|_H\) are finite,
\begin{equation}
\mathcal D^{(\boldsymbol\omega)}(\pi)
=
\tfrac{\pi^2}{2}\bigl\{\langle M,M\rangle_H+\langle Q,Q\rangle_H+2\langle M,Q\rangle_H\bigr\}+o(\pi^2)
\label{eq:wave_split_main}
\end{equation}
and \(\langle M,Q\rangle_H=C_{EQ}-C_{DQ}\).
\end{lemma}

\begin{proof}
At every order date \(s\ge0\),
\(\dot{\widehat x}_{ji,s}=(\dot{\widehat e}_{i,s}-\dot{\widehat d}_{j,s})+\dot{\widehat q}_{j,s}\)
by \eqref{eq:bundle_closed} and \eqref{eq:supply_response_identity},
and both sides vanish for \(s<0\). At \(s=t-k_{ji}\) this is
\(\dot{\widehat x}_{ji,t-k_{ji}}=M_{ji,t}+Q_{ji,t}\). For
\(t\ge\widehat k\) every order date is post-shock. By
\eqref{eq:centered_budget_impulse} and the nominal envelope
\eqref{eq:stab_nominal_envelope} of the paper,
\(|\dot{\widehat e}_{i,s}-1|\le C_\lambda C_\zeta\lambda_2(\mathbf A)^s(1-\theta)/(1-\overline\beta)\),
and since \(\dot{\widehat d}_{j,s}-1=\sum_l\Omega_{jl}(\dot{\widehat e}_{l,s}-1)\)
with \(\sum_l\Omega_{jl}\le1\), the same bound holds for
\(|\dot{\widehat d}_{j,s}-1|\), so
\(|M_{ji,t}|\le2C_\lambda C_\zeta(1-\theta)\lambda_2(\mathbf A)^{t-\widehat k}/(1-\overline\beta)\).
By the output envelope of Proposition~\ref{prop:stability_stationary}
differentiated at zero, \(|Q_{ji,t}|\le C_S\kappa^{t-\widehat k}(1-\theta)\),
with \(\kappa\) the contraction rate of that proposition. Since
\(\langle y,y\rangle_i\le(1+\overline\rho)\max_jy_j^2\) for a
supplier vector \(y\), the date-\(t\) terms of \(\|M\|_H^2\) and
\(\|Q\|_H^2\) are bounded by constants times \(\kappa^{2(t-\widehat k)}\)
for \(t\ge\widehat k\), using \(\lambda_2(\mathbf A)\le\kappa\), and the
finitely many dates before \(\widehat k\) contribute finite terms. Both
norms are finite, and \(|\langle M,Q\rangle_H|\le\|M\|_H\|Q\|_H\).
Substituting \(\dot{\widehat x}_{ji,t-k_{ji}}=M_{ji,t}+Q_{ji,t}\) into
the expansion \eqref{eq:cumulative_inefficiency_expansion} of
Theorem~\ref{theorem:finite_cumulative_inefficiency} and using the
bilinearity and symmetry of \(\langle\cdot,\cdot\rangle_i\),
\(\langle\dot{\widehat x}_{i,t},\dot{\widehat x}_{i,t}\rangle_i=\langle M_{i,t},M_{i,t}\rangle_i+\langle Q_{i,t},Q_{i,t}\rangle_i+2\langle M_{i,t},Q_{i,t}\rangle_i\),
and summing against \(\omega_i(1-\beta_i)\) over firms and dates gives
\eqref{eq:wave_split_main}. Bilinearity in the first argument,
\(\langle M_{i,t},Q_{i,t}\rangle_i=\langle\dot{\widehat e}_{i,t-k_{ji}},Q_{i,t}\rangle_i-\langle\dot{\widehat d}_{j,t-k_{ji}},Q_{i,t}\rangle_i\),
summed in the same way gives \(\langle M,Q\rangle_H=C_{EQ}-C_{DQ}\).
\end{proof}

Let \(\mathbf k'\), the translation \(\mathsf S_\Delta\) and the leading
coefficient \(C^{(\boldsymbol\omega)}(\mathbf k)\) be as in
Section~\ref{subsec:lag_profile} of the paper, so that
\(\mathcal D^{(\boldsymbol\omega)}(\pi,\mathbf k)=\pi^2C^{(\boldsymbol\omega)}(\mathbf k)+o(\pi^2)\).
The advance \(\mathsf S^{*}_\Delta\) moves a collection forward by
\(\Delta\) dates, \((\mathsf S^{*}_\Delta z)_t=z_{t+\Delta}\). A prime
marks an object under \(\mathbf k'\), so that \(M'\) and \(Q'\) are the
monetary wave and the real wave under \(\mathbf k'\).
Write \(\Xi_{ij,s}\) for the dated contrast \eqref{eq:dated_contrast}
of buyer \(i\) on supplier \(j\) at order date \(s\), so that
\(M_{ji,t}=\Xi_{ij,t-k_{ji}}\), and
\(\xi_{i,t}=\sum_{j\in\mathcal S_i:\,k_{ji}\le t}a_{ji}\Xi_{ij,t-k_{ji}}\)
for the dated relative expenditure impulse \eqref{eq:dated_forcing}.
By Lemma~\ref{lem:supplier_recursion}, the input-bundle scale is the
Neumann series \(\mathcal F=\sum_{r\ge0}\mathsf T^{r}\boldsymbol\xi\)
of the dated transmission operator \eqref{eq:transmission_operator},
so by \eqref{eq:supply_response_identity}
\begin{equation}
\dot{\widehat q}_{j,s}=\sum_{r\ge1}\dot{\widehat q}^{(r)}_{j,s},
\qquad
\dot{\widehat q}^{(r)}_{j,s}:=(1-\beta_j)\,(\mathsf T^{r-1}\boldsymbol\xi)_{j,s}
\label{eq:real_wave_rounds}
\end{equation}
Write \(Q^{(r)}_{ji,t}:=\dot{\widehat q}^{(r)}_{j,t-k_{ji}}\) for
\(t\ge k_{ji}\) and zero before, the \(r\)-th round of the real wave,
so that \(Q=\sum_{r\ge1}Q^{(r)}\). With \(W^{(0)}:=M\) and
\(W^{(r)}:=Q^{(r)}\) for \(r\ge1\), as in the paper,
\(M+Q=\sum_{r\ge0}W^{(r)}\).

\begin{proof}[Proof of Corollary~\ref{cor:lag_shift_inefficiency}]
By \eqref{eq:dot_e_main} and \eqref{eq:dyn_nominal_demand} of the paper,
\(\dot{\widehat e}_{i,s}\) and \(\dot{\widehat d}_{j,s}\) are linear
functionals of the first-order response of the balance vector, whose
path is given by the balance law \eqref{eq:firm_balance_law_pre} of the
paper and does not involve the lead-time profile. The contrasts \(\Xi_{ij,s}\)
are therefore the same functions of the order date \(s\) under
\(\mathbf k\) and \(\mathbf k'\), and vanish for \(s<0\). Hence
\(M'_{ji,t}=\Xi_{ij,t-k_{ji}-\Delta k}=M_{ji,t-\Delta k}\) for
\(t\ge\Delta k\) and \(M'_{ji,t}=0\) for \(t<\Delta k\), that is,
\(M'=\mathsf S_{\Delta k}M\). Under \(\mathbf k'\) the dated relative
expenditure impulse is
\(\xi'_{i,t}=\sum_{j:\,k_{ji}+\Delta k\le t}a_{ji}\Xi_{ij,t-k_{ji}-\Delta k}=\xi_{i,t-\Delta k}\)
for \(t\ge\Delta k\) and zero before, so
\(\boldsymbol\xi'=\mathsf S_{\Delta k}\boldsymbol\xi\), and for every
bounded profile \(z\),
\((\mathsf T'z)_{i,t}=\sum_{j:\,k_{ji}+\Delta k\le t}(1-\beta_j)a_{ji}z_{j,t-k_{ji}-\Delta k}=(\mathsf T\mathsf S_{\Delta k}z)_{i,t}=(\mathsf S_{\Delta k}\mathsf Tz)_{i,t}\),
so \(\mathsf T'=\mathsf T\mathsf S_{\Delta k}=\mathsf S_{\Delta k}\mathsf T\).
Therefore \(\mathsf T'^{\,r-1}\boldsymbol\xi'=\mathsf S_{r\Delta k}\mathsf T^{r-1}\boldsymbol\xi\),
\(\dot{\widehat q}'^{(r)}=\mathsf S_{r\Delta k}\dot{\widehat q}^{(r)}\),
and
\(Q'^{(r)}_{ji,t}=\dot{\widehat q}'^{(r)}_{j,t-k_{ji}-\Delta k}=\dot{\widehat q}^{(r)}_{j,t-k_{ji}-(r+1)\Delta k}=(\mathsf S_{(r+1)\Delta k}Q^{(r)})_{ji,t}\).
Every wave is therefore delayed, \(W'^{(r)}=\mathsf S_{(r+1)\Delta k}W^{(r)}\)
for every \(r\ge0\).

The shift of every lead time by \(\Delta k\) preserves the vintage
partition \(\{\mathcal S_i^{(k)}\}_k\) of every bundle and, by
hypothesis, the shares \(a_{ji}\), the labor shares and the weighting.
Hence the forms \(\langle\cdot,\cdot\rangle_i\) and
\(\langle\cdot,\cdot\rangle_H\) are the same under the two profiles.
For profiles \(y,z\) with finite \(H\)-norm and
\(\Delta\ge\Delta'\ge0\), reindexing with \(u=t-\Delta\),
\begin{equation}
\begin{aligned}
\langle\mathsf S_\Delta y,\mathsf S_{\Delta'}z\rangle_H
&=\sum_{t\ge\Delta}\sum_{i\in N}\omega_i(1-\beta_i)\langle y_{i,t-\Delta},z_{i,t-\Delta'}\rangle_i
=\sum_{u\ge0}\sum_{i\in N}\omega_i(1-\beta_i)\langle y_{i,u},z_{i,u+\Delta-\Delta'}\rangle_i\\
&=\langle y,\mathsf S^{*}_{\Delta-\Delta'}z\rangle_H
=\langle\mathsf S_{\Delta-\Delta'}y,z\rangle_H
\end{aligned}
\label{eq:translation_identity}
\end{equation}
the last equality by the same reindexing. In particular
\(\|\mathsf S_\Delta y\|_H=\|y\|_H\) and
\(\langle y,\mathsf S_\Delta z\rangle_H=\langle\mathsf S^{*}_\Delta y,z\rangle_H\).

The rounds are summable. By \eqref{eq:centered_budget_impulse} and
the nominal envelope \eqref{eq:stab_nominal_envelope} of the paper,
as in the proof of Lemma~\ref{lem:wave_comovement},
\(|\Xi_{ij,s}|\le C_\Xi\overline\lambda^{\,s}(1-\theta)\) for every
\(s\ge0\), with \(C_\Xi:=2C_\lambda C_\zeta/(1-\overline\beta)\) and
\(\overline\lambda\in(0,1)\) the bound on \(\lambda_2(\mathbf A)\) of
Assumption~\ref{assump:monotone_decay_spectrum}. Hence
\(|\xi_{i,u}|\le C_\Xi(1-\theta)\overline\lambda^{\,u-\widehat k}\)
for every \(u\ge0\), and since
\(|(\mathsf Tz)_{i,t}|\le(1-\underline\beta)\max_{j:\,k_{ji}\le t}|z_{j,t-k_{ji}}|\)
with \(t-k_{ji}\ge t-\widehat k\), induction gives
\(|(\mathsf T^{r-1}\boldsymbol\xi)_{j,s}|\le(1-\underline\beta)^{r-1}C_\Xi(1-\theta)\overline\lambda^{\,\max\{s-(r-1)\widehat k,0\}-\widehat k}\),
so
\(|Q^{(r)}_{ji,t}|\le(1-\underline\beta)^{r}C_\Xi(1-\theta)\overline\lambda^{\,\max\{t-r\widehat k,0\}-\widehat k}\).
Since \(\langle y,y\rangle_i\le(1+\overline\rho)\max_jy_j^2\) and
\(\sum_i\omega_i(1-\beta_i)\le1\),
\begin{equation}
\|Q^{(r)}\|_H
\;\le\;
(1+\overline\rho)^{1/2}C_\Xi(1-\theta)\overline\lambda^{\,-\widehat k}\,(1-\underline\beta)^{r}
\Bigl(r\widehat k+\frac{1}{1-\overline\lambda^{2}}\Bigr)^{1/2}
\label{eq:round_norm_bound}
\end{equation}
which is summable in \(r\). The series \(\sum_rQ^{(r)}\) converges to
\(Q\) at every firm and date, by the Neumann series of
Lemma~\ref{lem:supplier_recursion}, and in \(H\)-norm by
\eqref{eq:round_norm_bound}, so \(Q=\sum_rQ^{(r)}\) in every inner
product of the form \(\langle\cdot,\cdot\rangle_H\). The same holds
under \(\mathbf k'\), with \(\widehat k+\Delta k\) in place of
\(\widehat k\), and by \eqref{eq:translation_identity} the translated
rounds have the same norms. Every double sum over rounds below is
therefore absolutely convergent, by the Cauchy--Schwarz inequality for
the positive semidefinite form.

By \eqref{eq:wave_split_main} under each profile,
\(C^{(\boldsymbol\omega)}(\mathbf k)=\tfrac12\bigl\|\sum_{r\ge0}W^{(r)}\bigr\|_H^2\)
and
\(C^{(\boldsymbol\omega)}(\mathbf k')=\tfrac12\bigl\|\sum_{r\ge0}\mathsf S_{(r+1)\Delta k}W^{(r)}\bigr\|_H^2\).
Expanding both squares, the own terms agree, since
\(\|\mathsf S_\Delta y\|_H=\|y\|_H\). For \(r<r'\), the symmetry of the
form and \eqref{eq:translation_identity} give
\(\langle\mathsf S_{(r+1)\Delta k}W^{(r)},\mathsf S_{(r'+1)\Delta k}W^{(r')}\rangle_H=\langle W^{(r)},\mathsf S_{(r'-r)\Delta k}W^{(r')}\rangle_H\).
Subtracting gives \eqref{eq:lag_shift_exact}.

For the second part, the contrast \(\Xi_{ij,s}\) is a linear functional
of the balance misalignment
\(\dot{\mathbf m}_s^\perp=\mathbf A^{s}\dot{\mathbf m}_0^\perp\), by
\eqref{eq:centered_budget_impulse} and \eqref{eq:dyn_nominal_demand} of
the paper. When \(\dot{\mathbf m}_0^\perp\) is an eigenvector of
\(\mathbf A\) with the eigenvalue \(\lambda\),
\(\Xi_{ij,s}=\lambda^{s}\Xi_{ij,0}\) for every link and every
\(s\ge0\). Let \(k\) be the common lead time. Then
\(M_{ji,t}=\lambda^{t-k}\Xi_{ij,0}\) for \(t\ge k\) and
\(\xi_{i,u}=\lambda^{u-k}\xi_{i,k}\) for \(u\ge k\), and both vanish
before \(k\). With the common lead time,
\((\mathsf Tz)_{i,t}=\sum_{j\in\mathcal S_i}(1-\beta_j)a_{ji}z_{j,t-k}\)
for \(t\ge k\) and zero before. The operator \(\mathsf T\) therefore maps
a profile that vanishes before a date \(s_0\) and decays at the rate
\(\lambda\) from \(s_0\) on to one that vanishes before \(s_0+k\) and
decays at the rate \(\lambda\) from \(s_0+k\) on. By induction
\(\mathsf T^{r-1}\boldsymbol\xi\) vanishes before \(rk\) and decays at
the rate \(\lambda\) from \(rk\) on, and by \eqref{eq:real_wave_rounds}
so does \(\dot{\widehat q}^{(r)}\). Every wave \(W^{(r)}\), \(r\ge0\),
therefore vanishes before the date \((r+1)k\) and decays at the rate
\(\lambda\) from that date on. Hence, for every \(\Delta\ge1\),
\(\mathsf S^{*}_\Delta W^{(r)}-\lambda^{\Delta}W^{(r)}\) vanishes outside
the dates \((r+1)k-\Delta\le t<(r+1)k\), at which \(W^{(r')}\) vanishes
for every \(r'>r\), since \((r'+1)k\ge(r+2)k\). By
\eqref{eq:translation_identity}, for \(r<r'\),
\[
\bigl\langle W^{(r)},\mathsf S_{(r'-r)\Delta k}W^{(r')}\bigr\rangle_H
=\bigl\langle\mathsf S^{*}_{(r'-r)\Delta k}W^{(r)},W^{(r')}\bigr\rangle_H
=\lambda^{(r'-r)\Delta k}\bigl\langle W^{(r)},W^{(r')}\bigr\rangle_H
\]
and substituting into \eqref{eq:lag_shift_exact} gives
\eqref{eq:lag_shift_geometric}.
\end{proof}

Write \(\sigma_{i,s}:=\dot{\widehat e}_{i,s}-1\) and
\(\overline\sigma_{j,s}:=\dot{\widehat d}_{j,s}-1\), so that
\(\Xi_{ij,s}=\sigma_{i,s}-\overline\sigma_{j,s}\). By
\eqref{eq:centered_budget_impulse} of the paper,
\(\sigma_{i,s}=\dot m^\perp_{i,s}/((1-\beta_i)m_i^\ast)\) with
\(\dot{\mathbf m}_s^\perp=\mathbf A^{s}h_\theta\), and by
\(\dot{\widehat d}_{j,s}=\sum_l\Omega_{jl}\dot{\widehat e}_{l,s}+r_j\)
with \(\Omega_{jl}=(1-\beta_l)a_{jl}m_l^\ast/m_j^\ast\),
\(\overline\sigma_{j,s}=\sum_la_{jl}\dot m^\perp_{l,s}/m_j^\ast=\dot m^\perp_{j,s+1}/m_j^\ast\).
Hence
\begin{equation}
\overline\sigma_{j,s}=(1-\beta_j)\,\sigma_{j,s+1}
\label{eq:demand_is_next_expenditure}
\end{equation}
Write \(\iota_j=h_{\theta,j}/m_j^\ast\) for the impact misalignment
per unit of sales, \eqref{eq:incidence_ratio} of the paper, and
\((\mathsf Bz)_j:=(1-\beta_j)\sum_{l\in\mathcal S_j}a_{lj}z_l\) for
the pass-through-weighted average of a firm-indexed profile \(z\) over
the suppliers of firm \(j\).

\begin{lemma}[Mutual dampening of the monetary wave and the real wave]
\label{lem:wave_dampening}
Suppose that Assumptions~\ref{assump:diffuse_incidence} and~\ref{assump:monotone_decay_spectrum} of the paper
hold, that every purchase has the same lead time
\(k\ge1\), and that the impact misalignment \(h_\theta\) is an
eigenvector of \(\mathbf A\) with a real eigenvalue \(\lambda\) in
\((0,1)\). Fix a weighting \(\boldsymbol\omega\) of the firms and let
\(\upsilon:=(\mathbf I-\lambda^{k}\mathsf B)^{-1}(\mathbf I-\lambda\mathsf B)\,\iota\).
Then
\begin{equation}
\langle M,Q\rangle_H
=
-\,(1-\rho)\,\frac{\lambda^{k+1}}{1-\lambda^{2}}
\sum_{i\in N}\omega_i(1-\beta_i)\operatorname{Cov}_{\mathcal F_i}(\iota,\upsilon)
\label{eq:cross_term_sign}
\end{equation}
The monetary wave and the real wave therefore dampen each other,
\(\langle M,Q\rangle_H<0\), whenever
\(\operatorname{Cov}_{\mathcal F_i}(\iota,\upsilon)\ge0\) at every firm in
the support of \(\boldsymbol\omega\) and
\(\operatorname{Cov}_{\mathcal F_i}(\iota,\upsilon)>0\) at one of them.
When \(k=1\), \(\upsilon=\iota\), and this holds as soon as some firm in
the support of \(\boldsymbol\omega\) has suppliers with different
\(\iota_j\). When \(k\ge2\), it holds whenever some firm in the
support of \(\boldsymbol\omega\) has suppliers with different \(\iota_j\)
and, at every such firm \(i\),
\begin{equation}
\operatorname{Var}_{\mathcal F_i}(\mathsf B^{r}\iota)^{1/2}
\;\le\;
C_{\mathsf B}\,\operatorname{Var}_{\mathcal F_i}(\iota)^{1/2}
\quad\text{for every }r\ge1,
\qquad
C_{\mathsf B}<\frac{1-\lambda^{k}}{\lambda-\lambda^{k}}
\label{eq:dispersion_condition}
\end{equation}
In particular it holds whenever averaging over suppliers does not increase the
dispersion of \(\iota\) among the suppliers of any such firm, which is
\eqref{eq:dispersion_condition} with \(C_{\mathsf B}=1\).
\end{lemma}

\begin{proof}[Proof of Lemma~\ref{lem:wave_dampening}]
Since \(\mathbf Ah_\theta=\lambda h_\theta\),
\(\dot{\mathbf m}_s^\perp=\lambda^{s}h_\theta\) for \(s\ge0\), so
\(\sigma_{j,s}=\lambda^{s}\iota_j/(1-\beta_j)\) and, by
\eqref{eq:demand_is_next_expenditure},
\(\overline\sigma_{j,s}=\lambda^{s+1}\iota_j\). Hence
\(\Xi_{ij,s}=\lambda^{s}(\sigma_{i,0}-\lambda\iota_j)\) for \(s\ge0\), and
with the common lead time \(M_{ji,t}=\lambda^{t-k}(\sigma_{i,0}-\lambda\iota_j)\)
for \(t\ge k\) and \(M_{ji,t}=0\) for \(t<k\).

With the common lead time the
dated relative expenditure impulse is
\(\xi_{j,s}=\sum_{l\in\mathcal S_j}a_{lj}\Xi_{jl,s-k}\) for \(s\ge k\)
and zero before, so, with \(\Xi_{jl,s}=\lambda^{s}(\sigma_{j,0}-\lambda\iota_l)\) and \((1-\beta_j)\sigma_{j,0}=\iota_j\),
\[
\dot{\widehat q}^{(1)}_{j,s}=(1-\beta_j)\xi_{j,s}
=\lambda^{s-k}\Bigl[(1-\beta_j)\sigma_{j,0}-\lambda(1-\beta_j)\sum_{l\in\mathcal S_j}a_{lj}\iota_l\Bigr]
=\lambda^{s-k}\upsilon^{(1)}_j,
\qquad
\upsilon^{(1)}:=(\mathbf I-\lambda\mathsf B)\iota
\]
for \(s\ge k\). The dated transmission operator is
\((\mathsf Tz)_{j,s}=\sum_{l\in\mathcal S_j}(1-\beta_l)a_{lj}z_{l,s-k}\)
for \(s\ge k\), so
\(\dot{\widehat q}^{(r)}_{j,s}=(1-\beta_j)\sum_la_{lj}\dot{\widehat q}^{(r-1)}_{l,s-k}\),
and by induction
\(\dot{\widehat q}^{(r)}_{j,s}=\lambda^{s-rk}(\mathsf B^{r-1}\upsilon^{(1)})_j\)
for \(s\ge rk\) and zero before. Hence
\(Q^{(r)}_{ji,t}=\lambda^{t-(r+1)k}(\mathsf B^{r-1}\upsilon^{(1)})_j\) for
\(t\ge(r+1)k\) and zero before.

With one vintage layer in every
bundle, \(\langle y,z\rangle_i=(1-\rho)\operatorname{Cov}_{\mathcal F_i}(y,z)\).
For \(t\ge(r+1)k\), the constant \(\sigma_{i,0}\) drops out of the
covariance and
\[
\operatorname{Cov}_{\mathcal F_i}\bigl(M_{i,t},Q^{(r)}_{i,t}\bigr)
=-\lambda^{2t-(r+2)k+1}\operatorname{Cov}_{\mathcal F_i}\bigl(\iota,\mathsf B^{r-1}\upsilon^{(1)}\bigr)
\]
while for \(t<(r+1)k\) the term vanishes. Summing over \(t\),
\[
\bigl\langle M,Q^{(r)}\bigr\rangle_H
=
-\,(1-\rho)\,\frac{\lambda^{rk+1}}{1-\lambda^{2}}
\sum_{i\in N}\omega_i(1-\beta_i)\operatorname{Cov}_{\mathcal F_i}(\iota,\mathsf B^{r-1}\upsilon^{(1)})
\]
Since \(\|\mathsf B\|_\infty\le1-\underline\beta<1\), the series
\(\sum_{r\ge1}\lambda^{rk}\mathsf B^{r-1}=\lambda^{k}(\mathbf I-\lambda^{k}\mathsf B)^{-1}\)
converges absolutely, and by \eqref{eq:round_norm_bound} so does
\(\sum_r\langle M,Q^{(r)}\rangle_H=\langle M,Q\rangle_H\). Summing
over \(r\) gives \eqref{eq:cross_term_sign} with
\(\upsilon=(\mathbf I-\lambda^{k}\mathsf B)^{-1}\upsilon^{(1)}\).

The factor
\((1-\rho)\lambda^{k+1}/(1-\lambda^{2})\) is positive, so the sign of
\(\langle M,Q\rangle_H\) is minus that of
\(\sum_i\omega_i(1-\beta_i)\operatorname{Cov}_{\mathcal F_i}(\iota,\upsilon)\),
which is the first claim. For \(k=1\),
\(\upsilon=(\mathbf I-\lambda\mathsf B)^{-1}(\mathbf I-\lambda\mathsf B)\iota=\iota\)
and \(\operatorname{Cov}_{\mathcal F_i}(\iota,\upsilon)=\operatorname{Var}_{\mathcal F_i}(\iota)\).
For \(k\ge2\), a firm whose suppliers share the same \(\iota_j\) has
\(\operatorname{Cov}_{\mathcal F_i}(\iota,\upsilon)=0\). At the other
firms,
\(\upsilon-\iota=(\mathbf I-\lambda^{k}\mathsf B)^{-1}(\lambda^{k}-\lambda)\mathsf B\iota=-(\lambda-\lambda^{k})\sum_{r\ge0}\lambda^{kr}\mathsf B^{r+1}\iota\),
so by the Cauchy--Schwarz inequality for
\(\operatorname{Cov}_{\mathcal F_i}\) and
\eqref{eq:dispersion_condition},
\[
\operatorname{Cov}_{\mathcal F_i}(\iota,\upsilon)
\;\ge\;
\operatorname{Var}_{\mathcal F_i}(\iota)
-(\lambda-\lambda^{k})\sum_{r\ge0}\lambda^{kr}\operatorname{Var}_{\mathcal F_i}(\iota)^{1/2}\operatorname{Var}_{\mathcal F_i}(\mathsf B^{r+1}\iota)^{1/2}
\;\ge\;
\operatorname{Var}_{\mathcal F_i}(\iota)\Bigl[1-\frac{(\lambda-\lambda^{k})\,C_{\mathsf B}}{1-\lambda^{k}}\Bigr]
\]
which is positive because
\(\operatorname{Var}_{\mathcal F_i}(\iota)>0\) and
\(C_{\mathsf B}<(1-\lambda^{k})/(\lambda-\lambda^{k})\). The value \(C_{\mathsf B}=1\)
satisfies this because \((1-\lambda^{k})-(\lambda-\lambda^{k})=1-\lambda>0\).
\end{proof}

\subsection{Lemmas for Section~\ref{sec:reallocation}: Monetary Reallocation}
\label{app:lemmas_reallocation}

Write
\begin{equation}
\overline\ell_j:=\sum_{l\in N}\Omega_{jl}\,\ell_l,
\qquad
\Delta\ell_i^{(k)}:=\sum_{j\in\mathcal S_i^{(k)}}\widetilde a^{(k)}_{ji}\bigl(\overline\ell_j-\ell_i\bigr)
\label{eq:buyer_average_size}
\end{equation}
for the average centered log size of the buyers of good \(j\) and for
the log-size gap of firm \(i\) on its inputs of lead time \(k\). Write
\begin{equation}
\sigma^{\circ}_l:=-\frac{\ell_l}{1-\beta_l},
\qquad
\xi_i^{(k)\circ}:=\sigma^{\circ}_i-\sum_{j\in\mathcal S_i^{(k)}}\widetilde a^{(k)}_{ji}\sum_{l\in N}\Omega_{jl}\,\sigma^{\circ}_l
\label{eq:exposure_advantage}
\end{equation}
for the incidence exposure of firm \(l\) and for the exposure
advantage of firm \(i\) on its inputs of lead time \(k\). Let
\begin{equation}
C_\ell
:=
\frac{2}{1-\overline\beta}\,
L_\ell^{2}\Bigl(1+\frac{L_\ell}{4}\Bigr)e^{L_\ell},
\qquad
\Delta_\beta
:=
\frac{2L_\ell\,(\overline\beta-\underline\beta)}
{(1-\overline\beta)(1-\underline\beta)}
\label{eq:exposure_constant}
\end{equation}
be the remainder constant and the labor-share correction.

\begin{lemma}[The relative expenditure impulse to leading order in the slack]
\label{lem:impulse_exposure}
Suppose that the log-size bound \eqref{eq:log_size_bound} and
Assumption~\ref{assump:diffuse_incidence} of the paper hold. Then at
every firm \(i\) and every active vintage \(k\),
\begin{equation}
\Bigl|\xi_i^{(k)}-(1-\theta)\,\xi_i^{(k)\circ}\Bigr|
\;\le\;
C_\ell\,(1-\theta)^2
\label{eq:exposure_remainder}
\end{equation}
and the exposure advantage lies within the labor-share correction of
the log-size gap over the intermediate-input share,
\begin{equation}
\frac{\Delta\ell_i^{(k)}}{1-\overline\beta}-\Delta_\beta
\;\le\;
\xi_i^{(k)\circ}
\;\le\;
\frac{\Delta\ell_i^{(k)}}{1-\overline\beta}+\Delta_\beta
\label{eq:size_to_exposure}
\end{equation}
The impulse \(\xi_{\mathcal C}^{(k)}\) of a set of firms,
\eqref{eq:set_impulse} of the paper, obeys the same two bounds with the
exposure advantage and the log-size gap averaged with the same
weights.
\end{lemma}

\begin{proof}[Proof of Lemma~\ref{lem:impulse_exposure}]
By Assumption~\ref{assump:diffuse_incidence} of the paper, \(1-\theta\in(0,\tfrac12]\).
By \eqref{eq:expenditure_impulse} and the form of the incidence ratio
given before \eqref{eq:incidence_ratio} of the paper, the expenditure
impulse of firm \(l\) is
\[
\sigma_l
=\frac{1}{1-\beta_l}\Bigl(\frac{e^{-(1-\theta)\ell_l}}{Z}-1\Bigr),
\qquad
Z:=\sum_{j\in N}\frac{m_j^\ast}{\overline M}\,e^{-(1-\theta)\ell_j}
\]
and by \eqref{eq:buyer_average_impulse} and
\eqref{eq:vintage_impulse_def} of the paper
\[
\xi_i^{(k)}
=\sum_{j\in\mathcal S_i^{(k)}}\widetilde a^{(k)}_{ji}
\Bigl(\sigma_i-\sum_{l\in N}\Omega_{jl}\sigma_l\Bigr)
\]

For
\(|z|\le(1-\theta)L_\ell\), \(e^{-z}=1-z+\varepsilon(z)\) with
\(|\varepsilon(z)|\le\tfrac12z^2e^{|z|}\). Write \(\varepsilon_l:=\varepsilon((1-\theta)\ell_l)\), so
that every \(|\varepsilon_l|\) is at most \(\varepsilon_{\max}:=\tfrac12((1-\theta)L_\ell)^2e^{(1-\theta)L_\ell}\).
Since \(\ell\) is centered under the stationary weights,
\(\sum_j(m_j^\ast/\overline M)\ell_j=0\), the denominator is
\(Z=1+\overline\varepsilon\) with
\(\overline\varepsilon:=\sum_j(m_j^\ast/\overline M)\varepsilon_j\),
\(|\overline\varepsilon|\le\varepsilon_{\max}\), and \(Z\ge e^{-(1-\theta)L_\ell}\). Hence
\[
\frac{e^{-(1-\theta)\ell_l}}{Z}-1+(1-\theta)\ell_l
=\frac{\varepsilon_l-\overline\varepsilon\,(1-(1-\theta)\ell_l)}{Z}
\]
so that
\[
\Bigl|\frac{e^{-(1-\theta)\ell_l}}{Z}-1+(1-\theta)\ell_l\Bigr|
\le\varepsilon_{\max}\,(2+(1-\theta)L_\ell)\,e^{(1-\theta)L_\ell}
\le C_0\,(1-\theta)^2,
\qquad
C_0:=L_\ell^2\Bigl(1+\frac{L_\ell}{4}\Bigr)e^{L_\ell}
\]
the last step using \(1-\theta\le\tfrac12\). Dividing by \(1-\beta_l\),
\(|\sigma_l-(1-\theta)\,\sigma^{\circ}_l|\le C_0(1-\theta)^2/(1-\beta_l)\le C_0(1-\theta)^2/(1-\overline\beta)\)
at every firm.

The difference
\(\sigma_i-\sum_l\Omega_{jl}\sigma_l\) differs from
\((1-\theta)\,(\sigma^{\circ}_i-\sum_l\Omega_{jl}\sigma^{\circ}_l)\) by at most
\(C_0(1-\theta)^2(1+\sum_l\Omega_{jl})/(1-\overline\beta)\le2C_0(1-\theta)^2/(1-\overline\beta)=C_\ell(1-\theta)^2\).
Averaging over the suppliers of lead time \(k\) with the within-vintage shares,
which sum to one, gives \eqref{eq:exposure_remainder}.

For \eqref{eq:size_to_exposure}, for every firm \(l\) write
\(1/(1-\beta_l)=1/(1-\overline\beta)-\Delta_{\beta,l}\) with
\(\Delta_{\beta,l}:=1/(1-\overline\beta)-1/(1-\beta_l)\), which lies in
\([0,(\overline\beta-\underline\beta)/((1-\overline\beta)(1-\underline\beta))]\),
so that \(|\ell_l\Delta_{\beta,l}|\le\Delta_\beta/2\). For a buyer \(i\) and
a supplier \(j\),
\[
\sigma^{\circ}_i-\sum_{l\in N}\Omega_{jl}\sigma^{\circ}_l
=
\frac{\overline\ell_j-\ell_i}{1-\overline\beta}
+\ell_i\Delta_{\beta,i}-\sum_{l\in N}\Omega_{jl}\,\ell_l\Delta_{\beta,l}
\]
and the last two terms together lie in \([-\Delta_\beta,\Delta_\beta]\)
since \(\sum_l\Omega_{jl}\le1\). Averaging over firm \(i\)'s suppliers of lead time \(k\)
with the within-vintage shares puts \(\xi_i^{(k)\circ}\) within
\(\Delta_\beta\) of \(\Delta\ell_i^{(k)}/(1-\overline\beta)\), which is
\eqref{eq:size_to_exposure}. The average \eqref{eq:set_impulse} of the paper has nonnegative
weights that sum to one, so both bounds pass to it.
\end{proof}

\begin{lemma}[Decomposition of the linearized order response]
\label{lem:order_response_decomposition}
At every post-shock order date \(s\ge 0\), the
first-order proportional response of firm \(i\)'s intermediate-input
order from supplier \(j\), \(\dot{\widehat x}_{ji,s}\), decomposes into a
buyer-side expenditure channel and a supplier-side price channel,
\begin{equation}
\dot{\widehat x}_{ji,s}=\dot{\widehat e}_{i,s}-\dot{\widehat p}_{j,s}
\label{eq:bundle_response_decomp}
\end{equation}
Here \(\dot{\widehat e}_{i,s}\) is the response of firm \(i\)'s
intermediate-input expenditure and \(\dot{\widehat p}_{j,s}\) the response
of supplier \(j\)'s price. The expenditure channel admits the explicit form
\begin{equation}
\dot{\widehat e}_{i,s}
=
\frac{\dot m_{i,s}-(w^\ast/\overline M)(\mathbf 1^\top\dot{\mathbf m}_s)\,l_i^\ast}
{(1-\beta_i)m_i^\ast}
\label{eq:dot_e_main}
\end{equation}
linear in the first-order firm-balance perturbation \(\dot{\mathbf m}_s\)
at date \(s\). The price channel is fixed by the linearized market-clearing
identity
\begin{equation}
\dot{\widehat p}_{j,t}=\dot{\widehat d}_{j,t}-\dot{\widehat q}_{j,t}
\label{eq:price_response_identity}
\end{equation}
together with the supply-response identity
\begin{equation}
\dot{\widehat q}_{j,t}=(1-\beta_j)\,\mathcal F_j(\dot{\widehat x}_{j,t})
\label{eq:supply_response_identity}
\end{equation}
where \(\mathcal F_j(\dot{\widehat x}_{j,t})\) is the first-order proportional response
of firm \(j\)'s input-bundle scale \eqref{eq:bundle_scale_def}. Consequently
\begin{equation}
\dot{\widehat p}_{j,s}=\dot{\widehat d}_{j,s}-(1-\beta_j)\,\mathcal F_j(\dot{\widehat x}_{j,s})
\label{eq:price_closure}
\end{equation}
and
\begin{equation}
\dot{\widehat x}_{ji,s}
=
\dot{\widehat e}_{i,s}-\dot{\widehat d}_{j,s}+(1-\beta_j)\,\mathcal F_j(\dot{\widehat x}_{j,s})
\label{eq:bundle_closed}
\end{equation}
\end{lemma}

\begin{proof}[Proof of Lemma~\ref{lem:order_response_decomposition}]
The order rule \eqref{eq:dynamic_real_alloc} reads
\(x_{ji,t}=a_{ji}(m_{i,t}-\wM\,l_i^\ast)/p_{j,t}\). Taking logarithms and
the derivative in \(\pi\) at \(\pi=0\) gives
\eqref{eq:bundle_response_decomp}.

For the expenditure channel \eqref{eq:dot_e_main}, the proportional
deviation of the expenditure from its stationary value
\((1-\beta_i)m_i^\ast\) is
\(\widehat e_{i,t}=[m_{i,t}-\wM\,l_i^\ast-(1-\beta_i)m_i^\ast]/[(1-\beta_i)m_i^\ast]\).
Differentiating in \(\pi\) at \(\pi=0\) gives
\[\dot{\widehat e}_{i,t}=[\dot m_{i,t}-\dot \wM\,l_i^\ast]/[(1-\beta_i)m_i^\ast]\]
and Assumption~\ref{assump:short_run_nominal_rigidity} fixes
\(\wM=(w^\ast/\overline M)\mathbf 1^\top\mathbf m_t\), so that
\(\dot \wM=(w^\ast/\overline M)\mathbf 1^\top\dot{\mathbf m}_t\) and
\eqref{eq:dot_e_main} follows.

The price identity \eqref{eq:price_response_identity} is the clearing
rule \(p_{j,t}=d_{j,t}/q_{j,t}\) of \eqref{eq:dyn_price_update} in
logarithms, differentiated at \(\pi=0\). With labor fixed at
\(l_j^\ast\) by Assumption~\ref{assump:short_run_nominal_rigidity},
\(\log q_j=\beta_j\log l_j^\ast+(1-\beta_j)\log X_j\), and
differentiating at \(\pi=0\) with
\(\nabla_j^\ast\!\cdot\dot{\mathbf x}_{j,t}=(1-\beta_j)\mathcal F_j(\dot{\widehat x}_{j,t})\)
from \eqref{eq:bundle_scale_def} gives
\eqref{eq:supply_response_identity}. Equations \eqref{eq:price_closure}
and \eqref{eq:bundle_closed} follow by substitution.
\end{proof}

Let \(\Xi_{ij,s}:=\dot{\widehat e}_{i,s}-\dot{\widehat d}_{j,s}\) be
the dated contrast \eqref{eq:dated_contrast} of the paper of buyer
\(i\) on supplier \(j\) at order date \(s\), where
\(\dot{\widehat e}_{i,s}=1+\dot m^\perp_{i,s}/\bigl((1-\beta_i)m_i^\ast\bigr)\)
and \(\dot{\widehat d}_{j,s}=\sum_l\Omega_{jl}\dot{\widehat e}_{l,s}+r_j\)
at every date, by \eqref{eq:centered_budget_impulse} in the proof of
Proposition~\ref{prop:impact_reallocation}. The
\emph{dated relative expenditure impulse} of firm \(i\) at date \(t\)
takes each supplier's contrast at the date on which the input usable at
\(t\) was ordered,
\begin{equation}
\xi_{i,t}
:=
\sum_{j\in\mathcal S_i:\,k_{ji}\le t}a_{ji}\,\Xi_{ij,\,t-k_{ji}}
\label{eq:dated_forcing}
\end{equation}
and coincides at \(t=\underline k\) with the impact-date relative expenditure
impulse \eqref{eq:impact_forcing_def}.

\begin{lemma}[Transmission of the relative expenditure impulses]
\label{lem:supplier_recursion}
For every firm \(i\in N\) and date \(t\ge0\), the input-bundle scale
obeys
\begin{equation}
\mathcal F_i(\dot{\widehat x}_{i,t})
=
\xi_{i,t}
+
\sum_{j\in\mathcal S_i:\,k_{ji}\le t}(1-\beta_j)\,a_{ji}\,\mathcal F_j(\dot{\widehat x}_{j,\,t-k_{ji}})
\label{eq:supplier_recursion_lemma}
\end{equation}
Suppose that Assumptions~\ref{assump:diffuse_incidence} and~\ref{assump:monotone_decay_spectrum} of the paper hold. Then the profile
\(\bigl(\mathcal F_i(\dot{\widehat x}_{i,t})\bigr)_{i,t}\) is the unique
bounded solution of that system. And the first-order response of the position-tilted log-output index at every
date is the non-negative combination \eqref{eq:forced_representation} of
dated contrasts, with weights \(\varpi^{(t,s)}_{ij}(b)\ge0\) whose total at
any date is at most \((1-\underline\beta)/\underline\beta\).
\end{lemma}

\begin{proof}[Proof of Lemma~\ref{lem:supplier_recursion}]
By \eqref{eq:bundle_closed} of
Lemma~\ref{lem:order_response_decomposition} at supplier \(j\)'s order
date \(s=t-k_{ji}\), and by \eqref{eq:dated_contrast},
\(\dot{\widehat x}_{ji,t-k_{ji}}=\Xi_{ij,t-k_{ji}}+(1-\beta_j)\mathcal F_j(\dot{\widehat x}_{j,t-k_{ji}})\)
when \(t-k_{ji}\ge0\), while orders placed before the shock are
independent of \(\pi\) and have \(\dot{\widehat x}_{ji,t-k_{ji}}=0\).
Averaging against the shares \(a_{ji}\) over \(j\in\mathcal S_i\) gives
\eqref{eq:supplier_recursion_lemma}.

Define the dated transmission operator on
bounded profiles \(z=(z_{i,t})_{i\in N,\,t\ge0}\) by
\begin{equation}
(\mathsf T z)_{i,t}
:=
\sum_{j\in\mathcal S_i:\,k_{ji}\le t}(1-\beta_j)\,a_{ji}\,z_{j,\,t-k_{ji}}
\label{eq:transmission_operator}
\end{equation}
Since \(\sum_j(1-\beta_j)a_{ji}\le1-\underline\beta\),
\(\|\mathsf T\|_\infty\le1-\underline\beta<1\) on the space of bounded
profiles with the sup norm, and \eqref{eq:supplier_recursion_lemma}
reads \(\mathcal F=\boldsymbol\xi+\mathsf T\mathcal F\) with
\(\boldsymbol\xi=(\xi_{i,t})_{i,t}\) bounded by the nominal envelope.
The unique bounded solution is the Neumann series
\(\mathcal F=\sum_{r\ge0}\mathsf T^r\boldsymbol\xi\), and the
first-order scale profile, bounded uniformly in \(t\) by the
geometric decay of Proposition~\ref{prop:stability_stationary}, is that
solution.

By induction on \(r\), for every
\(r\ge1\),
\[
(\mathsf T^{r-1}\boldsymbol\xi)_{i_0,t}
=
\sum_{i_1,\dots,i_{r-1}}\;\prod_{m=1}^{r-1}(1-\beta_{i_m})a_{i_mi_{m-1}}\;
\xi_{i_{r-1},\,t-\sum_{m=1}^{r-1}k_{i_mi_{m-1}}}
\]
the sum running over the supplier paths \(i_0\to i_1\to\cdots\to i_{r-1}\)
with \(\sum_{m=1}^{r-1}k_{i_mi_{m-1}}\le t\), the term \(r=1\) being
\(\xi_{i_0,t}\). And by \eqref{eq:dated_forcing} each \(\xi\) is the
\(a\)-weighted sum over a last link \(i_{r-1}\to i_r\) of the contrast
\(\Xi_{i_{r-1}i_r,\,t-k(\mathcal P)}\), \(k(\mathcal P):=\sum_{m=1}^rk_{i_mi_{m-1}}\le t\).
Multiplying by \(\omega_{i_0}(b)(1-\beta_{i_0})\), summing over \(i_0\) and
over \(r\ge1\), and using
\(D_\pi\mathcal Q_t(b)=\sum_{i}\omega_i(b)(1-\beta_i)\mathcal F_i(\dot{\widehat x}_{i,t})\)
gives \eqref{eq:forced_representation} with
\[
\varpi^{(t,s)}_{ij}(b)
:=
\sum_{\mathcal P:\,i_{r-1}=i,\;i_r=j,\;k(\mathcal P)=t-s}
\omega_{i_0}(b)(1-\beta_{i_0})\prod_{m=1}^{r}a_{i_mi_{m-1}}\prod_{m=1}^{r-1}(1-\beta_{i_m})
\;\ge\;0
\]
the sum running over the paths of every length \(r\ge1\) whose last
link is \(i\to j\) and whose summed lead time is \(t-s\).

For the paths of \(r\) links,
\[
\sum_{i_0,\dots,i_r}\omega_{i_0}(b)(1-\beta_{i_0})\prod_{m=1}^{r}a_{i_mi_{m-1}}\prod_{m=1}^{r-1}(1-\beta_{i_m})
\;\le\;
(1-\underline\beta)^r\sum_{i_0}\omega_{i_0}(b)\prod_{m=1}^{r}\Bigl(\sum_{i_m}a_{i_mi_{m-1}}\Bigr)
\;=\;(1-\underline\beta)^r
\]
since the supplier shares of each buyer sum to one and each of the
\(r\) pass-through factors is at most \(1-\underline\beta\). Summing over
\(r\ge1\) bounds the total weight at any date by
\(\sum_{r\ge1}(1-\underline\beta)^r=(1-\underline\beta)/\underline\beta\).
\end{proof}

Let \(F_i(T)\) be the share \eqref{eq:delivery_time} of the paper of
firm \(i\)'s output that reaches final demand within \(T\) periods.
And for a set of firms \(\mathcal C\) write
\(F_{\mathcal C}(T):=\sum_{i\in\mathcal C}(m_i^\ast/w^\ast)F_i(T)\) for the sales of
\(\mathcal C\) delivered to final demand within \(T\) periods, relative to
GDP. For a supplier path \(\mathcal P'=(i_0\to\cdots\to i_r)\), each \(i_m\) a
supplier of \(i_{m-1}\), with summed lead time
\(k(\mathcal P'):=\sum_{m=1}^rk_{i_mi_{m-1}}\), the sum of the lead times of its
links, write
\[
\mu_{\boldsymbol\gamma}(\mathcal P')
:=
\gamma_{i_0}(1-\beta_{i_0})\prod_{m=1}^{r-1}(1-\beta_{i_m})\prod_{m=1}^{r}a_{i_mi_{m-1}}
\ \text{ for }r\ge1,
\qquad
\mu_{\boldsymbol\gamma}(\mathcal P'):=\gamma_{i_0}\ \text{ for }r=0
\]
for its mass from the household shares. Here \(\mu_{\boldsymbol\gamma}(\mathcal P')\) is the mass \(\mu(\mathcal P')\) of the
proof of Proposition~\ref{prop:finite_horizon_sign_persistence} of the
paper with \(\gamma_{i_0}\) in place of \(\omega_{i_0}(b)\). And, for
\(L\ge0\), write \(\mu_{\boldsymbol\gamma}(l;L)\) for the total mass of the paths that end at
the buyer \(l=i_r\) with \(k(\mathcal P')\le L\). Write
\(A_t=\sum_{i\in\mathcal R}\gamma_i(1-\beta_i)\mathcal F_i(\dot{\widehat x}_{i,t})\)
for the reallocation impulse of Section~\ref{subsec:asymmetric_nonneutrality} of
the paper.

\begin{lemma}[Delivery time to final demand and the pass-through to GDP]
\label{lem:gdp_pass_through}
Suppose that \(\underline k\ge1\). Then for every firm \(l\) and
every \(L\ge0\)
\begin{equation}
\mu_{\boldsymbol\gamma}(l;L)=\frac{m_l^\ast}{w^\ast}\,F_l(L)
\label{eq:domar_pass_through}
\end{equation}
The share \(F_l(L)\) is nondecreasing in \(L\) and rises to one. So
the total mass of the paths ending at \(l\) is its Domar weight
\(m_l^\ast/w^\ast\), and the paths ending in a set \(\mathcal C\) with summed
lead time at most \(L\) have mass \(F_{\mathcal C}(L)\). If
Assumption~\ref{assump:retail_outflow_cap} of the paper holds, then
for every tilt \(b\in[0,\infty]\) and every retailer
\(i\in\mathcal R\),
\begin{equation}
\omega_i(b)\;\le\;\epsilon_h^{-1}\,\gamma_i
\label{eq:tilt_household_comparison}
\end{equation}
So the paths from the tilt weights that start on the retail tier, end
in a set \(\mathcal C\) and have summed lead time at most \(L\) have mass at most
\(\epsilon_h^{-1}F_{\mathcal C}(L)\). If
Assumptions~\ref{assump:diffuse_incidence}
and~\ref{assump:monotone_decay_spectrum} of the paper hold, then the
reallocation impulse has the expansion \eqref{eq:forced_representation}
with the household shares in place of the tilt weights,
\begin{equation}
A_t
=
\sum_{s=0}^{t}\;\sum_{(i,j):\,j\in\mathcal S_i}
\varpi^{(t,s)}_{ij}(\boldsymbol\gamma)\,\Xi_{ij,s}
=:I_t(\boldsymbol\gamma)+J_t(\boldsymbol\gamma)
\label{eq:gdp_path_expansion}
\end{equation}
Here \(I_t(\boldsymbol\gamma)\) is the part carried by the impact-date
contrasts and \(J_t(\boldsymbol\gamma)\) the rest. And for every
horizon \(T\ge0\)
\begin{equation}
\sum_{t=0}^{T}I_t(\boldsymbol\gamma)
=
\sum_{l\in N}\sum_{k\in\mathcal K_l}\frac{m_l^\ast}{w^\ast}\,F_l(T-k)\,s_l^{(k)}\phi_l^{(k)},
\qquad
J_t(\boldsymbol\gamma)=0\ \text{ for every }t\le\underline k
\label{eq:window_impact_component}
\end{equation}
\end{lemma}

\begin{proof}
A path of one or more links ending at
\(l\) is a path \(\mathcal P''\) with one link fewer, ending at a buyer \(j\) of
which \(l\) is a supplier, extended by the link \(j\to l\), with
\(\mu_{\boldsymbol\gamma}(\mathcal P')=\mu_{\boldsymbol\gamma}(\mathcal P'')(1-\beta_j)a_{lj}\)
and \(k(\mathcal P')=k(\mathcal P'')+k_{lj}\). What firm \(j\) spends on good \(l\) is the
share \(\Omega_{lj}\) of \(l\)'s sales,
\(m_j^\ast(1-\beta_j)a_{lj}=m_l^\ast\Omega_{lj}\), so
\(\mu_{\boldsymbol\gamma}(\mathcal P')=\mu_{\boldsymbol\gamma}(\mathcal P'')\,(m_l^\ast/m_j^\ast)\,\Omega_{lj}\).
The one-firm path at \(l\) has mass \(\gamma_l=r_lm_l^\ast/w^\ast\)
and summed lead time zero. Summing over the paths with \(k(\mathcal P')\le L\), and
setting \(\mu_{\boldsymbol\gamma}(l;L):=0\) for \(L<0\), since no path has a negative
summed lead time,
\[
\mu_{\boldsymbol\gamma}(l;L)
=
\frac{m_l^\ast}{w^\ast}\,r_l
+\sum_{j\in N}\frac{m_l^\ast}{m_j^\ast}\,\Omega_{lj}\,\mu_{\boldsymbol\gamma}(j;L-k_{lj}),
\qquad L\ge0
\]
so \(\widetilde F_l(L):=w^\ast\mu_{\boldsymbol\gamma}(l;L)/m_l^\ast\) obeys
\(\widetilde F_l(L)=r_l+\sum_j\Omega_{lj}\widetilde F_j(L-k_{lj})\) for \(L\ge0\) and
\(\widetilde F_l(L)=0\) for \(L<0\), the recursion \eqref{eq:delivery_time} of
the paper. Every lead time is at least one, so the recursion determines its
solution level by level in \(L\) from the values at negative \(L\),
and \(\widetilde F=F\), which is \eqref{eq:domar_pass_through}. By induction on
\(L\), \(F_l(L)\) is nondecreasing in \(L\) and at most
\(r_l+\sum_j\Omega_{lj}=1\), so it has a limit \(F_l(\infty)\), which
obeys \(F_l(\infty)=r_l+\sum_j\Omega_{lj}F_j(\infty)\). The matrix
\(\boldsymbol\Omega=\mathbf D_m^{-1}\mathbf A_\beta\mathbf D_m\) has
the spectral radius of \(\mathbf A_\beta\), which is below one at the
stationary equilibrium, so this system has the unique solution
\((\mathbf I-\boldsymbol\Omega)^{-1}\mathbf r\), and since
\(\boldsymbol\Omega\mathbf 1=\mathbf 1-\mathbf r\) that solution is
\(\mathbf 1\). The total mass of the paths ending at \(l\) is therefore
\(m_l^\ast/w^\ast\). Summing \eqref{eq:domar_pass_through} over
\(l\in\mathcal C\) gives the mass of the paths ending in \(\mathcal C\).

Every retailer has \(r_i\ge\epsilon_h\)
by Assumption~\ref{assump:retail_outflow_cap} and
\(\gamma_i=r_im_i^\ast/w^\ast\), so
\(m_i^\ast\le\gamma_iw^\ast/\epsilon_h\). The capped position
\(\widehat\psi\) of \eqref{eq:w_b_def} is one on the retail tier, so
the denominator of the tilt weights is at least its retail part,
\(\sum_{j\in N}m_j^\ast\widehat\psi_j^{\,b}\ge\sum_{j\in\mathcal R}m_j^\ast\ge w^\ast\),
since \(m_j^\ast\ge\gamma_jw^\ast\) and the household shares sum to
one. Hence \(\omega_i(b)\le m_i^\ast/w^\ast\le\epsilon_h^{-1}\gamma_i\)
for every \(i\in\mathcal R\) and every \(b\), which is
\eqref{eq:tilt_household_comparison}. The mass \(\mu(\mathcal P')\) of a path
from the tilt weights is \(\omega_{i_0}(b)\) times a product that does not
involve the weights. Hence a path starting at a retailer \(i_0\) has at
most \(\epsilon_h^{-1}\) times its mass from the household shares,
and the bound on the paths that start on the retail tier follows from
\eqref{eq:domar_pass_through}.

The proof of
Lemma~\ref{lem:supplier_recursion} uses the weights only through the
linear combination
\(\sum_i\omega_i(1-\beta_i)\mathcal F_i(\dot{\widehat x}_{i,t})\), so the
expansion holds for any non-negative weights, and \(A_t\) is that
combination at the household shares, which is
\eqref{eq:gdp_path_expansion}. For the identity \eqref{eq:window_impact_component}, a path
\(\mathcal P'\) ending at \(l\) carries the impact-date impulse of the vintage
\(k=t-k(\mathcal P')\) to the date \(t\), grouped by vintage as in
\eqref{eq:vintage_grouping} of the paper. Hence the terms of
\(I_t(\boldsymbol\gamma)\) with \(t\le T\) are those with
\(k(\mathcal P')\le T-k\), whose mass is \(\mu_{\boldsymbol\gamma}(l;T-k)\), and \eqref{eq:domar_pass_through} gives
the first identity. For the post-impact part, each path \(\mathcal P'\)
extended by a link \(l\to j\) carries a contrast of date \(s\ge1\) to
the single date \(t=k(\mathcal P')+k_{jl}+s\), which is at least
\(\underline k+1\) because every lead time is at least \(\underline k\). Hence
\(J_t(\boldsymbol\gamma)=0\) at every \(t\le\underline k\).
\end{proof}

For a downstream-concentrated tilt \(b\) and a date \(t\), let
\[
\mu_t(\mathcal T)
:=
\sum_{(i,j):\,i\in\mathcal T}\frac{\varpi^{(t,0)}_{ij}(b)}{1-\beta_i}
\]
be the tilt mass on the impact-date contrasts at buyers in the top
tier \(\mathcal T\), scaled by the buyers' intermediate-input shares,
as in Section~\ref{subsec:sign_reversing_irf} of the paper. And let
\(r_a\), \(\overline\tau\) and \(\underline\mu\) be the constants
\eqref{eq:reversal_constants} of the paper.

\begin{lemma}[Arrival of the cascade at the retail tier]
\label{lem:layered_routing}
Suppose that Assumptions~\ref{assump:transmission_delay},
\ref{assump:retail_outflow_cap} and~\ref{assump:lag_position_sorting}
of the paper hold and that \(\underline k\ge1\). Then for every
downstream-concentrated tilt \(b\) there is a date
\(\tau^-\in[\underline k+k_\ast,\overline\tau]\) with
\(\mu_{\tau^-}(\mathcal T)\ge\underline\mu\). Every retail lead-time
date is below \(k_\ast\), so \(\tau^-\) lies beyond the retail lead-time
dates \(\mathcal W\).
\end{lemma}

\begin{proof}[Proof of Lemma~\ref{lem:layered_routing}]
Let \(\mathsf X\) be the supplier walk started
from the tilt, \(\Pr(\mathsf X_0=i)=\omega_i(b)\), moving from a buyer \(l\) to the
supplier \(i\) with probability \(a_{il}\), and stopped on first entry into
the top tier, at \(r_{\mathcal T}:=\inf\{r\ge0:\mathsf X_r\in\mathcal T\}\). We first show
that the walk is absorbed within \(r_a\) steps with probability at least
one half. While the walk
runs it sits at firms outside \(\mathcal T\), where
Assumption~\ref{assump:transmission_delay} applies, so that
\[
\mathbb E\bigl[\psi_{\mathsf X_{r+1}}\,\big|\,\mathsf X_r=l\bigr]
=
\mathcal F_l(\psi)
\;\le\;
\psi_l-\Delta_\psi,
\qquad l\notin\mathcal T
\]
The walk therefore drifts upstream by at least \(\Delta_\psi\) per step until it
is absorbed. Set \(\mathsf M_r:=\psi_{\mathsf X_{r\wedge r_{\mathcal T}}}+\Delta_\psi\,(r\wedge r_{\mathcal T})\). On
\(\{r<r_{\mathcal T}\}\),
\(\mathbb E[\mathsf M_{r+1}\mid \mathsf X_0,\dots,\mathsf X_r]\le\psi_{\mathsf X_r}-\Delta_\psi+\Delta_\psi(r+1)=\mathsf M_r\),
and on \(\{r\ge r_{\mathcal T}\}\), \(\mathsf M_{r+1}=\mathsf M_r\), so \(\mathsf M\) is a supermartingale
and \(\mathbb E[\mathsf M_r\mid \mathsf X_0]\le \mathsf M_0=\psi_{\mathsf X_0}\le1\). Since \(\psi\ge0\),
\(\Delta_\psi\,\mathbb E[r\wedge r_{\mathcal T}\mid \mathsf X_0]\le1\) for every \(r\), and
monotone convergence gives
\[
\mathbb E[r_{\mathcal T}\mid \mathsf X_0]\;\le\;\frac{1}{\Delta_\psi},
\qquad
\Pr(r_{\mathcal T}>r_a\mid \mathsf X_0)
\;\le\;
\frac{\mathbb E[r_{\mathcal T}\mid \mathsf X_0]}{r_a+1}
\;<\;
\frac{1}{\Delta_\psi\,r_a}
\;\le\;\tfrac12
\]
by Markov's inequality and \(r_a=\lceil2/\Delta_\psi\rceil\).
The walk is absorbed within \(r_a\) steps with probability at least one
half, whatever its starting firm. Restricting to the starts on the retail
tier, which carries tilt mass at least \(1-\overline\epsilon_w\ge7/8\),
since \(\overline\epsilon_w<1/8\) by \eqref{eq:downstream_threshold} of
the paper, and is disjoint from \(\mathcal T\), since
\(\psi_T<\psi_R\),
\[
\Pr\bigl(\mathsf X_0\in\mathcal R,\ r_{\mathcal T}\le r_a\bigr)
\;\ge\;
\tfrac{7}{16}
\]

We now convert walk probabilities into tilt mass. A realization
\((\mathsf X_0,\dots,\mathsf X_{r_{\mathcal T}})\) with \(r_{\mathcal T}=r\) is a supplier path \(\mathcal P'\) of \(r\)
links ending at the buyer \(\mathsf X_{r_{\mathcal T}}\in\mathcal T\), with walk probability
\(\omega_{\mathsf X_0}(b)\prod_{m=1}^ra_{\mathsf X_m\mathsf X_{m-1}}\) and tilt mass
\(\mu(\mathcal P')=\omega_{\mathsf X_0}(b)(1-\beta_{\mathsf X_0})\prod_{m=1}^{r-1}(1-\beta_{\mathsf X_m})\prod_{m=1}^ra_{\mathsf X_m\mathsf X_{m-1}}\)
in the notation of the proof of
Proposition~\ref{prop:finite_horizon_sign_persistence}. The two differ by
the pass-through factors of the \(r\) firms \(\mathsf X_0,\dots,\mathsf X_{r-1}\), so
\(\mu(\mathcal P')\) is at least \((1-\overline\beta)^r\) times the walk
probability, and at least \((1-\overline\beta)^{r_a}\) times it when
\(r\le r_a\). Extending \(\mathcal P'\) by one link to a supplier \(j\) of the buyer
\(l=\mathsf X_{r_{\mathcal T}}\) carries the contrast \(\Xi_{lj,0}\) at the date \(t=k(\mathcal P')+k_{jl}\),
and the mass it carries, scaled by \(1-\beta_l\), is \(\mu(\mathcal P')a_{jl}\).
The buyer \(l\) lies in the top tier, so by
Assumption~\ref{assump:lag_position_sorting} its suppliers of lead time
\(k_\ast\) or more carry at least the share \(1-\epsilon_k\) of its
expenditure, \(\sum_{j\in\mathcal S_l:\,k_{jl}\ge k_\ast}a_{jl}\ge1-\epsilon_k\).
For these suppliers the date \(k(\mathcal P')+k_{jl}\) lies in
\([\underline k+k_\ast,\overline\tau]\): the walk makes at least one
step, since \(\mathsf X_0\in\mathcal R\) and \(\mathsf X_{r_{\mathcal T}}\in\mathcal T\), so
\(k(\mathcal P')\ge r\underline k\ge\underline k\), while
\(k(\mathcal P')\le r\widehat k\le r_a\widehat k\) and \(k_{jl}\le\widehat k\).
Distinct realizations are distinct paths, and every path counted here
has its last buyer in \(\mathcal T\), so
\begin{align*}
\sum_{t=\underline k+k_\ast}^{\overline\tau}\mu_t(\mathcal T)
&\;\ge\;
(1-\epsilon_k)\sum_{\mathcal P'}\mu(\mathcal P')
\;\ge\;
(1-\epsilon_k)(1-\overline\beta)^{r_a}\Pr\bigl(\mathsf X_0\in\mathcal R,\ r_{\mathcal T}\le r_a\bigr)\\
&\;\ge\;
\tfrac{7}{16}\,(1-\epsilon_k)(1-\overline\beta)^{r_a}
\end{align*}
the middle sum running over the realizations with \(\mathsf X_0\in\mathcal R\)
and \(r_{\mathcal T}\le r_a\). Some date \(\tau^-\) among the
\(\overline\tau-\underline k-k_\ast+1\) dates of the window therefore
carries at least \(\underline\mu\).

Finally, at a retail lead-time date \(t\) the retail tier places at least
the share \(s_R\) of its input expenditure on suppliers of lead time \(t\),
by the definition \eqref{eq:retail_lag_dates}. But by
Assumption~\ref{assump:lag_position_sorting} every retailer, and so
the retail tier, places at most the share \(\epsilon_k<s_R\) on
suppliers of lead time \(k_\ast\) or more, so \(t<k_\ast\), and
\(\tau^-\ge\underline k+k_\ast>t\).
\end{proof}

\begin{lemma}[The economy without dated inputs]
\label{lem:no_dated_inputs}
Suppose that Assumptions~\ref{assump:diffuse_incidence}
and~\ref{assump:monotone_decay_spectrum} of the paper hold, that
\(|\pi|\le\overline\pi\), and that \(k_{ji}=0\) for every purchase.
Then at every firm \(i\) and date \(t\),
\(V_i^{\mathrm b}(\widehat x_{i,t})=0\) and
\(\mathcal D_{i,t}=\tfrac12(1-\beta_i)(1-\rho)\operatorname{Var}_{\mathcal F_i}(\widehat x_{i,t})\),
and at every date \(t\ge0\), \(P_t=0\) and
\(f_t(\pi)=-\Lambda_t(\pi)\le0\).
\end{lemma}

\begin{proof}
Every bundle has a single vintage layer, so the between-vintage
variance in \eqref{eq:ltv_decomposition} vanishes, and the loss from
production inefficiency \eqref{eq:firm_inefficiency_def} of the paper
reduces to the stated form. Every order is used in the period in which
it is placed, so no order is outstanding at any date and \(P_t=0\). The resource identity
\eqref{eq:resource_identity} of the paper then reads
\(f_t(\pi)=-\Lambda_t(\pi)\), and \(\Lambda_t\ge0\)
(Lemma~\ref{lem:resource_accounting_date}).
\end{proof}

\subsection{Lemmas for Section~\ref{sec:asymmetry}: Monetary Non-neutrality}
\label{app:lemmas_asymmetry}

\begin{lemma}[Asymmetry of a concave response]
\label{lem:concave_response_asymmetry}
Fix a finite nonempty set of dates \(W\), and suppose that at every
\(t\in W\) the function \(f_t\) is twice continuously differentiable
near zero with \(f_t(0)=0\), and that
\[
\min_{t\in W}\dot f_t>0,
\qquad
\max_{t\in W}\ddot f_t<0
\]
Then there is \(\pi_\ast>0\) such that, for every \(t\in W\) and every
\(0<\pi\le\pi_\ast\),
\begin{equation}
f_t(\pi)>0>f_t(-\pi),
\qquad
|f_t(-\pi)|>|f_t(\pi)|
\label{eq:asymmetry_inequality_lemma}
\end{equation}
\end{lemma}

\begin{proof}
Write \(f_t(\pi)=\dot f_t\pi+\tfrac12\ddot f_t\pi^2+\varepsilon_t(\pi)\) with
\(\varepsilon_t(\pi)=o(\pi^2)\), and set
\(c_{\dot f}:=\min_W\dot f_t>0\) and \(c_{\ddot f}:=\min_W(-\ddot f_t)>0\). Since \(W\) is
finite there is a single \(\pi_\ast>0\) such that, for every \(t\in W\) and
\(0<\pi\le\pi_\ast\), \(|\varepsilon_t(\pm\pi)|\le\tfrac14c_{\ddot f}\pi^2\) and
\(\tfrac12|\ddot f_t|\pi+\tfrac14c_{\ddot f}\pi\le\tfrac12c_{\dot f}\). Then
\(f_t(\pi)\ge \dot f_t\pi-\tfrac12|\ddot f_t|\pi^2-\tfrac14c_{\ddot f}\pi^2\ge\tfrac12c_{\dot f}\pi>0\)
and
\(f_t(-\pi)\le-\dot f_t\pi+\tfrac12|\ddot f_t|\pi^2+\tfrac14c_{\ddot f}\pi^2\le-\tfrac12c_{\dot f}\pi<0\),
which are the two signs, and
\[
f_t(\pi)+f_t(-\pi)
=
\ddot f_t\pi^2+\varepsilon_t(\pi)+\varepsilon_t(-\pi)
\le
-c_{\ddot f}\pi^2+\tfrac12c_{\ddot f}\pi^2
<0
\]
so \(|f_t(-\pi)|-|f_t(\pi)|=-f_t(-\pi)-f_t(\pi)>0\), which is
\eqref{eq:asymmetry_inequality_lemma}.
\end{proof}

Write
\begin{equation}
\Delta\mathcal Q_t(b)
:=
\mathcal Q_t(b)-\mathcal Q^\ast(b)
=
\sum_{i\in N}
\omega_i(b)\Delta\log q_{i,t}
\label{eq:Delta_Qb_def}
\end{equation}
for the post-shock deviation of the position-tilted log-output index
of Definition~\ref{def:position_tilted_output_index} from its
stationary value, the second equality by \eqref{eq:w_b_def}.

Write \(r_{i,t}:=c_{i,t}/q_{i,t}\) for the share of firm \(i\)'s
date-\(t\) output bought by the household, and
\begin{equation}
\Delta\log Y_t=\sum_{i\in\mathcal R}\gamma_i\Delta\log c_{i,t},
\qquad
\Theta_t:=\sum_{i\in\mathcal R}\gamma_i\Delta\log r_{i,t}
\label{eq:gdp_change_def}
\end{equation}
for the log change in GDP and the consumption--output wedge, the
sums running over the retail tier
\(\mathcal R=\operatorname{supp}\boldsymbol\gamma\).

\begin{lemma}[GDP accounting and the consumption--output wedge]
\label{lem:gdp_labor_income_anchor}
The exact identity
\begin{equation}
\Delta\log Y_t
=
\sum_{i\in\mathcal R}\gamma_i\Delta\log q_{i,t}+\Theta_t
\label{eq:gdp_gamma_form}
\end{equation}
holds at every date, under no assumption. Suppose that
Assumptions~\ref{assump:diffuse_incidence},
\ref{assump:monotone_decay_spectrum} and~\ref{assump:retail_outflow_cap}
of the paper hold, with \(\kappa\) the contraction rate of
Proposition~\ref{prop:stability_stationary}. Then for \(|\pi|\le\overline\pi\),
\begin{equation}
\bigl|\Theta_t\bigr|
\;\le\;
\frac{4C_S}{1-\overline\beta}\,(1-\epsilon_h)\,\kappa^t\,|\pi|\,(1-\theta),
\qquad t\ge0
\label{eq:wedge_bound}
\end{equation}
uniformly in the cross-section size. The balance law being affine in
\(\pi\), the even part of the consumption--output wedge is exact,
\begin{equation}
\tfrac12\bigl[\Theta_t(\pi)+\Theta_t(-\pi)\bigr]
=
\tfrac12\sum_{i\in\mathcal R}\gamma_i\log\frac{1-\pi^2}{1-(1+\dot{\widehat m}^\perp_{i,t+1})^2\pi^2}
=
\pi^2\sum_{i\in\mathcal R}\gamma_i\bigl(\dot{\widehat m}^\perp_{i,t+1}+\tfrac12(\dot{\widehat m}^\perp_{i,t+1})^2\bigr)+O(\pi^4)
\label{eq:wedge_even_part}
\end{equation}
where \(\dot{\widehat m}^\perp_{i,t+1}\), independent of \(\pi\), is the first-order
proportional misalignment of the retailer's next-period balance,
\(m_{i,t+1}(\pi)=m_i^\ast\{1+\pi(1+\dot{\widehat m}^\perp_{i,t+1})\}\). Consequently,
\begin{equation}
\dot\Theta_t=-\sum_{i\in\mathcal R}\gamma_i\dot{\widehat m}^\perp_{i,t+1},
\qquad
\ddot\Theta_t+2\dot\Theta_t=\sum_{i\in\mathcal R}\gamma_i(\dot{\widehat m}^\perp_{i,t+1})^2\;\ge\;0
\label{eq:wedge_net_curvature}
\end{equation}
\end{lemma}

\begin{proof}
For \(i\in\mathcal R\), household demand \(c_{i,t}=\gamma_i\wM/p_{i,t}\) and
market clearing \(p_{i,t}q_{i,t}=d_{i,t}=m_{i,t+1}\) give
\[
r_{i,t}=\frac{\gamma_i(1+\pi)w^\ast}{m_{i,t+1}},
\qquad
\frac{r_{i,t}}{r_i}=\frac{(1+\pi)m_i^\ast}{m_{i,t+1}}
\]
with \(r_i=\gamma_iw^\ast/m_i^\ast\) the stationary consumption share. Taking
logarithms of \(c_{i,t}=r_{i,t}q_{i,t}\) and summing against
\(\boldsymbol\gamma\) gives \eqref{eq:gdp_gamma_form}. For the bound \eqref{eq:wedge_bound} on the consumption--output wedge, let
\(\widehat m^\perp_{i,t+1}:=\bigl(m_{i,t+1}-(1+\pi)m_i^\ast\bigr)/m_i^\ast\) be the
proportional misalignment of the next-period balance. Subtracting the
post-shock stationary balance equation from the balance law
\eqref{eq:firm_balance_law_pre}, in which the wage bill is
\((1+\pi)\beta_lm_l^\ast\),
\[
\widehat m^\perp_{i,t+1}
=
\sum_{l\in N}\frac{a_{il}m_l^\ast}{m_i^\ast}\,
\frac{m_{l,t}-(1+\pi)m_l^\ast}{m_l^\ast}
=
\sum_{l\in N}\frac{\Omega_{il}}{1-\beta_l}\,\widehat m^\perp_{l,t}
\]
with \(\Omega_{il}=(1-\beta_l)a_{il}m_l^\ast/m_i^\ast\) the output shares
\eqref{eq:output_share_matrix}. For a retailer
\(\sum_l\Omega_{il}=1-r_i\le1-\epsilon_h\), and \(\widehat m^\perp_{l,t}\) is the balance
coordinate of the proportional state, so
\[
|\widehat m^\perp_{i,t+1}|
\;\le\;
\frac{1-\epsilon_h}{1-\overline\beta}\,\bigl\|\widehat{\mathbf S}_t(\pi)\bigr\|_\infty
\;\le\;
\frac{1-\epsilon_h}{1-\overline\beta}\,C_S\kappa^t|\pi|(1-\theta),
\qquad i\in\mathcal R
\]
by \eqref{eq:post_shock_geometric_decay}. Since
\(r_{i,t}/r_i=1/(1+\widehat m^\perp_{i,t+1}/(1+\pi))\), the conditions
\(|\widehat m^\perp_{i,t+1}|/(1+\pi)\le\tfrac12\) and \(|\pi|\le\tfrac12\),
which the bound \(\overline\pi\) imposes, keep every consumption share
within a factor of two of its stationary value. In
particular \(r_{i,t}\ge\epsilon_h/2>0\), so that the logarithm is defined,
and \(|\Delta\log r_{i,t}|=|\log(1+\widehat m^\perp_{i,t+1}/(1+\pi))|\le4|\widehat m^\perp_{i,t+1}|\).
Multiplying by \(\gamma_i\) and summing over \(\mathcal R\), where the
shares sum to one, gives \eqref{eq:wedge_bound}, with a constant that
depends only on \(\epsilon_h\), \(\overline\beta\) and \(C_S\). For the even part, the
balance law is affine in \(\pi\) with
\(m_{i,t+1}(\pi)=m_i^\ast\{1+\pi(1+\dot{\widehat m}^\perp_{i,t+1})\}\), so
\(\Delta\log r_{i,t}=\log(1+\pi)-\log\bigl(1+\pi(1+\dot{\widehat m}^\perp_{i,t+1})\bigr)\), and
averaging the values at \(\pi\) and \(-\pi\) gives
\eqref{eq:wedge_even_part}, and expanding the logarithm gives its leading
term. Differentiating \(\Delta\log r_{i,t}\) once and twice at \(\pi=0\)
gives \(-\dot{\widehat m}^\perp_{i,t+1}\) and \(2\dot{\widehat m}^\perp_{i,t+1}+(\dot{\widehat m}^\perp_{i,t+1})^2\), and summing against
\(\boldsymbol\gamma\), which has unit mass on \(\mathcal R\), gives
\eqref{eq:wedge_net_curvature}.
\end{proof}

Over the first reporting interval, the dates from the shock to the
shortest lead time \(\underline k\), the first-order response of GDP
takes a simple form.
Every lead time is at least \(\underline k\), so no dated input moves
before the orders placed at impact mature, and no order placed after
impact matures within the interval. Write
\(\dot f_+:=(\underline k+1)^{-1}\sum_{t\le\underline k}\dot f_t\) for the
first-order response of GDP over the first reporting interval.

\begin{lemma}[GDP over the first reporting interval at first order]
\label{lem:first_interval_gdp}
Suppose that Assumptions~\ref{assump:diffuse_incidence},
\ref{assump:monotone_decay_spectrum}, \ref{assump:retail_outflow_cap}
and~\ref{assump:impact_forcing_ordering} of the paper hold, with
\(C_S\) and \(\kappa\) the decay constants of
Proposition~\ref{prop:stability_stationary}, that \(\underline k\ge1\),
and that the shortest lead time is a retail lead-time date,
\(\underline k\in\mathcal W\). Then
\begin{equation}
\sum_{t=0}^{\underline k}\dot f_t
=\sum_{i\in\mathcal R}\gamma_i(1-\beta_i)\,s_i^{(\underline k)}\xi_i^{(\underline k)}
+\sum_{t=0}^{\underline k}\dot\Theta_t
\label{eq:first_interval_identity}
\end{equation}
And with the GDP margin
\begin{equation}
c_Y:=\frac{1}{\underline k+1}\Bigl[\bigl(c_\xi-(1-\epsilon_h)C_\xi\bigr)s_R(1-\overline\beta)
-(1-\epsilon_h)\,\frac{4C_S}{1-\overline\beta}\sum_{t=0}^{\underline k}\kappa^t\Bigr]
\label{eq:gdp_margin}
\end{equation}
the first-order response of GDP over the first reporting interval obeys
\(\dot f_+\ge c_Y(1-\theta)\) whenever \(c_Y>0\).
\end{lemma}

\begin{proof}
Every lead time is at least \(\underline k\), so no dated input of any
firm differs from its stationary value before \(\underline k\), and
\(D_\pi\log q_{i,t}=0\) for \(t<\underline k\). At \(\underline k\) only
the inputs ordered at impact from suppliers of lead time
\(\underline k\) have moved. By Proposition~\ref{prop:impact_reallocation}
of the paper the response of firm \(i\) is then its scaled relative
expenditure impulse,
\(D_\pi\log q_{i,\underline k}=s_i^{(\underline k)}\phi_i^{(\underline k)}=(1-\beta_i)s_i^{(\underline k)}\xi_i^{(\underline k)}\).
The identity \eqref{eq:gdp_gamma_form} of
Lemma~\ref{lem:gdp_labor_income_anchor}, differentiated at zero, gives
\(\dot f_t=\sum_{i\in\mathcal R}\gamma_iD_\pi\log q_{i,t}+\dot\Theta_t\)
at every date, and summing over the interval gives
\eqref{eq:first_interval_identity}.

For the bound, write the first term of
\eqref{eq:first_interval_identity} as
\(\sum_{i\in\mathcal R}r_im_i^\ast(1-\beta_i)s_i^{(\underline k)}\xi_i^{(\underline k)}/w^\ast\),
since \(\gamma_i=r_im_i^\ast/w^\ast\). By
Assumption~\ref{assump:impact_forcing_ordering} of the paper at
\(\underline k\in\mathcal W\) and the definition \eqref{eq:set_impulse} of
the paper,
\(\sum_{i\in\mathcal R}m_i^\ast(1-\beta_i)s_i^{(\underline k)}\xi_i^{(\underline k)}=e_{\mathcal R}^{(\underline k)}\xi_{\mathcal R}^{(\underline k)}\ge c_\xi(1-\theta)\,e_{\mathcal R}^{(\underline k)}\).
And \(0\le1-r_i\le1-\epsilon_h\) by
Assumption~\ref{assump:retail_outflow_cap} of the paper, while
\(|\xi_i^{(\underline k)}|\le C_\xi(1-\theta)\) by
\eqref{eq:forcing_upper_bound} of the paper. Hence
\[
\sum_{i\in\mathcal R}r_im_i^\ast(1-\beta_i)s_i^{(\underline k)}\xi_i^{(\underline k)}
=e_{\mathcal R}^{(\underline k)}\xi_{\mathcal R}^{(\underline k)}
-\sum_{i\in\mathcal R}(1-r_i)\,m_i^\ast(1-\beta_i)s_i^{(\underline k)}\xi_i^{(\underline k)}
\;\ge\;\bigl(c_\xi-(1-\epsilon_h)C_\xi\bigr)(1-\theta)\,e_{\mathcal R}^{(\underline k)}
\]
The shortest lead time is a retail lead-time date, so
\(e_{\mathcal R}^{(\underline k)}\ge s_R\sum_{i\in\mathcal R}m_i^\ast(1-\beta_i)\ge s_R(1-\overline\beta)\,w^\ast\),
for the sales of the retail tier cover the household's purchases,
\(\sum_{i\in\mathcal R}m_i^\ast\ge w^\ast\). This lower bound applies
because \(c_Y>0\) forces \(c_\xi>(1-\epsilon_h)C_\xi\). The first-order
consumption--output wedge obeys
\(|\dot\Theta_t|\le(1-\epsilon_h)\,4C_S\kappa^t(1-\theta)/(1-\overline\beta)\),
the bound \eqref{eq:wedge_bound} differentiated at zero. Summing over
\(t\le\underline k\) and dividing by \(\underline k+1\) gives
\(\dot f_+\ge c_Y(1-\theta)\).
\end{proof}

Write
\[
R_t^{(Y)}(\pi):=\sum_{i\in\mathcal R}\gamma_i(1-\beta_i)\,\mathcal F_i(\widehat x_{i,t}),
\qquad
\mathcal D_t^{(\boldsymbol\gamma)}(\pi)=\sum_{i\in\mathcal R}\gamma_i\,\mathcal D_{i,t}
\]
for the \emph{input-scale aggregate} and the loss from production
inefficiency of the retailers at the household shares
\(\boldsymbol\omega=\boldsymbol\gamma\). Write
\(B_t:=\sum_{i\in\mathcal R}\gamma_i(1-\beta_i)\,\mathcal F_i(\ddot{\widehat x}_{i,t})\),
\(C_t:=\sum_{i\in\mathcal R}\gamma_i(1-\beta_i)\,\{\mathcal F_i(\dot{\widehat x}_{i,t})\}^2\)
and
\(V_t:=\sum_{i\in\mathcal R}\gamma_i(1-\beta_i)\,\langle\dot{\widehat x}_{i,t},\dot{\widehat x}_{i,t}\rangle_i\)
for the curvature of the input-scale aggregate, the curvature of the
retailers' log output in the scale of their bundles, and the retailers'
loss from production inefficiency per unit of squared shock.

\begin{lemma}[Second-order expansion of the GDP response]
\label{lem:gdp_expansion}
Let \(W\) be a finite set of dates. At every date the GDP response
decomposes as
\begin{equation}
\Delta\log Y_t
=
R_t^{(Y)}(\pi)
-\tfrac12\sum_{i\in\mathcal R}\gamma_i(1-\beta_i)\{\mathcal F_i(\widehat x_{i,t})\}^2
-\mathcal D_t^{(\boldsymbol\gamma)}(\pi)
+\Theta_t(\pi)
+o\!\left(\max_{i\in\mathcal R}\|\Delta\mathbf x_{i,t}\|^2\right)
\label{eq:gdp_theorem_decomposition}
\end{equation}
The expansion \eqref{eq:gdp_expansion_coefficients} of the paper holds
uniformly over \(t\in W\), with the second derivative
\begin{equation}
\ddot f_t=B_t-C_t-V_t+\ddot\Theta_t
\label{eq:gdp_second_derivative}
\end{equation}
and with \(C_t\ge0\), \(V_t\ge0\) and
\(\mathcal D_t^{(\boldsymbol\gamma)}(\pi)=\tfrac12\pi^2V_t+o(\pi^2)\). If \(A_t\ge c_A>0\)
on \(W\), the input-scale aggregate \(R_t^{(Y)}(\pi)\) has the sign of
\(\pi\) for every sufficiently small \(\pi\neq0\), uniformly on \(W\).
\end{lemma}

\begin{proof}
Fix \(t\in W\) and \(i\in\mathcal R\). By
Lemma~\ref{lem:propagation_linearization} the proportional deviation of the dated
inputs \(\widehat x_{i,t}(\pi)\) is real-analytic in \(\pi\) near
zero, so
\(\widehat x_{i,t}(\pi)=\pi\dot{\widehat x}_{i,t}+\tfrac12\pi^2\ddot{\widehat x}_{i,t}+o(\pi^2)\).
Since \(\mathcal F_i\) is linear and \(\langle\cdot,\cdot\rangle_i\) is
bilinear,
\begin{gather*}
\mathcal F_i(\widehat x_{i,t})=\pi\mathcal F_i(\dot{\widehat x}_{i,t})+\tfrac12\pi^2\mathcal F_i(\ddot{\widehat x}_{i,t})+o(\pi^2),
\qquad
\{\mathcal F_i(\widehat x_{i,t})\}^2=\pi^2\{\mathcal F_i(\dot{\widehat x}_{i,t})\}^2+o(\pi^2),\\
\langle\widehat x_{i,t},\widehat x_{i,t}\rangle_i=\pi^2\langle\dot{\widehat x}_{i,t},\dot{\widehat x}_{i,t}\rangle_i+o(\pi^2)
\end{gather*}
Substituting into the expansion \eqref{eq:log_q_scale_composition} of
log output,
\[
\Delta\log q_{i,t}
=
\pi(1-\beta_i)\mathcal F_i(\dot{\widehat x}_{i,t})
+\tfrac12\pi^2(1-\beta_i)\Bigl[\mathcal F_i(\ddot{\widehat x}_{i,t})-\{\mathcal F_i(\dot{\widehat x}_{i,t})\}^2-\langle\dot{\widehat x}_{i,t},\dot{\widehat x}_{i,t}\rangle_i\Bigr]
+o(\pi^2)
\]
and \(\mathcal D_{i,t}(\pi)=\tfrac12\pi^2(1-\beta_i)\langle\dot{\widehat x}_{i,t},\dot{\widehat x}_{i,t}\rangle_i+o(\pi^2)\)
by \eqref{eq:firm_inefficiency_def}. Multiplying by \(\gamma_i\) and
summing over \(\mathcal R\) gives, with the definitions of \(B_t\),
\(C_t\) and \(V_t\) above and of \(A_t\) in
Section~\ref{subsec:asymmetric_nonneutrality} of the paper,
\(\sum_{i\in\mathcal R}\gamma_i\Delta\log q_{i,t}=\pi A_t+\tfrac12\pi^2(B_t-C_t-V_t)+o(\pi^2)\)
and \(\mathcal D_t^{(\boldsymbol\gamma)}(\pi)=\sum_{i\in\mathcal R}\gamma_i\mathcal D_{i,t}(\pi)=\tfrac12\pi^2V_t+o(\pi^2)\).
The wedge
\(\Theta_t(\pi)=\sum_{i\in\mathcal R}\gamma_i[\log(1+\pi)-\log(1+\pi(1+\dot{\widehat m}^\perp_{i,t+1}))]\)
of Lemma~\ref{lem:gdp_labor_income_anchor} is real-analytic near zero,
so \(\Theta_t(\pi)=\dot\Theta_t\pi+\tfrac12\ddot\Theta_t\pi^2+o(\pi^2)\).
Adding the two by the exact identity \eqref{eq:gdp_gamma_form} gives
\eqref{eq:gdp_theorem_decomposition} and, collecting the coefficients
of \(\pi\) and \(\pi^2\), \eqref{eq:gdp_expansion_coefficients} of the
paper with \(\dot f_t=A_t+\dot\Theta_t\), and \eqref{eq:gdp_second_derivative}. Each
remainder is \(o(\pi^2)\) at a fixed date and \(W\) is finite, so the
maximum over \(W\) is \(o(\pi^2)\). \(C_t\) is a sum of squares with
non-negative weights, and \(V_t\) is a sum with non-negative weights of
values of the form \(\langle\cdot,\cdot\rangle_i\), which is positive
semidefinite because \(1-\rho>0\) and \(1-\rho_c>0\), so both are
non-negative. Finally
\(R_t^{(Y)}(\pi)=\sum_{i\in\mathcal R}\gamma_i(1-\beta_i)\mathcal F_i(\widehat x_{i,t})=\pi A_t+O(\pi^2)\)
uniformly on \(W\), so if \(A_t\ge c_A>0\) on \(W\) there is
\(\pi_1>0\) such that \(R_t^{(Y)}(\pi)/\pi\ge c_A/2>0\) for
\(0<|\pi|\le\pi_1\) and every \(t\in W\).
\end{proof}

\begin{lemma}[The normalized shock and the curvature of the response]
\label{lem:concave_scaling}
For every date \(t\) there is a map \(\mathcal Y_t\), real-analytic near
\(\mathbf 0\) with \(\mathcal Y_t(\mathbf 0)=0\) and independent of \(\pi\), such
that \eqref{eq:concave_scaling} holds, and the same form holds
for the input-scale aggregate, with
\(R_t^{(Y)}(\pi)=\mathcal Y^{R}_t(\tilde\pi\,h_\theta)\) and \(B_t=-2A_t+D^2\mathcal Y^{R}_t(\mathbf 0)[h_\theta,h_\theta]\). Suppose that
Assumption~\ref{assump:diffuse_incidence} holds and, on a finite set of
dates \(W\), \(\|D^2\mathcal Y_t(\mathbf 0)\|_m\le K\) and \(\dot f_t\ge c\,(1-\theta)\)
with \(c>0\). Here
\(\|D^2\mathcal Y_t(\mathbf 0)\|_m:=\sup\{|D^2\mathcal Y_t(\mathbf 0)[\mathbf v,\mathbf v']|:\mathbf v,\mathbf v'\in\mathcal Z_{\mathbf A},\ \|\mathbf v\|_m\le1,\ \|\mathbf v'\|_m\le1\}\)
is the Hessian norm on the conservation subspace
\(\mathcal Z_{\mathbf A}\) of the paper, in which the misalignment
\(h_\theta\) lies, taken in the proportional coordinates. Then
\begin{equation}
\ddot f_t\;\le\;-c\,(1-\theta)\;<\;0
\qquad\text{on }W\text{ whenever }1-\theta\le\frac{c}{K C_\zeta^2}
\label{eq:concave_scaling_bound}
\end{equation}
and the GDP response is, to leading order in the slack, the first-order
response applied to the normalized shock,
\begin{equation}
f_t(\pi)\;=\;\dot f_t\,\frac{\pi}{1+\pi}\;+\;O\bigl((1-\theta)^2\pi^2\bigr)
\label{eq:gdp_deflator_reading}
\end{equation}
locally in \(\pi\).
\end{lemma}

\begin{proof}
Consider the within-period map of Appendix~\ref{app:within_period} of
the paper at a wage \(w>0\). It sends a start-of-period state
\((\mathbf m,\boldsymbol{\mathcal X})\) to expenditures
\(e_i=m_i-wl_i^\ast\), demands \(d_j=\sum_ia_{ji}e_i+\gamma_jw\) and
outputs \(q_j\) produced from the inherited orders and the current
zero-lead-time orders. Prices are \(p_j=d_j/q_j\), orders are
\(x_{ji}=a_{ji}e_i/p_j\), household consumption is
\(c_i=\gamma_iw/p_i\), the next balances are \(\mathbf m'=\mathbf d\),
and the pipeline \(\boldsymbol{\mathcal X}'\) is the pipeline aged by
one period. Let \(\varkappa>0\). At
\((\varkappa\mathbf m,\varkappa w,\boldsymbol{\mathcal X})\) the
expenditures and demands are \(\varkappa e_i\) and \(\varkappa d_j\). The
within-period fixed point of Lemma~\ref{lem:propagation_linearization}
is unique, and the quantities \((q_j,\varkappa p_j,x_{ji},c_i)\) solve
it, since \(\varkappa p_j=\varkappa d_j/q_j\),
\(x_{ji}=a_{ji}\varkappa e_i/(\varkappa p_j)\) and
\(c_i=\gamma_i\varkappa w/(\varkappa p_i)\). Hence the outputs, orders,
consumption and next pipeline are unchanged and the next balances are
\(\varkappa\mathbf m'\). By induction over dates, the real path from
\((\varkappa\mathbf m_0,\varkappa w,\boldsymbol{\mathcal X}_0)\) equals the
real path from \((\mathbf m_0,w,\boldsymbol{\mathcal X}_0)\).

The post-shock economy has wage \(\wM=(1+\pi)w^\ast\) and impact
balances
\(\mathbf m_0(\pi)=\mathbf m^\ast+\pi\overline M\boldsymbol\zeta=(1+\pi)\bigl(\mathbf m^\ast+\tilde\pi\,h_\theta\bigr)\),
with \(\tilde\pi:=\pi/(1+\pi)\) and \(h_\theta:=\overline M\boldsymbol\zeta-\mathbf m^\ast\),
and the stationary pipeline. Taking \(\varkappa=(1+\pi)^{-1}\), every
real quantity along the post-shock path, and in particular
\(\Delta\log Y_t(\pi)=\sum_{i\in\mathcal R}\gamma_i\Delta\log c_{i,t}\),
equals the corresponding quantity of the economy at the stationary
wage \(w^\ast\) started from
\((\mathbf m^\ast+\tilde\pi\,h_\theta,\boldsymbol{\mathcal X}^\ast)\).
For \(\mathbf v\in\mathbb R^n\) near \(\mathbf 0\) let
\(\mathcal Y_t(\mathbf v)\) be \(\Delta\log Y_t\) at date \(t\) in the
economy at wage \(w^\ast\) started from
\((\mathbf m^\ast+\mathbf v,\boldsymbol{\mathcal X}^\ast)\). The
initial state is affine in \(\mathbf v\) and equals \(\mathbf S^\ast\)
at \(\mathbf v=\mathbf 0\), so the argument of the proof of
Lemma~\ref{lem:propagation_linearization}, which uses only these two
properties, makes \(\mathcal Y_t\) real-analytic near \(\mathbf 0\)
with \(\mathcal Y_t(\mathbf 0)=0\). The map \(\mathcal Y_t\) does not
involve \(\pi\), and \(f_t(\pi)=\mathcal Y_t(\tilde\pi\,h_\theta)\).
Derivatives of \(\mathcal Y_t\) are measured in the proportional
coordinates of the paper, through the norm
\(\|\mathbf v\|_m=\|\mathbf D_m^{-1}\mathbf v\|_\infty\) and the
Hessian norm \(\|D^2\mathcal Y_t(\mathbf 0)\|_m\) built on it, and by
\eqref{eq:impact_misalignment_constant} of the paper
\(\|h_\theta\|_m=\max_{i\in N}|\overline M\zeta_i/m_i^\ast-1|\le C_\zeta(1-\theta)\).
The input-scale aggregate is a real quantity of the same economy, so
the same construction gives \(\mathcal Y_t^R\) with
\(R_t^{(Y)}(\pi)=\mathcal Y_t^R(\tilde\pi\,h_\theta)\).

With \(\tilde\pi(\pi)=\pi/(1+\pi)\), \(\tilde\pi(0)=0\),
\(\tilde\pi'(0)=1\) and \(\tilde\pi''(0)=-2\), the chain rule gives
\(f_t'(\pi)=D\mathcal Y_t(\tilde\pi h_\theta)[h_\theta]\,\tilde\pi'(\pi)\) and
\[
f_t''(\pi)=D^2\mathcal Y_t(\tilde\pi h_\theta)[h_\theta,h_\theta]\,\tilde\pi'(\pi)^2+D\mathcal Y_t(\tilde\pi h_\theta)[h_\theta]\,\tilde\pi''(\pi)
\]
which at \(\pi=0\) is \eqref{eq:concave_scaling}, with
\(\dot f_t=D\mathcal Y_t(\mathbf 0)[h_\theta]\). The same computation for
\(\mathcal Y_t^R\) gives \(B_t=-2A_t+D^2\mathcal Y_t^R(\mathbf 0)[h_\theta,h_\theta]\).

The misalignment \(h_\theta=\overline M\boldsymbol\zeta-\mathbf m^\ast\)
lies in \(\mathcal Z_{\mathbf A}\), since the incidence shares and the
stationary balances both sum to \(\overline M\). With
\(\|D^2\mathcal Y_t(\mathbf 0)\|_m\le K\) on \(W\), the bilinear bound gives
\(|D^2\mathcal Y_t(\mathbf 0)[h_\theta,h_\theta]|\le K\|h_\theta\|_m^2\le KC_\zeta^2(1-\theta)^2\),
and with \(\dot f_t\ge c(1-\theta)\),
\[
\ddot f_t\le-2c(1-\theta)+KC_\zeta^2(1-\theta)^2\le-c(1-\theta)
\qquad\text{when }KC_\zeta^2(1-\theta)\le c
\]
which is \eqref{eq:concave_scaling_bound}. Finally, Taylor's theorem
for the real-analytic \(\mathcal Y_t\) along the segment
\(\{\tilde\pi\,h_\theta\}\) gives
\(f_t(\pi)=\dot f_t\tilde\pi+\tfrac12D^2\mathcal Y_t(\mathbf 0)[h_\theta,h_\theta]\tilde\pi^2+\widetilde\varepsilon_t(\tilde\pi)\)
with \(|\widetilde\varepsilon_t(\tilde\pi)|\le K_3\|h_\theta\|_m^3|\tilde\pi|^3\) for \(|\tilde\pi|\) small,
where \(K_3\) bounds the third derivative of \(\mathcal Y_t\) near
\(\mathbf 0\) on \(\mathcal Z_{\mathbf A}\) in the same norm. For
\(\pi\ge-\tfrac12\), \(|\tilde\pi|\le2|\pi|\), so the second term is
at most \(2KC_\zeta^2(1-\theta)^2\pi^2\) in absolute value and the
remainder at most
\(8K_3C_\zeta^3(1-\theta)^3|\pi|^3\le8K_3C_\zeta^3(1-\theta)^2\pi^2\)
on \(|\pi|\le1\), which is \eqref{eq:gdp_deflator_reading}.
\end{proof}

The resource identity \eqref{eq:resource_identity} of the paper values
quantities at stationary prices. The excess value of the goods in
transit at date \(t\) is
\begin{equation}
P_t
:=
\sum_{i\in N}\sum_{j\in\mathcal S_i}p_j^\ast\sum_{u=t-k_{ji}}^{t-1}\bigl(x_{ji,u}-x_{ji}^\ast\bigr)
\label{eq:pipeline_value}
\end{equation}
the value of the orders outstanding at date \(t\) in excess of the
stationary pipeline. The value of the inputs that the firms used in
excess of their stationary bundles, less the value of the output they
made with them, is
\begin{equation}
E_t
:=
\sum_{i\in N}E_{i,t},
\qquad
E_{i,t}
:=
\sum_{j\in\mathcal S_i}p_j^\ast\bigl(x_{ji,t-k_{ji}}-x_{ji}^\ast\bigr)-p_i^\ast\bigl(q_{i,t}-q_i^\ast\bigr)
\label{eq:production_loss}
\end{equation}
with \(E_{i,t}\) the contribution of firm \(i\). And the consumption
loss, the value of the change in consumption at stationary prices less
its value in the log index, is
\begin{equation}
H_t
:=
w^\ast\sum_{i\in\mathcal R}\gamma_i\Bigl[\frac{c_{i,t}}{c_i^\ast}-1-\log\frac{c_{i,t}}{c_i^\ast}\Bigr]
\label{eq:consumption_loss}
\end{equation}
The deadweight loss of the paper is \(\Lambda_t:=(E_t+H_t)/w^\ast\).
Lemma~\ref{lem:resource_accounting_date} splits \(E_t\) into two
nonnegative parts, \(E_t=E_t^{\mathrm{mix}}+E_t^{\mathrm{scale}}\), each
the sum over firms of the corresponding term of
\eqref{eq:production_loss_split}. Both nests of the technology have
curvatures below one, so the composite of a bundle whose inputs have
moved by different proportions from their stationary levels falls short
of the value of the bundle at stationary prices. The first part,
\(E_t^{\mathrm{mix}}\), is the output lost in this way, and
\(E_t^{\mathrm{mix}}/w^\ast\) is the loss from production inefficiency
at date \(t\). Output is concave in the scale of the bundle at fixed
labor. The second part, \(E_t^{\mathrm{scale}}\), is therefore the
output lost because a firm's bundle as a whole is larger or smaller than
at the stationary equilibrium, relative to the workers it employs. With
the consumption loss, it makes up the loss from misallocation
\((E_t^{\mathrm{scale}}+H_t)/w^\ast\).

\begin{lemma}[Resource accounting at a date]
\label{lem:resource_accounting_date}
Suppose that Assumptions~\ref{assump:diffuse_incidence}
and~\ref{assump:monotone_decay_spectrum} of the paper hold and that
\(|\pi|\le\overline\pi\). Then at every date \(t\ge0\) the resource
identity \eqref{eq:resource_identity} of the paper holds, and the
consumption loss \eqref{eq:consumption_loss} is nonnegative, and
positive when \(c_{i,t}\ne c_i^\ast\) for some \(i\in\mathcal R\).
The contribution of each firm splits as
\(E_{i,t}=E_{i,t}^{\mathrm{mix}}+E_{i,t}^{\mathrm{scale}}\), with
\begin{equation}
\begin{aligned}
E_{i,t}^{\mathrm{mix}}
&:=
m_i^\ast\Bigl[\bigl(1+\mathcal F_i(\widehat x_{i,t})\bigr)^{1-\beta_i}-\Bigl(\frac{X_{i,t}}{X_i^\ast}\Bigr)^{1-\beta_i}\Bigr],\\
E_{i,t}^{\mathrm{scale}}
&:=
m_i^\ast\Bigl[(1-\beta_i)\,\mathcal F_i(\widehat x_{i,t})+1-\bigl(1+\mathcal F_i(\widehat x_{i,t})\bigr)^{1-\beta_i}\Bigr]
\end{aligned}
\label{eq:production_loss_split}
\end{equation}
where \(1+\mathcal F_i(\widehat x_{i,t})\) is the value of the dated input
bundle of firm \(i\) at stationary prices relative to the stationary
bundle and \(X_{i,t}/X_i^\ast\) is the change in its intermediate
composite. Both \(E_{i,t}^{\mathrm{mix}}\) and \(E_{i,t}^{\mathrm{scale}}\) are
nonnegative. The first vanishes exactly when the dated inputs of firm
\(i\) all move in the same proportion, and the second exactly when
\(\mathcal F_i(\widehat x_{i,t})=0\). At a fixed date, as \(\pi\to0\),
\begin{equation}
\begin{gathered}
E^{\mathrm{mix}}_{i,t}=m_i^\ast\,\mathcal D_{i,t}+o(\pi^2),
\qquad
E^{\mathrm{scale}}_{i,t}=\tfrac12\beta_i(1-\beta_i)\,m_i^\ast\{\mathcal F_i(\widehat x_{i,t})\}^2+o(\pi^2),\\
H_t=\tfrac12w^\ast\sum_{i\in\mathcal R}\gamma_i\,\widehat c_{i,t}^{\,2}+o(\pi^2)
\end{gathered}
\label{eq:deadweight_second_order}
\end{equation}
The drawdown \(\Pi_t\) of the goods in transit, the deadweight loss
\(\Lambda_t\) and the excess value \(P_t\) of the goods in transit depend
on the shock only through the normalized shock. There are maps
\(\mathcal Y^{\Pi}_t\), \(\mathcal Y^{\Lambda}_t\) and \(\mathcal Y^{P}_t\),
real-analytic near \(\mathbf 0\), vanishing there and not involving
\(\pi\), such that
\(\Pi_t(\pi)=\mathcal Y^{\Pi}_t(\tilde\pi\,h_\theta)\),
\(\Lambda_t(\pi)=\mathcal Y^{\Lambda}_t(\tilde\pi\,h_\theta)\) and
\(P_t(\pi)=\mathcal Y^{P}_t(\tilde\pi\,h_\theta)\). They satisfy
\(\mathcal Y^{\Pi}_t=(\mathcal Y^{P}_t-\mathcal Y^{P}_{t+1})/w^\ast\),
\(\mathcal Y^{P}_0\equiv0\), \(D\mathcal Y^{\Pi}_t(\mathbf 0)h_\theta=\dot f_t\)
and \(D^2\mathcal Y^{\Lambda}_t(\mathbf 0)[h_\theta,h_\theta]\ge0\). The first derivative of the excess value of the goods in transit at \(\pi=0\) obeys
\(\dot P_{T+1}=D\mathcal Y^{P}_{T+1}(\mathbf 0)h_\theta=-w^\ast\sum_{t\le T}\dot f_t\)
for every \(T\ge0\). The even
part of the GDP response obeys \eqref{eq:even_part} of the paper, and the
residual curvature of \eqref{eq:concave_scaling} of the paper is
\begin{equation}
D^2\mathcal Y_t(\mathbf 0)[h_\theta,h_\theta]
=
D^2\mathcal Y^{\Pi}_t(\mathbf 0)[h_\theta,h_\theta]-D^2\mathcal Y^{\Lambda}_t(\mathbf 0)[h_\theta,h_\theta]
\label{eq:residual_curvature_split}
\end{equation}
\end{lemma}

\begin{proof}
By the order rule
\eqref{eq:dynamic_real_alloc}, the nominal demand
\eqref{eq:dyn_nominal_demand} and the clearing rule
\eqref{eq:dyn_price_update} of the paper, the output of good \(j\) at
date \(t\) is the sum of the orders placed on it and of household
consumption \(c_{j,t}=\gamma_j\wM/p_{j,t}\),
\[
q_{j,t}
=
\sum_{i\in N}x_{ji,t}+c_{j,t}
\]
And the same holds at the stationary equilibrium. Subtracting the two,
multiplying by \(p_j^\ast\) and summing over goods,
\begin{align*}
\sum_{j\in N}p_j^\ast\bigl(c_{j,t}-c_j^\ast\bigr)
&=
\sum_{i\in N}p_i^\ast\bigl(q_{i,t}-q_i^\ast\bigr)
-\sum_{i\in N}\sum_{j\in\mathcal S_i}p_j^\ast\bigl(x_{ji,t}-x_{ji}^\ast\bigr)\\
&=
\sum_{i\in N}\sum_{j\in\mathcal S_i}p_j^\ast\bigl[(x_{ji,t-k_{ji}}-x_{ji}^\ast)-(x_{ji,t}-x_{ji}^\ast)\bigr]-E_t
\end{align*}
by the definition \eqref{eq:production_loss} of \(E_t\). The orders
outstanding at date \(t+1\) are those outstanding at date \(t\), less
the orders placed at the dates \(t-k_{ji}\), which mature at \(t\), plus
the orders placed at \(t\), so the bracketed sum is \(P_t-P_{t+1}\) by
\eqref{eq:pipeline_value}. A purchase with \(k_{ji}=0\) contributes to
neither. Since \(p_j^\ast c_j^\ast=\gamma_jw^\ast\) for
\(j\in\mathcal R\) and \(c_{j,t}=c_j^\ast=0\) otherwise,
\(\sum_jp_j^\ast(c_{j,t}-c_j^\ast)=w^\ast\sum_{i\in\mathcal R}\gamma_i(c_{i,t}/c_i^\ast-1)=w^\ast f_t(\pi)+H_t\)
by \eqref{eq:gdp_def_main} of the paper and \eqref{eq:consumption_loss}. Dividing by \(w^\ast\) gives \eqref{eq:resource_identity}. And
\(z-1-\log z\ge0\) for \(z>0\), with equality only at \(z=1\), so
\(H_t\ge0\), with \(H_t>0\) when \(c_{i,t}\ne c_i^\ast\) for some
\(i\in\mathcal R\).

For the split, fix \(i\) and \(t\), write
\(y_j:=x_{ji,t-k_{ji}}/x_{ji}^\ast=1+\widehat x_{ji,t}\) for the dated
input ratios of firm \(i\), in the notation of
Section~\ref{sec:inefficiency} of the paper, and
\(\overline y:=\sum_{j\in\mathcal S_i}a_{ji}y_j=1+\mathcal F_i(\widehat x_{i,t})\).
By the supplier shares \eqref{eq:supplier_share_def} and
\(m_i^\ast=p_i^\ast q_i^\ast\),
\(p_j^\ast x_{ji}^\ast=a_{ji}(1-\beta_i)m_i^\ast\), so
\(\sum_jp_j^\ast(x_{ji,t-k_{ji}}-x_{ji}^\ast)=(1-\beta_i)m_i^\ast(\overline y-1)\).
Labor stays at its stationary value
(Assumption~\ref{assump:short_run_nominal_rigidity} of the paper), so
\(q_{i,t}/q_i^\ast=(X_{i,t}/X_i^\ast)^{1-\beta_i}\) by \eqref{eq:top_cd},
and
\begin{equation}
E_{i,t}
=
m_i^\ast\Bigl[(1-\beta_i)(\overline y-1)+1-\bigl(X_{i,t}/X_i^\ast\bigr)^{1-\beta_i}\Bigr]
=
E^{\mathrm{mix}}_{i,t}+E^{\mathrm{scale}}_{i,t}
\label{eq:E_exact}
\end{equation}
The map \(z\mapsto z^{1-\beta_i}\) is strictly concave with
slope \(1-\beta_i\) at \(z=1\), so
\(\overline y^{\,1-\beta_i}\le1+(1-\beta_i)(\overline y-1)\), with equality exactly
at \(\overline y=1\), which is the claim for \(E^{\mathrm{scale}}_{i,t}\). By
the stationary shares \eqref{eq:share_quantity_form}, the two nests
normalized at the stationary equilibrium read
\[
\frac{X^{(k)}_{i,t}}{X^{(k)\ast}_i}
=
\Bigl(\sum_{j\in\mathcal S_i^{(k)}}\widetilde a^{(k)}_{ji}\,y_j^{\rho}\Bigr)^{1/\rho},
\qquad
\frac{X_{i,t}}{X_i^\ast}
=
\Bigl(\sum_{k\in\mathcal K_i}s_i^{(k)}\Bigl(\frac{X^{(k)}_{i,t}}{X^{(k)\ast}_i}\Bigr)^{\rho_c}\Bigr)^{1/\rho_c}
\]
power means with exponents \(\rho\) and \(\rho_c\) below one and
weights that sum to one. A power mean with exponent below one is
increasing in each argument and lies below the arithmetic mean, with
equality exactly when its arguments are equal. Hence
\(X^{(k)}_{i,t}/X^{(k)\ast}_i\le\overline y^{(k)}:=\sum_{j\in\mathcal S_i^{(k)}}\widetilde a^{(k)}_{ji}y_j\),
and
\[
\frac{X_{i,t}}{X_i^\ast}
\;\le\;
\Bigl(\sum_{k\in\mathcal K_i}s_i^{(k)}\bigl(\overline y^{(k)}\bigr)^{\rho_c}\Bigr)^{1/\rho_c}
\;\le\;
\sum_{k\in\mathcal K_i}s_i^{(k)}\overline y^{(k)}
\;=\;
\overline y
\]
the last equality by \(a_{ji}=s_i^{(k_{ji})}\widetilde a^{(k_{ji})}_{ji}\). So
\((X_{i,t}/X_i^\ast)^{1-\beta_i}\le\overline y^{\,1-\beta_i}\) and
\(E^{\mathrm{mix}}_{i,t}\ge0\), with equality exactly when all the
\(y_j\) are equal.

At a fixed date the dated input ratios are
real-analytic in \(\pi\) with \(y_j-1=O(\pi)\)
(Lemma~\ref{lem:propagation_linearization}). By the expansion
\eqref{eq:log_q_scale_composition} and
Definition~\ref{def:production_inefficiency} of the paper,
\((1-\beta_i)\log(X_{i,t}/X_i^\ast)=\Delta\log q_{i,t}=(1-\beta_i)\bigl((\overline y-1)-\tfrac12(\overline y-1)^2\bigr)-\mathcal D_{i,t}+o(\pi^2)\),
while \(\log\overline y=(\overline y-1)-\tfrac12(\overline y-1)^2+O(\pi^3)\). So
\((1-\beta_i)\log\bigl(\overline y X_i^\ast/X_{i,t}\bigr)=\mathcal D_{i,t}+o(\pi^2)\), and
\[
E^{\mathrm{mix}}_{i,t}
=
m_i^\ast\overline y^{\,1-\beta_i}\Bigl[1-\exp\Bigl(-(1-\beta_i)\log\frac{\overline y X_i^\ast}{X_{i,t}}\Bigr)\Bigr]
=
m_i^\ast\mathcal D_{i,t}+o(\pi^2)
\]
since \(\overline y^{\,1-\beta_i}=1+O(\pi)\) and \(\mathcal D_{i,t}=O(\pi^2)\).
Expanding
\(\overline y^{\,1-\beta_i}=1+(1-\beta_i)(\overline y-1)-\tfrac12\beta_i(1-\beta_i)(\overline y-1)^2+O(\pi^3)\)
gives the form of \(E^{\mathrm{scale}}_{i,t}\), and
\(z-1-\log z=\tfrac12(z-1)^2+O((z-1)^3)\) at \(z=1+\widehat c_{i,t}\)
gives that of \(H_t\).

Every real quantity of the post-shock path depends on \(\pi\) only
through \(\tilde\pi\,h_\theta\), and is real-analytic in the initial
balance vector near the stationary one, by the proof of
Lemma~\ref{lem:concave_scaling}. The
quantities \(P_t\), \(E_t\) and \(H_t\) are values at the fixed
stationary prices of real quantities, which gives the maps
\(\mathcal Y^{P}_t\), \(\mathcal Y^{\Pi}_t\) and \(\mathcal Y^{\Lambda}_t\).
By the definition of \(\Pi_t\),
\(\mathcal Y^{\Pi}_t=(\mathcal Y^{P}_t-\mathcal Y^{P}_{t+1})/w^\ast\), and
\(\mathcal Y^{P}_0\equiv0\) since every order outstanding at date \(0\)
was placed before the shock. The identity \eqref{eq:resource_identity}
and the split \eqref{eq:E_exact} hold along every path of the economy
at the stationary wage started from the stationary pipeline, so
\(\mathcal Y_t=\mathcal Y^{\Pi}_t-\mathcal Y^{\Lambda}_t\) near
\(\mathbf 0\), and \(\mathcal Y^{\Lambda}_t\ge0\) there with
\(\mathcal Y^{\Lambda}_t(\mathbf 0)=0\). The map
\(\mathcal Y^{\Lambda}_t\) is therefore smallest at \(\mathbf 0\), where
its derivative vanishes and its Hessian is positive semidefinite.
Differentiating \(\mathcal Y_t=\mathcal Y^{\Pi}_t-\mathcal Y^{\Lambda}_t\)
once in the direction \(h_\theta\) gives
\(D\mathcal Y^{\Pi}_t(\mathbf 0)h_\theta=D\mathcal Y_t(\mathbf 0)h_\theta=\dot f_t\),
and twice gives \eqref{eq:residual_curvature_split}. Summing
\(\mathcal Y^{\Pi}_t\) over \(t\le T\) gives
\(\sum_{t\le T}\mathcal Y^{\Pi}_t=-\mathcal Y^{P}_{T+1}/w^\ast\), and
differentiating in the direction \(h_\theta\) gives
\(\dot P_{T+1}=-w^\ast\sum_{t\le T}\dot f_t\). Averaging
\eqref{eq:resource_identity} at \(\pi\) and \(-\pi\) gives
\(-\overline f_t=\tfrac12[\Lambda_t(\pi)+\Lambda_t(-\pi)]-\tfrac12[\Pi_t(\pi)+\Pi_t(-\pi)]\)
exactly. With \(\tilde\pi(\pi)=\pi-\pi^2+O(\pi^3)\), Taylor's theorem along the
segment \(\{\tilde\pi\,h_\theta\}\) gives
\(\Pi_t(\pi)=\dot f_t\pi+\bigl(-\dot f_t+\tfrac12D^2\mathcal Y^{\Pi}_t(\mathbf 0)[h_\theta,h_\theta]\bigr)\pi^2+O(\pi^3)\),
so
\(\tfrac12[\Pi_t(\pi)+\Pi_t(-\pi)]=-\dot f_t\pi^2+\tfrac12\pi^2D^2\mathcal Y^{\Pi}_t(\mathbf 0)[h_\theta,h_\theta]+o(\pi^2)\),
which is \eqref{eq:even_part}.
\end{proof}

With \(\tilde\pi(\pi)=\pi/(1+\pi)\), for \(|\pi|<1\),
\begin{equation}
-\bigl[\tilde\pi(\pi)+\tilde\pi(-\pi)\bigr]
=\frac{2\pi^2}{1-\pi^2}
\label{eq:normalized_sizes}
\end{equation}
which for \(0\le\pi<1\) is \(|\tilde\pi(-\pi)|-\tilde\pi(\pi)\), the
excess of the normalized shock of the monetary contraction over that
of the monetary expansion in absolute value.
Write
\[
\Pi^\sharp_t(\pi):=\Pi_t(\pi)-\dot f_t\,\tilde\pi,
\qquad
P^\sharp_t(\pi):=P_t(\pi)-\dot P_t\,\tilde\pi
\]
for the parts of the drawdown of the goods in transit and of their excess value beyond first order in the normalized shock. For each of the maps \(\mathcal Y\) of
Lemmas~\ref{lem:concave_scaling} and~\ref{lem:resource_accounting_date},
write
\(\|D\mathcal Y(\mathbf 0)\|_m:=\sup\{|D\mathcal Y(\mathbf 0)\mathbf v|:\mathbf v\in\mathcal Z_{\mathbf A},\ \|\mathbf v\|_m\le1\}\)
for the norm of its derivative, and \(\|D^2\mathcal Y(\mathbf 0)\|_m\)
for the norm of its Hessian, as in Lemma~\ref{lem:concave_scaling}.

\begin{lemma}[The gap between the responses to a monetary expansion and to a monetary contraction]
\label{lem:gap_exact}
Suppose that Assumptions~\ref{assump:diffuse_incidence}
and~\ref{assump:monotone_decay_spectrum} of the paper hold and that
\(|\pi|\le\overline\pi\). Then at every date \(t\) and for every
\(T\ge0\), exactly,
\begin{equation}
-\bigl[f_t(\pi)+f_t(-\pi)\bigr]
=\bigl[\Lambda_t(\pi)+\Lambda_t(-\pi)\bigr]
+\frac{2\pi^2}{1-\pi^2}\,\dot f_t
-\bigl[\Pi^\sharp_t(\pi)+\Pi^\sharp_t(-\pi)\bigr]
\label{eq:gap_date_exact}
\end{equation}
and
\begin{equation}
-\sum_{t\le T}\bigl[f_t(\pi)+f_t(-\pi)\bigr]
=\sum_{t\le T}\bigl[\Lambda_t(\pi)+\Lambda_t(-\pi)\bigr]
+\frac{2\pi^2}{1-\pi^2}\sum_{t\le T}\dot f_t
+\frac{1}{w^\ast}\bigl[P^\sharp_{T+1}(\pi)+P^\sharp_{T+1}(-\pi)\bigr]
\label{eq:gap_interval_exact}
\end{equation}
At a fixed date and horizon, as \(\pi\to0\),
\begin{equation}
\begin{aligned}
\Lambda_t(\pi)+\Lambda_t(-\pi)&=\pi^2D^2\mathcal Y^{\Lambda}_t(\mathbf 0)[h_\theta,h_\theta]+O(\pi^4)\\
\Pi^\sharp_t(\pi)+\Pi^\sharp_t(-\pi)&=\pi^2D^2\mathcal Y^{\Pi}_t(\mathbf 0)[h_\theta,h_\theta]+O(\pi^4)\\
P^\sharp_{T+1}(\pi)+P^\sharp_{T+1}(-\pi)&=\pi^2D^2\mathcal Y^{P}_{T+1}(\mathbf 0)[h_\theta,h_\theta]+O(\pi^4)
\end{aligned}
\label{eq:gap_expansions}
\end{equation}
and the coefficients of the terms of \eqref{eq:gap_date_exact} and
\eqref{eq:gap_interval_exact} are bounded by the slack,
\begin{equation}
\begin{aligned}
|\dot f_t|&\le\|D\mathcal Y_t(\mathbf 0)\|_m\,C_\zeta(1-\theta)\\
0\;\le\;D^2\mathcal Y^{\Lambda}_t(\mathbf 0)[h_\theta,h_\theta]&\le\|D^2\mathcal Y^{\Lambda}_t(\mathbf 0)\|_m\,C_\zeta^2(1-\theta)^2\\
\bigl|D^2\mathcal Y^{\Pi}_t(\mathbf 0)[h_\theta,h_\theta]\bigr|&\le\|D^2\mathcal Y^{\Pi}_t(\mathbf 0)\|_m\,C_\zeta^2(1-\theta)^2\\
\bigl|D^2\mathcal Y^{P}_{T+1}(\mathbf 0)[h_\theta,h_\theta]\bigr|&\le\|D^2\mathcal Y^{P}_{T+1}(\mathbf 0)\|_m\,C_\zeta^2(1-\theta)^2
\end{aligned}
\label{eq:gap_slack_orders}
\end{equation}
\end{lemma}

\begin{proof}
By \eqref{eq:resource_identity} of the
paper at \(\pi\) and at \(-\pi\),
\(f_t(\pm\pi)=\Pi_t(\pm\pi)-\Lambda_t(\pm\pi)\). By the definition of
\(\Pi^\sharp_t\),
\(\Pi_t(\pi)+\Pi_t(-\pi)=\dot f_t\,[\tilde\pi(\pi)+\tilde\pi(-\pi)]+\Pi^\sharp_t(\pi)+\Pi^\sharp_t(-\pi)\),
and \eqref{eq:normalized_sizes} gives \eqref{eq:gap_date_exact}.
Summing \eqref{eq:resource_identity} over \(t\le T\) and using
\(P_0=0\) gives
\(\sum_{t\le T}f_t(\pi)=-P_{T+1}(\pi)/w^\ast-\sum_{t\le T}\Lambda_t(\pi)\),
as in the proof of Proposition~\ref{prop:resource_accounting} of the
paper. By Lemma~\ref{lem:resource_accounting_date},
\(P_{T+1}(\pi)=-w^\ast\bigl(\sum_{t\le T}\dot f_t\bigr)\tilde\pi+P^\sharp_{T+1}(\pi)\).
Adding the values at \(\pi\) and at \(-\pi\) and using
\eqref{eq:normalized_sizes} gives \eqref{eq:gap_interval_exact}.

For \eqref{eq:gap_expansions}, let \(g\) be a function of the
normalized shock, real-analytic near zero, with \(g(0)=g'(0)=0\). Then
\[
g\bigl(\tilde\pi(\pi)\bigr)+g\bigl(\tilde\pi(-\pi)\bigr)
=\tfrac12g''(0)\bigl[\tilde\pi(\pi)^2+\tilde\pi(-\pi)^2\bigr]
+\tfrac16g'''(0)\bigl[\tilde\pi(\pi)^3+\tilde\pi(-\pi)^3\bigr]+O(\pi^4)
\]
with \(\tilde\pi(\pi)^2+\tilde\pi(-\pi)^2=2\pi^2+O(\pi^4)\) and
\(\tilde\pi(\pi)^3+\tilde\pi(-\pi)^3=O(\pi^4)\). Take for \(g\) the
functions \(\tilde\pi\mapsto\mathcal Y^{\Lambda}_t(\tilde\pi h_\theta)\),
\(\tilde\pi\mapsto\mathcal Y^{\Pi}_t(\tilde\pi h_\theta)-\dot f_t\tilde\pi\)
and
\(\tilde\pi\mapsto\mathcal Y^{P}_{T+1}(\tilde\pi h_\theta)-\dot P_{T+1}\tilde\pi\).
By Lemma~\ref{lem:resource_accounting_date} they are real-analytic near
zero and vanish there together with their first derivatives, the first
since \(\Lambda_t=\Pi_t-f_t\) and
\(D\mathcal Y^{\Pi}_t(\mathbf 0)h_\theta=\dot f_t\). Their values at
\(\tilde\pi(\pm\pi)\) are \(\Lambda_t(\pm\pi)\), \(\Pi^\sharp_t(\pm\pi)\)
and \(P^\sharp_{T+1}(\pm\pi)\), and their second derivatives at zero
are the three Hessians in the direction \(h_\theta\), which gives
\eqref{eq:gap_expansions}.

The misalignment \(h_\theta\) lies in \(\mathcal Z_{\mathbf A}\) and
satisfies \(\|h_\theta\|_m\le C_\zeta(1-\theta)\) (proof of
Lemma~\ref{lem:concave_scaling}), and
\(\dot f_t=D\mathcal Y_t(\mathbf 0)h_\theta\). The linear and bilinear
bounds give \eqref{eq:gap_slack_orders}. The lower bound on the
deadweight loss is the positive semidefinite Hessian of
\(\mathcal Y^{\Lambda}_t\) at \(\mathbf 0\)
(Lemma~\ref{lem:resource_accounting_date}).
\end{proof}

\begin{lemma}[The loss on consumed output and the loss in transit]
\label{lem:consumed_transit}
Suppose that Assumptions~\ref{assump:diffuse_incidence}
and~\ref{assump:monotone_decay_spectrum} of the paper hold, that
\(|\pi|\le\overline\pi\) and that \(\underline k\ge1\). Then at every
date the deadweight loss splits exactly as
\(\Lambda_t=\Lambda^{\mathrm c}_t+\Lambda^{\mathrm{tr}}_t\), into the
loss on consumed output and the loss in transit,
\begin{equation}
\Lambda^{\mathrm c}_t:=\frac{1}{w^\ast}\Bigl[\sum_{i\in\mathcal R}r_{i,t}E_{i,t}+H_t\Bigr],
\qquad
\Lambda^{\mathrm{tr}}_t:=\frac{1}{w^\ast}\sum_{i\in N}(1-r_{i,t})E_{i,t}
\label{eq:loss_split}
\end{equation}
Both are nonnegative. The loss in transit is also part of the drawdown \(\Pi_t\) of the goods in transit,
\begin{equation}
\Pi_t-\Lambda^{\mathrm{tr}}_t
=f_t+\Lambda^{\mathrm c}_t
=\sum_{i\in\mathcal R}\gamma_i\Bigl[\bigl(1+(1-\beta_i)\,\mathcal F_i(\widehat x_{i,t})\bigr)\frac{r_{i,t}}{r_i}-1\Bigr]
\label{eq:lossfree_consumption}
\end{equation}
and, summed over the transition,
\(\sum_{t\ge0}(\Pi_t-\Lambda^{\mathrm{tr}}_t)=-\sum_{t\ge0}\Lambda^{\mathrm{tr}}_t\).
\end{lemma}

\begin{proof}
By \eqref{eq:E_exact} and \(q_{i,t}/q_i^\ast=(X_{i,t}/X_i^\ast)^{1-\beta_i}\),
\(E_{i,t}=m_i^\ast\bigl[(1-\beta_i)\mathcal F_i(\widehat x_{i,t})+1-q_{i,t}/q_i^\ast\bigr]\)
exactly, so that
\(q_{i,t}/q_i^\ast=1+(1-\beta_i)\mathcal F_i(\widehat x_{i,t})-E_{i,t}/m_i^\ast\).
For a retailer, \(c_{i,t}/c_i^\ast=(q_{i,t}/q_i^\ast)(r_{i,t}/r_i)\),
and \(\gamma_i/(r_im_i^\ast)=1/w^\ast\) by
\eqref{eq:retail_absorption_def} of the paper. Hence
\[
\sum_{i\in\mathcal R}\gamma_i\Bigl(\frac{c_{i,t}}{c_i^\ast}-1\Bigr)
=\sum_{i\in\mathcal R}\gamma_i\Bigl[\bigl(1+(1-\beta_i)\mathcal F_i(\widehat x_{i,t})\bigr)\frac{r_{i,t}}{r_i}-1\Bigr]
-\frac{1}{w^\ast}\sum_{i\in\mathcal R}r_{i,t}E_{i,t}
\]
By \eqref{eq:gdp_def_main} of the paper and \eqref{eq:consumption_loss},
\(f_t=\sum_{i\in\mathcal R}\gamma_i(c_{i,t}/c_i^\ast-1)-H_t/w^\ast\),
which gives the expression for \(f_t+\Lambda^{\mathrm c}_t\). The split
of \(\Lambda_t=(E_t+H_t)/w^\ast\) is immediate, and
\eqref{eq:resource_identity} of the paper gives
\(\Pi_t-\Lambda^{\mathrm{tr}}_t=f_t+\Lambda^{\mathrm c}_t\). Both
losses are nonnegative, since \(E_{i,t}\ge0\) and \(H_t\ge0\)
(Lemma~\ref{lem:resource_accounting_date}) and \(0\le r_{i,t}\le1\),
the output of a firm being what the household and the firms buy of it.
Finally, the drawdowns \(\Pi_t\) of the goods in transit sum to zero over the transition, since
\(P_0=0\) and \(P_t\to0\), and the losses
\(0\le\Lambda^{\mathrm{tr}}_t\le\Lambda_t\) are summable
(Proposition~\ref{prop:resource_accounting} of the paper), which gives
the last claim.
\end{proof}

 At the date
\(\underline k\), write \(\dot{\widehat x}_i:=\dot{\widehat x}_{i,\underline k}\)
for the first-order dated input deviations of firm \(i\), and
\(\dot{\widehat m}^\perp_i:=\dot{\widehat m}^\perp_{i,\underline k+1}\)
for the first-order misalignment of the next balance of retailer \(i\)
of Lemma~\ref{lem:gdp_labor_income_anchor}. With \(\phi_i\) the scaled
relative expenditure impulse \eqref{eq:scaled_impulse_def} of the paper, write
\(\dot{\widehat c}_i:=\phi_i-\dot{\widehat m}^\perp_i\) for the
first-order response of the household's consumption of good \(i\).
Write \(\ddot\Lambda^{\mathrm c}_{\underline k}\) and
\(\ddot\Lambda^{\mathrm{tr}}_{\underline k}\) for the second
derivatives at \(\pi=0\) of the loss on consumed output and of the loss
in transit.

\begin{lemma}[The deadweight loss and the drawdown of the goods in transit at the shortest lead time]
\label{lem:shortest_lag_curvature}
Suppose that Assumptions~\ref{assump:diffuse_incidence}
and~\ref{assump:monotone_decay_spectrum} of the paper hold and that
\(\underline k\ge1\). Then at the date \(\underline k\) the first-order
dated input deviations are \(\dot{\widehat x}_{ji}=\sigma_i-\overline\sigma_j\)
on firm \(i\)'s suppliers of lead time \(\underline k\) and zero on its
other suppliers, so that
\(\mathcal F_i(\dot{\widehat x}_i)=\xi_{i,\underline k}\), and
\begin{equation}
\begin{aligned}
\ddot\Lambda^{\mathrm c}_{\underline k}
&=\sum_{i\in\mathcal R}\gamma_i\Bigl[(1-\beta_i)\bigl(\langle\dot{\widehat x}_i,\dot{\widehat x}_i\rangle_i+\beta_i\xi_{i,\underline k}^2\bigr)+\dot{\widehat c}_i^{\,2}\Bigr]\\
\ddot\Lambda^{\mathrm{tr}}_{\underline k}
&=\sum_{i\in N}\Bigl(\frac{m_i^\ast}{w^\ast}-\gamma_i\Bigr)(1-\beta_i)\bigl(\langle\dot{\widehat x}_i,\dot{\widehat x}_i\rangle_i+\beta_i\xi_{i,\underline k}^2\bigr)\\
D^2\mathcal Y^{\Pi}_{\underline k}(\mathbf 0)[h_\theta,h_\theta]-\ddot\Lambda^{\mathrm{tr}}_{\underline k}
&=-2\sum_{i\in\mathcal R}\gamma_i(1-\beta_i)\sum_{j\in\mathcal S_i}a_{ji}\,\overline\sigma_j\,\dot{\widehat x}_{ji}
+2\sum_{i\in\mathcal R}\gamma_i\,\dot{\widehat m}^\perp_i\bigl(\dot{\widehat m}^\perp_i-\phi_i\bigr)
\end{aligned}
\label{eq:shortest_lag_curvatures}
\end{equation}
\end{lemma}

\begin{proof}
At impact firm \(i\) holds
\(m_i^\ast+\pi\overline M\zeta_i=m_i^\ast(1+\pi+\pi\iota_i)\) and pays
the wage bill \((1+\pi)\beta_im_i^\ast\), so its intermediate
expenditure is \(e_{i,0}=(1-\beta_i)m_i^\ast\bigl(1+\pi(1+\sigma_i)\bigr)\).
The household spends \((1+\pi)\gamma_jw^\ast\) on good \(j\). Since
\(\sum_ia_{ji}(1-\beta_i)m_i^\ast+\gamma_jw^\ast=m_j^\ast\), the nominal
demand for good \(j\) is
\(d_{j,0}=m_j^\ast(1+\pi)+\pi\sum_ia_{ji}(1-\beta_i)m_i^\ast\sigma_i=m_j^\ast\bigl(1+\pi(1+\overline\sigma_j)\bigr)\)
by \eqref{eq:output_share_matrix} and
\eqref{eq:buyer_average_impulse} of the paper. Since
\(\underline k\ge1\), every input used at date \(0\) was ordered before
the shock, and labor is fixed, so \(q_{j,0}=q_j^\ast\). Hence
\(p_{j,0}/p_j^\ast=1+\pi(1+\overline\sigma_j)=(1+\pi)(1+\tilde\pi\,\overline\sigma_j)\),
and the order rule gives
\(x_{ji,0}/x_{ji}^\ast=(e_{i,0}/e_i^\ast)(p_j^\ast/p_{j,0})=(1+\tilde\pi\sigma_i)/(1+\tilde\pi\,\overline\sigma_j)\).
Since \(\underline k\) is the shortest lead time, no input used before
date \(\underline k\) was ordered after the shock. At date
\(\underline k\) firm \(i\) uses the orders of date \(0\) on its
suppliers of lead time \(\underline k\) and orders placed before the
shock on the others. As functions of \(\tilde\pi\), its dated input
ratios and the consumption shares of the retailers are therefore
\begin{equation}
y_{ji}=\frac{1+\tilde\pi\sigma_i}{1+\tilde\pi\,\overline\sigma_j}\ \ (k_{ji}=\underline k),
\qquad
y_{ji}=1\ \ (k_{ji}>\underline k),
\qquad
\frac{r_{i,\underline k}}{r_i}=\frac{1}{1+\tilde\pi\,\dot{\widehat m}^\perp_i}\ \ (i\in\mathcal R)
\label{eq:shortest_lag_inputs}
\end{equation}
the last by Lemma~\ref{lem:gdp_labor_income_anchor}, since
\(r_{i,\underline k}/r_i=(1+\pi)m_i^\ast/m_{i,\underline k+1}\) with
\(m_{i,\underline k+1}=m_i^\ast\{1+\pi(1+\dot{\widehat m}^\perp_i)\}\).
At zero the first derivatives of the \(y_{ji}\) are
\(\sigma_i-\overline\sigma_j\), which gives \(\dot{\widehat x}_i\), and
\(\mathcal F_i(\dot{\widehat x}_i)=\xi_{i,\underline k}\) by
\eqref{eq:impact_forcing_def} of the paper. Their second derivatives at
zero are \(y_{ji}''(0)=-2\,\overline\sigma_j\,\dot{\widehat x}_{ji}\).
Dots keep their meaning of derivatives in \(\pi\), so that
\(\ddot{\widehat x}_{ji,\underline k}=y_{ji}''(0)-2\dot{\widehat x}_{ji}\).

Write \(\widehat x_i:=\widehat x_{i,\underline k}\) and
\(\widehat c_i:=\widehat c_{i,\underline k}\). By
\eqref{eq:production_loss_split} and \eqref{eq:deadweight_second_order}
of Lemma~\ref{lem:resource_accounting_date} and by
\eqref{eq:firm_inefficiency_def} of the paper, up to terms of third
order,
\[
E_{i,\underline k}=\tfrac12m_i^\ast(1-\beta_i)\bigl(\langle\widehat x_i,\widehat x_i\rangle_i+\beta_i\mathcal F_i(\widehat x_i)^2\bigr),
\qquad
H_{\underline k}=\tfrac12w^\ast\sum_{i\in\mathcal R}\gamma_i\widehat c_i^{\,2}
\]
Here \(\widehat x_i=\tilde\pi\,\dot{\widehat x}_i+O(\tilde\pi^2)\) by
\eqref{eq:shortest_lag_inputs}, and
\(\widehat c_i=\tilde\pi\,\dot{\widehat c}_i+O(\tilde\pi^2)\) by
\(c_{i,\underline k}/c_i^\ast=(q_{i,\underline k}/q_i^\ast)(r_{i,\underline k}/r_i)\),
\eqref{eq:scaled_impulse_def} of the paper and
\eqref{eq:shortest_lag_inputs}. Since
\(r_{i,\underline k}=r_i+O(\tilde\pi)\) and
\(E_{i,\underline k}=O(\tilde\pi^2)\), differentiating
\eqref{eq:loss_split} of Lemma~\ref{lem:consumed_transit} twice at zero
and using \(r_im_i^\ast/w^\ast=\gamma_i\) gives the first two lines of
\eqref{eq:shortest_lag_curvatures}. Both losses vanish at zero together
with their first derivatives, so their second derivatives are the same
in \(\pi\) and in \(\tilde\pi\).

The function \(1+(1-\beta_i)\mathcal F_i(\widehat x_i)\) of
\(\tilde\pi\) equals one at zero, with first derivative \(\phi_i\) and
second derivative \((1-\beta_i)\mathcal F_i(y_i''(0))\) there. By
\eqref{eq:shortest_lag_inputs}, the second derivative at zero of its
product with \(r_{i,\underline k}/r_i\) is
\((1-\beta_i)\mathcal F_i(y_i''(0))-2\phi_i\dot{\widehat m}^\perp_i+2(\dot{\widehat m}^\perp_i)^2\).
Summing against \(\boldsymbol\gamma\) over the retail tier and using
\eqref{eq:lossfree_consumption} of Lemma~\ref{lem:consumed_transit}
gives the third line of \eqref{eq:shortest_lag_curvatures}.
\end{proof}

\begin{example}
\label{ex:drawdown_curvature_sign}
Neither term on the right of the third line of
\eqref{eq:shortest_lag_curvatures} has a determined sign, nor does
\(D^2\mathcal Y^{\Pi}_{\underline k}(\mathbf 0)[h_\theta,h_\theta]\).
Take two retailers, with \(\theta=0.8\), \(\rho=-1\), every lead time
equal to one and labor shares of one half. The household spends the
share \(2/5\) on the first retailer. The second buys only its own
good, and the first buys the share \(a\) of its inputs from itself and
the rest from the second. By \eqref{eq:shortest_lag_curvatures}, at
\(a=1/20\) the first term, the sum of the two terms and
\(D^2\mathcal Y^{\Pi}_1(\mathbf 0)[h_\theta,h_\theta]\) are about
\(2.2\times10^{-4}\), \(9.9\times10^{-5}\) and \(4.0\times10^{-3}\). At
\(a=1/2\) they are about \(-3.1\times10^{-3}\), \(-6.8\times10^{-3}\)
and \(-2.2\times10^{-3}\). The first-order response \(\dot f_1\) is
positive at both values of \(a\). The second term is negative at both
values of \(a\), and positive, about \(1.7\times10^{-2}\), at
\(a=19/20\) when the household spends the share \(1/10\) on the first
retailer.
\end{example}

\appendix
\setcounter{section}{2}
\counterwithin{equation}{section}
\counterwithin{figure}{section}
\counterwithin{table}{section}
\renewcommand{\theHequation}{\thesection.\arabic{equation}}
\renewcommand{\theHfigure}{\thesection.\arabic{figure}}
\renewcommand{\theHtable}{\thesection.\arabic{table}}
\singlespacing

\clearpage
\section{Calibration, Notation and Dependency Map}
\label{app:construction}

\subsection{Parameter Calibration}
\label{app:calibration}
\label{app:calibration_summary}

The within-vintage curvature is set at \(\rho=-0.8\), the value
implied by the most direct firm-level estimate of the elasticity of
substitution across same-product suppliers, about \(0.55\)
\citep{cevallosfujiy2024production}. Table~\ref{tab:elasticities}
restates the leading estimates as values of \(\rho=1-1/\eta\), with
\(\eta\) the elasticity of substitution. The baseline sets
\(\rho_c=\rho\), so the reported numbers carry no cross-vintage
amplification. Table~\ref{tab:baseline} lists the baseline values of both supply chains and the range of each sweep.
Every link is assigned an integer lead time. The lead time of the link
from buyer \(i\) to supplier \(j\) is drawn uniformly on the integers from
\(\operatorname{round}\bigl(k_{\mathrm{lo}}(1-\epsilon_{\downarrow}\psi_i)\bigr)\)
to
\(\operatorname{round}\bigl(k_{\mathrm{hi}}(1-\epsilon_{\downarrow}\psi_i)\bigr)\),
and at least from \(1\) to \(2\), with \(\psi_i\) the downstreamness index
of the buyer (Definition~\ref{def:downstreamness} of the paper). At the
baseline the range is three to eight months for the most upstream buyers
and two to five months for the most downstream buyers. Every buyer keeps
at least \(\operatorname{round}\bigl(\underline n_{\mathrm s}(1-\psi_i)\bigr)\)
links with a lead time of \(\underline k_{\mathrm s}\) months or more.

\begin{table}[H]
\centering
\footnotesize
\caption{Short-run substitutability of intermediate inputs, expressed as the
within-vintage curvature \(\rho\). Each entry is the curvature implied by the
study's estimated elasticity of substitution \(\eta\) through
\(\rho=1-1/\eta\), or, where the study estimates none, the sign its findings
imply. A negative curvature denotes gross complements
and \(\rho\to-\infty\) the Leontief limit.}
\label{tab:elasticities}
\renewcommand{\arraystretch}{1.25}
\begin{tabular}{@{}>{\raggedright\arraybackslash}p{0.22\textwidth}>{\centering\arraybackslash}p{0.10\textwidth}>{\raggedright\arraybackslash}p{0.33\textwidth}>{\raggedright\arraybackslash}p{0.26\textwidth}@{}}
\toprule
Study & Implied \(\rho\) & Context & Sample and horizon \\
\midrule
\citet{barrot2016input} & \(\ll 0\) & Propagation of supplier shocks to customers when inputs are specific; no elasticity estimated & US firms and suppliers; \(\le\)1\,yr \\
\citet{atalay2017sectoral} & \(\lesssim-4\) & Sectoral shocks and aggregate volatility & US industries; annual \\
\citet{boehm2019input} & \(\lesssim-4\)  & Transmission of input shocks via multinationals; imported and domestic inputs close to Leontief & US affiliates of Japanese MNEs; months \\
\citet{carvalho2021supply} & \(\approx0.15\) & Supply chain disruption; intermediate inputs weak gross substitutes (\(\eta\approx1.18\)), primary and intermediate inputs complements (\(\eta\approx0.6\)) & Japanese firm network; \(\le\)1\,yr \\

\citet{cevallosfujiy2024production} & \(\approx-0.8\) & Firm-level substitution across same-product suppliers & Indian firms; COVID lockdowns; monthly \\
\bottomrule
\end{tabular}
\end{table}

\begin{table}[tbp]
\centering
\footnotesize
\caption{Baseline configuration of the two supply chains, organized by parameter class. For a swept parameter the baseline is the point from which the sweep in Section~\ref{sec:quant-robustness} departs, and the range of the sweep is recorded with the rationale.}
\label{tab:baseline}
\renewcommand{\arraystretch}{1.2}
\begin{tabular}{@{}>{\raggedright\arraybackslash}p{4.1cm}>{\raggedright\arraybackslash}p{2.5cm}>{\raggedright\arraybackslash}p{2.6cm}>{\raggedright\arraybackslash}p{5.5cm}@{}}
\toprule
Parameter & Synthetic supply chain & US supply chain & Source or rationale \\
\midrule
\multicolumn{4}{@{}l}{\itshape Fixed at conventional values}\\
$\rho$\, within-vintage curvature & $-0.8$ & $-0.8$ & $\eta\approx0.56$; gross complements \citep{cevallosfujiy2024production}; swept over $[-1.4,-0.1]$, plus $+0.3$ \\
$\rho-\rho_c$\, curvature gap & $0$ & $0$ & baseline carries no cross-vintage amplification (Lemma~\ref{lem:inefficiency_gap}) \\
$\beta$\, labor share & $0.4$ & $0.4$ & Cobb--Douglas top nest; the same at every firm; standard range \\
labor along the transition & fixed at $l_i^\ast$ & fixed at $l_i^\ast$ & Assumption~\ref{assump:short_run_nominal_rigidity}; wage indexed to the money stock \\
$\pi$\, shock size & $\pm0.05$ & $\pm0.05$ & \(5\%\) of the money stock; swept over $0.002$--$0.10$ \\

period length & one month & one month & three periods aggregate to a quarter \\
\addlinespace
\multicolumn{4}{@{}l}{\itshape Network, lead times and incidence}\\
network source & tiered generator & gravity reconstruction & controlled gradient vs.\ census and input--output data \\
$n$\, number of firms & $10^3$--$10^5$ & $10^4$--$6.46\times10^{6}$ & scale ladder \\
tiers and relative widths & five, $1\!:\!3\!:\!9\!:\!27\!:\!81$ & --- & top of the supply chain to the retail tier \\
mean number of suppliers & $8$ & $50$ & draws from upstream tiers on the synthetic supply chain; gravity fit on the US supply chain \\
Pareto tail index of sizes within a tier & $4$ & --- & sizes on the US supply chain are the census \\
sourcing & distance decay up the supply chain & gravity reconstruction & size--position sorting on the synthetic supply chain (rank correlation $-0.35$); data on the US supply chain \\
$k_{\mathrm{lo}}\!:\!k_{\mathrm{hi}}$\, base lead times & $3\!:\!8$ & $3\!:\!8$ & one to three quarters order-to-delivery \\
$\epsilon_{\downarrow}$\, downstream compression & $0.4$ & $0.4$ & time to build \\

$\underline k_{\mathrm{s}},\underline n_{\mathrm{s}}$\, slow links & $4,\;3$ & $4,\;3$ & upstream time-to-build bottleneck \\
$\theta$\, incidence concavity & $0.5$ & $0.5$ & small firms struck harder \\
$\theta_{\mathrm{h}}$\, household incidence exponent & $0.9$ & $0.9$ & the household's share of the shock (Section~\ref{sec:quant-setup}) \\
$\epsilon_{\mathrm{cap}}$\, transfer cap & $0.5$ & $0.5$ & no-shutdown at the smallest firms (the margin \eqref{eq:no_shutdown_margin} of the paper) \\
\bottomrule
\end{tabular}
\end{table}

\clearpage
\subsection{Notation}
\label{app:notation}

This subsection lists the notation of the paper, grouped by theme and, within
each group, in roughly the order of first use. The last column gives the
section of the paper, of Appendix~\ref{app:analytical} or of this appendix, in which the symbol is
introduced.

{\footnotesize
\renewcommand{\arraystretch}{1.12}
\begin{longtable}{@{}p{0.26\textwidth} p{0.56\textwidth} p{0.125\textwidth}@{}}
\toprule
Symbol & Meaning & Where \\
\midrule
\endfirsthead
\toprule
Symbol & Meaning & Where \\
\midrule
\endhead
\endfoot
\bottomrule
\endlastfoot

\multicolumn{3}{@{}l}{\itshape Sets and indices} \\
$N=\{1,\dots,n\}$, $h$ & Firms and the representative household & \textsection\ref{subsec:environment} \\
$i,j,l\in N$ & Firm indices; in $a_{ji}$, $x_{ji}$ and $k_{ji}$ the first index is the supplier and the second the buyer & \textsection\ref{subsec:environment} \\
$\mathcal S_i$, $\mathcal S_i^{(k)}$ & Suppliers of firm $i$; those delivering at lead time $k$ & \textsection\ref{subsubsec:technology} \\
$\mathcal K_i$ & Active vintages of firm $i$, the lead times at which it has a supplier & \textsection\ref{subsubsec:technology} \\
$\mathcal R=\operatorname{supp}\boldsymbol\gamma$; $\psi_R$ & Firms the household buys from, the retail tier; the position of its least downstream firm & \textsection\ref{subsec:downstreamness} \\
$\mathcal T$, $\psi_T$; $\mathcal G$, $\Delta_+$ & Top tier \eqref{eq:top_tier} and its level; neighborhood of the retail tier \eqref{eq:retail_neighborhood} and its width & \textsection\ref{subsec:downstreamness}, \textsection\ref{subsec:finite_horizon_persistence} \\
$\mathcal R^{+}=\mathcal R\cup\mathcal G$ & Neighborhood of final demand, the retail tier with its neighborhood & \textsection\ref{subsec:finite_horizon_persistence} \\
$\mathcal C$ & A set of firms, in Proposition~\ref{prop:impact_reallocation} one on which the shock lands more heavily than on the other buyers of its shortest-lead-time inputs & \textsection\ref{subsec:downstreamness} \\
$t,s$; $k$ & Dates; lead times & \textsection\ref{subsec:environment} \\

\midrule
\multicolumn{3}{@{}l}{\itshape Technology, network and household} \\
$\beta_i\in[\underline\beta,\overline\beta]$, $\boldsymbol\beta$ & Labor share of firm $i$ and the vector of shares & \textsection\ref{subsubsec:technology} \\
$\rho$, $\rho_c$; $\overline\rho,\underline\rho$ & Within- and cross-vintage CES curvatures, $-\overline\rho\le\rho_c\le\rho\le-\underline\rho$ & \textsection\ref{subsubsec:technology} \\
$\rho-\rho_c$ & Curvature gap & \textsection\ref{subsubsec:technology} \\
$\eta=1/(1-\rho)$, $\eta_c=1/(1-\rho_c)$ & The two elasticities of substitution & \textsection\ref{app:proof_existence} \\
$X_i$, $X_i^{(k)}$ & Intermediate composite of firm $i$ and its vintage composite of lead time $k$ & \textsection\ref{subsubsec:technology} \\
$\nu_{ji}$ & Supplier weight of firm $i$ on good $j$ in the within-vintage nest & \textsection\ref{subsubsec:technology} \\
$\chi_i^{(k)}$ & Vintage weight of firm $i$ on its vintage of lead time $k$ in the outer nest & \textsection\ref{subsubsec:technology} \\
$a_{ji}=s_i^{(k_{ji})}\widetilde a_{ji}^{(k_{ji})}$ & Supplier share, the share of $i$'s input spending paid to $j$: cross-vintage share times within-vintage share & \textsection\ref{subsubsec:network} \\
$\mathbf A=(a_{ji})$ & Column-stochastic supplier-share matrix & \textsection\ref{subsubsec:network} \\
$\mathbf A_\beta=\mathbf A\,\mathrm{diag}(\mathbf 1-\boldsymbol\beta)$ & Matrix of firm-to-firm expenditure shares, column sums $1-\beta_i$ & \textsection\ref{subsec:equilibrium} \\
$\gamma_i$, $\boldsymbol\gamma$ & Household shares & \textsection\ref{subsubsec:household} \\
$U(c_1,\dots,c_n)$, $c_i$ & Cobb--Douglas household utility and consumption of good $i$ & \textsection\ref{subsubsec:household} \\
$a_{ji}^{1-\rho}$ & Calibrated value of the supplier weights $\nu_{ji}$ in the simulated supply chains & \textsection\ref{sec:quant} \\

\midrule
\multicolumn{3}{@{}l}{\itshape Lead times and the pipeline} \\
$k_{ji}$ & Lead time of firm $i$'s purchase from supplier $j$ & \textsection\ref{subsubsec:technology} \\
$\underline k$, $\widehat k$ & Minimum and maximum lead time & \textsection\ref{subsubsec:technology} \\
$s_i^{(k)}$ & Share of firm $i$'s input spending on its vintage layer of lead time $k$ & \textsection\ref{subsubsec:network} \\
$\mathbf k=(k_{ji})$ & Lead-time profile & \textsection\ref{subsec:stability} \\

$\boldsymbol{\mathcal X}_t^{(k)}$, $\boldsymbol{\mathcal X}_t$ & Pipeline layer at remaining lead time $k$ and the pipeline state & \textsection\ref{subsec:temporal_depth_formal} \\
$\Delta k$ & Uniform shift of the lead-time profile (Corollary~\ref{cor:lag_shift_inefficiency}) & \textsection\ref{subsec:lag_profile} \\
$\mathbf k'$ & A second lead-time profile, a buyerwise mean-preserving lead-time-variance increase of $\mathbf k$ (Corollary~\ref{cor:firm_level_vintage_spread_aggregate_inefficiency}) or its uniform shift by $\Delta k$ (Corollary~\ref{cor:lag_shift_inefficiency}); in the proof of that corollary a prime marks an object under $\mathbf k'$ & \textsection\ref{subsec:lag_profile}, \textsection\ref{app:proof_lag_shift} \\

\midrule
\multicolumn{3}{@{}l}{\itshape Downstreamness and the supply chain} \\
$\boldsymbol\Omega=\mathbf D_m^{-1}\mathbf A_\beta\mathbf D_m$ & Output-share matrix, $\Omega_{ji}$ the share of $j$'s sales bought by $i$ & \textsection\ref{sec:reallocation} \\
$r_j=\gamma_jw^\ast/m_j^\ast$ & Consumption share of firm $j$, the share of its output sold to the household & \textsection\ref{sec:reallocation} \\
$\delta$ & Per-round discount of downstreamness & \textsection\ref{sec:reallocation} \\
$\widetilde{\boldsymbol\psi}=(\mathbf I-\delta\boldsymbol\Omega)^{-1}\mathbf r$, $\boldsymbol\psi$ & Downstreamness and the downstreamness index & \textsection\ref{sec:reallocation} \\
$\Delta_\psi$ & Height of a tier of the supply chain (Assumption~\ref{assump:transmission_delay}) & \textsection\ref{subsec:downstreamness} \\
$\epsilon_h$ & Consumption-share floor at the retail tier (Assumption~\ref{assump:retail_outflow_cap}) & \textsection\ref{subsec:downstreamness} \\
$\epsilon_k$, $k_\ast$; $s_R$, $\mathcal W$, $t_{\mathcal W}$ & Share of a retailer's input expenditure that may arrive with lead time $k_\ast$ or more, and of a top-tier firm's that may arrive in fewer than $k_\ast$ periods, and the lead time that divides the two (Assumption~\ref{assump:lag_position_sorting}, time to build); share of the retail tier's input expenditure that makes a lead time a retail lead-time date, the retail lead-time dates \eqref{eq:retail_lag_dates} and the last of them & \textsection\ref{subsec:downstreamness}, \textsection\ref{subsec:finite_horizon_persistence} \\

\midrule
\multicolumn{3}{@{}l}{\itshape Stationary equilibrium} \\
$p_i,q_i,l_i,x_{ji},w$; $z^\ast$ & Prices, outputs, labor, orders, wage; a star marks the stationary value & \textsection\ref{subsec:environment} \\
$m_i^\ast=p_i^\ast q_i^\ast$, $\mathbf m$, $\mathbf m^\ast$; $\overline M=\mathbf 1^\top\mathbf m^\ast$ & Nominal sales, which are next period's balances; aggregate firm balances & \textsection\ref{subsec:equilibrium} \\
$\mathbf D_m=\mathrm{diag}(\mathbf m^\ast)$; $\|\mathbf v\|_m=\|\mathbf D_m^{-1}\mathbf v\|_\infty$ & Stationary balances on the diagonal; sup norm in proportional coordinates & \textsection\ref{subsec:stability} \\
$\Psi_i(\mathbf p)$; $\overline p_i^{(k)}(\mathbf p)$, $\overline p_i(\mathbf p)$ & Unit cost of firm $i$'s output; CES price indices of its vintage and composite (proof of Proposition~\ref{prop:existence_stationary}) & \textsection\ref{app:proof_existence} \\
$\check{\mathbf p}=\log\mathbf p$, $\check\Psi$ & Log prices and the log unit-cost map (proof of Proposition~\ref{prop:existence_stationary}) & \textsection\ref{app:proof_existence} \\
$\lambda_2(\mathbf A)$; $\overline\lambda$, $C_\lambda$ & Nominal redistribution rate; its bound and the multiplier of the power bound \eqref{eq:power_bound}, parts (a) and (b) of Assumption~\ref{assump:monotone_decay_spectrum} & \textsection\ref{subsec:stability} \\
$\mathcal Z_{\mathbf A}=\{\mathbf v:\mathbf 1^\top\mathbf v=0\}$ & Conservation subspace & \textsection\ref{subsec:stability} \\

\midrule
\multicolumn{3}{@{}l}{\itshape Monetary shock and incidence} \\
$\pi\in[-\overline\pi,\overline\pi]$; $\pi_\ast$ & Shock as a fraction of $\overline M$ and its bound, set by the constants of the assumptions; the threshold of Theorem~\ref{thm:asymmetric_nonneutrality}, which depends on the economy & \textsection\ref{subsec:money}, \textsection\ref{subsec:asymmetric_nonneutrality} \\
$\theta$; $\overline\epsilon_\theta$ & Incidence concavity; bound on $1-\theta$ (Assumption~\ref{assump:diffuse_incidence}) & \textsection\ref{subsec:money} \\
$\zeta_i=(m_i^\ast)^\theta/\sum_j(m_j^\ast)^\theta$ & Incidence share of firm $i$ & \textsection\ref{subsec:money} \\
$\ell_i$, $L_\ell$ & Centered log size of firm $i$ and its bound & \textsection\ref{subsec:money} \\
$Z$ & Normalizer of the incidence ratio, $\overline M\zeta_i/m_i^\ast=e^{-(1-\theta)\ell_i}/Z$, with $Z=\sum_j(m_j^\ast/\overline M)e^{-(1-\theta)\ell_j}$ & \textsection\ref{subsec:money} \\
$\iota_i$; $C_\zeta$ & Misalignment the shock creates at firm $i$ per unit of shock \eqref{eq:incidence_ratio}, the incidence ratio $\overline M\zeta_i/m_i^\ast$ less one; incidence range, $|\iota_i|\le C_\zeta(1-\theta)$, at most $2L_\ell e^{2\overline\epsilon_\theta L_\ell}$ & \textsection\ref{subsec:money} \\
$\mathbf m_0=\mathbf m^\ast+\pi\overline M\boldsymbol\zeta$ & Impact balances & \textsection\ref{subsec:money} \\
$h_\theta=\dot{\mathbf m}_0^\perp=\overline M\boldsymbol\zeta-\mathbf m^\ast$ & Impact misalignment, the cross-sectional balance misalignment at impact & \textsection\ref{subsec:stability}, \textsection\ref{subsec:asymmetric_nonneutrality} \\
$\varepsilon_\pi=|\pi|(1-\theta)$ & Size of the shock times the slack (proof of Proposition~\ref{prop:stability_stationary}) & \textsection\ref{app:proof_stability} \\

\midrule
\multicolumn{3}{@{}l}{\itshape Post-shock dynamics} \\
$\wM=(w^\ast/\overline M)\mathbf 1^\top\mathbf m_t=(1+\pi)w^\ast$ & Indexed nominal wage (Assumption~\ref{assump:short_run_nominal_rigidity}) & \textsection\ref{subsec:stability} \\
$m_{i,t}$, $d_{j,t}$, $p_{j,t}$, $q_{i,t}$, $x_{ji,t}$ & Date-$t$ balances, nominal demand, prices, outputs and orders & \textsection\ref{subsec:stability} \\
$e_{i,t}=m_{i,t}-\wM l_i^\ast$ & Intermediate-input expenditure of firm $i$ & \textsection\ref{subsec:stability} \\
$c_\beta=(1-\overline\beta)/(1+\overline\beta)$ & No-shutdown margin \eqref{eq:no_shutdown_margin}, implied for small shocks & \textsection\ref{app:proof_stability} \\
$\check x_{ji,t}$, $\check q_{j,t}$, $\check\Xi_{ij,t}$; $C_{\check\Xi}$, $\kappa_\beta$, $C_q$, $\Upsilon$; $\mathcal L_i$ & Log deviations of orders and outputs from their stationary values, and the log contrast \eqref{eq:stab_order_identity}; bound \eqref{eq:stab_forcing_bound} on the log contrasts, contraction factor $(1-\underline\beta)\kappa^{-\widehat k}$, output envelope constant and scaled output deviation of the induction; log-output map of firm $i$ (proof of Proposition~\ref{prop:stability_stationary}) & \textsection\ref{app:proof_stability} \\
$\mathbf S_t=(\mathbf m_t,\boldsymbol{\mathcal X}_t)$, $\mathbf S^\ast$, $\Phi$, $\widehat{\mathbf S}_t(\pi)$ & State, stationary state, transition map, proportional-deviation state & \textsection\ref{subsec:stability}, \textsection\ref{app:lemmas_model} \\
$\mathbf J_\Phi^\ast$, $\dot{\mathbf S}_0$; $\mathcal V_i$; $\mathcal N_i$, $\mathbf J_{\mathcal N}$, $\mathcal U$ & Jacobian of the transition map at the stationary state and the first-order response of the impact state; map that returns firm $i$'s dated input bundle from the state; within-period fixed-point map of the log outputs, its Jacobian and the neighborhood of the stationary state on which the maps are real-analytic (Lemma~\ref{lem:propagation_linearization}) & \textsection\ref{app:lemmas_model} \\
$\kappa$, $C_S$ & Contraction rate and prefactor of the post-shock envelope (Proposition~\ref{prop:stability_stationary}) & \textsection\ref{subsec:stability} \\
$\dot{\mathbf m}_t$; $\dot{\mathbf m}_t^\perp=\dot{\mathbf m}_t-\mathbf m^\ast$; $\mathbf m_t^\perp$, $\widehat{\mathbf m}_t^\perp$ & First-order response of the balances; the cross-sectional balance misalignment; the misalignment $\mathbf m_t-(1+\pi)\mathbf m^\ast$ of the balances, of which $\dot{\mathbf m}_t^\perp$ is the first-order response, and its proportional form $\mathbf D_m^{-1}\mathbf m_t^\perp$ (proof of Proposition~\ref{prop:stability_stationary}) & \textsection\ref{subsec:stability}, \textsection\ref{app:proof_stability} \\

\midrule
\multicolumn{3}{@{}l}{\itshape Deviations, impulses, scale and composition} \\
$\Delta z$, $\widehat z$; $\dot z$, $\ddot z$ & Absolute and proportional deviations; first-order response per unit of shock and curvature response & \textsection\ref{subsec:money} \\
$\mathbf x_{i,t}$, $\mathbf x_i^\ast$; $\Delta\mathbf x_{i,t}$, $\widehat x_{ji,t}$ & Dated input bundle of firm $i$ and its stationary value; deviation of the bundle and proportional deviation of input $(j,i)$ & \textsection\ref{sec:inefficiency} \\
$\nabla_i^\ast$ & Gradient of $\log q_i$ at the stationary bundle & \textsection\ref{sec:inefficiency} \\
$\mathcal F_i(y)=\sum_{j\in\mathcal S_i}a_{ji}y_j$ & Supplier-averaging operator; $\mathcal F_i(\widehat x_{i,t})$ is the input-bundle scale, and a dated argument is evaluated at its order date & \textsection\ref{sec:inefficiency} \\
$\mathcal F_i^{(k)}$ & Within-vintage average over vintage layer $k$ & \textsection\ref{sec:inefficiency} \\
$V_i^{\mathrm w},V_i^{\mathrm b}$ & Within- and between-vintage variance of a supplier-indexed quantity under the supplier shares of firm $i$ & \textsection\ref{sec:inefficiency} \\
$\operatorname{Cov}_{\mathcal F_i}$, $\operatorname{Var}_{\mathcal F_i}$; $\operatorname{Cov}^{\mathrm w}_i,\operatorname{Cov}^{\mathrm b}_i$ & Supplier-share covariance and variance and their within- and between-vintage components & \textsection\ref{app:lemmas_inefficiency} \\
$\langle y,z\rangle_i$; $\langle\cdot,\cdot\rangle_{i,\mathbf k}$ & Curvature-weighted covariance form; $\langle y,y\rangle_i$ is the curvature-weighted dispersion; the form under the vintage partition of the lead-time profile $\mathbf k$ (Corollary~\ref{cor:firm_level_vintage_spread_aggregate_inefficiency}) & \textsection\ref{sec:inefficiency}, \textsection\ref{app:lemmas_inefficiency} \\
$\mathbf z_{i,t}$ & Composition residual of the bundle deviation & \textsection\ref{sec:inefficiency} \\
$\sigma_i$; $\overline\sigma_j$ & Expenditure impulse of firm $i$ \eqref{eq:expenditure_impulse}, the misalignment $\iota_i$ divided by its intermediate-input share; average expenditure impulse of the buyers of good $j$ \eqref{eq:buyer_average_impulse}, the household counting at zero; their dated forms $\sigma_{i,s}=\dot{\widehat e}_{i,s}-1$ and $\overline\sigma_{j,s}=\dot{\widehat d}_{j,s}-1$ (Lemma~\ref{lem:wave_dampening}) & \textsection\ref{subsec:downstreamness} \\
$\xi_i^{(k)}$; $\xi_{\mathcal C}^{(k)}$, $e_{\mathcal C}^{(k)}$ & Relative expenditure impulse of firm $i$ on its inputs of lead time $k$ \eqref{eq:vintage_impulse_def}; that of a set of firms $\mathcal C$ and the set's expenditure on those inputs \eqref{eq:set_impulse} & \textsection\ref{subsec:downstreamness} \\
$\xi_{i,\underline k}$; $\phi_i$, $\phi_i^{(k)}$ & Relative expenditure impulse of firm $i$ on its shortest-lead-time inputs, weighted by their cross-vintage share, $s_i^{(\underline k)}\xi_i^{(\underline k)}$ \eqref{eq:impact_forcing_def}; the scaled relative expenditure impulses $(1-\beta_i)\xi_{i,\underline k}$ and $(1-\beta_i)\xi_i^{(k)}$ \eqref{eq:scaled_impulse_def} & \textsection\ref{subsec:impact_reallocation} \\
$C_\xi$, $C_\phi$ & Bounds on the relative expenditure impulses and on their scaled forms \eqref{eq:forcing_upper_bound}, $2C_\zeta/(1-\overline\beta)$ and $(1-\underline\beta)C_\xi$ & \textsection\ref{subsec:impact_reallocation} \\
$\Xi_{ij,s}=\dot{\widehat e}_{i,s}-\dot{\widehat d}_{j,s}$ & Dated contrast of buyer $i$ on supplier $j$ at order date $s$ & \textsection\ref{subsec:finite_horizon_persistence} \\
$\xi_{i,t}$, $\boldsymbol\xi$; $\mathsf T$; $\mathcal F$ & Dated relative expenditure impulse \eqref{eq:dated_forcing} and its profile; the dated transmission operator \eqref{eq:transmission_operator}; the profile $(\mathcal F_i(\dot{\widehat x}_{i,t}))_{i,t}$ of input-bundle scales (Lemma~\ref{lem:supplier_recursion}) & \textsection\ref{app:proof_lag_shift}, \textsection\ref{app:lemmas_reallocation} \\
$\varpi^{(t,s)}_{ij}(b)$ & Tilt mass carried to link $i\to j$ at order date $s$ in equation \eqref{eq:forced_representation} & \textsection\ref{subsec:finite_horizon_persistence} \\
$F_i(T)$; $F_{\mathcal C}(T)$ & Share of firm $i$'s output that reaches final demand within $T$ periods \eqref{eq:delivery_time}, the distribution of its delivery time to final demand; sales of a set of firms $\mathcal C$ delivered to final demand within $T$ periods, relative to GDP (Lemma~\ref{lem:gdp_pass_through}) & \textsection\ref{subsec:finite_horizon_persistence} \\
$\Gamma_{lk}[T_1,T_2]$; $\Gamma[T_1,T_2]$, $\Gamma_{\mathcal C}[T_1,T_2]$ & Weight of firm $l$'s purchase of lead time $k$ delivered to the household in the interval; the delivered purchases of the interval relative to GDP \eqref{eq:delivered_purchases}, and those of the firms of a set $\mathcal C$ & \textsection\ref{subsec:finite_horizon_persistence} \\
$\widetilde F_l$ & Path-mass solution of the recursion \eqref{eq:delivery_time} (proof of Lemma~\ref{lem:gdp_pass_through}) & \textsection\ref{app:lemmas_reallocation} \\
$T$; $T_0$, $t_1$ & A horizon, the last date of a window $[0,T]$; the dates of Proposition~\ref{prop:implicit_interest} & \textsection\ref{subsec:finite_horizon_persistence}, \textsection\ref{subsec:intertemporal_reallocation} \\
$M$, $Q$; $\langle\cdot,\cdot\rangle_H$; $\mathsf S_\Delta$, $\mathsf S^{*}_\Delta$ & Monetary and real waves; cumulative curvature-weighted form; causal time translation and the advance & \textsection\ref{subsec:lag_profile}, \textsection\ref{app:proof_lag_shift} \\
$C_{EQ}$, $C_{DQ}$ & Cumulative curvature-weighted comovements of the responses of expenditure and of demand with the real wave \eqref{eq:DMR_def} (Lemma~\ref{lem:wave_comovement}) & \textsection\ref{app:wave_comovement} \\
$Q^{(r)}$, $\dot{\widehat q}^{(r)}_{j,s}$; $W^{(r)}$; $\lambda$; $C_\Xi$ & $r$-th round of the real wave and of the response of output \eqref{eq:real_wave_rounds}; the waves counted by round, $W^{(0)}=M$ and $W^{(r)}=Q^{(r)}$ for $r\ge1$; a real eigenvalue of $\mathbf A$ with the impact misalignment as eigenvector, the rate at which the misalignment of balances dies out (Corollary~\ref{cor:lag_shift_inefficiency} and Lemma~\ref{lem:wave_dampening}); bound on the dated contrasts, $|\Xi_{ij,s}|\le C_\Xi\overline\lambda^{\,s}(1-\theta)$ (proof of Corollary~\ref{cor:lag_shift_inefficiency}) & \textsection\ref{subsec:lag_profile}, \textsection\ref{app:proof_lag_shift} \\
$\mathsf B$; $\upsilon$, $\upsilon^{(1)}$; $C_{\mathsf B}$ & Pass-through-weighted average over a firm's suppliers; the profiles $(\mathbf I-\lambda^{k}\mathsf B)^{-1}(\mathbf I-\lambda\mathsf B)\iota$ and $(\mathbf I-\lambda\mathsf B)\iota$; constant of the dispersion condition \eqref{eq:dispersion_condition} (Lemma~\ref{lem:wave_dampening}) & \textsection\ref{app:proof_lag_shift} \\
$\varsigma^{\mathbf k}_{i,t}$, $u^{\mathbf k}_{i,t}$ & Regression slope of the response on the lead time and its residual & \textsection\ref{subsec:lag_profile} \\

\midrule
\multicolumn{3}{@{}l}{\itshape Losses, tilted output and reallocation} \\
$\mathcal D_{i,t}(\pi)$, $\mathcal D_t^{(\boldsymbol\omega)}(\pi)$, $\mathcal D^{(\boldsymbol\omega)}(\pi)$ & Loss from production inefficiency of a firm, loss from production inefficiency at a date and cumulative loss from production inefficiency at the weighting $\boldsymbol\omega$ & \textsection\ref{sec:inefficiency} \\
$C^{(\boldsymbol\omega)}$ & Leading coefficient of the cumulative loss from production inefficiency & \textsection\ref{subsec:cumulative_loss} \\
$V^{\mathrm w},V^{\mathrm b}$ & Cumulative within- and between-vintage dispersions of the first-order responses & \textsection\ref{app:lemmas_inefficiency} \\
$C^{(\boldsymbol\omega)}(\mathbf k)$ & Leading coefficient of the cumulative loss from production inefficiency under the lead-time profile $\mathbf k$ & \textsection\ref{subsec:lag_profile} \\
$\mathcal D_{\mathrm M}^{(\boldsymbol\omega)},\mathcal D_{\mathrm I}^{(\boldsymbol\omega)},\mathcal D_{\mathrm C}^{(\boldsymbol\omega)}$ & Nominal-expenditure, input-price and covariance components (Lemma~\ref{lem:monetary_input_decomposition}) & \textsection\ref{app:lemmas_inefficiency} \\
$b$; $\widehat\psi_i$; $\omega_i(b)$ & Tilt on $[0,\infty]$; capped position of firm $i$, its downstreamness index relative to that of the least downstream retailer, capped at one; tilt weights \eqref{eq:w_b_def}, the sales shares at $b=0$ and the sales shares within the retail tier at $b=\infty$ & \textsection\ref{subsec:impact_reallocation} \\
$\omega_{\mathcal R}(b)$, $b_w$; $M_{\mathcal R}$, $a_m$ & Retail-tier mass of the tilt weights and the first downstream-concentrated tilt; sales of the retail tier and the sales off it relative to those on it (proof of Proposition~\ref{prop:impact_reallocation}) & \textsection\ref{subsec:impact_reallocation}, \textsection\ref{app:proof_impact_reallocation} \\
$\mathcal Q_t(b)$, $\Delta\mathcal Q_t(b)$, $D_\pi\mathcal Q_t(b)$ & Position-tilted log-output index, its deviation and first-order response & \textsection\ref{subsec:impact_reallocation}, \textsection\ref{app:lemmas_asymmetry} \\

$c_\xi$; $c_R$ & Impulse margin of Assumption~\ref{assump:impact_forcing_ordering} and of condition (U) on the top tier, per unit of slack; retail margin $s_R(1-\overline\beta)c_\xi$ of the retail impulse bound \eqref{eq:retail_impulse_bound} & \textsection\ref{subsec:downstreamness}, \textsection\ref{subsec:impact_reallocation} \\
$\overline\ell_j$; $\Delta\ell_i^{(k)}$ & Average centered log size of the buyers of good $j$ and log-size gap of firm $i$ on its inputs of lead time $k$ \eqref{eq:buyer_average_size} & \textsection\ref{app:lemmas_reallocation} \\
$\sigma^{\circ}_l$, $\xi_i^{(k)\circ}$; $C_\ell$, $\Delta_\beta$ & Incidence exposure of firm $l$ and exposure advantage of firm $i$ on its inputs of lead time $k$ \eqref{eq:exposure_advantage}; remainder constant and labor-share correction \eqref{eq:exposure_constant} of Lemma~\ref{lem:impulse_exposure}, which gives the impulse to leading order in the slack & \textsection\ref{app:lemmas_reallocation} \\
$C_0$; $\varepsilon_l$, $\varepsilon_{\max}$, $\overline\varepsilon$; $\Delta_{\beta,l}$ & Remainder constant of the expenditure impulses; remainders of the exponential at firm $l$, their bound and their stationary average; gap $1/(1-\overline\beta)-1/(1-\beta_l)$ at firm $l$ (proof of Lemma~\ref{lem:impulse_exposure}) & \textsection\ref{app:lemmas_reallocation} \\
$\overline\epsilon_w$; $\overline\epsilon_\xi$; $\overline\epsilon_F,\overline\epsilon_h$ & Threshold \eqref{eq:downstream_threshold} of the downstream-concentrated tilts; tolerance of condition (N) on the neighborhood of the retail tier and thresholds of condition (L), both in \eqref{eq:persistence_thresholds}, $\overline\epsilon_F$ for the purchases of the firms beyond the neighborhood of final demand delivered to the household by the last retail lead-time date and $\overline\epsilon_h$ for the retail slack & \textsection\ref{subsec:impact_reallocation}, \textsection\ref{subsec:finite_horizon_persistence}, \textsection\ref{app:reallocation} \\
$\mathcal P'$, $k(\mathcal P')$, $\mu(\mathcal P')$ & Supplier path, its summed lead time and its tilt mass (proof of Proposition~\ref{prop:finite_horizon_sign_persistence}) & \textsection\ref{app:proof_finite_horizon_sign_persistence} \\
$\mathcal P$, $k(\mathcal P)$; $\mathcal P''$ & Supplier path ending in a link from a buyer to a supplier, and its summed lead time (Lemma~\ref{lem:supplier_recursion}); a path one link shorter than $\mathcal P'$ (proof of Lemma~\ref{lem:gdp_pass_through}) & \textsection\ref{app:proof_finite_horizon_sign_persistence}, \textsection\ref{app:lemmas_reallocation} \\
$\mu^{(r)}(l)$ & Mass of the $r$-link paths from the tilt weights ending at buyer $l$ (proof of Proposition~\ref{prop:finite_horizon_sign_persistence}) & \textsection\ref{app:proof_finite_horizon_sign_persistence} \\
$\overline M(b)=\sum_jm_j^\ast\widehat\psi_j^{\,b}$ & Denominator of the tilt weights, with $\overline M(0)=\overline M$ (proof of Proposition~\ref{prop:finite_horizon_sign_persistence}) & \textsection\ref{app:proof_finite_horizon_sign_persistence} \\
$I_t(b)$, $J_t(b)$ & Impact-date and post-impact parts of $D_\pi\mathcal Q_t(b)$ & \textsection\ref{subsec:finite_horizon_persistence} \\
(N), (L); (U) & Tolerance on the impulses of the neighborhood of the retail tier; smallness of the purchases of the firms beyond the neighborhood of final demand delivered to the household by the last retail lead-time date, and of the retail slack, the thresholds \eqref{eq:persistence_thresholds}; margin on the impulses of the top tier & \textsection\ref{subsec:finite_horizon_persistence}, \textsection\ref{subsec:sign_reversing_irf} \\
(P), (R) & Hypotheses of Proposition~\ref{prop:finite_horizon_sign_persistence} and Theorem~\ref{thm:sign_reversing_irf} & \textsection\ref{subsec:finite_horizon_persistence}, \textsection\ref{subsec:sign_reversing_irf} \\
$I_t^{\mathcal C}(b)$, $\mu_t(\mathcal C)$; $-I_t^{\mathcal T}(b)$ & Part of $I_t(b)$ carried by the impact-date contrasts at buyers in a set of firms $\mathcal C$, and the tilt mass that carries it, scaled by the buyers' intermediate-input shares \eqref{eq:arriving_contraction}; the arriving contraction & \textsection\ref{subsec:sign_reversing_irf} \\
$r_a$, $\overline\tau$, $\underline\mu$, $\tau^-$ & Rounds of sourcing within which the supplier walk from the retail tier reaches the contracting region with probability at least one half, horizon of the arrival window, arrival mass, reversal date (Theorem~\ref{thm:sign_reversing_irf}, \eqref{eq:reversal_constants}, Lemma~\ref{lem:layered_routing}) & \textsection\ref{subsec:sign_reversing_irf}, \textsection\ref{app:proof_sign_reversing_irf} \\
$\mathsf X$, $r_{\mathcal T}$, $\mathsf M_r$; $\mathbb E$, $\Pr$ & Supplier walk started from the tilt weights, its first entry into the top tier and the supermartingale of the proof of Lemma~\ref{lem:layered_routing}; expectation and probability under the walk & \textsection\ref{app:lemmas_reallocation} \\

\midrule
\multicolumn{3}{@{}l}{\itshape GDP response} \\
$Y_t$; $f_t(\pi)=\Delta\log Y_t(\pi)$ & GDP, the household consumption index, and its log response & \textsection\ref{sec:asymmetry} \\
$P_t$, $\Pi_t$; $E_t$, $H_t$; $\Lambda_t$ & Value of the goods in transit relative to the stationary pipeline, and its drawdown; value of the inputs used in excess of the stationary bundles less the value of the output made with them, and the consumption loss; deadweight loss \eqref{eq:resource_identity} & \textsection\ref{sec:asymmetry}, \textsection\ref{app:lemmas_asymmetry} \\
$E_{i,t}$; $E_{i,t}^{\mathrm{mix}}$, $E_{i,t}^{\mathrm{scale}}$; $E_t^{\mathrm{mix}}$, $E_t^{\mathrm{scale}}$ & Contribution of firm $i$ to $E_t$, its terms \eqref{eq:production_loss_split} and their sums; $E_t^{\mathrm{mix}}/w^\ast$ is the loss from production inefficiency and $(E_t^{\mathrm{scale}}+H_t)/w^\ast$ the loss from misallocation & \textsection\ref{app:lemmas_asymmetry} \\
$R_t^{(Y)}$, $\mathcal D_t^{(\boldsymbol\gamma)}$, $\Theta_t$ & Input-scale aggregate, loss from production inefficiency of the retailers at the household shares, consumption--output wedge & \textsection\ref{subsec:asymmetric_nonneutrality}, \textsection\ref{app:lemmas_asymmetry} \\
$\dot f_t$, $\ddot f_t$; $A_t,B_t,C_t,V_t$; $\dot\Theta_t,\ddot\Theta_t$ & First and second derivatives of $f_t$ at zero; the reallocation impulse, the first-order response of the retailers' input-bundle scales at the household shares, in \eqref{eq:gdp_expansion_coefficients}, and the curvature of the input-scale aggregate, the curvature of the retailers' log output in the scale of their bundles and their loss from production inefficiency per unit of squared shock, in \eqref{eq:gdp_second_derivative}; derivatives of the consumption--output wedge & \textsection\ref{subsec:asymmetric_nonneutrality}, \textsection\ref{app:lemmas_asymmetry} \\
$\tilde\pi=\pi/(1+\pi)$, $\mathcal Y_t$, $\mathcal Y^{\Pi}_t$, $\mathcal Y^{\Lambda}_t$, $\mathcal Y^{P}_t$, $\mathcal Y^{R}_t$; $K$ & Normalized shock; normalized responses of GDP (Lemma~\ref{lem:concave_scaling}), of the drawdown of the goods in transit, of the deadweight loss and of the excess value of the goods in transit (Lemma~\ref{lem:resource_accounting_date}), and of the input-scale aggregate (Lemma~\ref{lem:concave_scaling}); bound on the Hessian of $\mathcal Y_t$ on the conservation subspace $\mathcal Z_{\mathbf A}$ in the norm $\|\cdot\|_m$ (Lemma~\ref{lem:concave_scaling}) & \textsection\ref{subsec:asymmetric_nonneutrality}, \textsection\ref{app:lemmas_asymmetry} \\
$c$; $K_3$, $\widetilde\varepsilon_t$; $\varkappa$ & A margin of a positive first-order response, $\dot f_t\ge c(1-\theta)$; bound on the third derivative of $\mathcal Y_t$ and the remainder of the second-order expansion of $f_t$ in the normalized shock; common rescaling of balances and wage (Lemma~\ref{lem:concave_scaling} and its proof) & \textsection\ref{subsec:asymmetric_nonneutrality}, \textsection\ref{app:lemmas_asymmetry} \\
$\overline f_t$ & Even part of the GDP response, decomposed in \eqref{eq:even_part} & \textsection\ref{subsec:asymmetric_nonneutrality} \\
$c_A$; $c_{\dot f}$, $c_{\ddot f}$, $\varepsilon_t$ & Lower bound on the reallocation impulse on $W$ (Lemma~\ref{lem:gdp_expansion}); smallest first derivative and smallest absolute second derivative on $W$, and the remainder of the second-order expansion of $f_t$ in the shock (proof of Lemma~\ref{lem:concave_response_asymmetry}) & \textsection\ref{app:lemmas_asymmetry} \\
$f_+(\pi)$, $\Lambda_+(\pi)$; $c_Y$; $K_+$, $K_P$ & GDP over the first reporting interval, the dates from the shock to the shortest lead time, and the deadweight loss averaged over it; GDP margin \eqref{eq:gdp_margin} (Lemma~\ref{lem:first_interval_gdp}); bounds on the Hessians of the average of the maps $\mathcal Y_t$ over the first reporting interval and of the excess value of the goods in transit at its end, divided by the number of its dates and by $w^\ast$ (Theorem~\ref{thm:asymmetric_nonneutrality}) & \textsection\ref{subsec:asymmetric_nonneutrality}, \textsection\ref{app:lemmas_asymmetry} \\
$I_t(\boldsymbol\gamma)$, $J_t(\boldsymbol\gamma)$; $\mu_{\boldsymbol\gamma}(\mathcal P')$, $\mu_{\boldsymbol\gamma}(l;L)$ & Impact-date and post-impact parts of the reallocation impulse $A_t$ at the household shares; mass of a supplier path from the household shares, and the mass of those ending at buyer $l$ with summed lead time at most $L$, equal to $(m_l^\ast/w^\ast)F_l(L)$ (Lemma~\ref{lem:gdp_pass_through}) & \textsection\ref{app:lemmas_reallocation} \\
$\Pi^\sharp_t$, $P^\sharp_t$; $\dot P_t$; $\|D\mathcal Y(\mathbf 0)\|_m$, $\|D^2\mathcal Y(\mathbf 0)\|_m$ & Parts of the drawdown of the goods in transit and of their excess value beyond first order in the normalized shock (Lemma~\ref{lem:gap_exact}); first derivative of the goods in transit at $\pi=0$ (Lemma~\ref{lem:resource_accounting_date}); norms of the derivative and of the Hessian of a normalized response on $\mathcal Z_{\mathbf A}$ in the proportional coordinates (Lemmas~\ref{lem:concave_scaling} and~\ref{lem:gap_exact}) & \textsection\ref{app:lemmas_asymmetry} \\
$\Lambda^{\mathrm c}_t$, $\Lambda^{\mathrm{tr}}_t$; $r_{i,t}$, $\dot{\widehat m}^\perp_{i,t+1}$ & Loss on consumed output and loss in transit, the deadweight loss on the output the household consumes and on the output that goes into transit (Lemma~\ref{lem:consumed_transit}); consumption share of firm $i$, the share of its output bought by the household, and the first-order proportional misalignment of its next-period balance (Lemma~\ref{lem:gdp_labor_income_anchor}) & \textsection\ref{subsec:asymmetric_nonneutrality}, \textsection\ref{app:lemmas_asymmetry} \\
$\dot{\widehat x}_i$, $\dot{\widehat m}^\perp_i$, $\dot{\widehat c}_i$; $\ddot\Lambda^{\mathrm c}_{\underline k}$, $\ddot\Lambda^{\mathrm{tr}}_{\underline k}$; $y_{ji}$, $\widehat x_i$, $\widehat c_i$ & At the shortest lead time: the first-order deviations $\dot{\widehat x}_{i,\underline k}$ of the dated inputs of firm $i$, the first-order misalignment $\dot{\widehat m}^\perp_{i,\underline k+1}$ of the next balance of retailer $i$ and the first-order response of the household's consumption of good $i$; the second derivatives there of the loss on consumed output and of the loss in transit (Lemma~\ref{lem:shortest_lag_curvature}); the ratios of the dated inputs as functions of $\tilde\pi$ and the deviations $\widehat x_{i,\underline k}$, $\widehat c_{i,\underline k}$ (its proof) & \textsection\ref{app:lemmas_asymmetry} \\
$C_c$ & Prefactor of the consumption envelope (proof of Proposition~\ref{prop:resource_accounting}) & \textsection\ref{app:proof_resource_accounting} \\
$G_Y(\pi)$, $L_Y(\pi)$, $\mathcal I(\pi)$ & Cumulative gain, cumulative loss and implicit interest rate of the monetary expansion (Proposition~\ref{prop:implicit_interest}) & \textsection\ref{subsec:intertemporal_reallocation} \\
$G_0$; $\mathcal E_i$, $C_{\mathcal E}$, $C_\Lambda$ & GDP responses accumulated up to $T_0$; contribution of firm $i$ to $E_t$ per unit of its sales as a function of the proportional changes of its dated inputs, its quadratic bound, and the bound $\sum_t(E_t+H_t)\le C_\Lambda\pi^2$ (proof of Proposition~\ref{prop:implicit_interest}) & \textsection\ref{app:proof_implicit_interest} \\

\midrule
\multicolumn{3}{@{}l}{\itshape Operators} \\
$\mathbf 1$, $\mathbf I$; $\ind_{\{\cdot\}}$ & All-ones vector, identity matrix; indicator & \textsection\ref{subsec:environment} \\
$D_\pi$; $D$, $D^2$; $O(\cdot)$, $o(\cdot)$ & Derivative in the shock; derivative and Hessian of a map at a point, $D^2\mathcal Y(\mathbf 0)[\mathbf v,\mathbf v']$ the Hessian as a bilinear form; asymptotic notation as $\pi\to0$ at a fixed economy & \textsection\ref{subsec:impact_reallocation}, \textsection\ref{subsec:asymmetric_nonneutrality}, \textsection\ref{subsec:money} \\

\midrule
\multicolumn{3}{@{}l}{\itshape Simulation and reported objects} \\
$k_{\mathrm{lo}}$, $k_{\mathrm{hi}}$; $\epsilon_{\downarrow}$; $\underline k_{\mathrm s},\underline n_{\mathrm s}$ & Bounds of the lead-time window of the most upstream buyers; fraction by which the buyer's downstreamness index compresses the window; lead time from which a link counts as slow and the base number of slow links a buyer keeps (Table~\ref{tab:baseline}) & \textsection\ref{app:calibration} \\
$M_{\mathrm h}$; $\theta_{\mathrm h}$; $\epsilon_{\mathrm{cap}}$ & Household balance; household incidence exponent; transfer cap, the largest fraction of its balance by which the shock moves a firm's balance (Table~\ref{tab:baseline}) & \textsection\ref{sec:quant-setup}, \textsection\ref{app:calibration} \\

\end{longtable}\addtocounter{table}{-1}}

\newpage
\subsection{Dependency map}
\label{app:depgraph}

\begin{figure}[H]
\centering
\includegraphics[width=\textwidth,height=0.78\textheight,keepaspectratio]{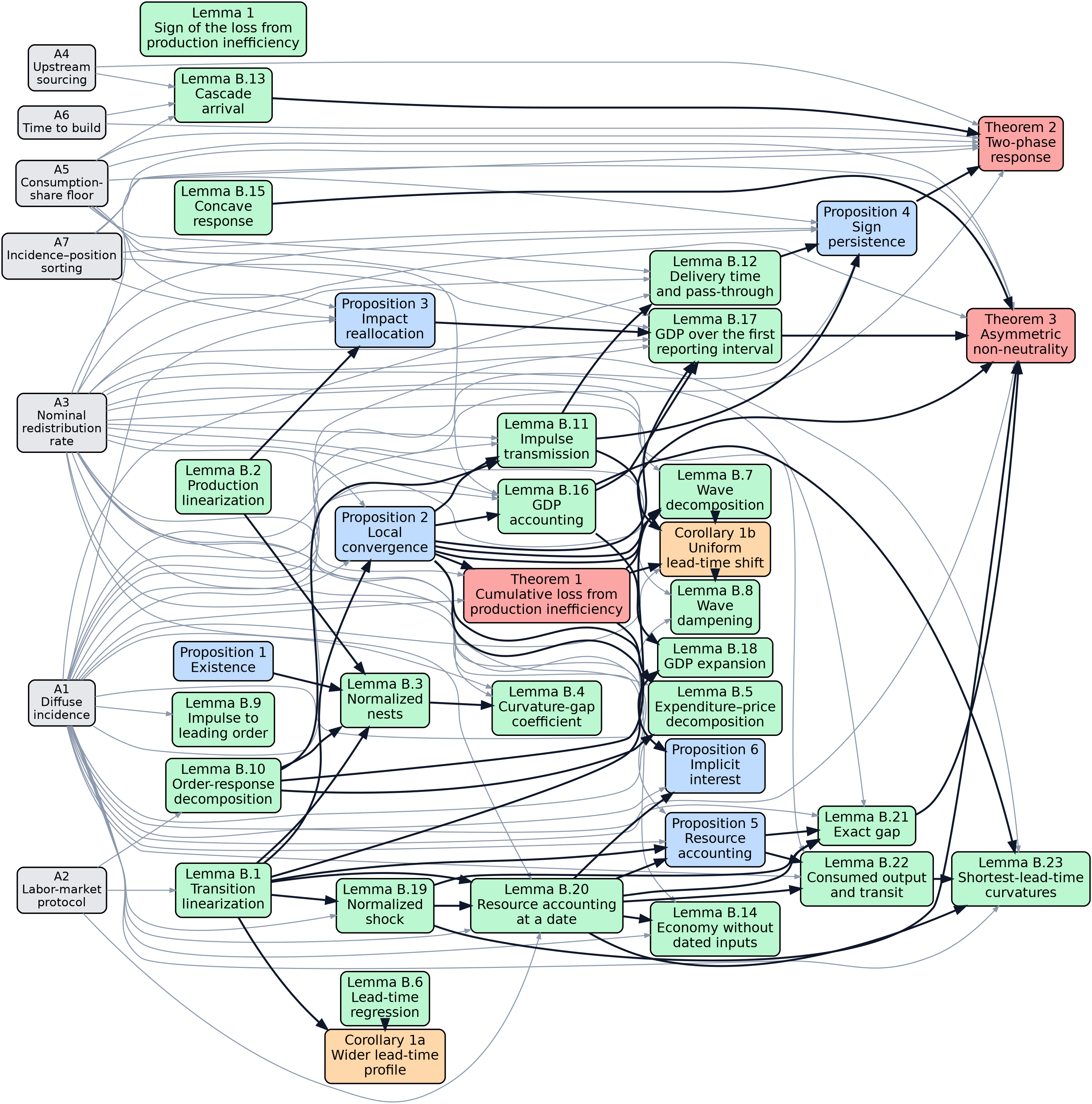}

\vspace{0.35em}
{\footnotesize
\fcolorbox{black!30}{depAssumption}{\strut\,Assumption\,}\;
\fcolorbox{black!30}{depProposition}{\strut\,Proposition\,}\;
\fcolorbox{black!30}{depLemma}{\strut\,Lemma\,}\;
\fcolorbox{black!30}{depTheorem}{\strut\,Theorem\,}\;
\fcolorbox{black!30}{depCorollary}{\strut\,Corollary\,}
}

\caption{Dependency map of the formal results of the paper and of
Appendix~\ref{app:analytical}. An arrow from \(X\) to \(Y\) means that the
statement or proof of \(Y\) invokes \(X\). Results are ordered left to
right by dependency depth.
Assumption~\ref{assump:short_run_nominal_rigidity}, maintained throughout,
receives an arrow only where a proof invokes it by name.}
\label{fig:depgraph_hierarchical}
\end{figure}
\clearpage

\end{document}